\documentclass[a4paper,oneside,fleqn,11pt]{article}
\usepackage{amsmath,amssymb,amsfonts,amsthm,type1cm,bm,xcolor,mathrsfs}
\usepackage{mathtools}
\usepackage{graphicx,psfrag,epsf}
\usepackage{rotating}
\usepackage{authblk}
\usepackage{yfonts}
\usepackage{pdflscape}
\usepackage{setspace}
\usepackage{natbib}
\usepackage{float} 
\usepackage{caption}
\usepackage{subcaption}
\RequirePackage[colorlinks,citecolor=blue,urlcolor=pink]{hyperref}
\usepackage{bigstrut}
\usepackage{comment}
\mathtoolsset{showonlyrefs=true}

\DeclareTextFontCommand{\texttt}{\ttfamily\upshape}

\DeclareMathOperator\diag{diag}

\DeclareMathOperator\E{\mathbb{E}}

\DeclareMathOperator\F{F}

\DeclareMathOperator*{\plim}{plim}

\DeclareMathAlphabet\mathbfcal{OMS}{cmsy}{b}{n}
\def\b0{\mathbf{0}}
\def\b1{\mathbf{1}}

\def\bA{\mathbf{A}}
\def\bB{\mathbf{B}}
\def\bC{\mathbf{C}}
\def\bD{\mathbf{D}}
\def\bE{\mathbf{E}}
\def\bF{\mathbf{F}}
\def\bG{\mathbf{G}}
\def\bH{\mathbf{H}}
\def\bI{\mathbf{I}}

\def\bM{\mathbf{M}}
\def\bN{\mathbf{N}}
\def\bP{\mathbf{P}}
\def\bQ{\mathbf{Q}}
\def\bR{\mathbf{R}}

\def\bV{\mathbf{V}}
\def\bW{\mathbf{W}}
\def\bX{\mathbf{X}}

\def\bZ{\mathbf{Z}}

\def\bb{\mathbf{b}}

\def\be{\mathbf{e}}
\def\bff{\mathbf{f}}

\def\bh{\mathbf{h}}

\def\bu{\mathbf{u}}

\def\bw{\mathbf{w}}

\def\by{\mathbf{y}}
\def\bz{\mathbf{z}}

\def\bGamma{\boldsymbol{\Gamma}}
\def\bgamma{\boldsymbol{\gamma}}
\def\bkappa{\boldsymbol{\kappa}}
\def\balpha{\boldsymbol{\alpha}}
\def\bbeta{\boldsymbol{\beta}}

\def\bdelta{\boldsymbol{\delta}}
\def\bLambda{\boldsymbol{\Lambda}}

\def\bPhi{\boldsymbol{\Phi}}

\def\bSigma{\boldsymbol{\Sigma}}
\def\bGamma{\boldsymbol{\Gamma}}

\def\bepsilon{\boldsymbol{\epsilon}}

\def\bzero{\mathbf{0}}
\def\bone{\mathbf{1}}

\newtheorem{thm}{Theorem}
\newtheorem{lem}{Lemma}
\newtheorem{cor}{Corollary}
\newtheorem{ass}{Assumption}

\newcommand{\CD}{\stackrel{d}{\longrightarrow}}
\newcommand{\CP}{\stackrel{p}{\longrightarrow}}

\makeatletter
\def\section{\@startsection {section}{1}{\z@}{-3.5ex plus -1ex minus-.2ex}{2.3ex plus .2ex}{\large\bf}}
\makeatother

\makeatletter
\def\subsection{\@startsection {subsection}{1}{\z@}{-3.5ex plus -1ex minus-.2ex}{2.3ex plus .2ex}{\normalsize\bf}}
\makeatother

\title{\textbf{Bias Correction in Factor-Augmented Regression Models with Weak Factors}}

\author[$\,\!$]{\textsc{Peiyun Jiang}$^*$}

\author[$\,\!$]{\textsc{Yoshimasa Uematsu}$^\dagger$}
\author[$\,\!$]{\textsc{Takashi Yamagata}$^\ddagger$}
	\affil[$*$]{\textit{Faculty of Economics and Business Administration,
	Tokyo Metropolitan University
	}}
	\affil[$\dagger$]{\textit{Department of Social Data Science, Hitotsubashi University}}
	\affil[$\ddagger$]{\textit{Department of Economics and Related Studies, University of York}}
	\affil[$\ddagger$]{\textit{Institute of Social and Economic Research, Osaka University}}
	
	\date{December 26, 2024}
	
\begin{document}



%
%
%

\maketitle

\begin{abstract}
	In this paper, we study the asymptotic bias of the factor-augmented regression estimator and its reduction, which is augmented by the $r$ factors extracted from a large number of $N$ variables with $T$ observations. In particular, we consider general weak latent factor models with $r$ signal eigenvalues that may diverge at different rates, $N^{\alpha _{k}}$, $0<\alpha _{k}\leq 1$, $k=1,\dots,r$. 
	In the existing literature, the bias has been derived using an approximation for the estimated factors with a specific data-dependent rotation matrix $\hat{\bH}$ for the model with $\alpha_{k}=1$ for all $k$, whereas we derive the bias for weak factor models.
	In addition, we derive the bias using the approximation with a different rotation matrix $\hat{\bH}_q$, which generally has a smaller bias than with $\hat{\bH}$. 
	We also derive the bias using our preferred approximation with a purely signal-dependent rotation ${\bH}$, which is unique and can be regarded as the population version of $\hat{\bH}$ and $\hat{\bH}_q$. Since this bias is parametrically inestimable, we propose a split-panel jackknife bias correction, and theory shows that it successfully reduces the bias.
	The extensive finite-sample experiments suggest that the proposed bias correction works very well, and the empirical application illustrates its usefulness in practice.

\end{abstract}
\textbf{Keywords.} Factor model, Asymptotic bias, Jackknife, Cross-sectional dependence, Weak factors.{\let\thefootnote\relax\footnote{$^*$Corresponding author: Email: jiang-peiyun@tmu.ac.jp; Address: Faculty of Economics and Business Administration, Tokyo Metropolitan University, 1-1 Minami-Osawa, Hachioji-shi, Tokyo, Japan 192-0397}}


\section{Introduction}

Factor-augmented regressions are widely used in financial and economic research. They are often used to forecast macroeconomic and financial time series. The forecast regression is augmented with a few common factors extracted from a large set of predictors. Specifically, the $h$-ahead forecast regression of $y_t$ is written as
\begin{align}\label{Augmodel}
	y_{t+h}={\bgamma^*}'\bff_t^*+\bbeta' \bw_t+\epsilon_{t+h}, \quad t=1,\dots,T,
\end{align}
where $\bff_t^*$ is an $r \times 1$ vector of latent predictive factors, $\bw_t$ is a $p \times 1$ vector of observable predictors, $(\bgamma^{*\prime},\bbeta')'$ is an $(r+p)\times 1$ vector of their coefficients, and $\epsilon_{t+h}$ is an error term. 
The latent $r$ factors drive a large number of $N$ predictors: 
\begin{align}\label{factormodel}
	x_{t,i}=\bb_{i}^{*\prime} \bff_t^* + e_{t,i},  
	\quad t=1,\dots,T,\quad i=1,\dots,N,
\end{align}
where $\bb_{i}^{*}$ is an $r\times 1$ vector of unknown factor loadings and $e_{t,i}$ is an error term. 

Since $\bff_t^*$ is unobserved, it is typically replaced by the principal component (PC) estimator, $\hat{\bff}_t$ such that $T^{-1}\sum_{t=1}^T\hat{\bff}_t\hat{\bff}_t^{\prime}=\bI_r$, obtained as $\sqrt{T}$ times the $r$ eigenvectors corresponding to the $r$ largest eigenvalues ($\hat{\lambda}_{1}>\dots>\hat{\lambda}_{r}$) of the $T \times T$ sample covariance matrix of $N$ predictors, $(x_{t,i})$. 
Let $(\hat{\bgamma}',\hat{\bbeta}')'$ be the least squares (LS) estimators of the regression of $y_{t+h}$ on $(\hat{\bff}_t^{\prime}, \bw_t^{\prime})^{\prime}$.
\cite{StockWatson2002JASA}, \cite{BaiNg2006} and \cite{GoncalvesPerron2014,gonccalves2020bootstrapping} employ the asymptotic approximation of the PC factor by rotated latent factors with a data dependent (but infeasible) rotation matrix $\hat{\bH}$:
\begin{align}\label{"consis"}
	\hat{\bff}_t=\hat{\bH}'\bff_t^* +o_p(1),
\end{align}
where $\hat{\bH}=  \sum_{i=1}^N{\bb}_{i}^{*}{\bb}_i^{*\prime}T^{-1}\sum_{t=1}^T{\bff}_t^{*}\hat{\bff}_t^{\prime} \hat{\bLambda}^{-1}$ with $\hat{\bLambda}=\diag{(\hat{\lambda}_{1},\dots,\hat{\lambda}_{r})}$. Using \eqref{"consis"}, we approximate the first term on the right-hand side of \eqref{Augmodel} as $ \bgamma^{*\prime}\bff_t^*=\bgamma_{\hat{\bH}}^{\prime}\hat{\bff}_t + o_p(1)$, 
where $\bgamma_{\hat{\bH}}:=\hat{\bH}^{-1}\bgamma^{*}$ is what $\hat{\bgamma}$ estimates. 
\cite{BaiNg2006} show that so long as $\sqrt{T}/N \to 0$, 
the limiting distribution of $\sqrt{T}(\hat{\bgamma}-\bgamma_{\hat{\bH}})$ is centered at zero (i.e., no asymptotic bias).

When $N$ is not sufficiently large relative to $T$ to satisfy $\sqrt{T}/N \to c\in(0,\infty)$,  
\cite{Ludvigson2011} show that $\sqrt{T}(\hat{\bgamma}-\bgamma_{\hat{\bH}})$ has an asymptotic bias and provide an analytical bias correction. \cite{GoncalvesPerron2014} refine the asymptotic bias expression and propose an analytical bias correction as well as a wild bootstrap method to remove the asymptotic bias. \cite{gonccalves2020bootstrapping} extend \cite{GoncalvesPerron2014} to allow the bias corrections for errors $e_{t,i}$ to be cross-correlated using the large covariance matrix estimator proposed by \cite{BickelLevina2008}. 

The existing literature derives the asymptotic bias by assuming that all the $r$ largest eigenvalues of the sample covariance matrix of $x_{t,i}$, $(\hat{\lambda}_{1},\dots,\hat{\lambda}_{r})$, diverge proportionally to $N$ for sufficiently large samples. This is known as the \textit{strong factor} (SF) model. To the best of our knowledge, this paper is the first to study the derivation of asymptotic bias for the more general \textit{weak factor} (WF) models, in which each $\hat{\lambda}_{k}$ can diverge at a different rate $N^{\alpha_k}$, with $\alpha_{1}\geq\dots \geq{\alpha}_{r}$, $\alpha_k \in (0,1]$, $k=1,\dots,r$ for sufficiently large samples.
A growing body of literature suggests that such WFs are widely observed in real data; see \citet{BaileyEtAl2016,BaileyEtAl2021}, \cite{DeMol2008}, \cite{Freyaldenhoven21JoE}, \cite{Onatski2010}, \cite{UY2019,UY2019inference}, \cite{WeiZhang2023}, among many others.
In particular, we show that if $\sqrt{T}/N^{(3\alpha_r - \alpha_1)/2} \to c_1\in(0,\infty)$, $\sqrt{T}(\hat{\bgamma}-\bgamma_{\hat{\bH}})$ has an asymptotic bias.
Using the derived asymptotic bias expression, we propose an analytical bias-corrected estimator, called the ``plug-in estimator'' in \cite{GoncalvesPerron2014,gonccalves2020bootstrapping}, using the POET estimator proposed by \cite{FanEtAl2013} and extended for WF models by \cite{RunyuEtAl2024}, allowing for cross-sectional and serial correlations in $e_{t,i}$. 

As introduced in \cite{BaiNg2023} and \cite{jiang2023revisiting}, there are variants of asymptotically equivalent data-dependent rotations other than $\hat{\bH}$. Among these, we consider $\hat{\bH}_q=(T^{-1}\sum_{t=1}^T \hat{\bff}_t{\bff}_t^{*\prime})^{-1}$, and show that if $\sqrt{T}/N^{\alpha_r}\to c_2\in(0,\infty)$, $\sqrt{T}(\hat{\bgamma}-\bgamma_{\hat{\bH}_q})$ will have an asymptotic bias, where $\bgamma_{\hat{\bH}_q}=\hat{\bH}_q^{-1}\bgamma^{*}$. Note that the rate $\sqrt{T}/N^{\alpha_r}$ is not slower than $\sqrt{T}/N^{(3\alpha_r - \alpha_1)/2}$. 
Furthermore, it is found that the bias of $\sqrt{T}(\hat{\bgamma}-\bgamma_{\hat{\bH}_q})$ is generally smaller in magnitude than that of $\sqrt{T}(\hat{\bgamma}-\bgamma_{\hat{\bH}})$. We propose an analytical bias correction that makes the bias virtually zero even with small samples. In addition, we show that the asymptotic bias of $\sqrt{T}(\hat{\bgamma}-\bgamma_{\hat{\bH}_q})$ becomes \textit{exactly} zero when $\bff_t^*$ and $\bw_t$ are uncorrelated. To exploit this property, we also suggest projecting out $\bw_t$ from $x_{t,i}$ and then extracting the factors, say $\hat{\bff}_w$ to augment the regression. In fact, this is often done in practice; see \cite{YamamotoHara2022}, for example.

We have discussed the asymptotic biases of $\hat{\bgamma}$ relative to $\bgamma_{\hat{\bH}}$ and $\bgamma_{\hat{\bH}_q}$, and a natural question that arises here is: ``Relative to which rotation matrix should the asymptotic bias of $\hat{\bgamma}$ be evaluated?'' Actually, for inferential purposes in particular, our preferred choice is neither ${\hat{\bH}}$ nor ${\hat{\bH}_q}$; instead, we argue for choosing $\bH$ composed only of the signals $\bff_t^*$ and $\bb_i^*$ such that  $T^{-1}\sum_{t=1}^T{\bff}_t^0 {\bff}_t^{0\prime}=\bI_r$, 
where 
\begin{align}\label{f0}
	\bff_t^0:=\bH'\bff_t^*.
\end{align}
\cite{jiang2023revisiting} have shown that up to sign the PC estimator, $\hat{\bff}_t$, is \textit{consistent} to $\bff_t^0$ and $\bH$ \textit{always} exists and is unique. Therefore, $\bH$ can be considered as the population version of the data-dependent rotation matrices, $\hat{\bH}$ and $\hat{\bH}_q$. Substituting \eqref{f0} into \eqref{Augmodel}, we can rewrite the model as
\begin{align}\label{consis}
	y_{t+h}={\bgamma^0}'\bff_t^0+\bbeta' \bw_t+\epsilon_{t+h},
\end{align}
where $\bgamma^0= {\bH}^{-1}\bgamma^{*}$. Thus, a regression of $y_{t+h}$ on $(\hat{\bff}_t,\bw_t)$ consistently estimates the parameters $(\bgamma^{0\prime},\bbeta')'$. We prefer to consider the asymptotic bias of $\sqrt{T}(\hat{\bgamma}-\bgamma^0)$. Note that $\bgamma^0$ is a function of the signal parameters $\{\bgamma^*,\bff_{t}^*,\bb_{i}^*\}$,  whereas $\bgamma_{\hat{\bH}}$ and $\bgamma_{\hat{\bH}_q}$ are functions of data, $(x_{t,i})$.
We show that if $\sqrt{T}/N^{(3\alpha_r - \alpha_1)/2} \to c_1\in(0,\infty)$, $\sqrt{T}(\hat{\bgamma}-\bgamma^0)$ has an asymptotic bias, 
but it seems difficult to obtain the explicit analytical expression. 
Therefore, we propose the \textit{split-panel Jackknife bias-correction}, and show that it effectively reduces the bias. Our Jackknife bias-correction is generally less computationally expensive than bootstrapping, while allowing for cross-sectional and serial correlations in $e_{t,i}$.

The finite sample behavior in terms of bias, standard deviation and tail distribution of the estimators $(\hat{\bgamma},\hat{\bbeta})$ and their bias-corrected versions relative to ($\bgamma_{\hat{\bH}}$, $\bgamma_{\hat{\bH}_q}$, $\bgamma^0$) and $\bbeta$, are investigated for SF and WF models with cross-sectional and serially correlated errors, $e_{t,i}$. Throughout the design, the analytically bias-corrected estimator with respect to $(\bgamma_{\hat{\bH}_q},\bbeta)$ has the least bias, the least size distortion of t-tests, with the smallest standard errors, while the performance of the estimators with respect to $(\bgamma_{\hat{\bH}},\bbeta)$ is far worse than others. Our preferred jackknife bias-corrected estimator with respect to the parameter $(\bgamma^0,\bbeta)$ comes in second, closely following the performance of the bias-corrected estimator with respect to $(\bgamma_{\hat{\bH}_q},\bbeta)$. The proposed jackknife method successfully reduces the bias with a small increase in the standard deviation, leading to the correct size of the test even for very weak factor models for sufficiently large sample sizes.
In addition, the empirical power curve of the significance test shows that, on balance, the split-panel jackknife estimator appears to be the most reliable among the estimators compared.

We apply the bias-corrected estimator to the factor-augmented forecast regression of bond yields $y_{t+h}$ on the factors extracted from 131 monthly macroeconomic series and one observed predictor, the \cite{CochranePiazzesi2005} factor, over the period January 1982 to December 2002. In the application, we have regarded $(\bgamma^0,\bbeta)$ as the parameter to be estimated, and the results show that the jackknife appears to effectively correct the bias of the LS estimator, thus providing more reliable inference. In particular, the jackknife estimator has given significantly different results from the rest of the estimators for testing the equal explanatory power of the factors.

The rest of the paper is organized as follows. Section \ref{sec:2} introduces models and assumptions, and also clarifies specific objects of our analysis.
Section \ref{sec:bias} derives the asymptotic biases of the augmented regression estimator with two different data dependent rotations, then derives the asymptotic bias with respect to the signal parameter and proposes a jackknife bias reduction.
To reduce the derived bias, Section \ref{sec:Mw} proposes to orthogonalize the predictor variables to the observable factor in the augmented regression before extracting the factors.
Section \ref{sec: MC} summarizes finite sample experiments and Section \ref{sec:emp} illustrates the practical usefulness of our proposed methods.
Section \ref{sec:con} contains some concluding remarks.
Mathematical proofs and additional experimental results can be found in the Online Appendix.

\noindent \textbf{Notations}: 
For any matrix $\bM=(m_{t,i})\in\mathbb{R}^{T\times N}$, we define the Frobenius norm
as 
$\|\bM\|_{\F}=(\sum_{t,i}m_{t,i}^2)^{1/2}$ and $\|\bM\|_2$  
denotes the square root of the largest eigenvalue for a positive semi-definite matrix $\bM'\bM$.
We denote the identity matrix of order $s$ by $\bI_s$ and $s\times 1$ vectors of ones and zeros by $\bone_s$ and $\bzero_s$. We use $\lesssim$ ($\gtrsim$) to represent $\leq$ ($\geq$) up to a positive constant factor. $\odot$ denotes the Hadamard product of matrices. For any positive sequences $a_n$ and $b_n$,  we write $a_n \asymp b_n$ if $a_n \lesssim b_n$ and $a_n \gtrsim b_n$. 
All asymptotic results are for cases where $N,T\to\infty$, and we do not specifically mention it. $M$ denotes a positive constant which does not depend on $N$ and $T$.

\section{Preliminaries}\label{sec:2}

\subsection{Factor-augmented regression}

The factor-augmented regression model \eqref{Augmodel} can be rewritten in a matrix form as:
\begin{align}\label{augmodel_mat}
	\by={\bF}^*\bgamma^* + \bW \bbeta + \bepsilon 
	=\bZ^* \bdelta^* + \bepsilon,
\end{align}
where $\by=(y_{1+h},\dots,y_{T+h})'$, $\bepsilon=(\epsilon_{1+h},\dots,\epsilon_{T+h})'$, $\bF^{*} = (\bff_1^{*},\dots,\bff_T^{*})'$, $\bW=(\bw_1,\dots, \bw_T)'$, $\bZ^* = (\bF^*,\bW)$, and $\bdelta^* = (\bgamma^{*\prime},\bbeta')'$. 
In line with \eqref{factormodel}, the latent factor model for the $T \times N$ matrix of predictors is given by
\begin{align}\label{factormodel_mat}
	\bX={\bF}^*\bB^{* \prime} + \bE,
\end{align}
where $\bX = (x_{t,i})$, $\bB^* = (\bb_1^* ,\dots,\bb_N^*)'$, and $\bE = (e_{t,i})$.
Let $(\lambda_1>\cdots>\lambda_r)$ 
be the $r$ largest eigenvalues of the signal covariance matrix of model \eqref{factormodel_mat}, $T^{-1}\bF^*\bB^{\ast\prime}\bB^{\ast}\bF^{\ast\prime}$, and set $\bLambda=\diag(\lambda_1,\dots,\lambda_r)$. We allow them to diverge at different rates; namely, we assume $\lambda_k \asymp N^{\alpha_k}$ with $0<\alpha_k\leq 1$ for $k=1,\dots,r$. We call the model in \eqref{factormodel_mat} with such eigenvalues a weak factor (WF) model and that with $\alpha_r=1$ a strong factor (SF) model as a special case. 


\cite{StockWatson2002JASA} proposed to extract PC factors from the predictors $\bX$, and then use them in the forecast regression. The PC estimator, $(\hat{\mathbf{F}}, \hat{\mathbf{B}})$, is defined as a solution to minimization of $\left\|\mathbf{X}-\mathbf{F B}^{\prime}\right\|_{\mathrm{F}}^2$ subject to the $r^2$ restrictions, $T^{-1} \mathbf{F}^{\prime} \mathbf{F}=\mathbf{I}_r$ and $\mathbf{B}^{\prime} \mathbf{B}$ diagonal with rank $r$. This reduces to the eigen-problem of $T^{-1} \mathbf{X X}^{\prime}$; the factor estimator $\hat{\mathbf{F}} \in \mathbb{R}^{T \times r}$ is obtained as $\sqrt{T}$ times the $r$ eigenvectors associated with the $r$ largest eigenvalues of $T^{-1} \mathbf{X} \mathbf{X}^{\prime}$ $(\hat{\lambda}_1>\cdots>\hat{\lambda}_r)$, 
and the loading estimator $\hat{\mathbf{B}} \in \mathbb{R}^{N \times r}$ is obtained by $\hat{\mathbf{B}}=T^{-1} \mathbf{X}^{\prime} \hat{\mathbf{F}}$. By the construction, we have $T^{-1} \hat{\mathbf{F}}^{\prime} \hat{\mathbf{F}}=\mathbf{I}_r$ and $\hat{\mathbf{B}}^{\prime} \hat{\mathbf{B}}=\hat{\boldsymbol{\Lambda}}=\operatorname{diag}(\hat{\lambda}_1, \ldots, \hat{\lambda}_r)$.
Setting $\hat{\bZ}=(\hat{\bF},\bW)$, we are interested in the LS estimator, 
\begin{align}\label{ols}
	\hat{\bdelta}=(\hat\bZ'\hat\bZ)^{-1}\hat\bZ'\by. 
\end{align}

%
%

\subsection{Rotation matrices}

The problem lies in the impossibility of separately identifying $(\bF^*,\bgamma^*)$ and $(\bF^*,\bB^*)$ due to the rotation indeterminacy when we consider estimation. More precisely, the products ${\bF}^*\bgamma^*$ and $\bF^* \bB^{*\prime}$ are observationally equivalent to $({\bF}^* \bR) (\bR^{-1} \bgamma^*)$ and $(\bF^* \bR) (\bR^{-1} \bB^{*\prime})$, respectively, for any $r\times r$ invertible matrix $\bR$ while the LS and PC estimators, $\hat{\bdelta}$ and  $(\hat{\bF},\hat{\bB})$, are uniquely determined. Therefore, replacing $\bF^*\bR$ with $\hat{\bF}$ in \eqref{augmodel_mat} raises the question; which rotation matrices $\bR$ justify approximation $\hat{\bF}\approx \bF^*\bR$ and $\hat{\bgamma}\approx \bR^{-1}\bgamma^*$? In this paper, we consider three rotation matrices, $\hat{\bH}$, $\hat{\bH}_q$, and $\bH$, for such $\bR$, where the first two matrices depend on data, but the last one consists only of $(\bF^*,\bB^*)$; see \cite{BaiNg2023} for more information on data-dependent rotations and \cite{jiang2023revisiting} for $\bH$.

\subsubsection{First data-dependent rotation matrix: $\hat{\mathbf{H}}$}
\cite{BaiNg2002,BaiNg2006} and \cite{StockWatson2002JASA} consider the approximation 
\begin{align}\label{FHhat}
	\hat{\bF}=\bF^* \hat{\bH} + o_p(1)~~~\text{with}~~~
	\hat{\bH} = {\bB^*}'\bB^*(T^{-1}{\bF^*}'\hat{\bF}) \hat{\bLambda}^{-1}, 
\end{align}
where and hereafter $o_p(1)$ is understood to apply row-wise. 
Using this approximation, we rewrite the model in \eqref{augmodel_mat} as 
\begin{align}\label{augmodel_Hhat}
	\by={\bF^*}\hat{\bH}\bgamma_{\hat{\bH}} + \bW \bbeta + \bepsilon
	=\hat{\bZ} \bdelta_{\hat{\bH}} + \bu 
\end{align}
with $\bu = \bepsilon - ({\hat\bF -\bF^* \hat{\bH}})\hat{\bH}^{-1}\bgamma^*$, where $\bdelta_{\hat{\bH}}=(\bgamma_{\hat{\bH}}',\bbeta')'$ and $\bgamma_{\hat{\bH}}=\hat{\bH}^{-1} \bgamma^*$.
Thus, the LS estimator $\hat{\bdelta}=(\hat{\bgamma}',\hat{\bbeta}')'$, obtained by the regression of $\by$ on $\hat{\bZ}$, is interpreted as the ``estimator'' of the data-dependent parameter $\bdelta_{\hat{\bH}}$. 
The error term $\bu$ contains $({\hat\bF -\bF^* \hat{\bH}})\hat{\bH}^{-1}\bgamma^*$ due to the approximation by $\hat{\bF}$, which is correlated with the regressor $\hat{\bZ}$ and can cause the asymptotic bias in $\sqrt{T}(\hat\bdelta - \bdelta_{\hat{\bH}})$.


\subsubsection{Second data-dependent rotation matrix: $\hat{\mathbf{H}}_q$}
As pointed out by \cite{BaiNg2023} and \cite{jiang2023revisiting}, there are many asymptotically equivalent rotation matrices to $\hat{\bH}$ for the approximation of $\hat{\bF}$.
In particular, \cite{BaiNg2023}, \cite{jiang2023revisiting} and \cite{WeiZhang2023} suggested the approximation
\begin{align}\label{FHhatq}
	\hat{\bF}=\bF^* \hat{\bH}_q + o_p(1) ~~~\text{with}~~~
	\hat{\bH}_q = (T^{-1}\hat{\bF}'{\bF^*})^{-1}. 
\end{align}
A similar discussion above will lead to a study of the asymptotic bias of $\sqrt{T}(\hat\bdelta - \bdelta_{\hat{\bH}_q})$,
where $\bdelta_{\hat{\bH}_q}=(\bgamma_{\hat{\bH}_q}',\bbeta')'$ and $\bgamma_{\hat{\bH}_q}=\hat{\bH}_q^{-1} \bgamma^*$.
%

\subsubsection{Population rotation matrix: $\mathbf{H}$}
We consider yet another rotation matrix, which is regarded as a population version of any data-dependent rotation matrix asymptotically equivalent to $\hat{\bH}$ and $\hat{\bH}_q$.
\cite{jiang2023revisiting} have shown that there always exists the unique (up to sign) rotation matrix $\bH$, which is a pure function of $(\bF^* , \bB^*)$, such that $T^{-1}\bF^{0 \prime}\bF^0 = \bI_r$ and $\bB^{0 \prime}\bB^0 = \bLambda$ with $\bF^0:=\bF^* \bH$ and $\bB^0:=\bB^* \bH^{-1\prime}$. Such a unique rotation matrix is indeed represented as
\begin{align}
	\bH = \bP\bV^{-1/2},
\end{align}
where $\bP$ is the eigenvector matrix of ${\bB^*}'\bB^*\left(T^{-1}{\bF^*}'\bF^*\right)$ corresponding to $(\lambda_1,\dots\lambda_r)$, and $\bV=\bP\left(T^{-1}{\bF^*}'\bF^*\right)\bP'$.
This rotation matrix is straightforwardly rewritten as
\begin{align}\label{HalaHhat}
	{\bH} = {\bB^*}'\bB^*(T^{-1}{\bF^*}'{\bF^0}) {\bLambda}^{-1}
	=(T^{-1}{\bF^0}'{\bF^*})^{-1},
\end{align}
insisting that $\bH$ is the population counterpart of $\hat{\bH}$ and $\hat{\bH}_q$. \cite{jiang2023revisiting} show that 
\begin{align}
	\hat{\bF} = \bF^0 + o_p(1)
\end{align}
and $\hat{\bB} = \bB^0 + o_p(1)$ in the WF setting. As in the same way, this approximation rewrites the model in \eqref{augmodel_mat} as
\begin{align}\label{pseudomodel}
	\by={\bF^0}\bgamma^0 + \bW \bbeta + \bepsilon
	=\hat{\bZ} \bdelta^0 + \bu^0 
\end{align}
with $\bu^0 = \bepsilon- ({\hat\bF -\bF^* {\bH}})\bgamma^0$  where $\bdelta^0=(\bgamma^{0\prime},\bbeta')'$ and $\bgamma^0={\bH}^{-1} \bgamma^*$.
We also study the asymptotic bias of
$\sqrt{T}(\hat\bdelta - \bdelta^0)$. Unlike the others, $\bdelta^0$ does not depend on the data.

\subsection{Assumptions}

For the asymptotic analysis that follows, we make the following assumptions.
\begin{ass}\normalfont\label{ass:eigen}
	The smallest eigenvalues of ${\bB^*}'\bB^*$ and $T^{-1}{\bF^*}'\bF^*$ are bounded away from zero.
\end{ass}

\begin{ass}[Idiosyncratic errors]\normalfont\label{ass:errors}${ }^{ }$\\
	(i) $\E[e_{t,i}]=0$ and $\E[e_{t,i}^4]\le M$ for all $i$ and $t$;\\
	(ii) For all $i$, $\left|\E[e_{s,i} e_{t,i}]\right| \leq\left|\gamma_{s, t}\right|$ for some $\gamma_{s, t}$ such that $\sum_{t=1}^T\left|\gamma_{ s, t}\right| \leq M$;\\
	(iii)  For all $t$, $\left|\E[e_{t,i} e_{t,j}]\right| \leq\left|\tau_{i, j}\right|$ for some $\tau_{i, j}$ such that $\sum_{j=1}^N\left|\tau_{ i,j}\right| \leq M$;\\
	(iv) $\left\|\bE \right\|_2^2=O_p(\max{\{N,T}\})$.
\end{ass}

\begin{ass}[Signal strength]\normalfont\label{ass:signal}
	There exist random or non-random variables $d_1,\dots,d_r>0$ and constants $0<\alpha_r\leq \dots \leq \alpha_1 \leq 1$ such that {$\lambda_k=d_kN^{\alpha_k}$} for $k=1,\dots, r$ with ordered $0<\lambda_r< \dots < \lambda_1$ for large $N$. If $d_k$'s are random, we have {$\E[d_k^2] \le M$} for all $k$.
\end{ass}

Denote $\bN=\diag(N^{\alpha_1}, \dots,N^{\alpha_r})$ and $\bD=\diag(d_1,\dots,d_r)$, so that we can write {$\bLambda=\bD \bN$}. We do not require any specific structure in $(\bF^{*}, \bB^{*})$, such as diagonality of $\bN^{-\frac{1}{2}}\bB^{*\prime}\bB^*\bN^{-\frac{1}{2}}$ in \citet[Section 5]{BaiNg2023} and/or $T^{-1}\bF^{*\prime}\bF^*=\bI_r$ in \cite{Freyaldenhoven21JoE}. 

As discussed earlier, the PC estimators $(\hat{\bF},\hat{\bB})$ are estimating the pseudo-true parameters $({\bF^0},{\bB^0})$. We directly impose the following assumptions on them.


\begin{ass}[Factors and Loadings]\normalfont\label{ass:factor and loadings}

	Denote $\bz^0_t=(\bff^{0\prime}_t, \bw_t')'$. \\
	(i) $\E\|{\bz_t^0}\|_2^4 \le M$ and $\E\|\bb_i^0\|_2^4 \le M$;\\ 
	(ii) $\E\| \bN^{-\frac{1}{2}}\sum_{i=1}^{N} \bb^0_ie_{t,i}   \|_2^2 \le M$ for each $t$; \\
	(iii) $\E\| T^{-\frac{1}{2}}\sum_{t=1}^{T}{\bz_t^0}e_{t,i}\|_2^2 \le M$ for each $i$; \\
	(iv) The $r \times r$ matrix satisfies $\E\| T^{-\frac{1}{2}}\bN^{-\frac{1}{2}} \sum_{t=1}^{T}\sum_{i=1}^{N}\bb_i^{0}e_{t,i}{\bz_t^{0\prime}} \|_{{2}}^2 \le M$;\\
	(v) As $N, T \rightarrow \infty$, $T^{-1} \sum_{t=1}^{T}  \left( \bN^{-\frac{1}{2}}\sum_{i=1}^{N} \bb_i^0e_{t,i}\right)  \left( \bN^{-\frac{1}{2}}\sum_{i=1}^{N} \bb_i^0e_{t,i}\right)'\CP \bGamma  $, where $\bGamma = \lim_{N, T \rightarrow \infty} T^{-1} \sum_{t=1}^T \bGamma_t>0$, and $\bGamma_t =   \operatorname{Var}\left(\bN^{-\frac{1}{2}}\sum_{i=1}^{N} \bb_i^0e_{t,i} \right)$.
\end{ass}

The moment restrictions in Assumption 4 (iii), (iv) 
are essentially similar to Assumptions D, F2 
in \cite{Bai2003}, and Assumption 4 (ii) 
is similar moment restriction related for $\bb_i^0$.
Assumption (v) is similar to Assumption 3(e) in \cite{GoncalvesPerron2014}.

Now we impose assumptions on the pseudo-true augmented model \eqref{pseudomodel}:

\begin{ass}[Weak dependence between idiosyncratic errors and regression errors]\normalfont\label{ass:2errors} 
	The $r \times r$ matrix satisfies $\E\left\|T^{-\frac{1}{2}}\bN^{-\frac{1}{2}} \sum_{t=1}^{T} \sum_{i=1}^N \bb_i^0 e_{t,i} \varepsilon_{t+h}\right\|_2^2 \leq M$. 
\end{ass}

\begin{ass}[Moments, parameters and CLT]
	\normalfont\label{ass:Aug_errors}${ }^{ }$\\
	(i) $\E[\epsilon_{t+h}] =0$ and $\E|\epsilon_{t+h}|^2 < M $;\\
	(ii) $||\bdelta^0||_{2}\leq M$ and $\bH\CP \bH_0$ which is fixed and invertible;\\
	(iii) $\E||\bz_t^0 ||^4 \leq M$,
	$T^{-1/2} {\bZ}^{0\prime}\bepsilon \CD N(\mathbf{0},\bSigma_{\bZ^0 \bepsilon})$, $T^{-1} {\bZ}^{0\prime}\bZ^0 \CP \bSigma_{\bZ^0 \bZ^0}$, where $\bSigma_{\bZ^0\bepsilon}$ and $\bSigma_{\bZ^0 \bZ^0}$ are fixed, positive definite and bounded.
	
\end{ass}
Assumptions \ref{ass:2errors} and \ref{ass:Aug_errors} are similar to Assumption 4 in \cite{GoncalvesPerron2014} and Assumption E in \cite{BaiNg2006}, respectively.



\section{Bias Analysis}\label{sec:bias}

\subsection{Bias analysis with $\hat{\bH}$}\label{subsec:Hhat}

The asymptotic bias of $\sqrt{T}(\hat{\bdelta}-\bdelta_{\hat{\bH}})$ has only been investigated in the literature, assuming the SF model. 
\cite{BaiNg2006} show that the limiting distribution of $\sqrt{T}(\hat{\bdelta}-\bdelta_{\hat{\bH}})$ is centered at zero (i.e., no asymptotic bias) so long as $\sqrt{T}/N \to 0$.
Under a relaxed condition of $\sqrt{T}/N \to c\in[0,\infty)$, 
\cite{Ludvigson2011} show that $\sqrt{T}(\hat{\bdelta}-\bdelta_{\hat{\bH}})$ has an asymptotic bias and provide an analytical correction. \cite{GoncalvesPerron2014} refine the bias expression, and propose an analytical correction as well as a wild bootstrap method to remove the bias. \cite{gonccalves2020bootstrapping} extend \cite{GoncalvesPerron2014} to allow errors $e_{t,i}$ to be cross-correlated using the large covariance matrices estimator proposed by \cite{BickelLevina2008}. 

We are ready to analyze the asymptotic bias of $\sqrt{T}(\hat{\bdelta}-\bdelta_{\hat{\bH}})$ in the WF setting. 

\begin{thm}\label{thm:bias_Hhat}
	Suppose Assumptions \ref{ass:eigen}--\ref{ass:Aug_errors} hold. If {$\alpha_r>\frac{1}{2}$}, $\frac{N^{1-\alpha_r}}{\sqrt{T}} \to  0$, and $\sqrt{T}N^{\frac{1}{2}\alpha_1-\frac{3}{2}\alpha_r} \to c_1 \in [0,\infty)$, as $N, T \to \infty$, we have
	\begin{align*}
		\sqrt{T}(\hat{\bdelta}-\bdelta_{\hat{\bH}}) \CD N\left(-c_1 \bkappa_{\bdelta^*}, \bSigma_{\bdelta}\right)
	\end{align*}
	with 
	\begin{align*}
		\bkappa_{\bdelta^*}=
		\bSigma_{\bZ^0 \bZ^0}^{-1} \binom{\bG + \nu\bD^{-1}\bGamma\bD^{-1}}{\bSigma_{\bW \bF^0} \, {\bG}} \, \bH_0^{-1} \, \bgamma^* 
		~~~ \text{and}~~~
		\bSigma_{\bdelta}=  \bSigma_{\bZ^0 \bZ^0}^{-1} \bSigma_{\bZ^0 \bepsilon} \bSigma_{\bZ^0\bZ^0}^{-1},
	\end{align*}
	where 
	$c_1 \bG = \lim_{N,T\to\infty} \sqrt{T}\bN^{\frac{1}{2}} \bGamma  \bD^{-2} \bN^{-\frac{3}{2}} $, $\nu = \lim_{N\to\infty} N^{-\frac{1}{2}(\alpha_1-\alpha_r)}$ 
	and\\ $\bSigma_{\bW \bF^0} = \plim_{N,T \rightarrow \infty} T^{-1}\bW'\bF^0$.

\end{thm}

Theorem \ref{thm:bias_Hhat} establishes the asymptotic normality of $\sqrt{T}(\hat{\bdelta}-\bdelta_{\hat{\bH}})$, with the bias appearing in the location. Because $\nu$ takes the values one (when $\alpha_1=\alpha_r$) or zero (when $\alpha_1>\alpha_r$) only, when $c_1=0$ the asymptotic bias becomes zero and the distribution is identical to the results in \citet[Theorem 4]{jiang2023revisiting}.

Essentially, asymptotic bias is due to the correlation between the regressor matrix $\hat{\bF}$ and the estimation error shown in \eqref{augmodel_Hhat}. Therefore, when $\bW$ and $\bF^*$ are orthogonal, $\hat{\bbeta}$ will not have an asymptotic bias since $\bSigma_{\bW \bF^0}=\bzero$ and $\bSigma_{\bZ^0 \bZ^0}$ is block diagonal.

Note that the bias term $\bkappa_{\bdelta^*}$ is estimable, and the associated bias-corrected estimator is given in \eqref{bcHhats} in Section \ref{sec: MC}.

In general, $\bSigma_{\bdelta}$ can be consistently estimated by $\hat{\bSigma}_{\bdelta}=(T^{-1}\hat{\bZ}^{\prime}\hat{\bZ})^{-1}
\hat{\bSigma}_{\bZ^0 \bepsilon}
(T^{-1}\hat{\bZ}^{\prime}\hat{\bZ})^{-1}
$
so long as $\hat{\bSigma}_{\bZ^0 \bepsilon}$ is consistent to ${\bSigma}_{\bZ^0 \bepsilon}$. 
A choice of $\hat{\bSigma}_{\bZ^0 \bepsilon}$ depends on the property of $\epsilon_t$. For example, when $\epsilon_t$ is heteroskedastic, choose $\hat{\bSigma}_{\bZ^0 \bepsilon}=
T^{-1}\sum_{t=1+h}^{T+h}\hat\bz_t \hat{\epsilon}_t^2 \hat\bz_t'$, where $\hat{\epsilon}_{t+h}=y_{t+h}-\hat{\bdelta}'\hat{\bz}_t$.

The expressions $c_1 \bG = \lim_{N, T \rightarrow \infty} \sqrt{T}\bN^{\frac{1}{2}} \bGamma  \bD^{-2} \bN^{-\frac{3}{2}} $ and 
$\nu \bD^{-1}\bGamma  \bD^{-1}$ 
suggest a complicated asymptotic bias structure, depending on the structure of $(\alpha_1,\dots,\alpha_r)$. When all the divergence rates are identical, in which $\alpha=\alpha_1=\cdots=\alpha_r$, $c_1 = \lim_{N, T \rightarrow \infty} \sqrt{T}/N^{\alpha}$ and $\nu=1$ so that $c_1\bG = c_1\bGamma \bD^{-2}$ and $\nu \bD^{-1}\bGamma \bD^{-1} = \bD^{-1}\bGamma \bD^{-1}$. 
When $\alpha_1>\alpha_r$, $\nu=0$ so that the term $\bD^{-1}\bGamma \bD^{-1}$ disappears, while $c_1 = \lim_{N, T \rightarrow \infty}\sqrt{T}N^{\frac{1}{2}\alpha_1-\frac{3}{2}\alpha_r}$ and the structure of the matrix $\bG$ depends on how many factors diverge at the rate of $N^{\alpha_1}$ and $N^{\alpha_r}$. To see this, for simplicity, suppose that all the factors diverge in different rates (i.e., $\alpha_1>\alpha_2>\cdots>\alpha_r$). In $c_1 \bG$, $\sqrt{T}\bGamma\bD^{-2}$ is pre-multiplied by $\bN^{\frac{1}{2}}=\diag(N^{\frac{1}{2}\alpha_1},...,N^{\frac{1}{2}\alpha_r})$ and post-multiplied by $\bN^{-\frac{3}{2}}=\diag(N^{-\frac{3}{2}\alpha_1},...,N^{-\frac{3}{2}\alpha_r})$. 
The combination of elements that disappear at the slowest rate $\sqrt{T}N^{\frac{1}{2}\alpha_1-\frac{3}{2}\alpha_r}$ is given by the product of the first element in $\bN^{\frac{1}{2}}$ and the last element in $\bN^{-\frac{3}{2}}$. Therefore, the only non-zero element in $c_1\bG$ is the $(1,r)$-th element, as other elements tend to zero. 
Now let the first $u$ factors and the last $s$ factors ($u+s\leq r)$ share the equal divergence rates (i.e. $\alpha_1=\cdots=\alpha_u>\cdots>\alpha_{r-s+1}=\cdots=\alpha_r$). Then, a similar discussion leads to the conclusion that 
the $u\times s$ sub-matrix in the upper right corner of $c_1\bG$ is non-zero.

The structure of $c_1\bG$ for different sets of $\alpha_k$, $k=1,\dots,r$, is summarized in the following corollary.

\begin{cor}\label{cor1}
Suppose that the conditions for the results in Theorem \ref{thm:bias_Hhat} are satisfied and $c_1 \in (0,\infty)$. 
Consider an $r\times 1$ vector $\balpha=(\alpha_1,\alpha_2,\ldots,\alpha_r)'$. Let an $r\times 1$ vector of a binary variable be $\be_{\alpha_1}$, which replaces elements in $\balpha$ with $1$ if they are $\alpha_1$ and $0$ otherwise. Similarly, define an $r\times 1$ vector of a binary variable $\be_{\alpha_r}$ for $\alpha_r$. By construction, $||\be_{\alpha_1}||_2^2+||\be_{\alpha_r}||_2^2\leq r$. 
If $\alpha_1=\alpha_r$, then
$c_1\bG = c_1 \bGamma  \bD^{-2}$.
If $\alpha_1>\alpha_r$, then
$c_1\bG = c_1 (\be_{\alpha_1} \be_{\alpha_r}')\odot\bGamma  \bD^{-2}$.
\end{cor}

Theorem \ref{thm:bias_Hhat} and Corollary \ref{cor1} can be viewed as the generalized version of the results in \cite{GoncalvesPerron2014,gonccalves2020bootstrapping} for WF models, which tell us the following. If $\alpha_1 =\alpha_r = 1$, then the model becomes a SF model, and the condition $\sqrt{T}N^{\frac{1}{2}\alpha_1-\frac{3}{2}\alpha_r}\to c_1$ reduces to $\sqrt{T}/N\to c_1$ and our result reduces to the result in Theorem 2.1 in \cite{GoncalvesPerron2014}.
The asymptotic bias expression remains the same so long as $\alpha_1=\alpha_r(=\alpha)$, whilst the associated condition $\sqrt{T}/N^{\alpha}\to c_1$ implies that, other things being equal, the weaker the model, the larger the bias.
For $\alpha_1\neq\alpha_r$, the condition $\sqrt{T}N^{\frac{1}{2}\alpha_1-\frac{3}{2}\alpha_r}\to c_1$ implies that, all else being equal, the larger the deviation between $\alpha_1$ and $\alpha_r$ or the weaker the model, the larger the bias. Furthermore, all else being equal, the more heterogeneous the divergence rates are, the more the asymptotic bias deviates from that with $\alpha_1=\alpha_r$.



\subsection{Bias analysis with $\hat{\bH}_q$}


We next show the asymptotic bias of $\sqrt{T}(\hat{\bdelta}-\bdelta_{\hat{\bH}_q})$. We reveal that the bias is generally smaller than that of $\sqrt{T}(\hat{\bdelta}-\bdelta_{\hat{\bH}})$.

\begin{thm}\label{thm:bias_H3}
Suppose Assumptions \ref{ass:eigen}--\ref{ass:Aug_errors} hold. If $\alpha_r>\frac{1}{2}$, $\frac{N^{1-\alpha_r}}{\sqrt{T}} \to  0$, and $\sqrt{T}N^{-\alpha_r} \to c_2 \in [0,\infty)$, as $N, T \to \infty$, we have 
\[
\sqrt{T}(\hat{\bdelta}-\bdelta_{\hat{\bH}_q}) \CD N\left(c_2 \bar{\bkappa}_{\bdelta^*}, \bSigma_{\bdelta}\right)
\]
with
\begin{align*}
\bar{\bkappa}_{\bdelta^*} = \bSigma_{\bZ^0 \bZ^0}^{-1}  \binom{\mathbf{0}}{\bSigma_{\bW \bF^0} \bar{\bG} } \bH_0^{-1}  \bgamma^*,
\end{align*}
where $c_2 \bar{\bG} = \lim_{N, T \rightarrow \infty}\sqrt{T}\bN^{-\frac{1}{2}} \bD^{-1}\bGamma  \bD^{-1}\bN^{-\frac{1}{2}} $. 
If $\alpha_1=\alpha_r$, then $c_1=c_2$ and $c_2\bar{\bG}=c_2\bD^{-1}\bGamma  \bD^{-1}$. 
If $\alpha_1 > \alpha_r$, then   
$c_2\bar{\bG} = c_2(\be_{\alpha_r} \be_{\alpha_r}')
\odot\bD^{-1}\bGamma  \bD^{-1}$ .

\end{thm}

Theorem \ref{thm:bias_H3} establishes the asymptotic normality of $\sqrt{T}(\hat{\bdelta}-\bdelta_{\hat{\bH}_q})$ with the bias generally smaller than that in Theorem \ref{thm:bias_Hhat}.
To see this, we compare their difference, $c_1{\bkappa}_{\bdelta^*}-c_2\bar{\bkappa}_{\bdelta^*}$, under the comparable condition of $c_1=c_2$ (and $\nu=1$), or, equivalently, $\alpha_1 = \alpha_r$. 
Under the simplest error assumption, $\be_{t}\sim \text{i.i.d.}(\bzero,\sigma_e^2 \bI_N)$, we have $\bGamma=\sigma_e^2 \bD$, and hence $\bD^{-1}\bGamma\bD^{-1}=\bGamma\bD^{-2}=\bG = \bar{\bG}$. Therefore, in this case a direct calculation yields 
\begin{align}
\|{\bkappa}_{\bdelta^*}-\bar{\bkappa}_{\bdelta^*}\|_2
=\left\|\bSigma_{\bZ^0 \bZ^0}^{-1}  \binom{{\bG}+\bD^{-1}\bGamma\bD^{-1}}{\bSigma_{\bW \bF^0} ({\bG}-\bar{\bG}) } \bH_0^{-1}  \bgamma^*\right\|_2
=\left\|\bSigma_{\bZ^0 \bZ^0}^{-1}  \binom{2\sigma_e^2\bD^{-1}}{\bzero } \bH_0^{-1}  \bgamma^*\right\|_2 \geq 0,
\end{align}
meaning that the magnitude of the bias of $\sqrt{T}(\hat{\bdelta}-\bdelta_{\hat{\bH}})$ is not smaller than that of $\sqrt{T}(\hat{\bdelta}-\bdelta_{\hat{\bH}_q})$. 
We have ${\bkappa}_{\bdelta^*}=\bar{\bkappa}_{\bdelta^*}$ only when $\bgamma^* = \bzero$.
It should also be noted that, other things being equal, the bias difference increases as the scaled noise to signal ratio ($\sigma_e^2\bD^{-1}$) and/or $||\bgamma^*| |_2$ increases. 

Regarding the conditions, only difference is that $a_{NT}:=\sqrt{T}N^{\frac{1}{2}\alpha_1-\frac{3}{2}\alpha_r} \to c_1$ in Theorem \ref{thm:bias_Hhat} is replaced by $b_{NT}:=\sqrt{T}N^{-\alpha_r} \to c_2$ in Theorem \ref{thm:bias_H3}.  We find that $b_{NT}$ is not slower than $a_{NT}$ by taking the ratio: $a_{NT}/b_{NT}=N^{\frac{1}{2}(\alpha_1-\alpha_r)}\geq 1$, and $c_1=c_2(\neq0)$ only when $\alpha_1=\alpha_r$. 
This is essentially because the approximation $T^{-1}\bW'(\hat{\bF}-\bF^* \hat{\bH}_{{q}})=o_p(1)$ has a much faster convergence rate than the approximation $T^{-1}\bW'(\hat{\bF}-\bF^* \hat{\bH})=o_p(1)$. 

The condition $\sqrt{T}/N^{\alpha_r}\to c_2$ in Theorem \ref{thm:bias_H3} implies that, all else being equal, the weaker the model, the larger the bias.
The bias term $\bar{\bkappa}_{\bdelta^*}$ is estimable and the associated bias-corrected estimator is given in \eqref{bcHhats} in Section \ref{sec: MC}. 

The expression $c_2 \bar{\bG} = \lim_{N, T \rightarrow \infty}\sqrt{T}\bN^{-\frac{1}{2}} \bD^{-1}\bGamma  \bD^{-1}\bN^{-\frac{1}{2}}$ again suggests a complicated asymptotic bias structure, depending on the structure of $(\alpha_1,\dots,\alpha_r)$. When all the divergence rates are identical, in which $\alpha=\alpha_1=\cdots=\alpha_r$, $c_2 = \lim_{N, T \rightarrow \infty} \sqrt{T}/N^{\alpha}$ and $c_2 \bar{\bG} = c_2\bD^{-1}\bGamma \bD^{-1}$. 
When $\alpha_1>\alpha_r$, $c_2 = \lim_{N, T \rightarrow \infty}\sqrt{T}/N^{\alpha_r}$ and the structure of the asymptotic bias depends on how many factors diverge at the rate of $N^{\alpha_r}$. In $c_2 \bar{\bG}$, $\sqrt{T}\bGamma$ is pre- and post-multiplied by the diagonal matrix $\bD^{-1}\bN^{-\frac{1}{2}}$ with $\bN^{-\frac{1}{2}}=\diag(N^{-\frac{1}{2}\alpha_1},\dots,N^{-\frac{1}{2}\alpha_r})$. 
If $s\leq r$ elements in $\balpha$ take the value $\alpha_r$ (i.e. $\alpha_1>\dots>\alpha_{r-s+1}=\dots=\alpha_r$), the elements that disappear at the slowest rate in $\bN^{-\frac{1}{2}}$ are the last $s$ diagonal elements. Therefore, the non-zero elements in $c_2\bar{\bG}$ are the $s\times s$ block at its bottom-right corner, which leads to the last sentence in Theorem \ref{thm:bias_H3}.

Note also that the bias $\bar{\bkappa}_{\bdelta^*}$ is completely eliminated (i.e. $\bar{\bkappa}_{\bdelta^*}=\mathbf{0}$) if the observed factor $\bW$ is uncorrelated with $\bF^*$ (i.e. $\bSigma_{\bW \bF^0}=\mathbf{0}$). 
To exploit this property, we consider extracting the factor from the predictors after projecting out the observable factor, $\bw_t$, in Section \ref{sec:Mw}.


\subsection{Bias analysis with $\mathbf{H}$}\label{sec:jackknife}
In this section we consider the asymptotic bias of $\sqrt{T}(\hat\bdelta  - \bdelta^0)$. For any choice of a data dependent invertible rotation matrix, say $\hat{\bR}$, we can decompose $\sqrt{T}(\hat\bdelta  - \bdelta^0)=\sqrt{T}(\hat{\bdelta}  - \bdelta_{\hat{\bR}})+\sqrt{T}({\bdelta}_{\hat{\bR}}  - \bdelta^0)$. Thus, in principle, for any choice of $\hat{\bR}$, the asymptotic bias of $\sqrt{T}(\hat\bdelta  - \bdelta^0)$ has to be the sum of the biases due to $\sqrt{T}(\hat{\bdelta}  - \bdelta_{\hat{\bR}})$ and $\sqrt{T}({\bdelta}_{\hat{\bR}} - \bdelta^0)$.

In this section, we employ the decomposition
\begin{align}
\sqrt{T}(\hat\bdelta  - \bdelta^0) = \sqrt{T}(\hat\bdelta  - \bdelta_{\hat{\bH}_q}) + \sqrt{T}(\bdelta_{\hat{\bH}_q}-\bdelta^0).
\end{align}
The first term has been investigated in Theorem \ref{thm:bias_H3} and the second term is evaluated as
\begin{align}
\sqrt{T}(\bdelta_{\hat{\bH}_q}-\bdelta^0)=\binom{\sqrt{T}(\hat{\bH}_q^{-1}-\bH^{-1})\bgamma^*}{\mathbf{0}}=O_p(\sqrt{T}N^{\frac{1}{2}\alpha_1-\frac{3}{2}\alpha_r}), 
\end{align}
but no explicit bias expression is provided. We derive the asymptotic bias by assuming that $\sqrt{T}(\bdelta_{\hat{\bH}_q}-\bdelta^0)$ tends to a constant vector in probability, say $c_1\bh_{\bgamma^*}$, when $\sqrt{T}N^{\frac{1}{2}\alpha_1-\frac{3}{2}\alpha_r} \to c_1 \in [0,\infty)$ as $N, T \to \infty$. 
For the purpose of deriving the asymptotic bias, this additional assumption is not overly restrictive, since it includes the case in which the bias tends to zero. 
\begin{thm}\label{thm:bias_H}
Suppose that Assumptions \ref{ass:eigen}--\ref{ass:Aug_errors} hold and that $\alpha_r>\frac{1}{2}$, $\frac{N^{1-\alpha_r}}{\sqrt{T}} \to  0$, $\sqrt{T}N^{\frac{1}{2}\alpha_1-\frac{3}{2}\alpha_r} \to c_1 \in [0,\infty)$, $\sqrt{T}N^{-\alpha_r} \to c_2 \in [0,\infty)$ as $N, T \to \infty$. Then, further assuming that $\sqrt{T}(\bdelta_{\hat{\bH}_q}-\bdelta^0) \CP c_1 \bh_{\bgamma^*}$ where $\bh_{\bgamma^*}$ is a constant vector whose last $p$ rows are zero, we have 
\begin{align*}
& \sqrt{T}(\hat\bdelta  - \bdelta^0)  \CD
N\left(c_1 \bh_{\bgamma^*} + c_2 \bar{\bkappa}_{\bdelta^*} ,\bSigma_{\bdelta}\right).
\end{align*}
\end{thm}

The theorem tells that the bias of $\sqrt{T}(\hat\bdelta  - \bdelta^0)$ has an additional bias due to the difference between $\hat{\bH}_q$ and ${\bH}$ compared to the bias of $\sqrt{T}(\hat\bdelta  - \bdelta_{\hat{\bH}_q})$. When $\alpha_r=\alpha_1$, we have $c_2=c_1$ and the asymptotic bias is $c_2( \bh_{\bgamma^*} + \bar{\bkappa}_{\bdelta^*})$. If $\alpha_r < \alpha_1$ and $c_1 \in (0,\infty)$, then $c_2=0$, hence the asymptotic bias is $c_1\bh_{\bgamma^*}$. The sign and the magnitude of $\bh_{\bgamma^*}$ cannot be identified analytically, but the experimental results below seem to suggest that it has the same sign as the sign of $\bar{\bkappa}_{\bdelta^*}$, rather than they canceling out. 

Unlike the asymptotic biases $\bkappa_{\bdelta^*}$ and $\bar{\bkappa}_{\bdelta^*}$ in Theorems \ref{thm:bias_Hhat} and \ref{thm:bias_H3}, 
the asymptotic bias of $\sqrt{T}(\hat\bdelta  - \bdelta^0)$ is not parametrically estimable as it contains $\bh_{\bgamma^*}$. 
Here, we propose a subsampling method to reduce the bias, called the split-panel jackknife; this has been originally proposed to correct the $O(T^{-1})$ bias of fixed-effects panel data estimators by \cite{DhaeneJochman2015}, and then extended to correct the $O(N^{-1})$ bias as well by \cite{FERNANDEZVAL2016291}. To the best of our knowledge, this is the first application to the estimation of factor-augmented models.

To define the split-panel jackknife, 
consider a partition of $\{1, \ldots, N\}$ into two half-panels, $\mathcal{N}_1 :=$ $\{1, \ldots, \lfloor N/ 2 \rfloor\}$ and $\mathcal{N}_2 :=\{\lfloor N/ 2 \rfloor+1, \ldots, N\}$. 
Let $\hat{\bF}_j$ and $\hat{\bdelta}_j$ be the PC estimator using the subsample $\mathcal{N}_j$ with $T$ observations and the associated augmented regression estimator for $j=1,2$. Then, the split-panel jackknife bias-corrected estimator is given by\footnote{For finite samples, we introduce randomization of the order of the cross-sectional units to avoid potentially biased information on factors in $\mathcal{N}_j$. See \eqref{bcjkest_R} in Section \ref{sec: MC} for the procedure.}
\begin{align*}\label{bcjkest_1}
\hat{\bdelta}_{bcjk} = 2 \hat{\bdelta}-\frac{1}{2}\left(\hat{\bdelta}_{1}+\hat{\bdelta}_{2}\right).
\end{align*}

Now we derive the asymptotic bias of the jackknife bias-corrected estimator relative to the parameter $\bdelta^0$. 
\begin{thm}\label{thm:bias_JC}
Suppose that the same assumptions hold as in Theorem \ref{thm:bias_H}. Then, we have 
\begin{align*}
& \sqrt{T}(\hat\bdelta_{bcjk}  - \bdelta^0)  \CD
N\left((2-2^{\frac{1}{2}(3\alpha_r-\alpha_1)})c_1 \bh_{\bgamma^*} + (2-2^{\alpha_r})c_2 \bar{\bkappa}_{\bdelta^*} ,\bSigma_{\bdelta}\right).
\end{align*} 
\end{thm}
The theorem tells that the proposed jackknife bias correction \textit{always} reduces the bias of $\sqrt{T}(\hat\bdelta  - \bdelta^0)$, but the effectiveness of the bias reduction depends on the weakness of the factor model. For the SF model with $\alpha_r=1$, the asymptotic bias is removed completely. When $\alpha_r=\alpha_1$, the asymptotic bias is $(2-2^{\alpha_r})c_2(\bh_{\bgamma^*} + \bar{\bkappa}_{\bdelta^*})$, which is always smaller than the asymptotic bias of $\sqrt{T}(\hat\bdelta  - \bdelta^0)$, but the correction becomes less effective as $\alpha_r$ deviates from unity. If $\alpha_r < \alpha_1$ and $c_1 \in(0,\infty)$, then the asymptotic bias is $(2-2^{\frac{1}{2}(3\alpha_r-\alpha_1)})c_1 \bh_{\bgamma^*}$, which is again always smaller than that of $\sqrt{T}(\hat\bdelta  - \bdelta^0)$, but the correction becomes less effective as $\frac{1}{2}(3\alpha_r - \alpha_1)$ deviates from unity.


\section{Bias Analysis After Model Transformation}\label{sec:Mw}

So far, we have analyzed the asymptotic biases of the augmented regression estimator $\hat{\bdelta}$ relative to the rotated parameter $\bdelta^*$ with the different rotation matrices, $\hat{\bH}$, $\hat{\bH}_q$, and $\bH$. All of the results indicate that the bias will be present due to the replacement of rotated $\bF^*$ with the PC estimator $\hat{\bF}$. 
In this section, we introduce an estimation procedure that forces the latent factor and the observed factor to be uncorrelated. Such a procedure will eliminate the bias. In particular, the bias of the estimator relative to the rotated parameter with $\hat{\bH}_q$ becomes zero.

To begin the discussion, let $
\mathbf{M}_w=\mathbf{I}_T- \bP_w$ with $\bP_w = \mathbf{W}\left(\mathbf{W}^{\prime} \mathbf{W}\right)^{-1} \mathbf{W}^{\prime}$. In the first step, we consider the transformed model of $\bX$, $\bX_w = \mathbf{M}_w \mathbf{X}$, which is written as
\begin{align}
\mathbf{X}_w=\bF_w^* \mathbf{B}^{* \prime}+\mathbf{E}_w=\mathbf{F}_w^0 \mathbf{B}_w^{0 \prime}+\mathbf{E}_w,
\end{align}
where $\mathbf{F}_w^* = \bM_w \mathbf{F}^*$, $\mathbf{F}_w^0=\bM_w\bF^* \bH_w$, $\mathbf{B}_w^0= \bB^* \bH_w^{'-1}$,  $\mathbf{E}_w=\bM_w \bE$, and the rotation matrix $\bH_w$ is the analogous counterpart of $\bH$ in the previous section but for the model of $\bX_w$. 
Suppose Assumptions \ref{ass:eigen}--\ref{ass:Aug_errors} hold to this model, but defining the variables and parameters with the subscript $w$ in Assumptions \ref{ass:signal}--\ref{ass:Aug_errors}.
$\hat\bF_w$ is $\sqrt{T}$ times the eigenvectors associated with the first $r$ largest eigenvalues of $T^{-1}{\bX}_w\bX_w'$. 
$\hat{\bH}_{w}$ and $\hat{\bH}_{q,w}$ are the data dependent rotation matrices corresponding to $\hat{\bH}$ and $\hat{\bH}_q$ in the previous sections but for the model of $\bX_w$.
Now consider the augmented regression model, which can be re-written using $\bF_w^*$ as
\begin{align}
\nonumber
\mathbf{y}
=\mathbf{F}^* \boldsymbol{\gamma}^*+\mathbf{W} \boldsymbol{\beta}+\boldsymbol{\epsilon}
=\mathbf{F}_w^* \boldsymbol{\gamma}^*+\mathbf{W} \bbeta_w +\bepsilon 
=\mathbf{F}_w^0 \boldsymbol{\gamma}_w^0+\mathbf{W} \bbeta_w +\bepsilon,
\end{align}
where $\bbeta_w = \left(\mathbf{W}^{\prime} \mathbf{W}\right)^{-1} \mathbf{W}^{\prime}\bF^*\bgamma^*+\bbeta$ and $\boldsymbol{\gamma}_w^0=\bH_w^{-1}\bgamma^*$.
The relevant feasible augmented regression estimator is 
\begin{align}
\hat{\bdelta}_w = (\hat{\bZ}_w'\hat{\bZ}_w)^{-1}\hat{\bZ}_w'\by,
\end{align}
where $\hat{\bZ}_w=(\hat{\bF}_w,\bW)$. As easily seen, the approximations $\hat{\bF}_w = \bF_w \hat{\bH}_w + o_p(1)$, 
$\hat{\bF}_w = \bF_w \hat{\bH}_{q,w} + o_p(1)$ and 
$\hat{\bF}_w = \bF_w^0  + o_p(1)$ lead us to study the asymptotic biases of 
$\sqrt{T}(\hat{\bdelta}_w - {\bdelta}_{\hat{\bH}_w})$, 
$\sqrt{T}(\hat{\bdelta}_w - {\bdelta}_{\hat{\bH}_{q,w}})$ and
$\sqrt{T}(\hat{\bdelta}_w - {\bdelta}_w^0)$,
where 
${\bdelta}_{\hat{\bH}_w}=(\bgamma_{\hat{\bH}_w}',\bbeta_w')'$ with $\bgamma_{\hat{\bH}_w}=\hat{\bH}_w^{-1}\bgamma^{*}$,
${\bdelta}_{\hat{\bH}_{q,w}}=(\bgamma_{\hat{\bH}_{q,w}}',\bbeta_w')'$ with $\bgamma_{\hat{\bH}_{q,w}}=\hat{\bH}_{q,w}^{-1}\bgamma^{*}$, and
${\bdelta}_w^0=(\bgamma_w^{0\prime},\bbeta_w')'$. 
It is straightforward to prove the ``$\bX_w$'' versions of Theorems \ref{thm:bias_Hhat}-\ref{thm:bias_JC}, putting $\plim_{N,T\to\infty} T^{-1}\bW'\bF_w^0 =\bSigma_{\bW \bF_w^0}= \mathbf{0}$.

There are a few comments to make. After the transformation, the augmented regression coefficient on $\bZ_w^* = (\bF_w^* , \bW)$ changes to $\bdelta_w^* = (\bgamma^{*\prime},\bbeta_w')'$, which is different from $\bdelta^*$ 
unless $\bW'\bF^* = \mathbf{0}$. 
Also, because of the orthogonality $\bW' \bF_w^0 = \mathbf{0}$, the bias of $\hat{\bbeta}_w-{\bbeta}_w$ is always zero, and the bias corrections are only applied to $\hat{\bgamma}_w$ with respect to the relevant ``parameters''.

We conclude this section by providing the result for the asymptotic bias of $\sqrt{T}(\hat{\bdelta}_w - {\bdelta}_{\hat{\bH}_{q,w}})$, which is completely eliminated due to the transformation.

\begin{thm}\label{thm:bias_Hw3}
Suppose Assumptions \ref{ass:eigen}--\ref{ass:Aug_errors} for the versions of $\bX_w$ hold. If $ \alpha_r>\frac{1}{2}$, $\frac{N^{1-\alpha_r}}{\sqrt{T}} \to  0$,  $\sqrt{T}N^{-\alpha_r} \to c_2 \in (0,\infty)$, as $N, T \to \infty$, we have 
\begin{align*}
& \sqrt{T}(\hat\bdelta_w - \bdelta_{\hat{\bH}_{q,w}})  \CD
N\left(\mathbf{0},\bSigma_{\bdelta_w}\right).
\end{align*} 
\end{thm}


\section{Monte Carlo Experiments}\label{sec: MC}

In this section, we examine the finite sample performance of the estimators of
the factor-augmented regressions. In particular, we focus on the bias, standard deviation and size of the t-test.

\subsection{Design}\label{sec:design}

$\mathbf{X=F}^{0}\mathbf{B}^{0\prime}+\mathbf{E}$, $\mathbf{X=(}%
x_{t,i}\mathbf{)}$, $i=1,\dots,N$, $t=1,\dots,T$, $\mathbf{F}^{0}\in\mathbb{R}%
^{T\times r}$, $\mathbf{B}^{0}\in\mathbb{R}^{N\times r}$ is generated as
follows. Define a positive definite matrix $\mathbf{D}=\diag(d_{1},\dots,d_{r})$
and $\mathbf{N}=\diag(N^{\alpha_{1}},\dots,N^{\alpha_{r}})$. Form a $T\times
N$ matrix $\mathbf{A}$ whose elements are independent draws from $N(0,1)$ for
each replication. Obtain the singular value decomposition $\mathbf{A=USV}%
^{\prime}$, and set $\mathbf{F}^{0}$ the first $r$ columns of $\mathbf{U}$
multiplied by $\sqrt{T}$ and $\mathbf{B}^{0}$ the first $r$ columns of
$\mathbf{V}$ post-multiplied by $\mathbf{D}^{1/2}\mathbf{N}^{1/2}$. Then, given an invertible square matrix of order $r$, $\bH$, set
$\mathbf{F}^{\ast}=\mathbf{F}^{0}\mathbf{H}^{-1}$ and $\mathbf{B}^{\ast
}=\mathbf{B}^{0}\mathbf{H}^{\prime}$. 
We set $\mathbf{H}=\left(
\begin{smallmatrix}
1 & 1/2\\
1/2 & 2
\end{smallmatrix}\right)$.
The $t^{th}$ rows of $\mathbf{E}$ are generated as $\be_{t}=\rho_{e}\mathbf{e}_{t-1}+(1-\rho_{e}^{2})^{1/2}\boldsymbol{\Sigma}_{e}^{1/2}\boldsymbol{\xi}_{t}$ for $t=2,\dots,T$ with $\be_{1}\sim N(\bzero,\bI_N)$, where
$\boldsymbol{\xi}_{t}\sim i.i.d.N(\mathbf{0},\mathbf{I}_{N})$. 
It is set to $\rho_{e}=0.2$. 
We consider $\boldsymbol{\Sigma}%
_{e}=\sigma_{e}^{2}\mathbf{R}_{s}$, 
where $\mathbf{R}_{s}$ is the correlation matrix of $\left(  \mathbf{I}%
_{N}-\theta\mathbf{S}_{s}\right)  \left(  \mathbf{I}_{N}-\theta\mathbf{S}%
_{s}\right)  ^{\prime}$, $\mathbf{S}_{s}$ is the row normalized $s^{th}$-order
rook contiguity spatial matrix. We set $s=2$, $\sigma_{e}=0.5$ and $\theta=0.5$. 
The factor-augmented regression is generated as%
\[
y_{t+1}=\mathbf{f}_{t}^{0\prime} \bgamma^0+\mathbf{w}_{t}%
^{\prime}\boldsymbol{\beta}+\epsilon_{t+1}\text{, }t=1,\dots,T,
\]
where $\mathbf{f}_{t}^{0\prime}$ is the $t^{th}$ row vector of $\mathbf{F}%
^{0}$ and $\mathbf{w}_{t}=(w_{t,1},\dots,w_{t,p})^{\prime}$ with $w_{t,p}=1$
and
\[
w_{t,\ell}=\sigma_{w}(\rho_{fw}\mathbf{f}_{t}^{0\prime}\mathbf{1}_{r}%
r^{-1/2}+\sqrt{1-\rho_{fw}^{2}}\zeta_{t,\ell}),
\]
where $\zeta_{t,\ell}\sim i.i.d.N(0,1\mathbf{)}$, $\ell=1,\dots,p-1$, and
$\epsilon_{t+1}\sim i.i.d.N(0,\sigma_{\epsilon}^{2})$. We set
$\boldsymbol{\gamma}^0=\mathbf{1}_{r}$ and $\boldsymbol{\beta}=\bone_{p}$, so
that $\boldsymbol{\gamma}^{\ast}=\mathbf{H}\boldsymbol{\gamma}^0$.

As discussed in the theory, the correlation between $\mathbf{w}_{t}$ and
$\mathbf{f}_{t}$ affects the asymptotic bias of the estimator. We consider $\rho_{fw}=\{0,0.6\}$, 
while setting to $\sigma_{w}^{2}=1$ and $\sigma_{\epsilon}^{2}=0.5$. We choose $r=2$ and $p=2$,
$(\alpha_{1},\alpha_{2})=(1,1)$, $(1,0.8),$ $(0.8,0.6)$ with $(d_{1}%
,d_{2})=(0.05,0.2)$, $(0.2,0.2)$ and $(0.2,0.2)$, respectively. Different
values of $d_{1}$ and $d_{2}$ for the model with $\alpha_{1}=\alpha_{2}$ are
necessary to ensure the identification of the two largest eigenvalues of
$\mathbb{E(}\mathbf{x}_{t}\mathbf{x}_{t}^{\prime}\mathbb{)}$, $\lambda_{1}$
and $\lambda_{2}$. The experiments are conducted for
$(T,N)=(50,50),(100,100),(200,200)$ with 1,000 replications.

Supposing $\{\mathbf{x}_{t},\mathbf{w}_{t},\mathbf{y}_{t}\}$ are observable in
practice, $\mathbf{F}^{0}$ is estimated by PC using $\mathbf{X}$, which is the
first $r$ largest eigenvectors of $\mathbf{XX}^{\prime}/T$ multiplied by
$\sqrt{T}$. The PC estimator is denoted by $\mathbf{\hat{F}}$. The factor
augmented model is estimated by regressing $y_{t+1}$ on $\mathbf{\hat{z}}%
_{t}=(\mathbf{\hat{f}}_{t}^{\prime},\mathbf{w}_{t}^{\prime})^{\prime}$, which
gives the estimates $\boldsymbol{\hat{\delta}}=(\boldsymbol{\hat{\gamma}%
}^{\prime},\boldsymbol{\hat{\beta}}^{\prime})^{\prime}$. 
Using the different rotation matrices, we compute the average (i.e. bias), the standard deviation, and two-sided t-test at the 5\% level for $\hat{\bdelta}-{\bdelta}_{\hat{\bH}}$, $\hat{\bdelta}-{\bdelta}_{\hat{\bH}_q}$ and $\hat{\bdelta}-{\bdelta}^0$.

In addition, we consider associated bias-corrected estimators. Let the bias estimates of $\hat{\bdelta}-{\bdelta}_{\hat{\bH}}$ and $\hat{\bdelta}-{\bdelta}_{\hat{\bH}_q}$ be
\begin{align*}
\hat{\bkappa}_{\bdelta^*}
=-(T^{-1}\mathbf{\hat{Z}}^{\prime}%
\mathbf{\hat{Z}})^{-1}
\begin{pmatrix}
\hat{\bG}+ \hat{{\bar{\bG}}}\\
T^{-1}\mathbf{W}^{\prime}\mathbf{\hat{F}\hat{G}}
\end{pmatrix}
\boldsymbol{\hat{\gamma}},\quad
\hat{\bar{\bkappa}}_{\bdelta^*}
=(T^{-1}\mathbf{\hat{Z}}^{\prime
}\mathbf{\hat{Z}})^{-1}\left(
\begin{array}
[c]{c}%
\mathbf{0}\\
T^{-1}\mathbf{W}^{\prime}\mathbf{\hat{F}}\hat{{\bar{\bG}}}
\end{array}
\right)\boldsymbol{\hat{\gamma}}
\end{align*}
with
\begin{align*}
\mathbf{\hat{G}=\mathbf{\hat{B}}^{\prime}}\boldsymbol{\hat{\Sigma}}_{e}\mathbf{\mathbf{\hat{B}}(\hat{B}}^{\prime}\mathbf{\hat{B})}^{-2},\quad
\hat{{\bar{\bG}}}=(\mathbf{\mathbf{\hat{B}}^{\prime}\mathbf{\hat{B}}%
})^{-1}\mathbf{\mathbf{\hat{B}}^{\prime}}\boldsymbol{\hat{\Sigma}}%
_{e}\mathbf{\mathbf{\hat{B}}}(\mathbf{\hat{B}}^{\prime}\mathbf{\hat{B}})^{-1},
\end{align*}
where $\boldsymbol{\hat{\Sigma}}_{e}$ is the POET estimator of \cite{RunyuEtAl2024}, which extends \cite{FanEtAl2013} for the estimation in WF models. The associated bias corrected estimators are defined by 
\begin{align}\label{bcHhats}
\hat{\bdelta}_{bc\hat{\bH}}=\hat{\bdelta}-\hat{{\bkappa}}_{\bdelta^*}, \quad
\boldsymbol{\hat{\delta}}_{bc\hat{\bH}_{q}}=\boldsymbol{\hat{\delta}}-\hat{\bar{\bkappa}}_{\bdelta^*},
\end{align}
and the bias, standard deviation, and size of the t-test at the 5\% level of $\hat{\bdelta}_{bc\hat{\bH}} - \bdelta_{\hat{\bH}}$ and $\hat{\bdelta}_{bc\hat{\bH}_q} - \bdelta_{\hat{\bH}_q}$ are reported.
The panel-split jackknife bias corrected estimator is computed as
\begin{align}\label{bcjkest_R}
\boldsymbol{\hat{\delta}}_{bcjk}=2\boldsymbol{\hat{\delta}}-\boldsymbol{\hat
{\delta}}_{jk}\ \text{with }\boldsymbol{\hat{\delta}}_{jk}=R^{-1}%
\sum\nolimits_{s=1}^{R}(\boldsymbol{\hat{\delta}}_{\mathcal{N}_{1}^{(s)}%
}+\boldsymbol{\hat{\delta}}_{\mathcal{N}_{2}^{(s)}})/2   
\end{align}
where $\boldsymbol{\hat{\delta}}_{\mathcal{N}_{j}^{(s)}}$ is obtained
regressing $\mathbf{y}$ on ($\mathbf{\hat{F}}_{\mathcal{N}_{j}^{(s)}%
},\mathbf{W}$), where $\mathbf{\hat{F}}_{\mathcal{N}_{j}^{(s)}}$ is the PC
factor extracted from $\mathbf{X}_{\mathcal{N}_{j}^{(s)}}$, where
$\mathbf{X}_{\mathcal{N}_{j}^{(s)}}=\{\mathbf{x}_{i\in\mathcal{N}%
_{j}^{(s)}}\}$, for $j=1,2$, with $\mathcal{N}_{1}^{(s)}$ is the
first half and $\mathcal{N}_{2}^{(s)}$ is the second half of $\mathbf{X}%
^{(s)}$, whose $N$ columns are randomly re-ordered over $s=1,\dots,R$. 
This randomization is to avoid potentially biased information on factors in $\mathcal{N}_j$. 
When $N$ is odd, $\mathcal{N}_{1}^{(s)}$ and $\mathcal{N}_{2}^{(s)}$ are chosen to
contain one common index. The order and the sign of the columns of $\mathbf{\hat
{F}}_{\mathcal{N}_{j}^{(s)}}$ are determined in line with those of $\mathbf{\hat{F}%
}$, based on the correlation between the pair $(\mathbf{\hat{F}}%
_{\mathcal{N}_{j}^{(s)}},\mathbf{\hat{F}})$, for each of $j=1,2$. We have chosen
$R=100$. The bias, standard deviation and size of the t-test at the 5\% level of $\hat{\bdelta}_{bcjk} - \bdelta^0$ are reported.

Using the same sample, we investigate similar statistics for the model with
the extracted factors from $\mathbf{M}_{w}\mathbf{X}$. Specifically, we define
$\mathbf{F}_{w}^{\ast}=\mathbf{M}_{w}\mathbf{F}^{\ast}$\ where $\mathbf{M}%
_{w}=\mathbf{I}_{T}-\mathbf{P}_{w}$ with $\mathbf{P}_{w}=\mathbf{W(W}^{\prime
}\mathbf{W)}^{-1}\mathbf{W}^{\prime}$, so that $\mathbf{F}_{w}^{0}%
=\mathbf{F}_{w}^{\ast}\mathbf{H}_{w}$ and $\mathbf{B}_{w}^{0}=\mathbf{B}%
^{\ast}\mathbf{H}_{w}^{\prime-1}$ with $\mathbf{H}_{w}=\mathbf{L}%
_{w}\mathbf{V}_{w}^{-1/2}\boldsymbol{\Pi}$, where $\mathbf{V}_{w}%
=\mathbf{L}_{w}^{\prime}(T^{-1}\mathbf{F}_{w}^{\ast\prime}\mathbf{F}_{w}%
^{\ast})\mathbf{L}_{w}$, $\mathbf{L}_{w}$ is the $r\times r$ eigenvector
matrix of $(\mathbf{\mathbf{B}^{\ast\prime}\mathbf{B}}^{\ast})(T^{-1}%
\mathbf{F}_{w}^{\ast\prime}\mathbf{F}_{w}^{\ast})$, $\mathbf{\Pi}$ is a
diagonal matrix with elements either $-1$ or $1$, which makes the diagonal of
$\mathbf{H}_{w}$ positive. The PC estimator $\mathbf{\hat{F}}_{w}$ is
the $\sqrt{T}$ times $r$ eigenvectors corresponding to the $r$ largest
eigenvalues of $\mathbf{M}_{w}\mathbf{XX}^{\prime}\mathbf{M}_{w}/T$ and
$\mathbf{\hat{B}}_{w}=\mathbf{X}^{\prime}\mathbf{\hat{F}}_{w}/T$. The
augmented model of the $T\times1$ vector is re-written as
$\mathbf{y}   =\mathbf{F}^{\ast}\boldsymbol{\gamma}^{\ast}+\mathbf{W}
\boldsymbol{\beta}+\boldsymbol{\epsilon} =\mathbf{F}_{w}^{*}\boldsymbol{\gamma}^{*}+\mathbf{W}%
\boldsymbol{\beta}_{w}+\boldsymbol{\epsilon}$, 
where $\boldsymbol{\beta}_{w}=\mathbf{(W}^{\prime}\mathbf{W)}^{-1}%
\mathbf{W}^{\prime}\mathbf{F}^{\ast}\boldsymbol{\gamma}^{\ast}%
+\boldsymbol{\beta}$.
Regression of $y$ on $\mathbf{z}%
_{w,t}=(\mathbf{\hat{f}}_{w}^{\prime},\mathbf{w}_{t}^{\prime})^{\prime}$ gives
$\boldsymbol{\hat{\delta}}_{w}=(\boldsymbol{\hat{\gamma}}_{w}^{\prime
},\boldsymbol{\hat{\beta}}_{w}^{\prime})^{\prime}$, which is the estimator of
$\boldsymbol{\delta}_{w}^{0}=(\boldsymbol{\gamma}_w^{0\prime}%
,\boldsymbol{\beta}_{w}^{\prime})^{\prime}$. We investigate analogous
counterpart statistics, $\boldsymbol{\delta}_{\hat{\bH}_w}$, $\boldsymbol{\delta
}_{\hat{\bH}_{q,w}}$, $\boldsymbol{\hat{\delta}}_{bc\hat{\bH}_w}$,
$\boldsymbol{\hat{\delta}}_{bc\hat{\bH}_{q,w}}$, $\boldsymbol{\hat{\delta}%
}_{bcjk,w}$ with respect to the relevant ``parameters''.

\subsection{Results}

\subsubsection{``Parameters'' }
Before discussing the performance of the estimators, we would like to draw our attention to the ``parameters" we are estimating. In the literature, including \cite{BaiNg2006} and \citet{GoncalvesPerron2014,gonccalves2020bootstrapping}, the LS estimator $\hat{\bdelta}$ is considered to estimate $\bdelta_{\hat{\bH}}=(\bgamma^{* \prime}\hat{\bH}^{-1} , \bbeta')'$, which is noise-dependent via $\hat{\bH}$. On the other hand, $\bdelta^0$ is a pure function of the signals. Normalizing $\bdelta^0 :=\mathbf{1}$, mean and standard deviation of the second elements of $\bgamma^0 :=\bgamma^{* \prime}{\bH}^{-1}$, $\bgamma_{\hat{\bH}_q} :=\bgamma^{* \prime}\hat{\bH}_q^{-1}$ and $\bgamma_{\hat{\bH}} :=\bgamma^{* \prime}\hat{\bH}^{-1}$ over the replications are plotted over $N=T=50,100,200$ and $(\alpha_1,\alpha_2)=(1.0,1.0),(1.0,0.8),(0.8,0.6)$ in Figures \ref{fig:mean_gamma_hats} and \ref{fig:sd_gamma_hats}, respectively.
\begin{figure}[!htb]	
\centering
\begin{subfigure}[b]{0.32\textwidth}
\centering
\includegraphics[width=\textwidth]{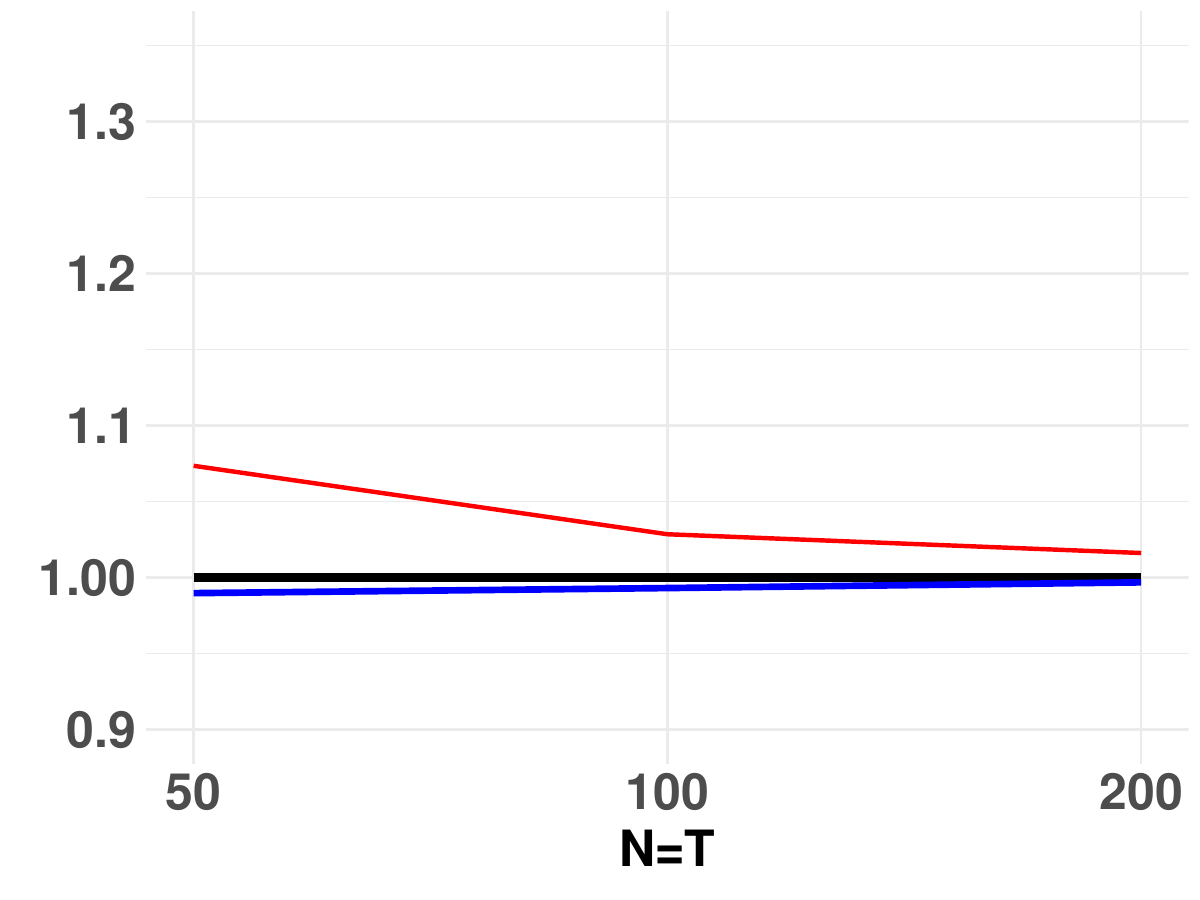}
\caption{$\alpha_2 = 1.0$}
\label{fig:bias_f2_06_10}
\end{subfigure}
\hfill
\begin{subfigure}[b]{0.32\textwidth}
\centering
\includegraphics[width=\textwidth]{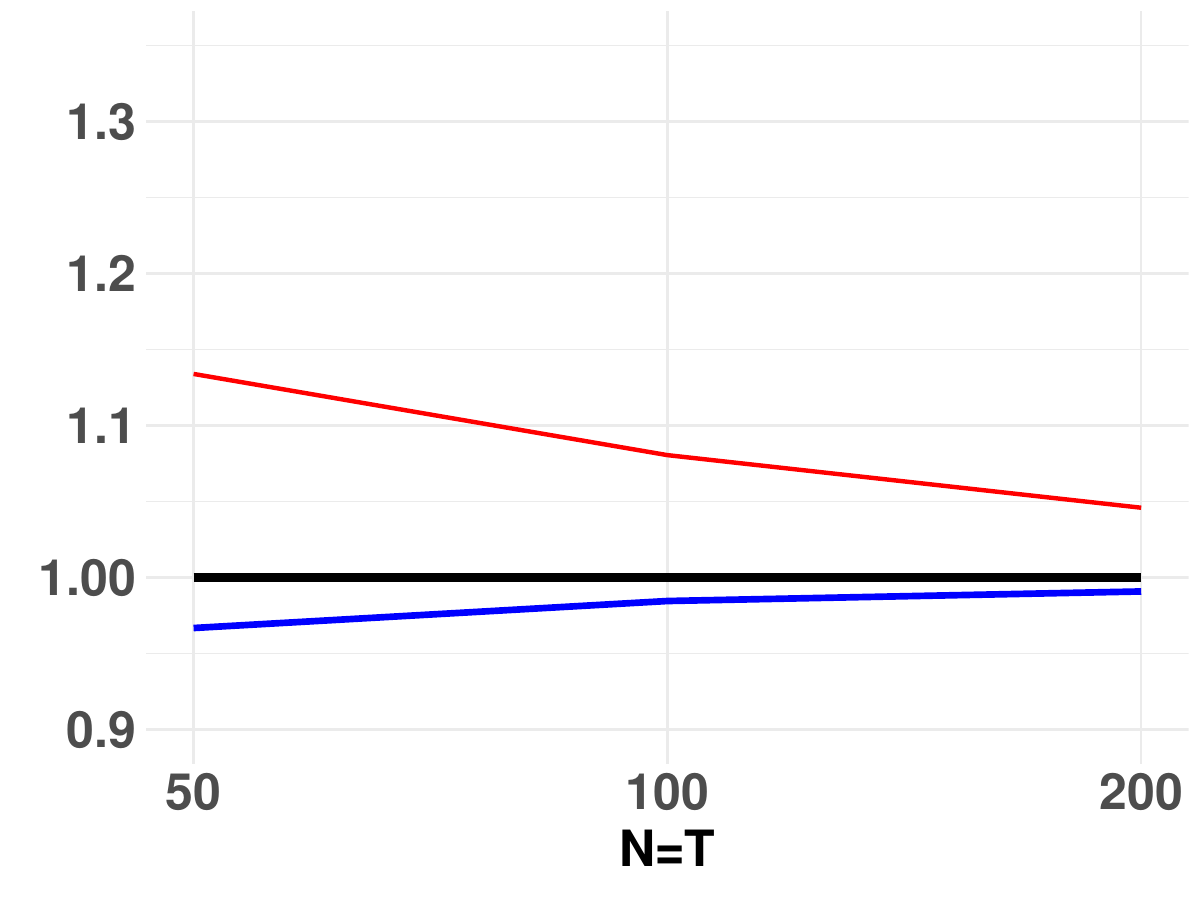}
\caption{$\alpha_2 = 0.8$}
\label{fig:bias_f2_06_08}
\end{subfigure}
\hfill
\begin{subfigure}[b]{0.32\textwidth}
\centering
\includegraphics[width=\textwidth]{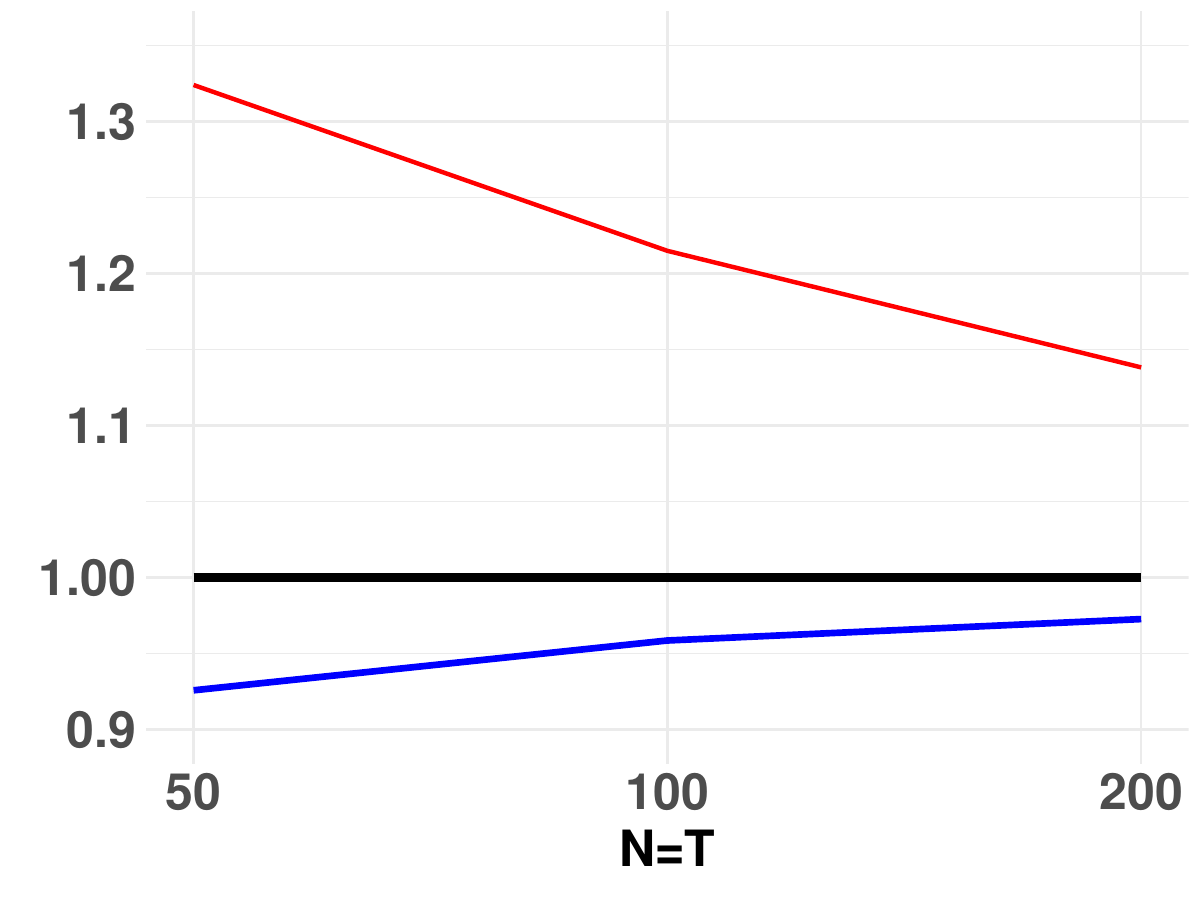}
\caption{$\alpha_2 = 0.6$}
\label{fig:bias_f2_06_06}
\end{subfigure}
\includegraphics[width=0.33\textwidth]{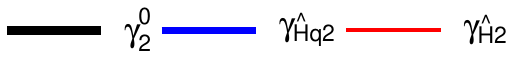}
\caption{Mean of ``parameters" $\gamma_2^0$, 
${\gamma}_{\hat{\bH}_{q},2}$ and
${\gamma}_{\hat{\bH},2}$ over the replications}

\label{fig:mean_gamma_hats}
\end{figure}


\begin{figure}[!htb]	
\centering
\begin{subfigure}[b]{0.32\textwidth}
\centering
\includegraphics[width=\textwidth]{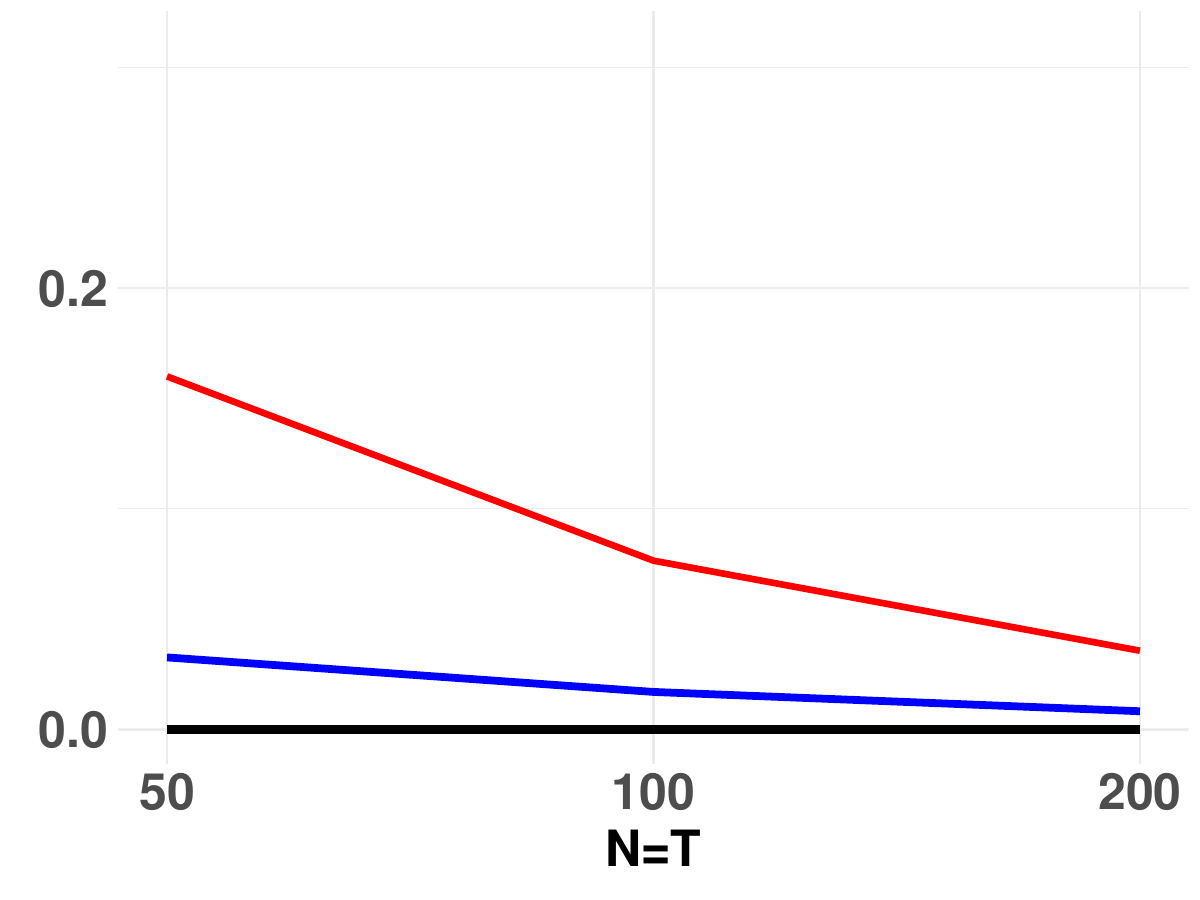}
\caption{$\alpha_2 = 1.0$}
\label{fig:bias_f2_06_10}
\end{subfigure}
\hfill
\begin{subfigure}[b]{0.32\textwidth}
\centering
\includegraphics[width=\textwidth]{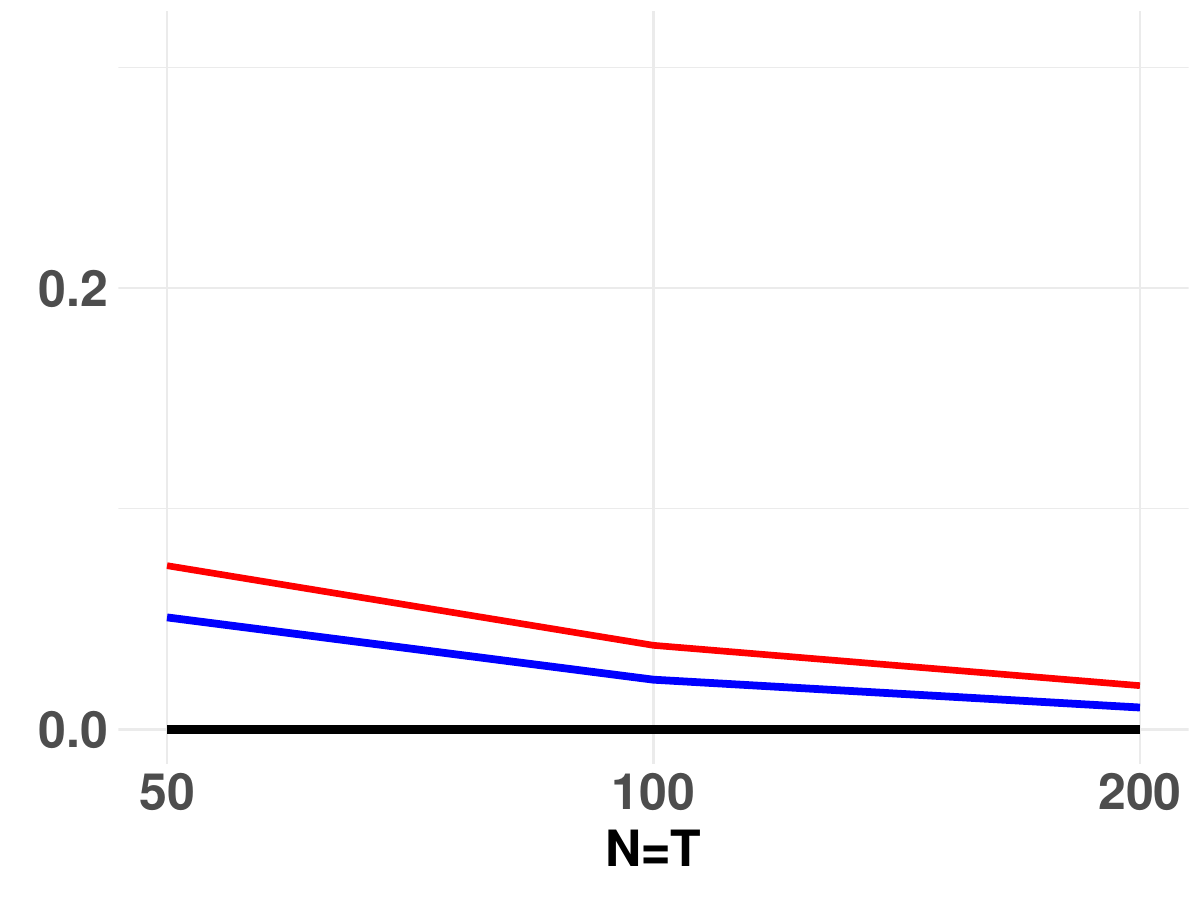}
\caption{$\alpha_2 = 0.8$}
\label{fig:bias_f2_06_08}
\end{subfigure}
\hfill
\begin{subfigure}[b]{0.32\textwidth}
\centering
\includegraphics[width=\textwidth]{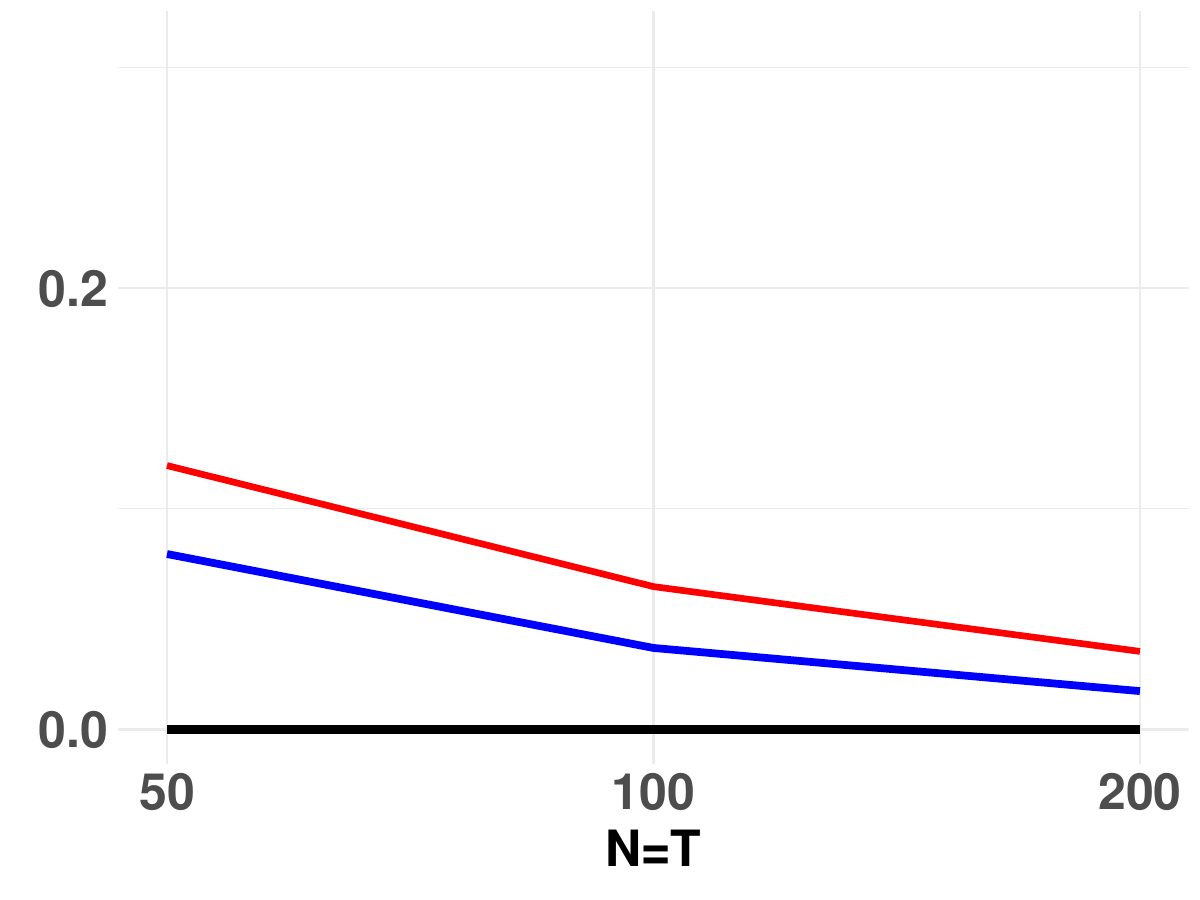}
\caption{$\alpha_2 = 0.6$}
\label{fig:bias_f2_06_06}
\end{subfigure}
\includegraphics[width=0.33\textwidth]{newimages/legend_Parameters_SigE3_rhoe2_bias_rhowf0_a1_f2.pdf}
\caption{Standard deviation of ``parameters" $\gamma_2^0$, 
${\gamma}_{\hat{\bH}_{q},2}$ and
${\gamma}_{\hat{\bH},2}$ over the replications}

\label{fig:sd_gamma_hats}
\end{figure}

As can be seen, the smaller the sample size and the weaker the model, the more $\bgamma_{\hat{\bH}}$ deviates upwards and $\bgamma_{\hat{\bH}_q}$ deviates downwards from $\bgamma^0$ and their variations tend to increase.
It is a visual confirmation that the values of the ``parameters'' $\bgamma^0$, $\bgamma_{\hat{\bH}}$ and $\bgamma_{\hat{\bH}_q}$ can be very different even for relatively large sample sizes despite their asymptotic equivalence.

Furthermore, even when all the elements in $\bgamma^0$ are equal, it is unlikely that a similar property holds for $\bgamma_{\hat{\bH}}$ and $\bgamma_{\hat{\bH}_q}$. For example, for $N=T=50$ and the model with $(\alpha_1,\alpha_2)=(1,1)$, in the above experiment, the average values of the $r\times1$ vectors $\bgamma^0$, $\bgamma_{\hat{\bH}}$ and $\bgamma_{\hat{\bH}_q}$ are $(1.00,1.00)'$, $(1.27,1.07)'$ and $(0.94,0.99)'$, respectively. Although all considered rotation matrices share the common probability limit, it does not seem plausible to treat $\bgamma^0$, $\bgamma_{\hat{\bH}}$ and $\bgamma_{\hat{\bH}_q}$ indifferently in practice. 

Therefore, researchers may want to be clear about which ``parameter'' is estimated by $\hat{\bgamma}$ in their analysis. Our preferred parameter is $\bgamma^0$, because it is the coefficient on the latent factor $\bff_t^0$ that the PC estimator $\hat{\bff}_t$ consistently estimates. 

In what follows we will discuss the bias of the LS estimator $\hat{\bdelta}$ and its bias-corrected versions relative to $\bdelta^0$, $\bdelta_{\hat{\bH}}$ and $\bdelta_{\hat{\bH}_q}$.

\subsubsection{Coefficient on the second factor, $\gamma_2$}

In this section we examine the finite sample performance of the estimated coefficient on the second factor $\hat{f}_{2,t}$, $\hat{\gamma}_2$ and its bias-corrected versions, relative to $\gamma_2^0$, $\gamma_{\hat{\bH}_q,2}$ and $\gamma_{\hat{\bH},2}$. Specifically, we report the bias, standard deviation (s.d.), and a t-test at the 5\% level of $\hat{\gamma}_2 - \gamma_2^0$, $\hat{\gamma}_2 - \gamma_{\hat{\bH}_q,2}$, $\hat{\gamma}_2 - \gamma_{\hat{\bH},2}$ and their bias-corrected versions, $\hat{\gamma}_{bcjk,2} - \gamma_2^0$, $\hat{\gamma}_{bc\hat{\bH}_q,2} - \gamma_{\hat{\bH}_q,2}$, $\hat{\gamma}_{bc\hat{\bH},2} - \gamma_{\hat{\bH},2}$, respectively. The bias-corrected versions are shown as dashed lines in the figures.

\begin{figure}[!htb]	
\centering
\begin{subfigure}[b]{0.32\textwidth}
\centering
\includegraphics[width=\textwidth]{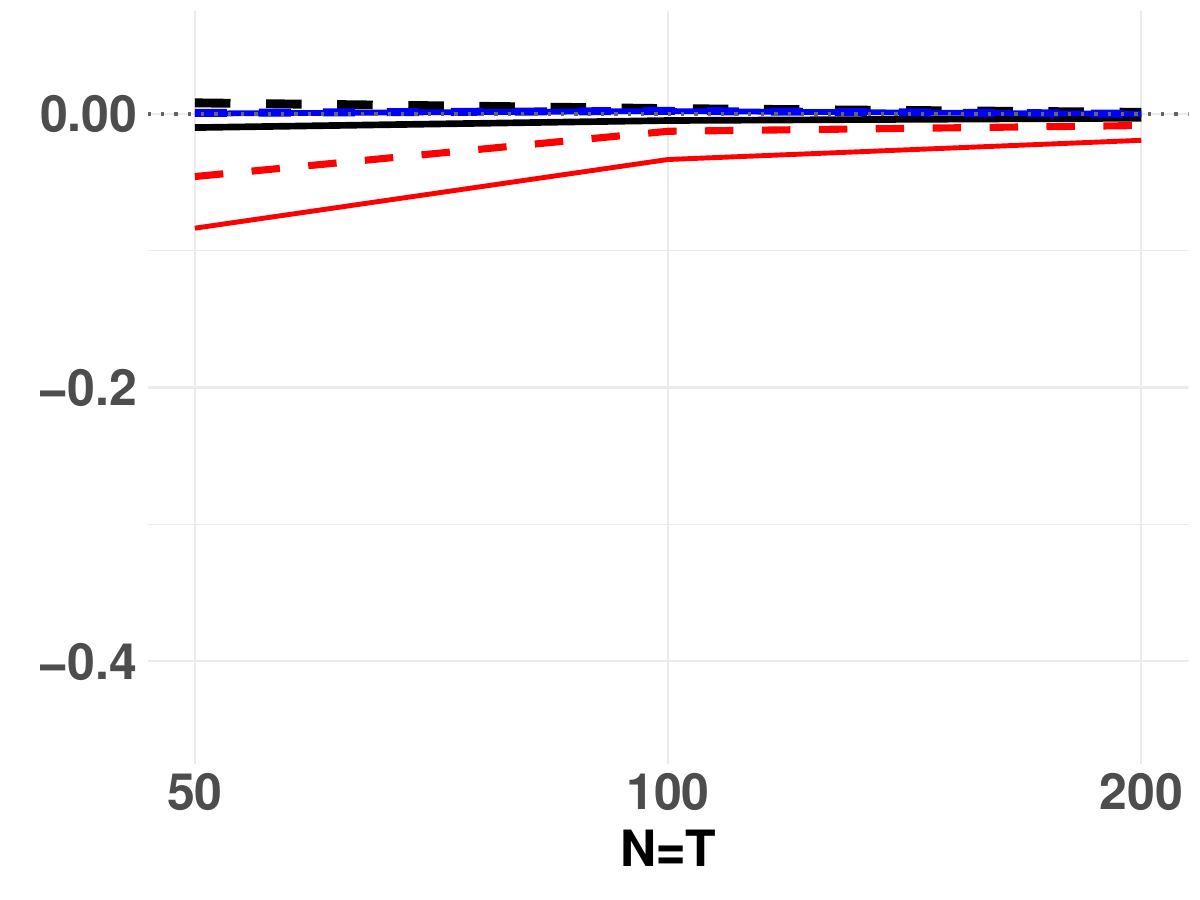}
\caption{$\rho_{wf}=0.0,\alpha_2 = 1.0$}
\label{fig:bias_f2_00_10}
\end{subfigure}
\hfill
\begin{subfigure}[b]{0.32\textwidth}
\centering
\includegraphics[width=\textwidth]{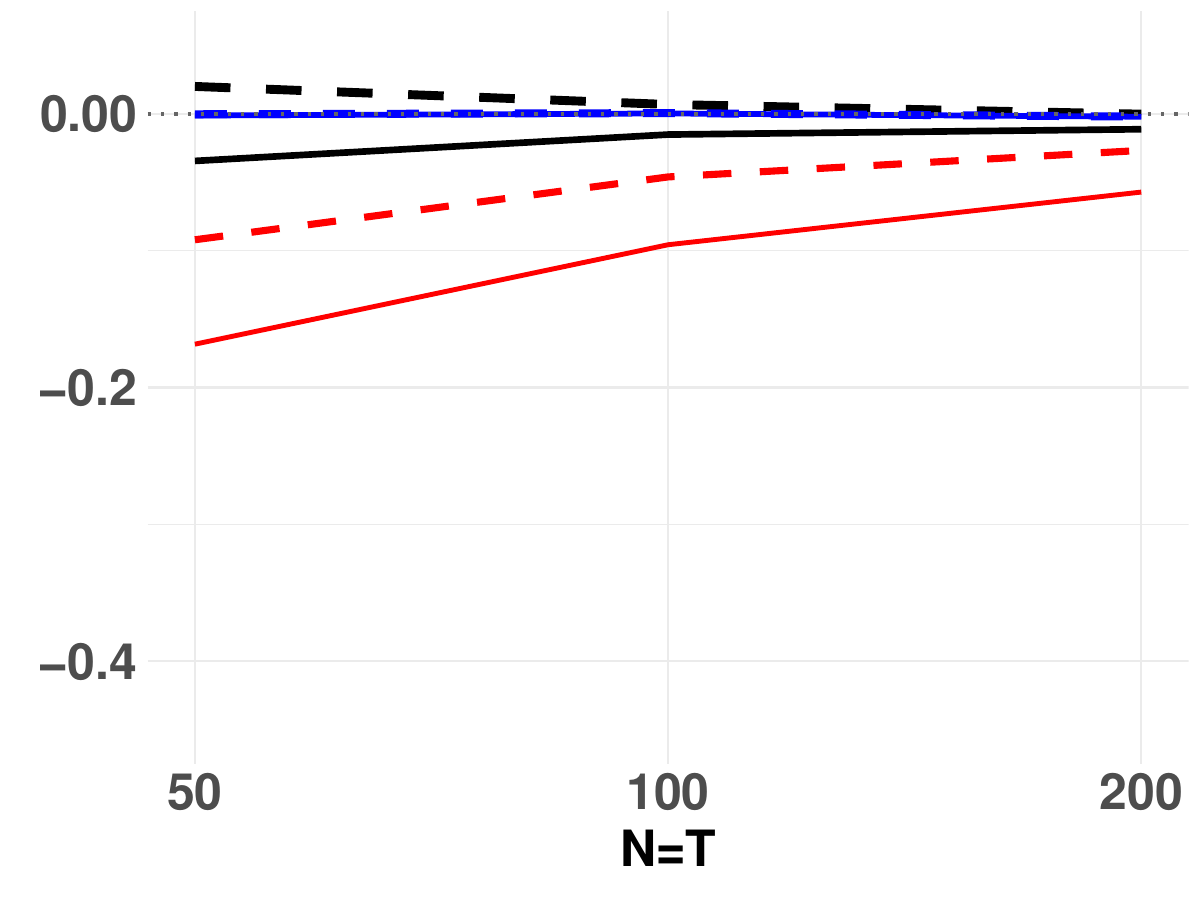}
\caption{$\rho_{wf}=0.0,\alpha_2 = 0.8$}
\label{fig:bias_f2_00_08}
\end{subfigure}
\hfill
\begin{subfigure}[b]{0.32\textwidth}
\centering
\includegraphics[width=\textwidth]{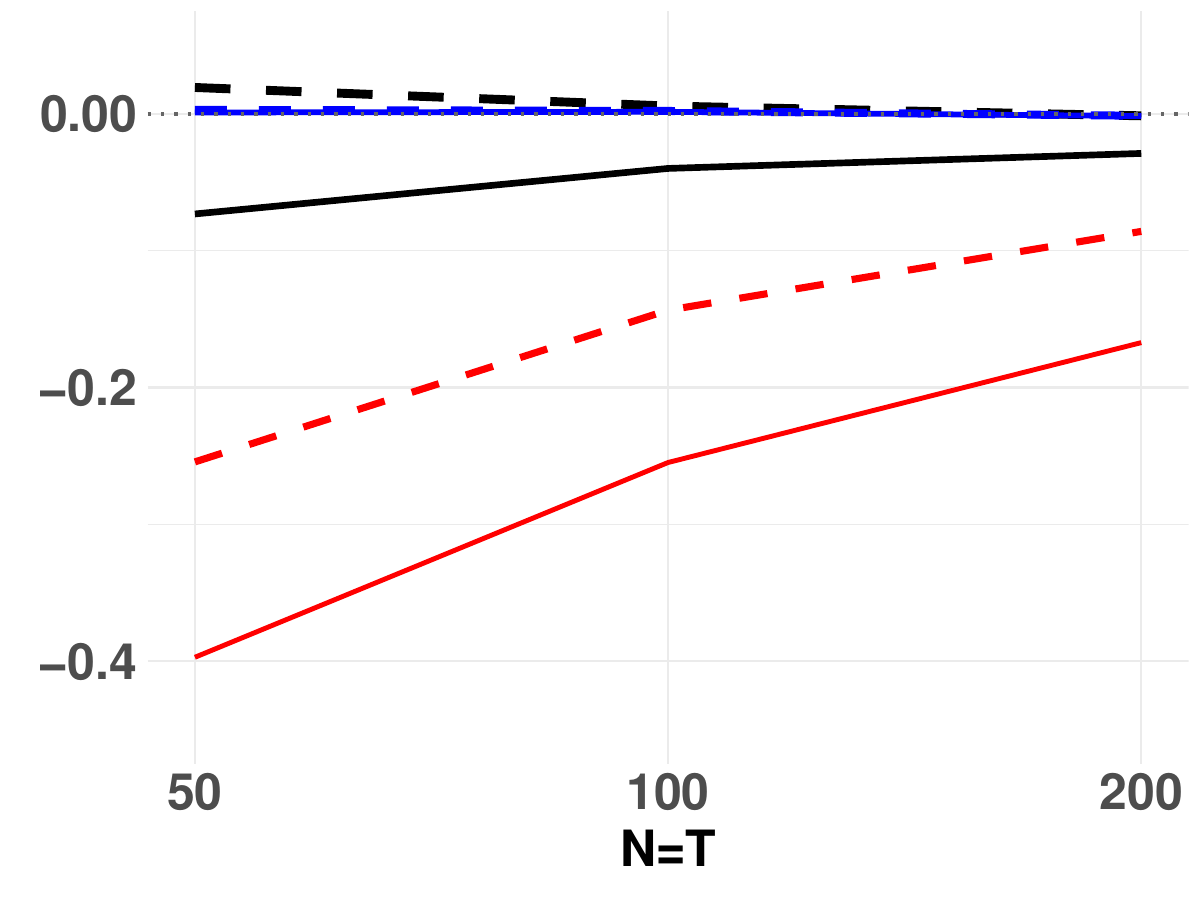}
\caption{$\rho_{wf}=0.0,\alpha_2 = 0.6$}
\label{fig:bias_f2_00_06}
\end{subfigure}

\centering
\begin{subfigure}[b]{0.32\textwidth}
\centering
\includegraphics[width=\textwidth]{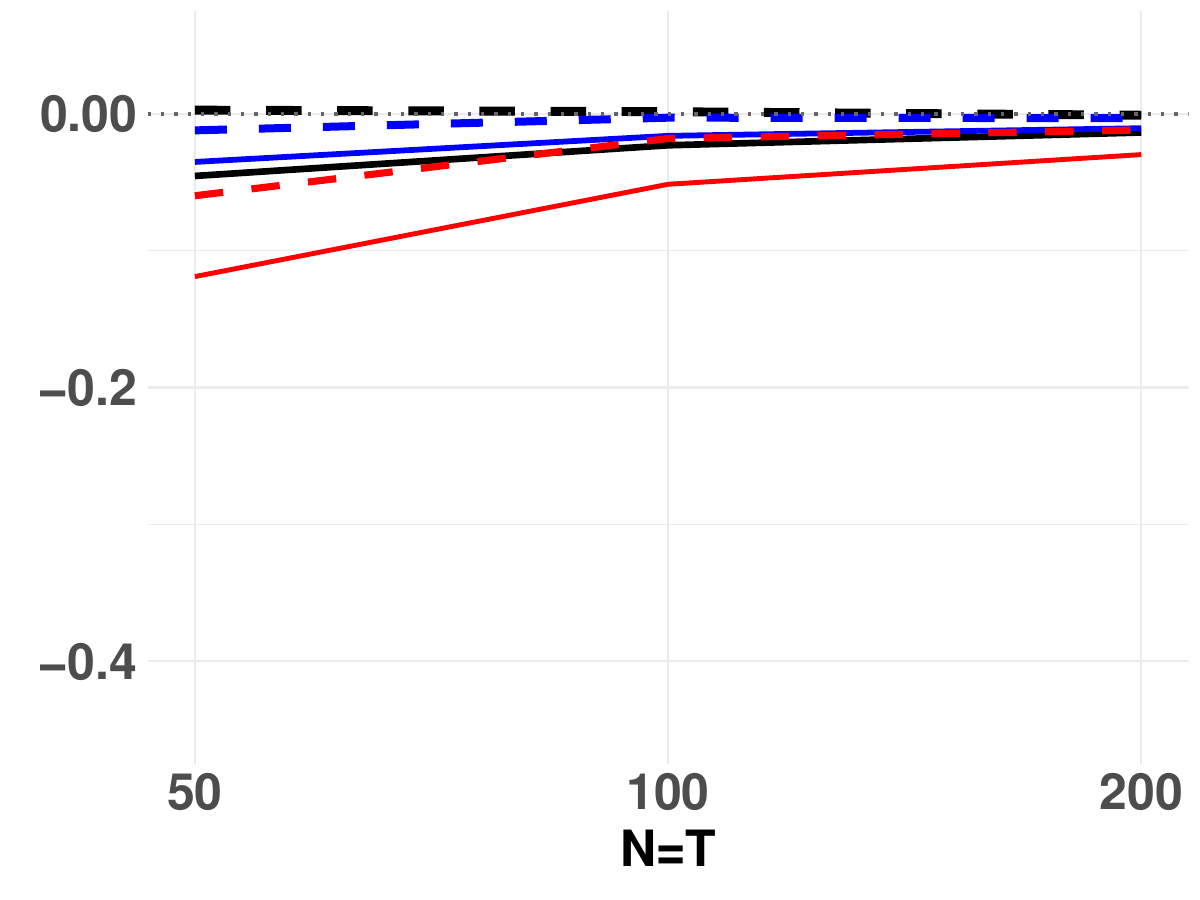}
\caption{$\rho_{wf}=0.6,\alpha_2 = 1.0$}
\label{fig:bias_f2_06_10}
\end{subfigure}
\hfill
\begin{subfigure}[b]{0.32\textwidth}
\centering
\includegraphics[width=\textwidth]{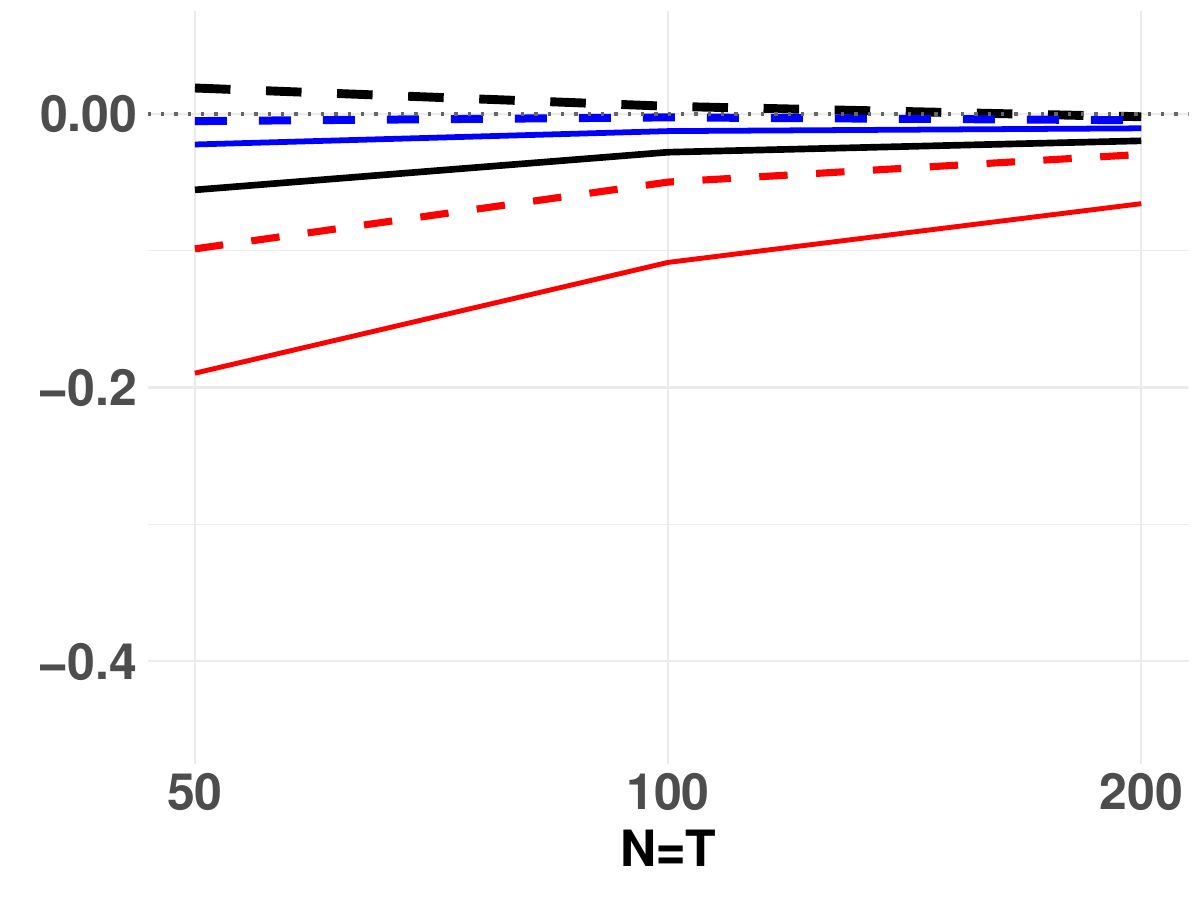}
\caption{$\rho_{wf}=0.6,\alpha_2 = 0.8$}
\label{fig:bias_f2_06_08}
\end{subfigure}
\hfill
\begin{subfigure}[b]{0.32\textwidth}
\centering
\includegraphics[width=\textwidth]{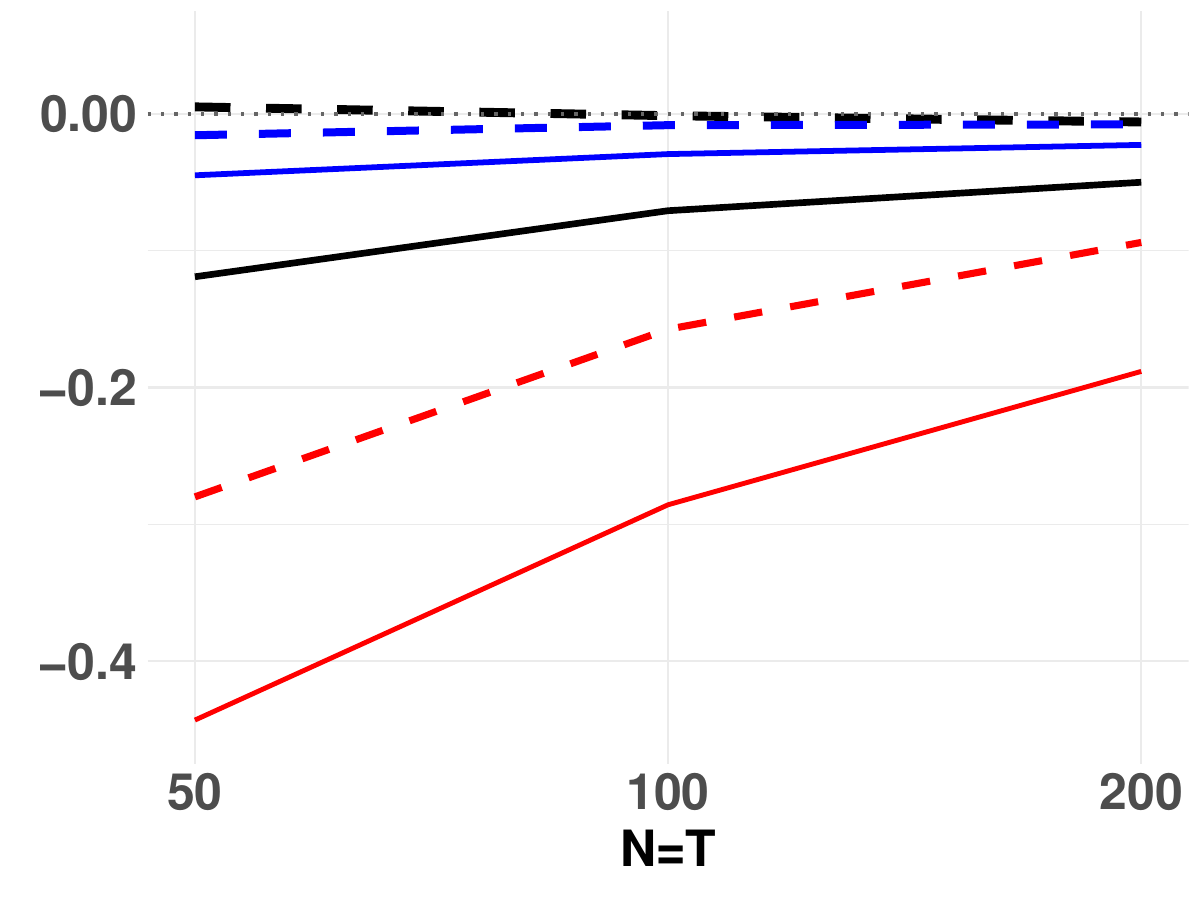}
\caption{$\rho_{wf}=0.6,\alpha_2 = 0.6$}
\label{fig:bias_f2_06_06}
\end{subfigure}

\centering
\includegraphics[width=0.50\textwidth]{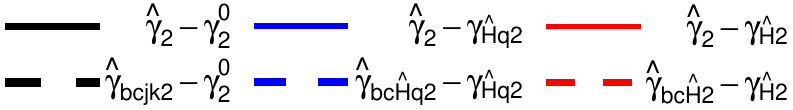}
\caption{Bias of $\hat{\gamma}_2$ and its bias corrected versions for cross and serially correlated $e_{t,i}$}
\label{fig:bias.f2}

\end{figure}

The average bias over the replications is reported in Figure \ref{fig:bias.f2}. The bias $\hat{\gamma}_2 - \gamma_{\hat{\bH},2}$ is always negative and the largest of all in magnitude, and the bias becomes more significant for the smaller sample sizes and weaker factors. With the bias-corrected version, $\hat{\gamma}_{bc\hat{\bH},2} - \gamma_{\hat{\bH},2}$, always has a smaller bias than without bias-correction, but still other estimators have smaller biases in magnitude.
In contrast, $\hat{\gamma}_2 - \gamma_{\hat{\bH}_q,2}$ has virtually no bias when $w_t$ and $\bff_t^*$ are uncorrelated (i.e. $\rho_{wf} = 0$). It has very small bias when $\rho_{wf} = 0.6$, but the bias-corrected version successfully reduces the bias; see $\hat{\gamma}_{bc\hat{\bH}_q,2} - \gamma_{\hat{\bH}_q,2}$ in the figure.

The most reasonable parameter that is considered to be estimated by $\hat{\gamma}_2$ is the parameter of the pure signals, $\gamma_2^0$, as we have argued. The bias $\hat{\gamma}_2 - \gamma_2^0$ is moderately negatively biased which gets worse for smaller sample sizes and weaker factors. In contrast, the proposed jackknife bias-correction reduces the bias very successfully; see $\hat{\gamma}_{bcjk,2} - \gamma_2^0$ in the figure.

\begin{figure}[!htb]	
\centering
\begin{subfigure}[b]{0.32\textwidth}
\centering
\includegraphics[width=\textwidth]{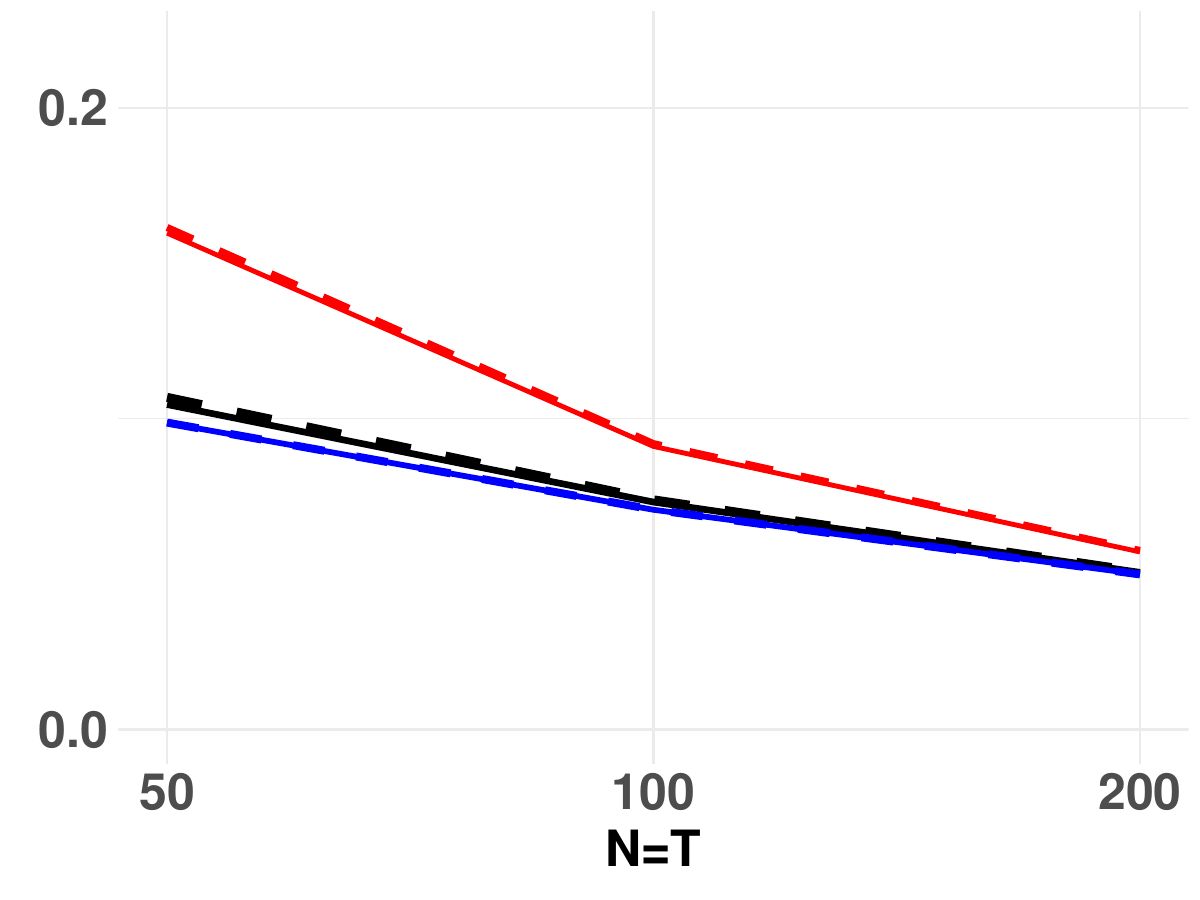}
\caption{$\rho_{wf}=0.0,\alpha_2 = 1.0$}
\label{fig:sd_f2_00_10}
\end{subfigure}
\hfill
\begin{subfigure}[b]{0.32\textwidth}
\centering
\includegraphics[width=\textwidth]{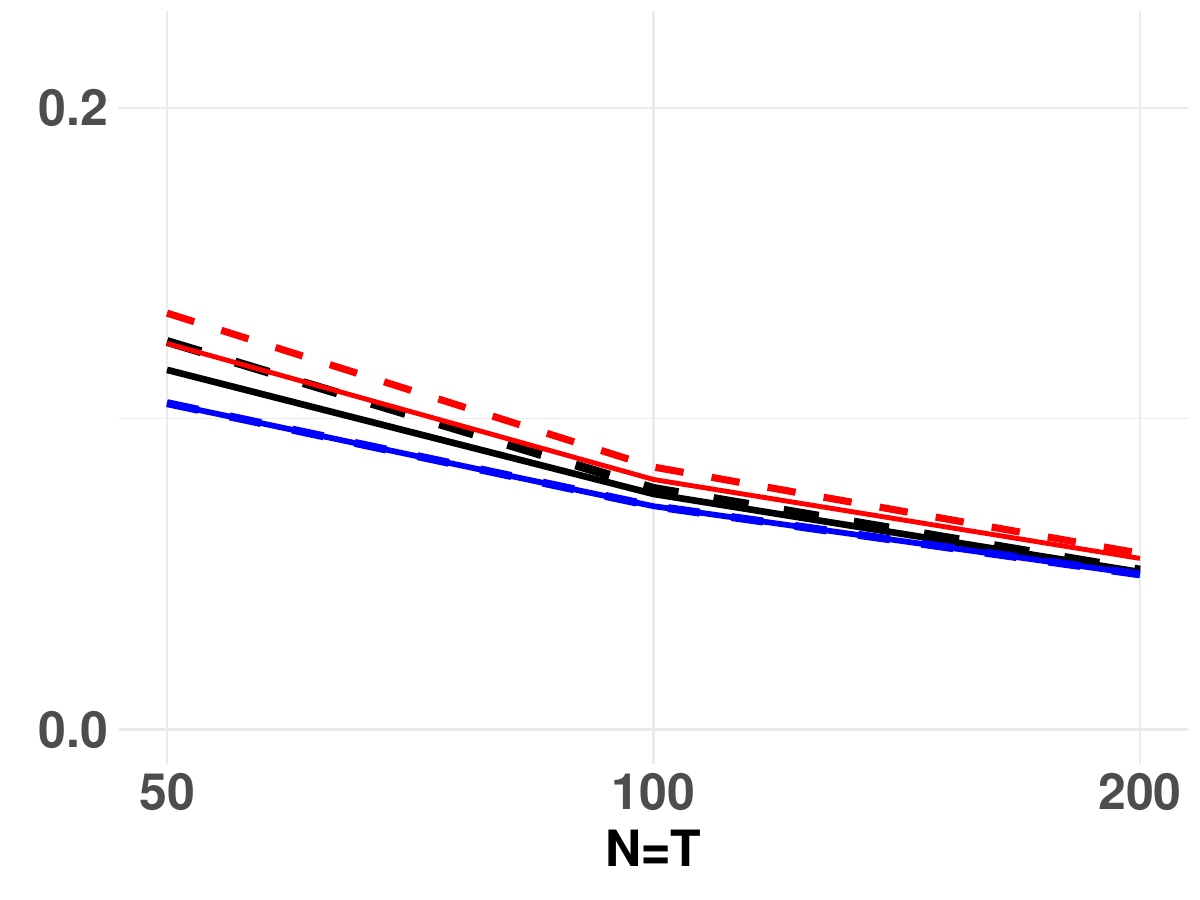}
\caption{$\rho_{wf}=0.0,\alpha_2 = 0.8$}
\label{fig:sd_f2_00_08}
\end{subfigure}
\hfill
\begin{subfigure}[b]{0.32\textwidth}
\centering
\includegraphics[width=\textwidth]{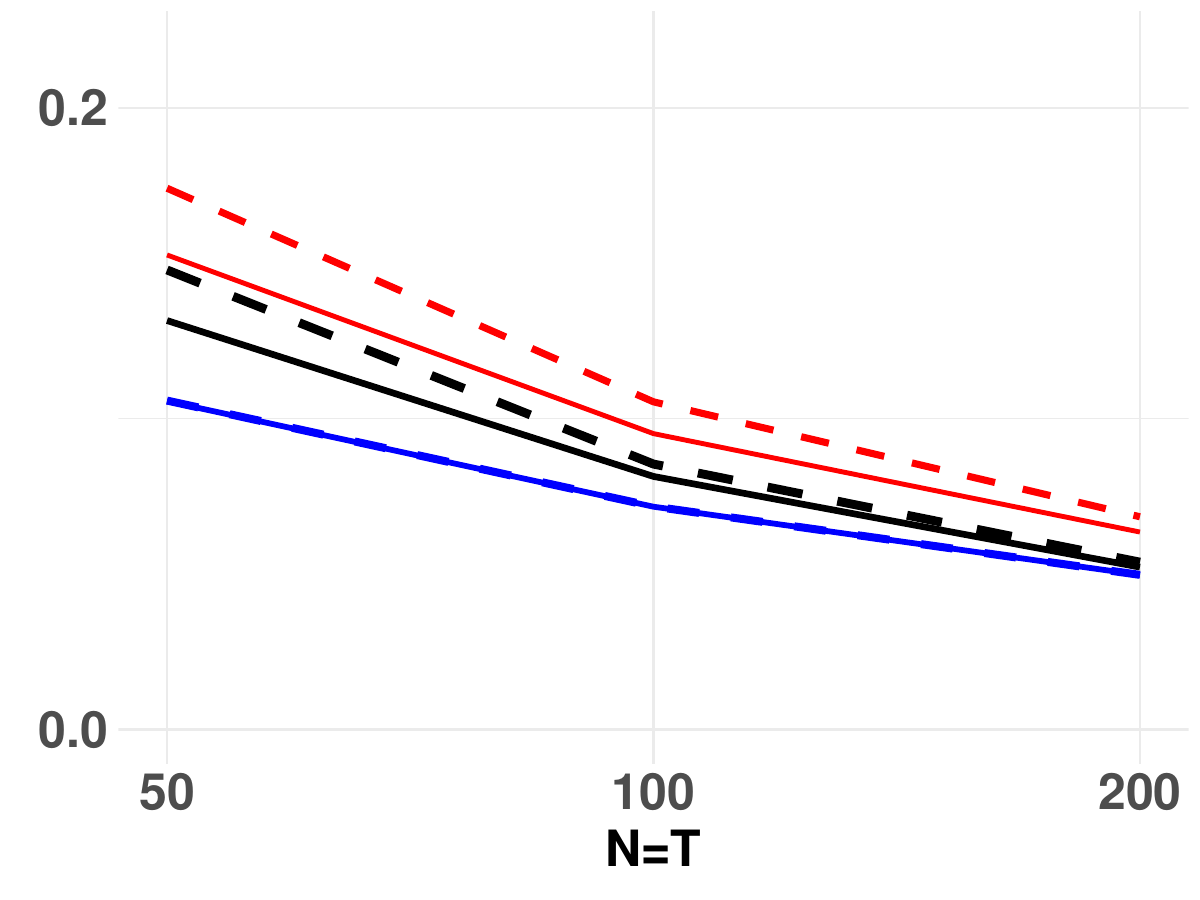}
\caption{$\rho_{wf}=0.0,\alpha_2 = 0.6$}
\label{fig:sd_f2_00_06}
\end{subfigure}

\centering
\begin{subfigure}[b]{0.32\textwidth}
\centering
\includegraphics[width=\textwidth]{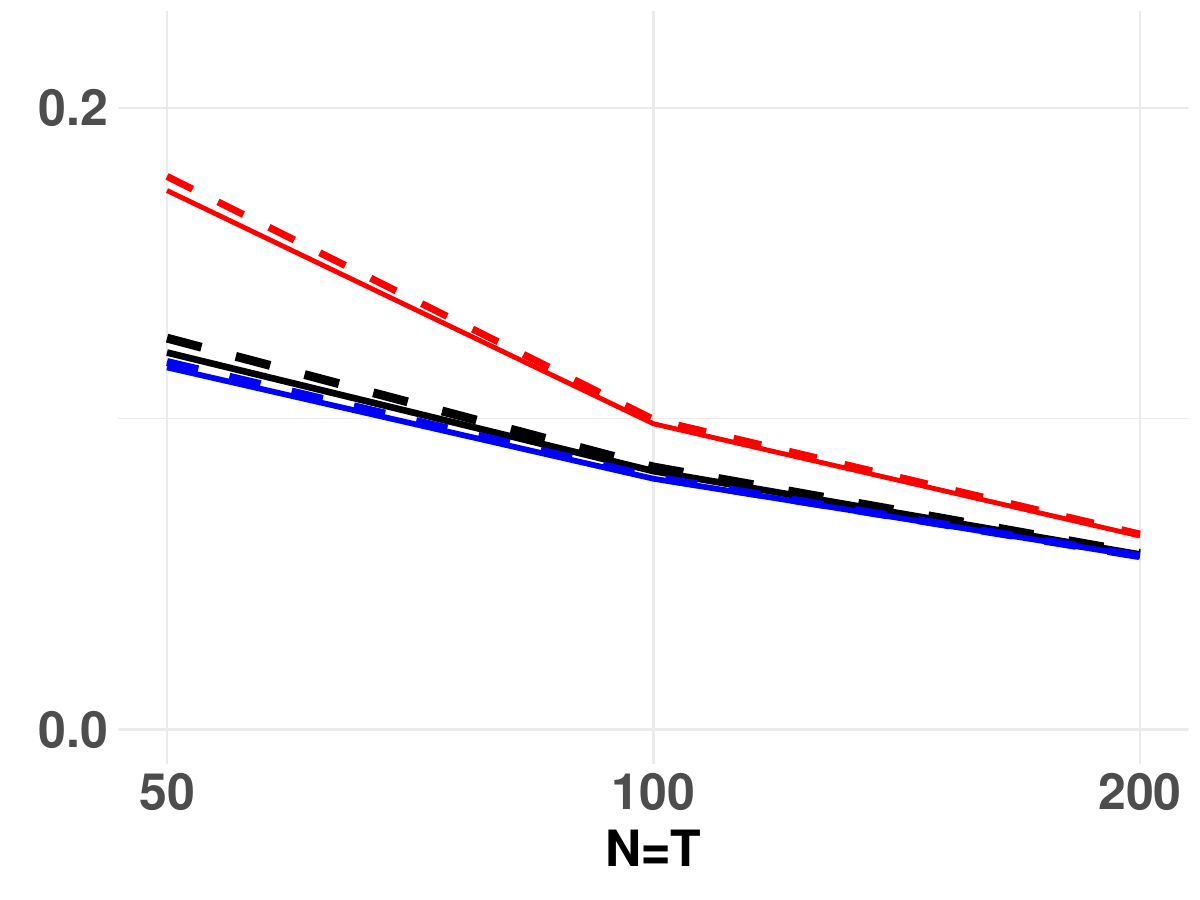}
\caption{$\rho_{wf}=0.6,\alpha_2 = 1.0$}
\label{fig:sd_f2_06_10}
\end{subfigure}
\hfill
\begin{subfigure}[b]{0.32\textwidth}
\centering
\includegraphics[width=\textwidth]{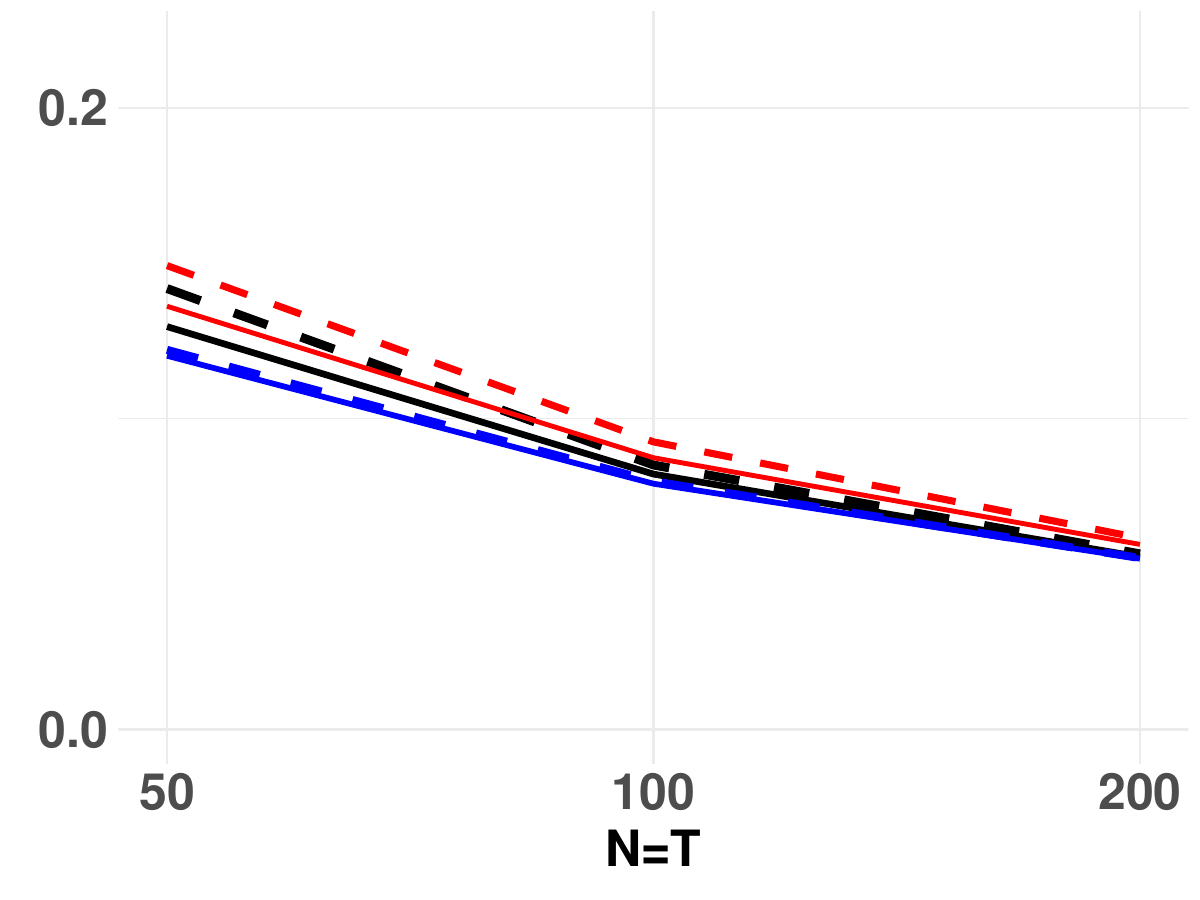}
\caption{$\rho_{wf}=0.6,\alpha_2 = 0.8$}
\label{fig:sd_f2_06_08}
\end{subfigure}
\hfill
\begin{subfigure}[b]{0.32\textwidth}
\centering
\includegraphics[width=\textwidth]{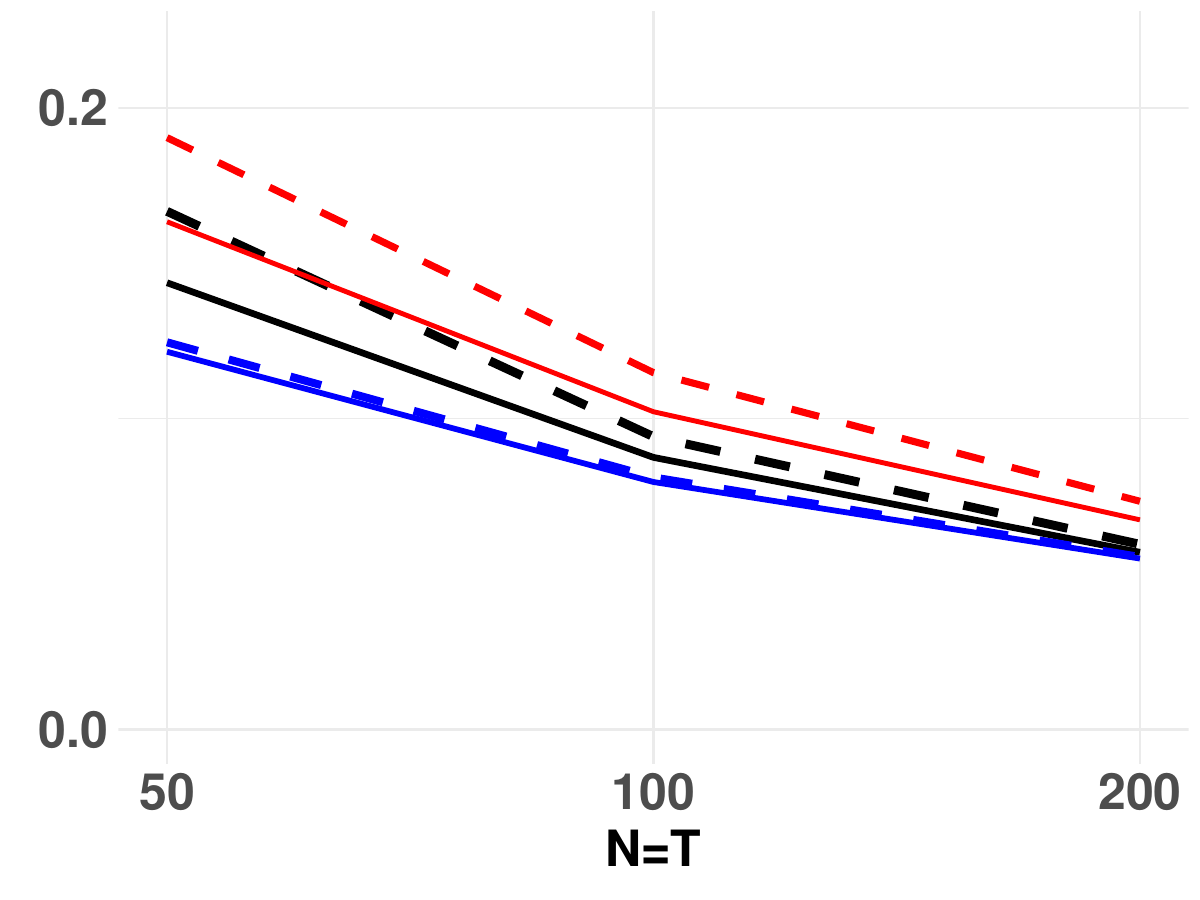}
\caption{$\rho_{wf}=0.6,\alpha_2 = 0.6$}
\label{fig:sd_f2_06_06}
\end{subfigure}

\centering
\includegraphics[width=0.50\textwidth]{newimages/legend_SigE3_rhoe2_bias_rhowf0_a1_f2.pdf}
\caption{Standard deviation of $\hat{\gamma}_2$ and its bias corrected versions for cross and serially correlated $e_{t,i}$}
\label{fig:sd.f2}

\end{figure}

\begin{figure}[!htb]	
\centering
\begin{subfigure}[b]{0.32\textwidth}
\centering
\includegraphics[width=\textwidth]{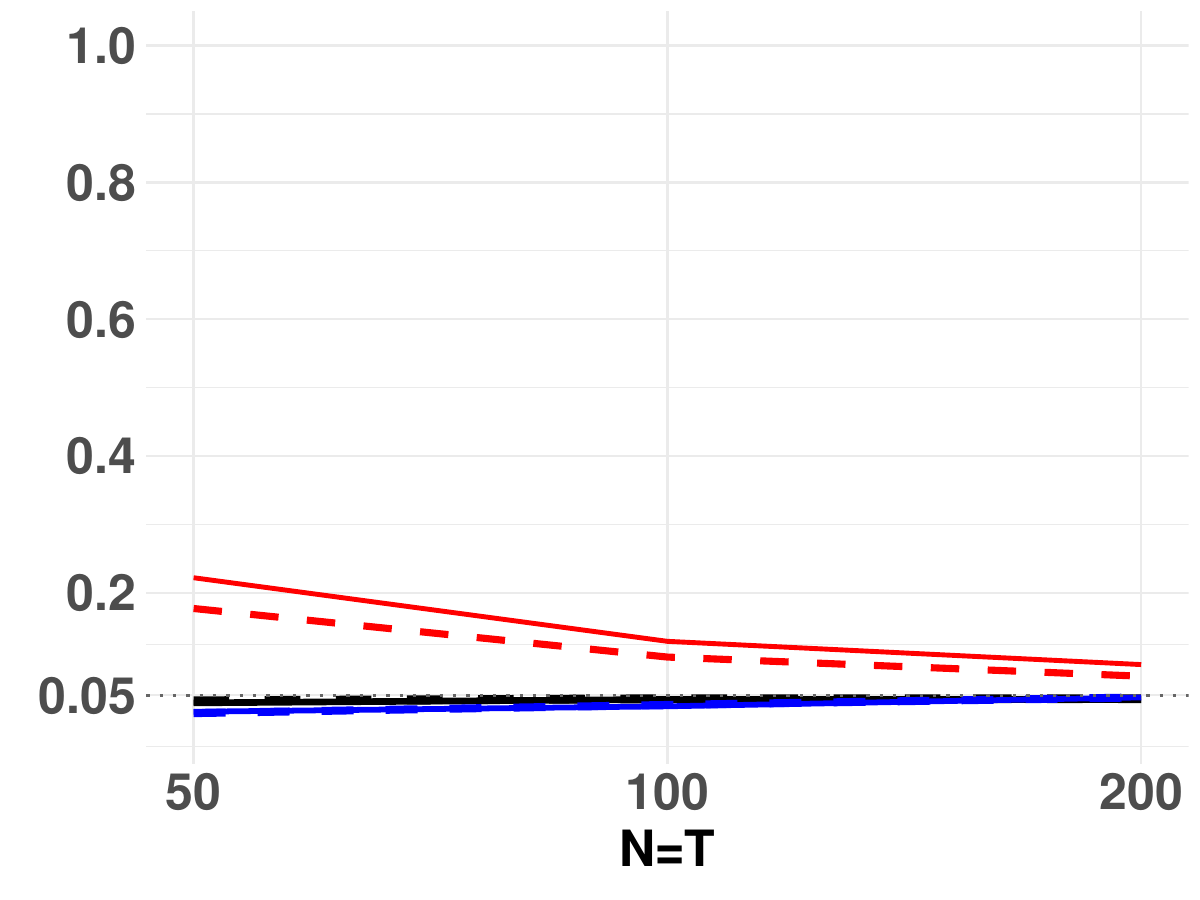}
\caption{$\rho_{wf}=0.0,\alpha_2 = 1.0$}
\label{fig:test_f2_00_10}
\end{subfigure}
\hfill
\begin{subfigure}[b]{0.32\textwidth}
\centering
\includegraphics[width=\textwidth]{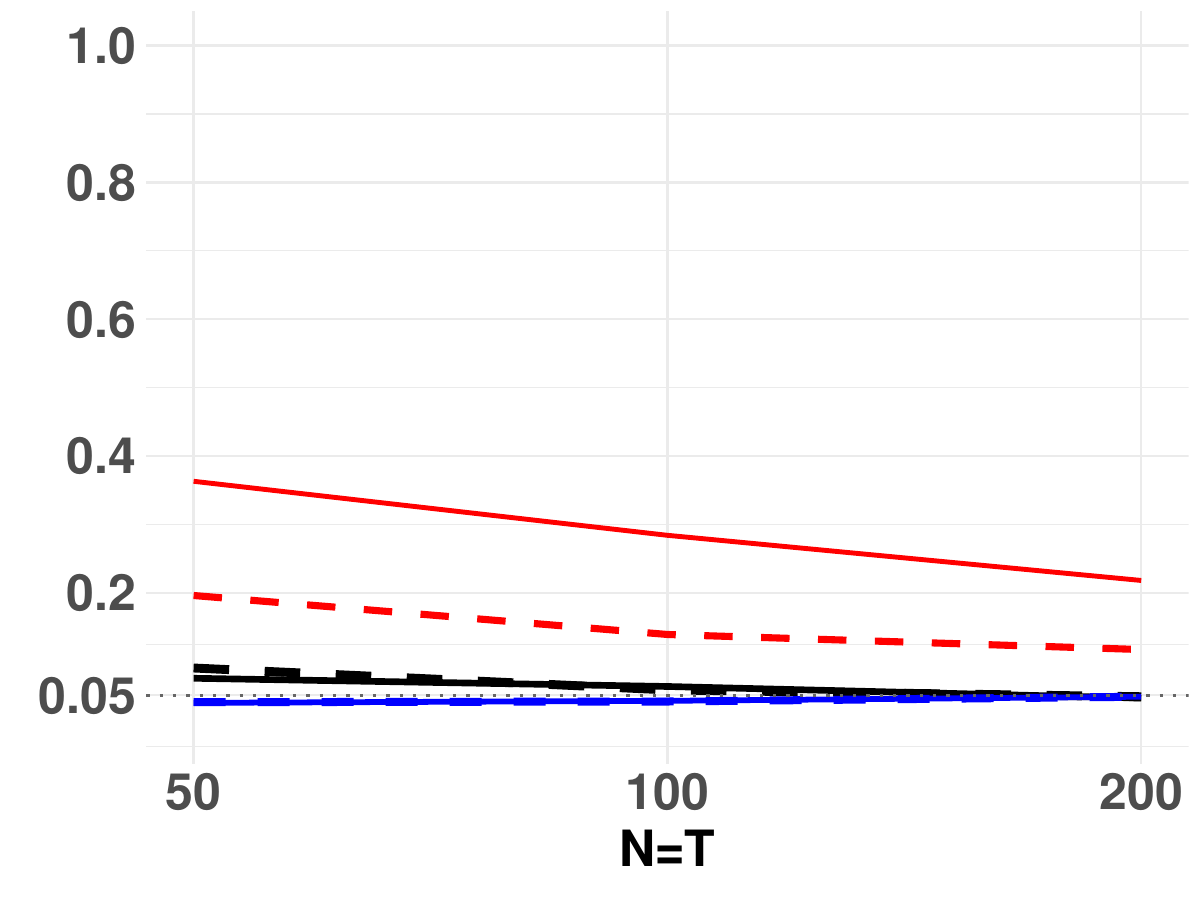}
\caption{$\rho_{wf}=0.0,\alpha_2 = 0.8$}
\label{fig:test_f2_00_08}
\end{subfigure}
\hfill
\begin{subfigure}[b]{0.32\textwidth}
\centering
\includegraphics[width=\textwidth]{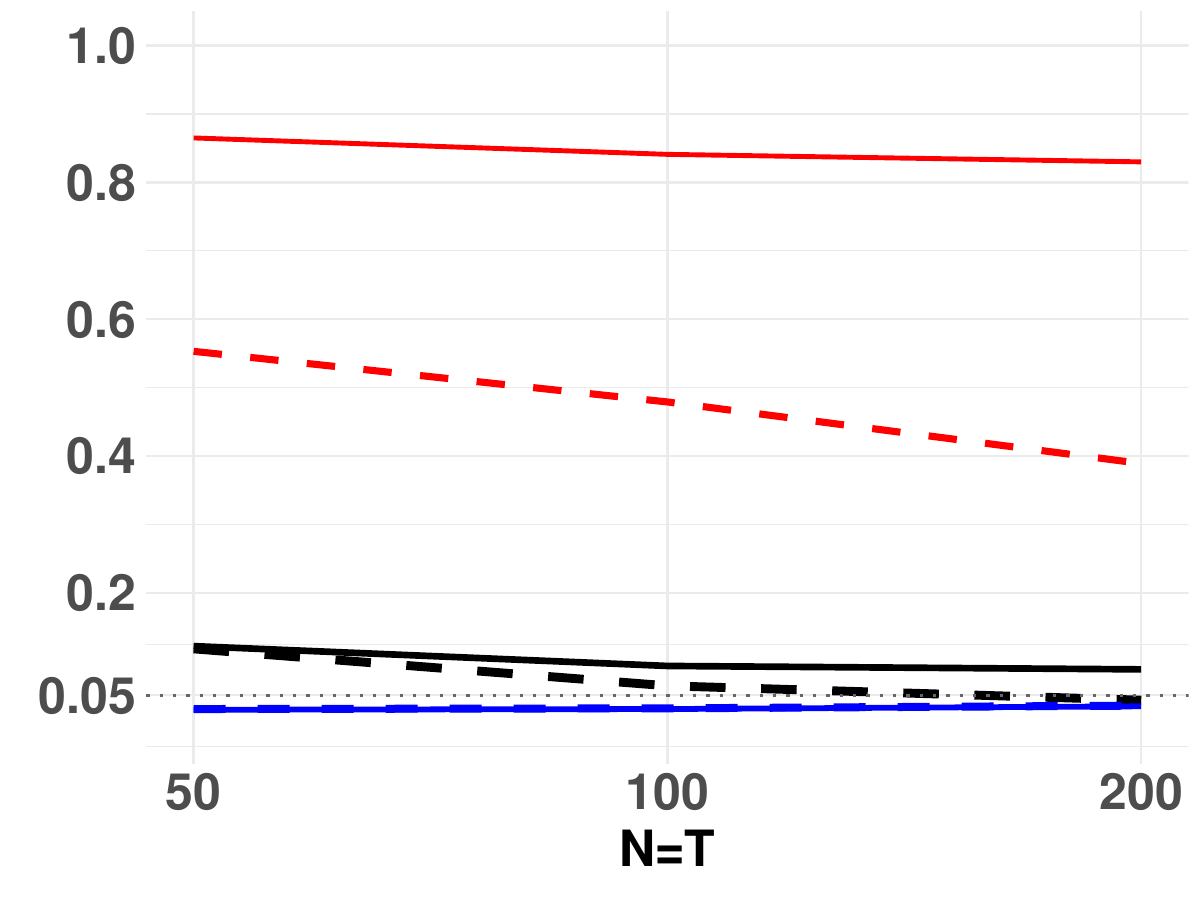}
\caption{$\rho_{wf}=0.0,\alpha_2 = 0.6$}
\label{fig:test_f2_00_06}
\end{subfigure}

\centering
\begin{subfigure}[b]{0.32\textwidth}
\centering
\includegraphics[width=\textwidth]{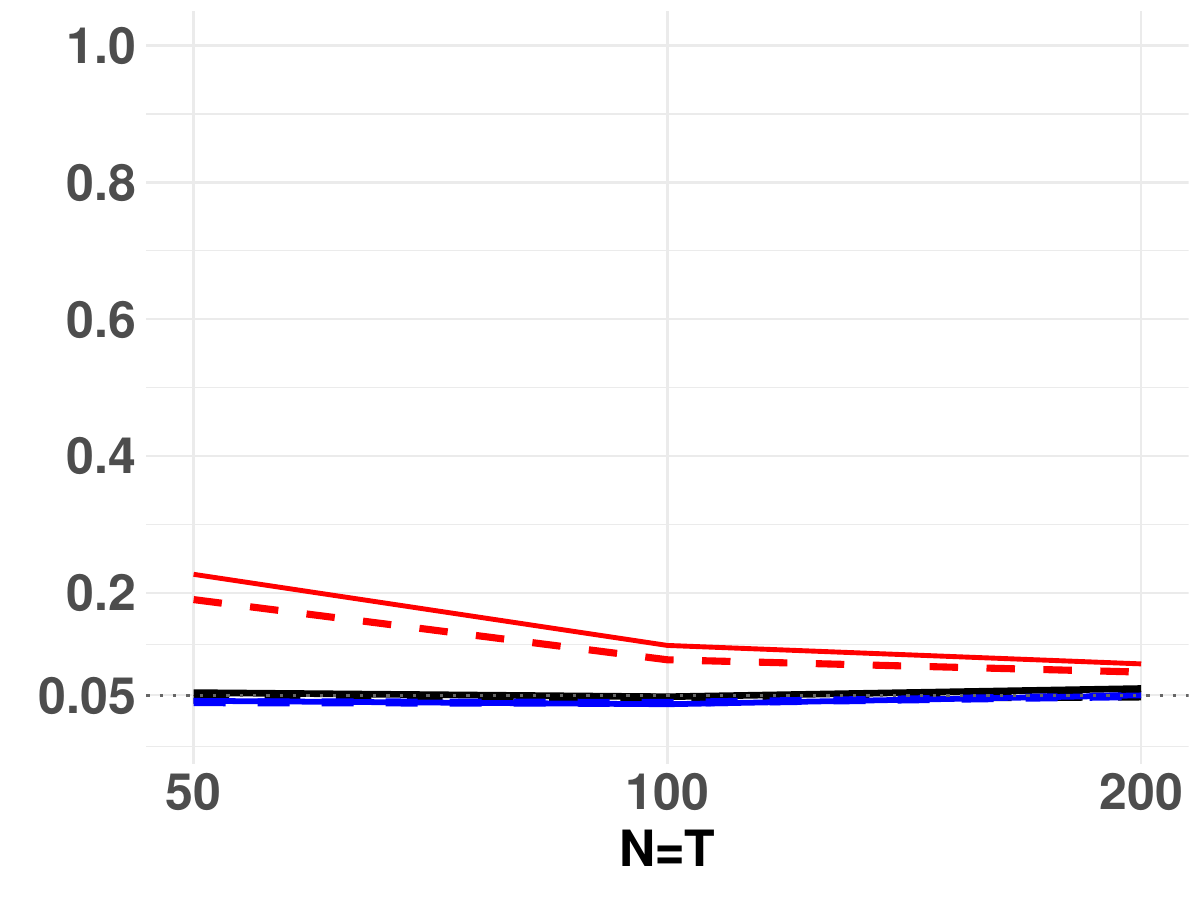}
\caption{$\rho_{wf}=0.6,\alpha_2 = 1.0$}
\label{fig:test_f2_06_10}
\end{subfigure}
\hfill
\begin{subfigure}[b]{0.32\textwidth}
\centering
\includegraphics[width=\textwidth]{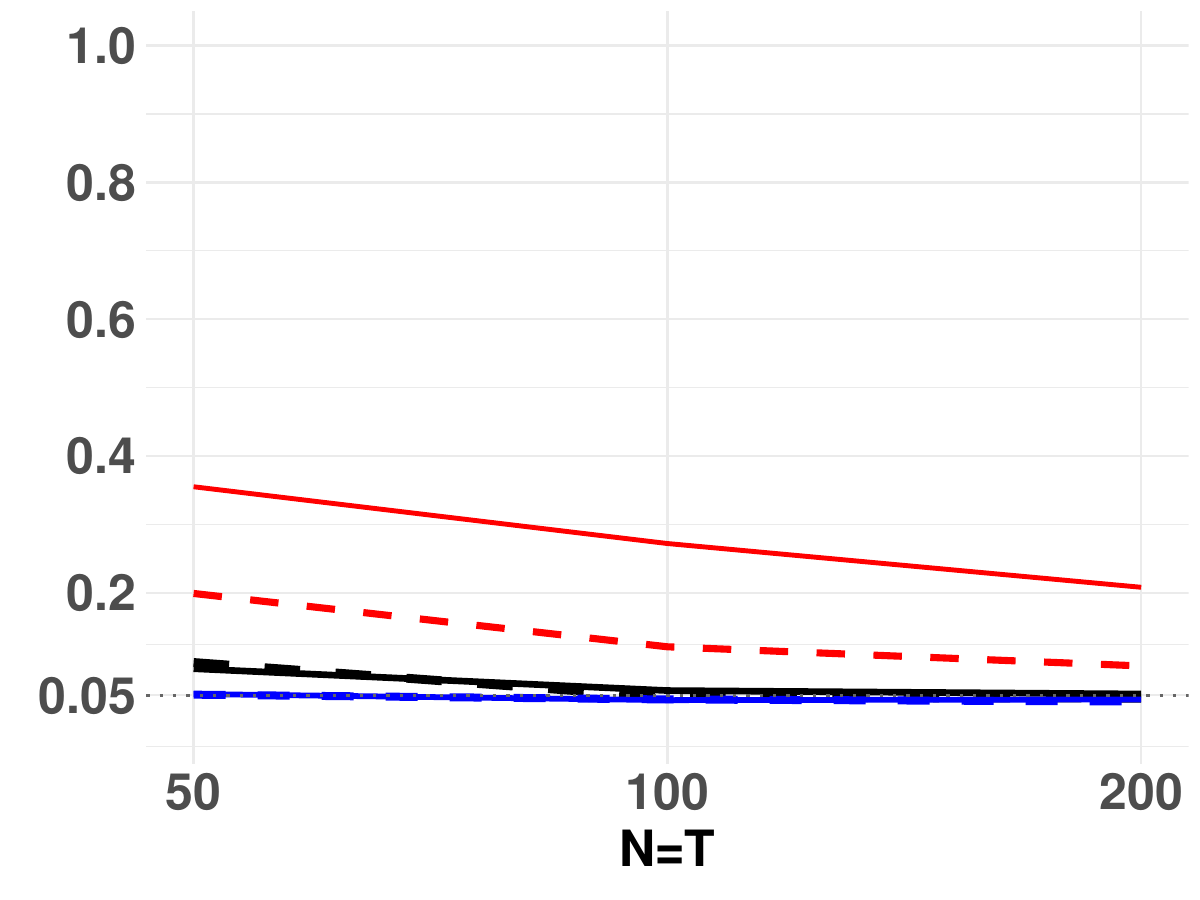}
\caption{$\rho_{wf}=0.6,\alpha_2 = 0.8$}
\label{fig:test_f2_06_08}
\end{subfigure}
\hfill
\begin{subfigure}[b]{0.32\textwidth}
\centering
\includegraphics[width=\textwidth]{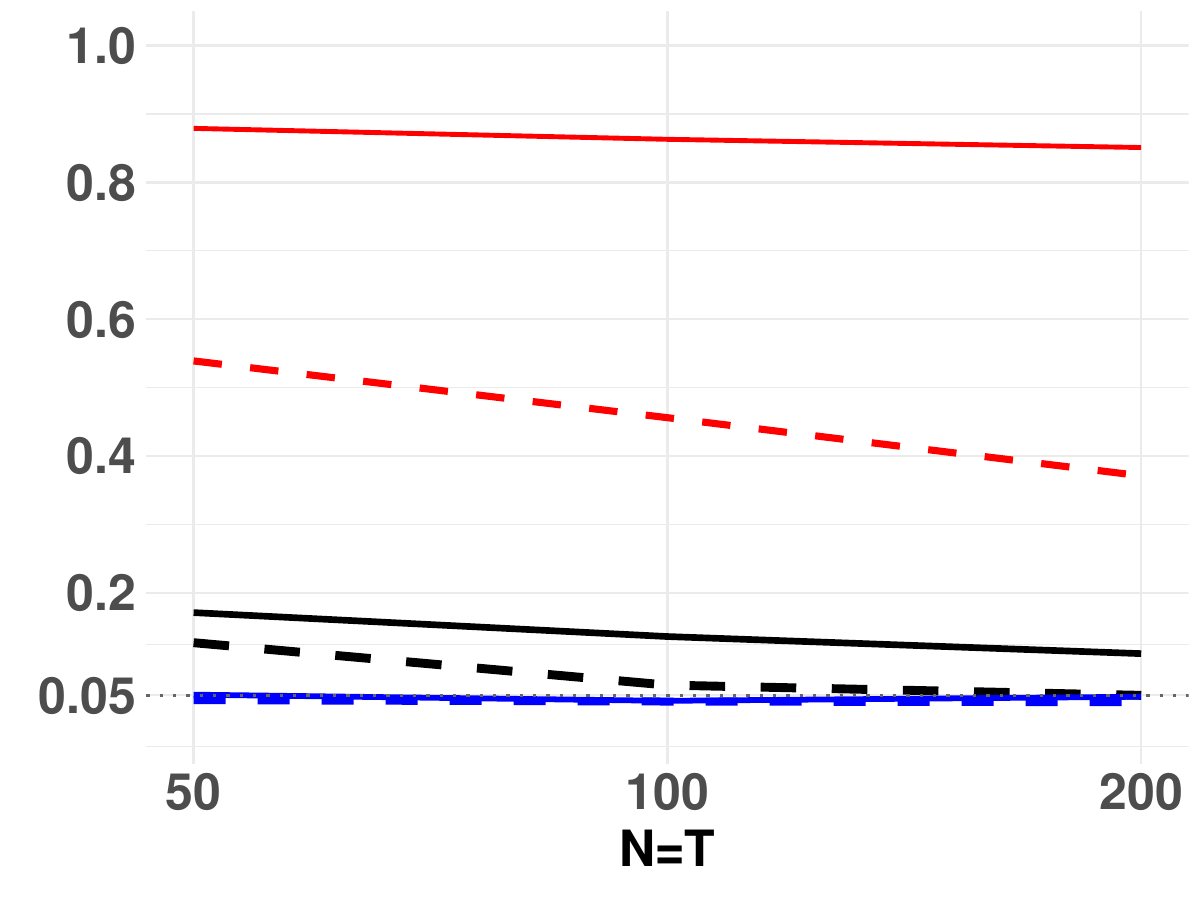}
\caption{$\rho_{wf}=0.6,\alpha_2 = 0.6$}
\label{fig:test_f2_06_06}
\end{subfigure}

\centering
\includegraphics[width=0.50\textwidth]{newimages/legend_SigE3_rhoe2_bias_rhowf0_a1_f2.pdf}
\caption{Size of the t-tests at the 5\% level using $\hat{\gamma}_2$ and its bias corrected versions for cross and serially correlated $e_{t,i}$}
\label{fig:test.f2}

\end{figure}

According to Figure \ref{fig:sd.f2}, the standard deviations of $\hat{\gamma}_2 - \gamma_{\hat{\bH}_q,2}$ and its bias-corrected one are the smallest, closely followed by that of $\hat{\gamma}_2 - \gamma_2^0$. For the strong factor model, the jackknife bias-correction does not increase the variation much, but it moderately does for the very weak factor model with the small sample size, which quickly goes down as the sample size increases. In most cases, the variation of $\hat{\gamma}_{bc\hat{\bH},2} - \gamma_{\hat{\bH},2}$ is the largest, followed by $\hat{\gamma}_{2} - \gamma_{\hat{\bH},2}$.

Reflecting the bias and the variation inflation, the size of the t-test, which is reported in Figure \ref{fig:test.f2}, is affected. The size of the test based on $\hat{\gamma}_{bcjk,2} - \gamma_2^0$ and its bias-corrected versions are always around the nominal level. The size of the test based on $\hat{\gamma}_2 - \gamma_2^0$ is correct unless the model is very weak. The test based on the jackknife corrected estimator is correct except for the very weak factor and the very small sample size.   
The tests based on $\hat{\gamma}_2 - \gamma_{\hat{\bH},2}$ and its bias-corrected are the most unreliable, suffering from enormous size distortion.

\subsubsection{Coefficient on the observed factor, $\beta$}
Let us turn our attention to the performance of the estimated coefficient on the observed factor $w_{t}$, $\beta$ and its bias-corrected versions. Note that the elements of the ``parameters'' $\bdelta^0$, $\bdelta_{\hat{\bH}_q}$ and $\bdelta_{\hat{\bH}}$ corresponding to $w_t$ are all $\beta$, and their uncorrected estimators are identical. However, the corresponding elements of the bias-corrected estimators may have different values if the latent factors and the observable factors are correlated. Specifically, we report the bias, standard deviation (s.d.), and a t-test at the 5\% level of $\hat{\beta} - \beta$ and the bias-corrected estimators relative to $\beta$, $\hat{\beta}_{bcjk} - \beta$, $\hat{\beta}_{bc\hat{\bH}_q} - \beta$, $\hat{\beta}_{bc\hat{\bH}} - \beta$. The bias-corrected versions are shown as dashed lines in the figures.

\begin{figure}[!htb]	
\centering
\begin{subfigure}[b]{0.32\textwidth}
\centering
\includegraphics[width=\textwidth]{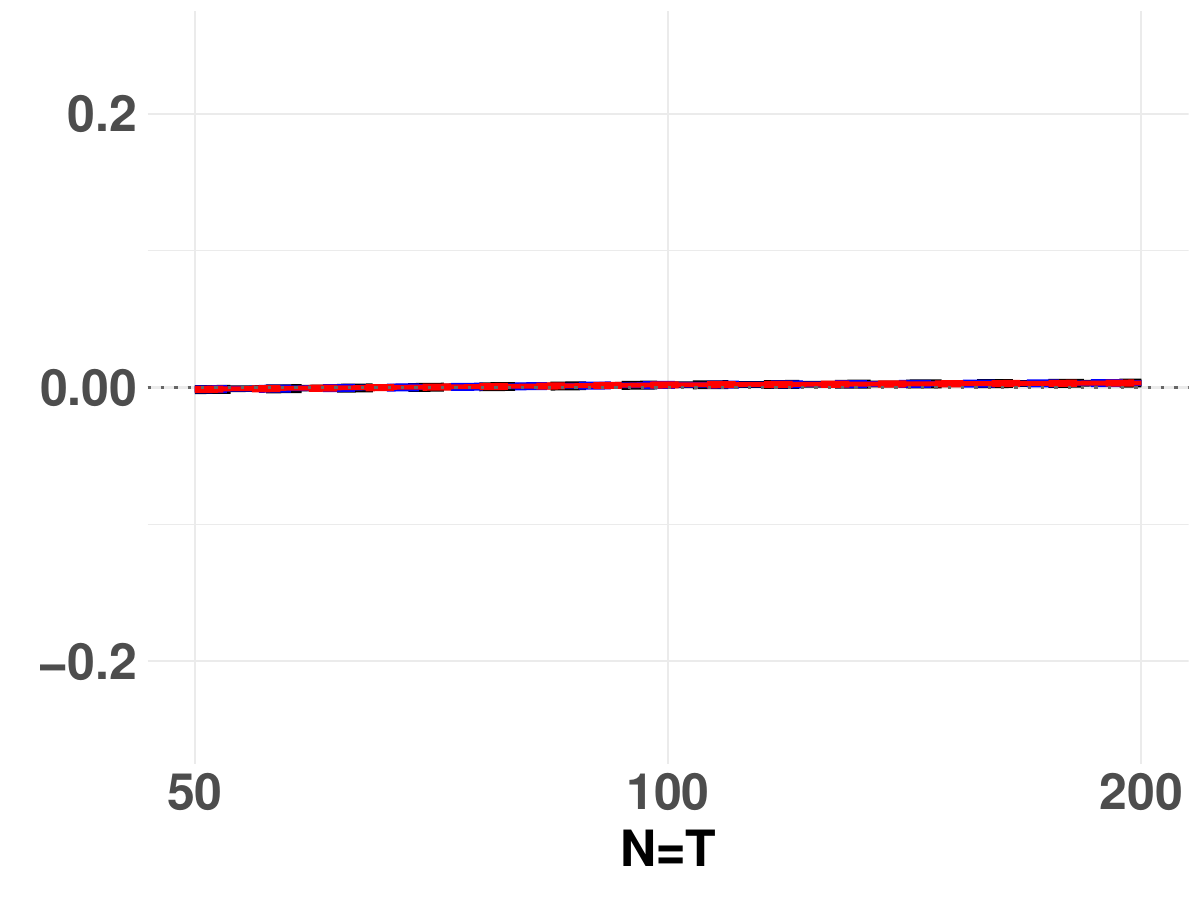}
\caption{$\rho_{wf}=0.0,\alpha_2 = 1.0$}
\label{fig:bias_w_00_10}
\end{subfigure}
\hfill
\begin{subfigure}[b]{0.32\textwidth}
\centering
\includegraphics[width=\textwidth]{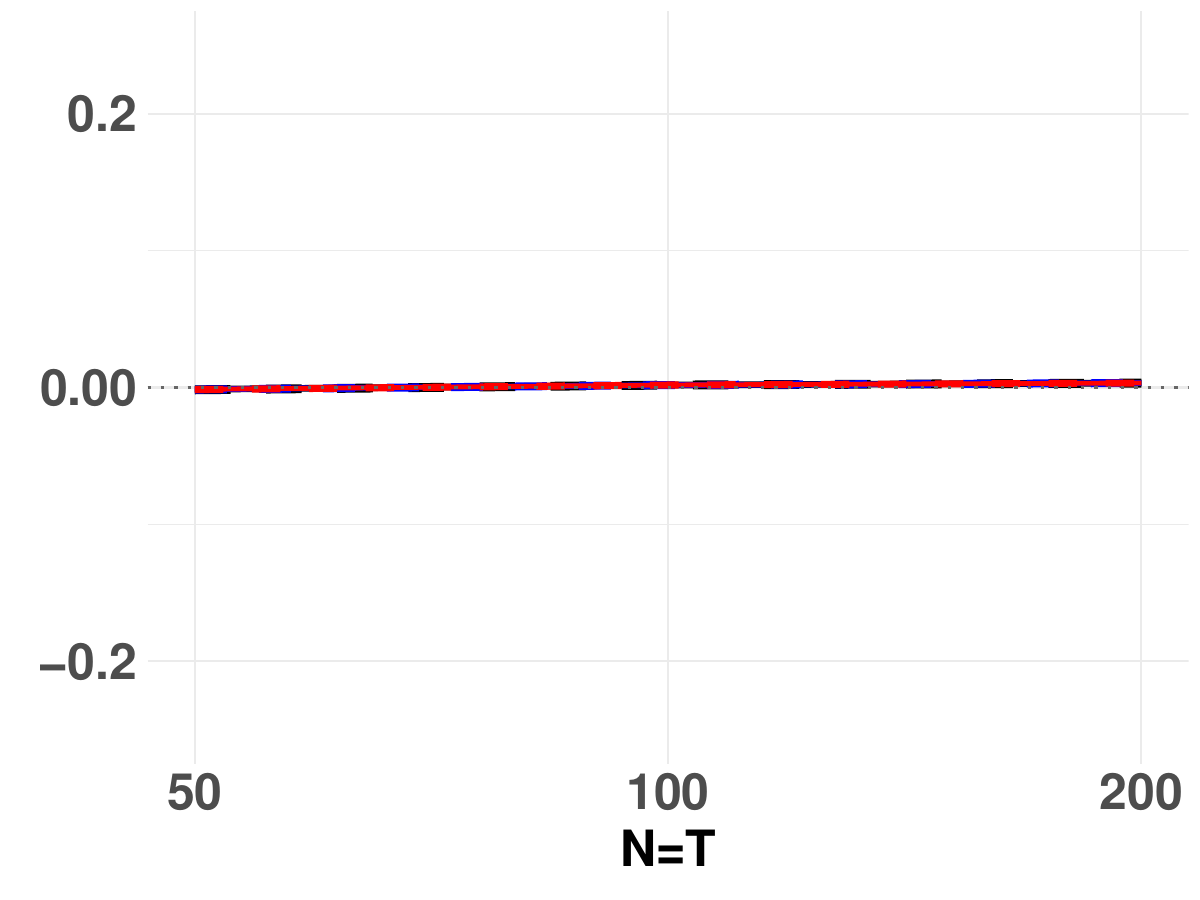}
\caption{$\rho_{wf}=0.0,\alpha_2 = 0.8$}
\label{fig:bias_w_00_08}
\end{subfigure}
\hfill
\begin{subfigure}[b]{0.32\textwidth}
\centering
\includegraphics[width=\textwidth]{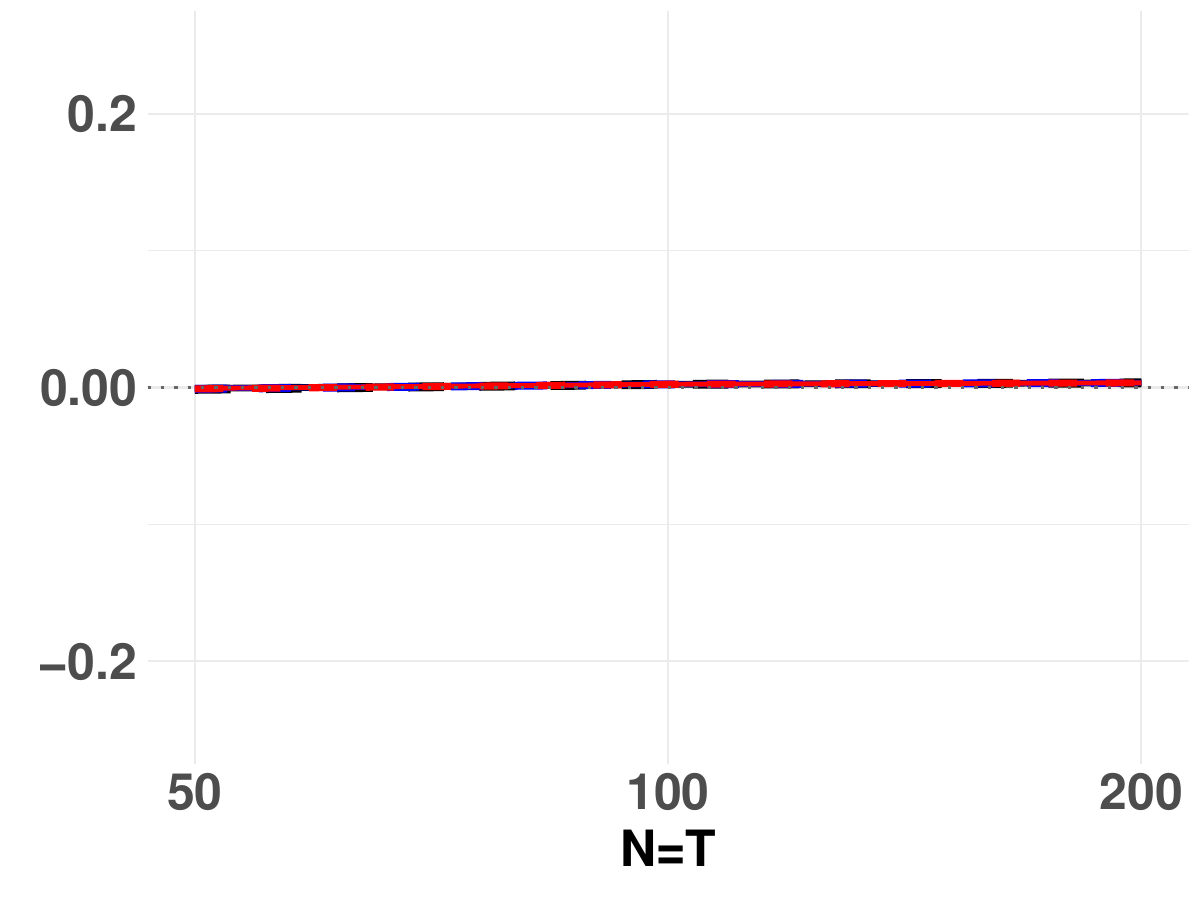}
\caption{$\rho_{wf}=0.0,\alpha_2 = 0.6$}
\label{fig:bias_w_00_06}
\end{subfigure}

\centering
\begin{subfigure}[b]{0.32\textwidth}
\centering
\includegraphics[width=\textwidth]{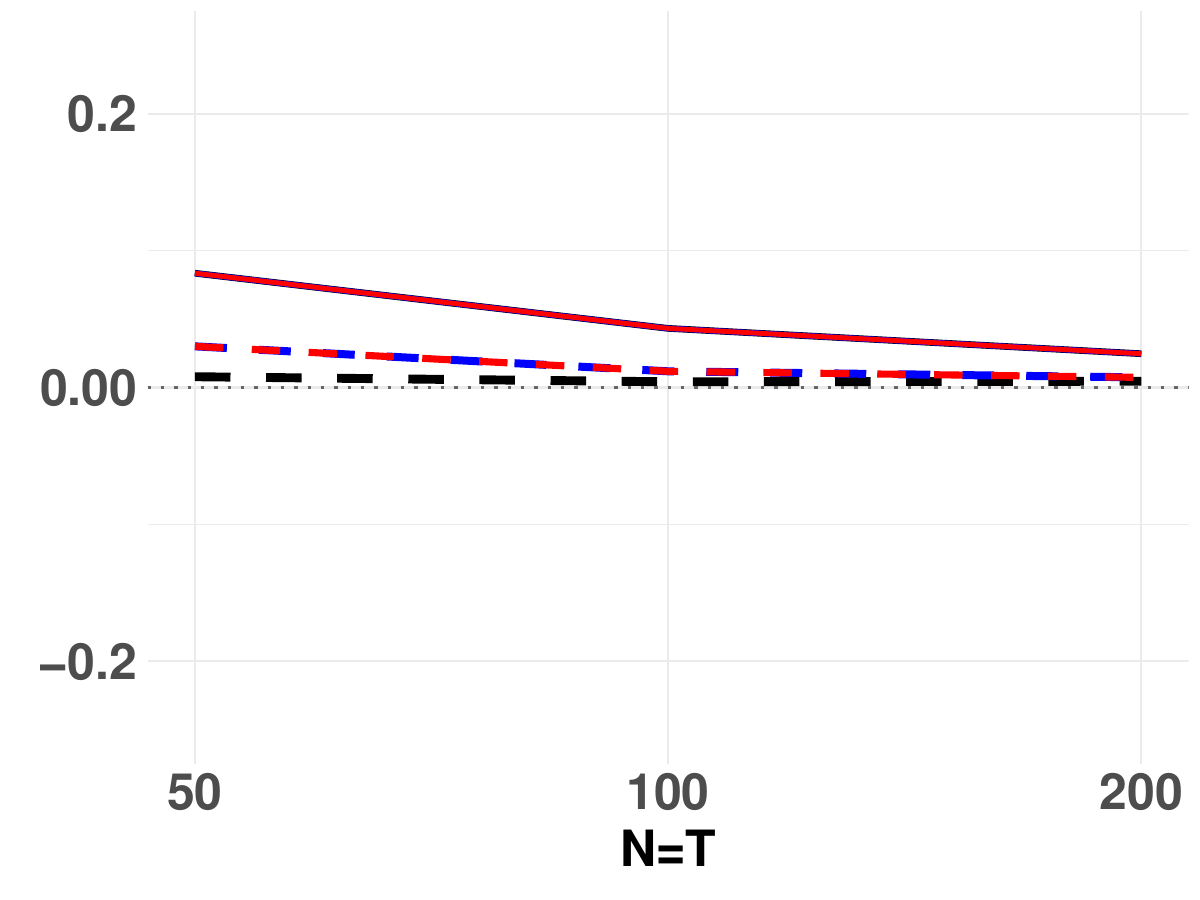}
\caption{$\rho_{wf}=0.6,\alpha_2 = 1.0$}
\label{fig:bias_w_06_10}
\end{subfigure}
\hfill
\begin{subfigure}[b]{0.32\textwidth}
\centering
\includegraphics[width=\textwidth]{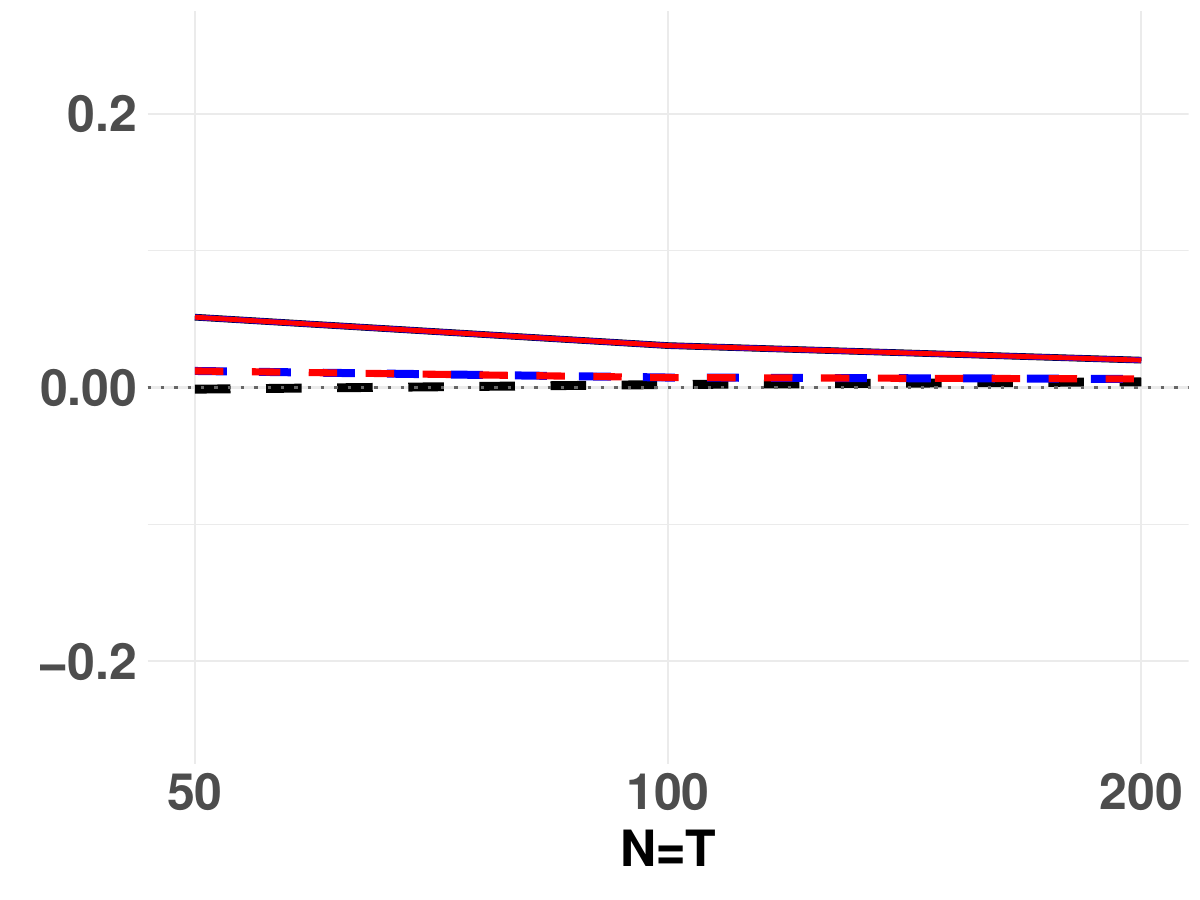}
\caption{$\rho_{wf}=0.6,\alpha_2 = 0.8$}
\label{fig:bias_w_06_08}
\end{subfigure}
\hfill
\begin{subfigure}[b]{0.32\textwidth}
\centering
\includegraphics[width=\textwidth]{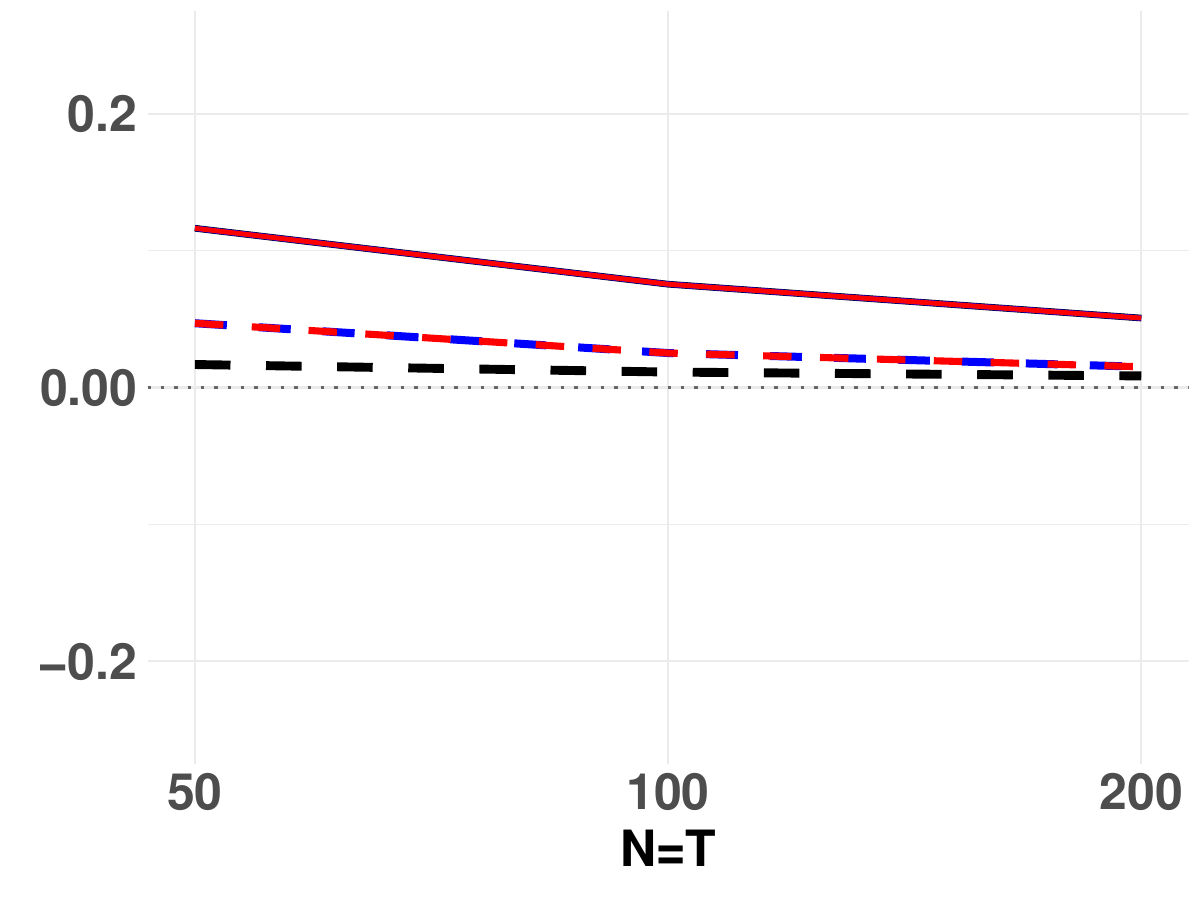}
\caption{$\rho_{wf}=0.6,\alpha_2 = 0.6$}
\label{fig:bias_w_06_06}
\end{subfigure}

\centering
\includegraphics[width=0.45\textwidth]{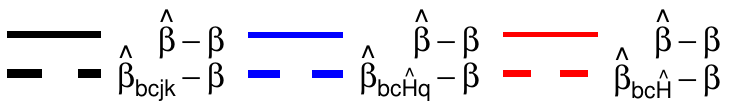}
\caption{Bias of $\hat{\beta}$ and its bias corrected versions for cross and serially correlated $e_{t,i}$}
\label{fig:bias.w}

\end{figure}
\begin{figure}[!htb]	
\centering
\begin{subfigure}[b]{0.32\textwidth}
\centering
\includegraphics[width=\textwidth]{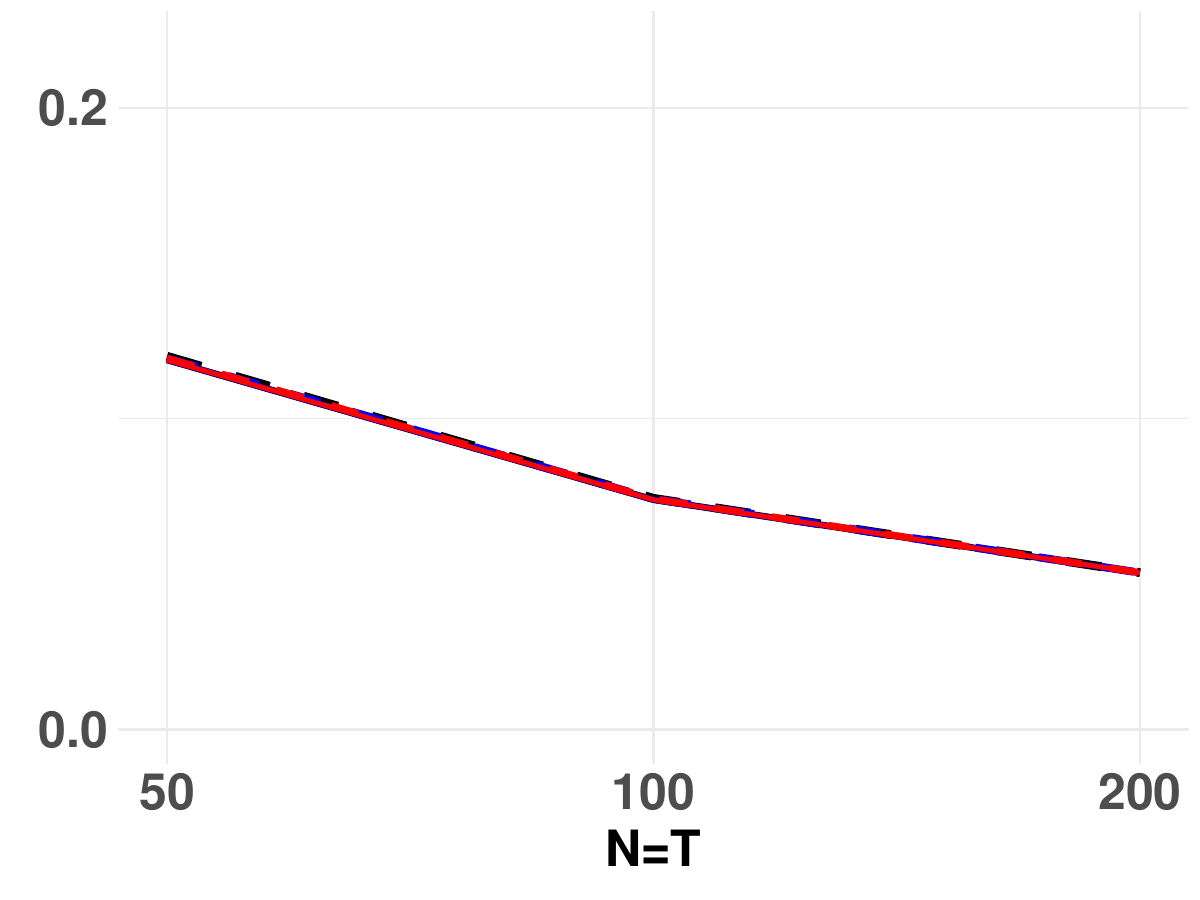}
\caption{$\rho_{wf}=0.0,\alpha_2 = 1.0$}
\label{fig:sd_w_00_10}
\end{subfigure}
\hfill
\begin{subfigure}[b]{0.32\textwidth}
\centering
\includegraphics[width=\textwidth]{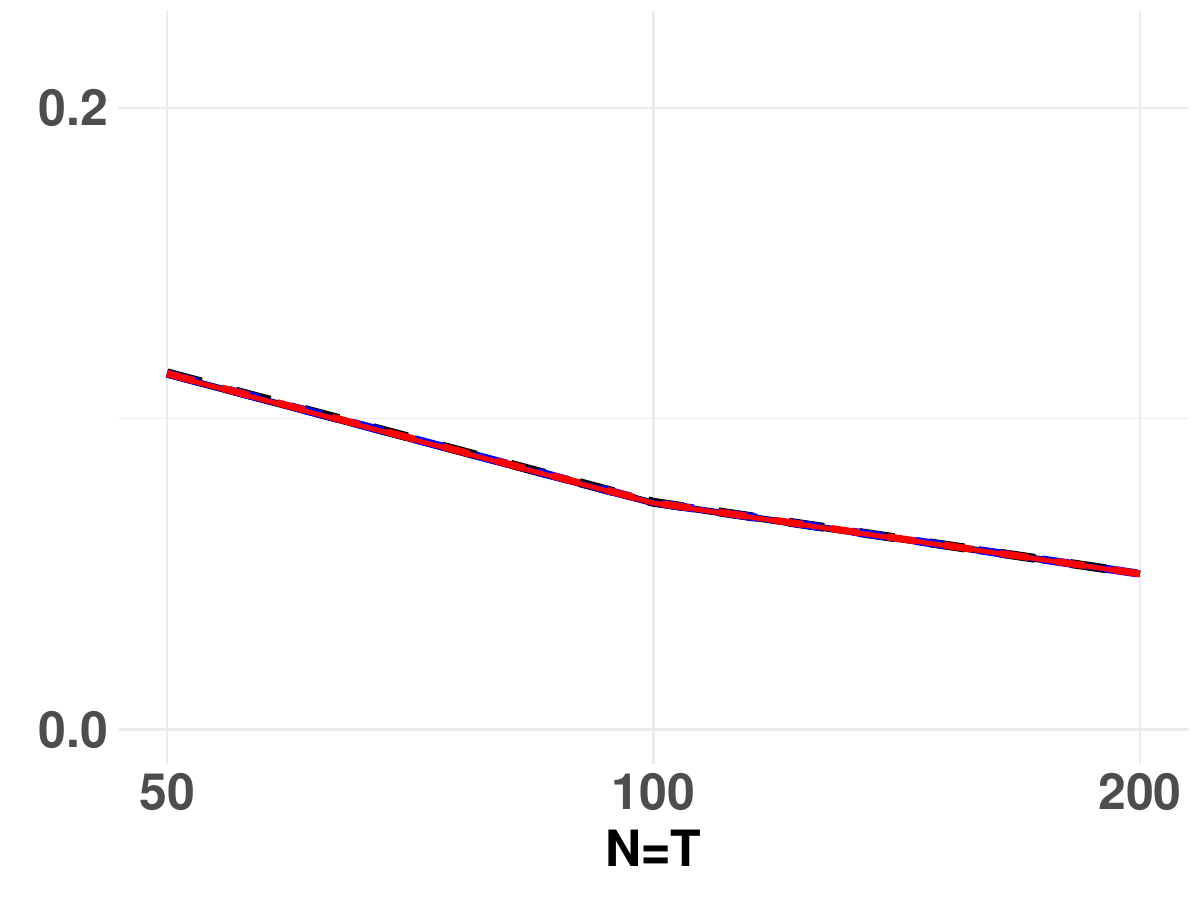}
\caption{$\rho_{wf}=0.0,\alpha_2 = 0.8$}
\label{fig:sd_w_00_08}
\end{subfigure}
\hfill
\begin{subfigure}[b]{0.32\textwidth}
\centering
\includegraphics[width=\textwidth]{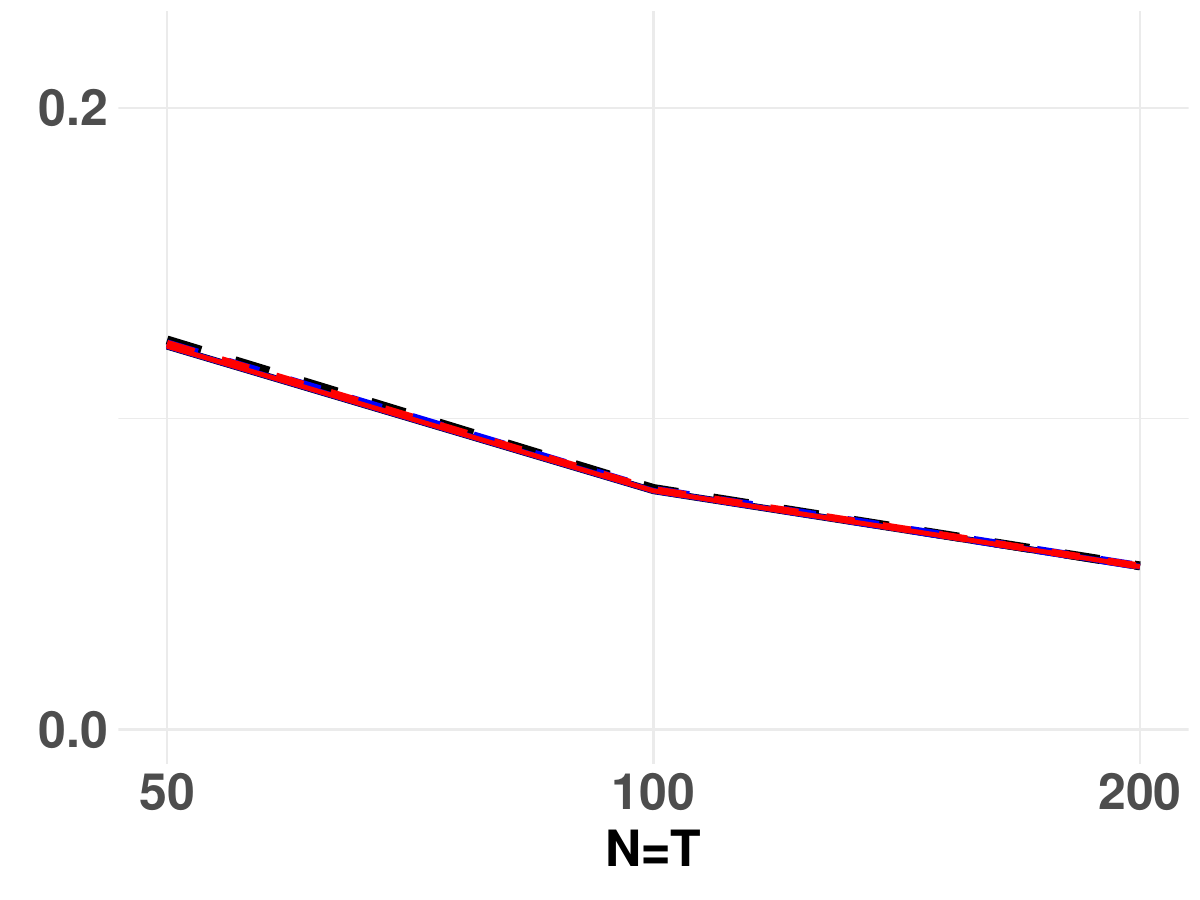}
\caption{$\rho_{wf}=0.0,\alpha_2 = 0.6$}
\label{fig:sd_w_00_06}
\end{subfigure}

\centering
\begin{subfigure}[b]{0.32\textwidth}
\centering
\includegraphics[width=\textwidth]{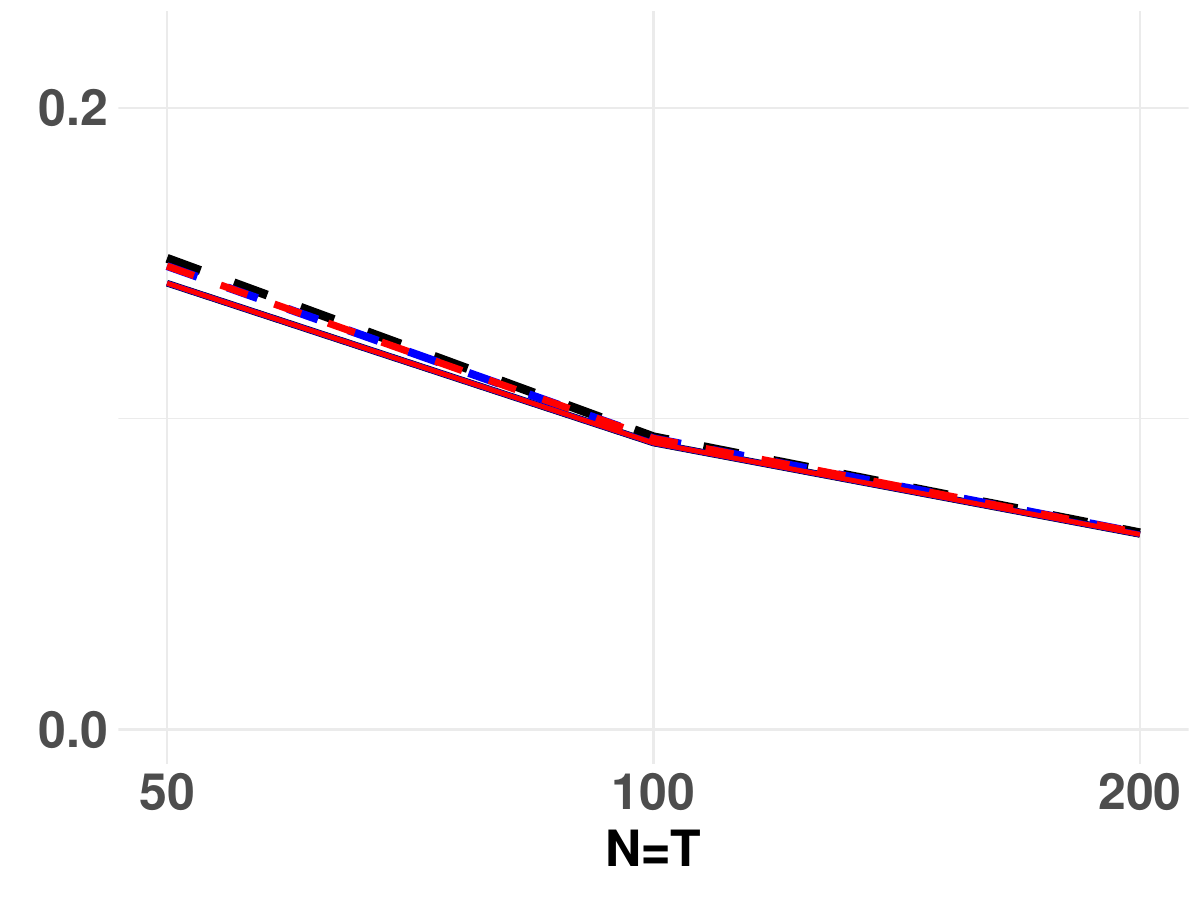}
\caption{$\rho_{wf}=0.6,\alpha_2 = 1.0$}
\label{fig:sd_w_06_10}
\end{subfigure}
\hfill
\begin{subfigure}[b]{0.32\textwidth}
\centering
\includegraphics[width=\textwidth]{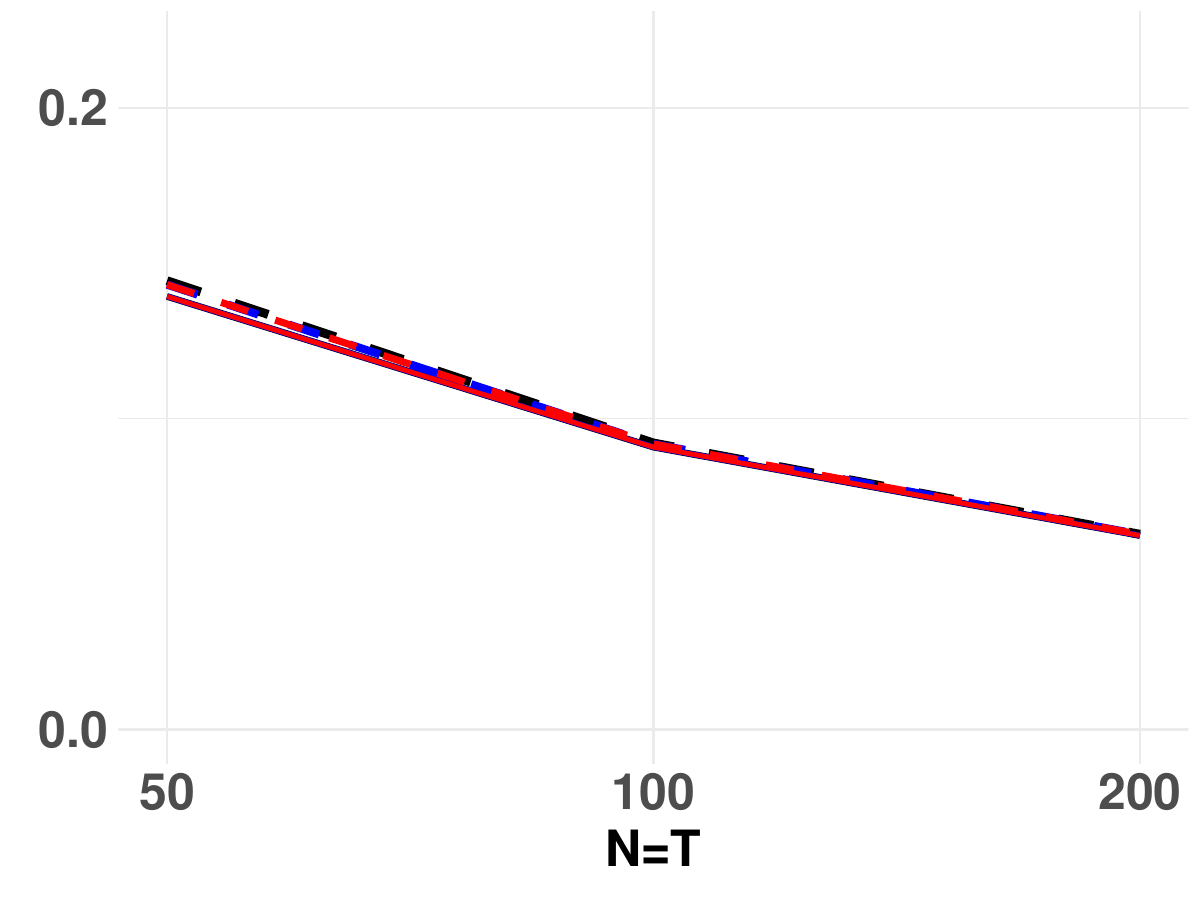}
\caption{$\rho_{wf}=0.6,\alpha_2 = 0.8$}
\label{fig:sd_w_06_08}
\end{subfigure}
\hfill
\begin{subfigure}[b]{0.32\textwidth}
\centering
\includegraphics[width=\textwidth]{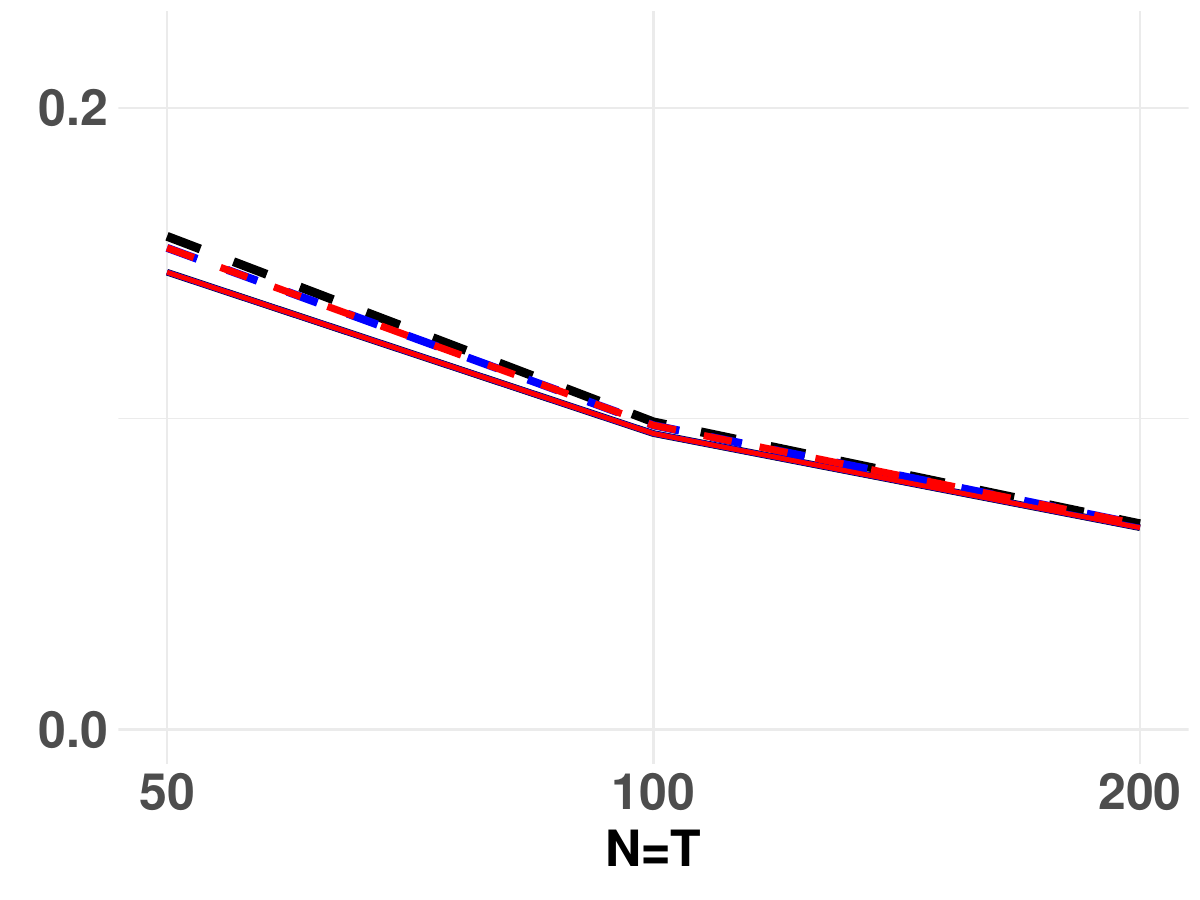}
\caption{$\rho_{wf}=0.6,\alpha_2 = 0.6$}
\label{fig:sd_w_06_06}
\end{subfigure}

\centering
\includegraphics[width=0.45\textwidth]{newimages/legend_SigE3_rhoe2_bias_rhowf0_a1_w.pdf}
\caption{Standard deviation of $\hat{\beta}$ and its bias corrected versions for cross and serially correlated $e_{t,i}$}
\label{fig:sd.w}

\end{figure}

Figure \ref{fig:bias.w} shows the average bias over the replications. The bias of all the estimators is zero when $\bff_t^*$ and $w_t$ are uncorrelated (i.e. $\rho_{wf}=0$) as expected, because the bias is caused by the estimation effect of $\hat{\bff}_t$ and it is not transmitted to the estimator of $\beta$. 
The picture changes dramatically when $\bff_t^*$ and $w_t$ are correlated (i.e. $\rho_{wf}=0.6$). The LS estimator $\hat{\beta}$ is biased, and the magnitude of the bias tends to increase for weaker models and for smaller sample sizes. The bias-corrections $\hat{\beta}_{bc\hat{\bH}_q}$ and $\hat{\beta}_{bc\hat{\bH}}$ successfully reduce the bias, always by almost the same amount. The clear winner in terms of bias reduction is the jackknife estimator, $\hat{\beta}_{bcjk}$. The standard deviations of all the estimators shown in Figure \ref{fig:sd.w} are almost identical when $\bff_t^*$ and $w_t$ are uncorrelated (i.e. $\rho_{wf}=0$), while the standard deviations of all the bias-corrected estimators are similar but very slightly larger than the non-corrected estimator when $\rho_{wf}=0.6$. 

\begin{figure}[!htb]	
\centering
\begin{subfigure}[b]{0.32\textwidth}
\centering
\includegraphics[width=\textwidth]{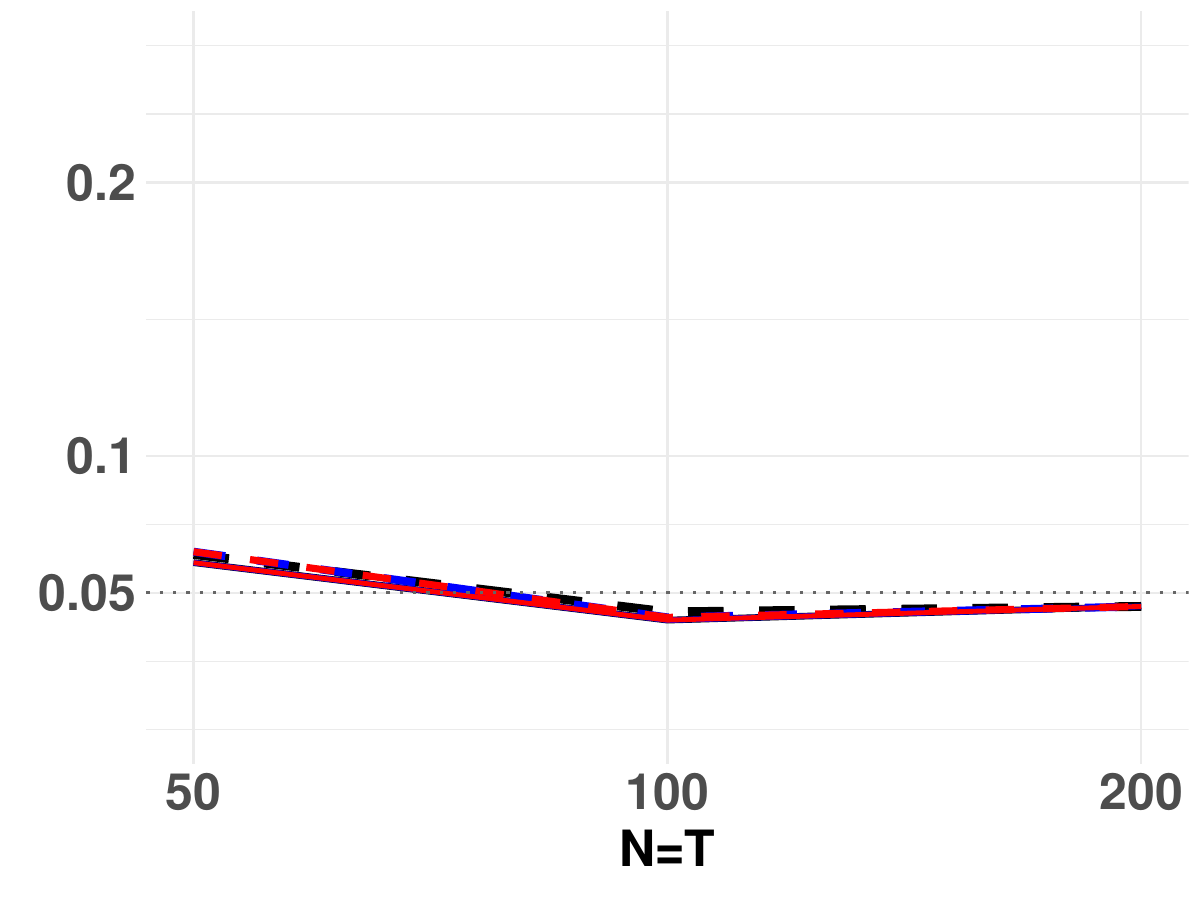}
\caption{$\rho_{wf}=0.0,\alpha_2 = 1.0$}
\label{fig:test_w_00_10}
\end{subfigure}
\hfill
\begin{subfigure}[b]{0.32\textwidth}
\centering
\includegraphics[width=\textwidth]{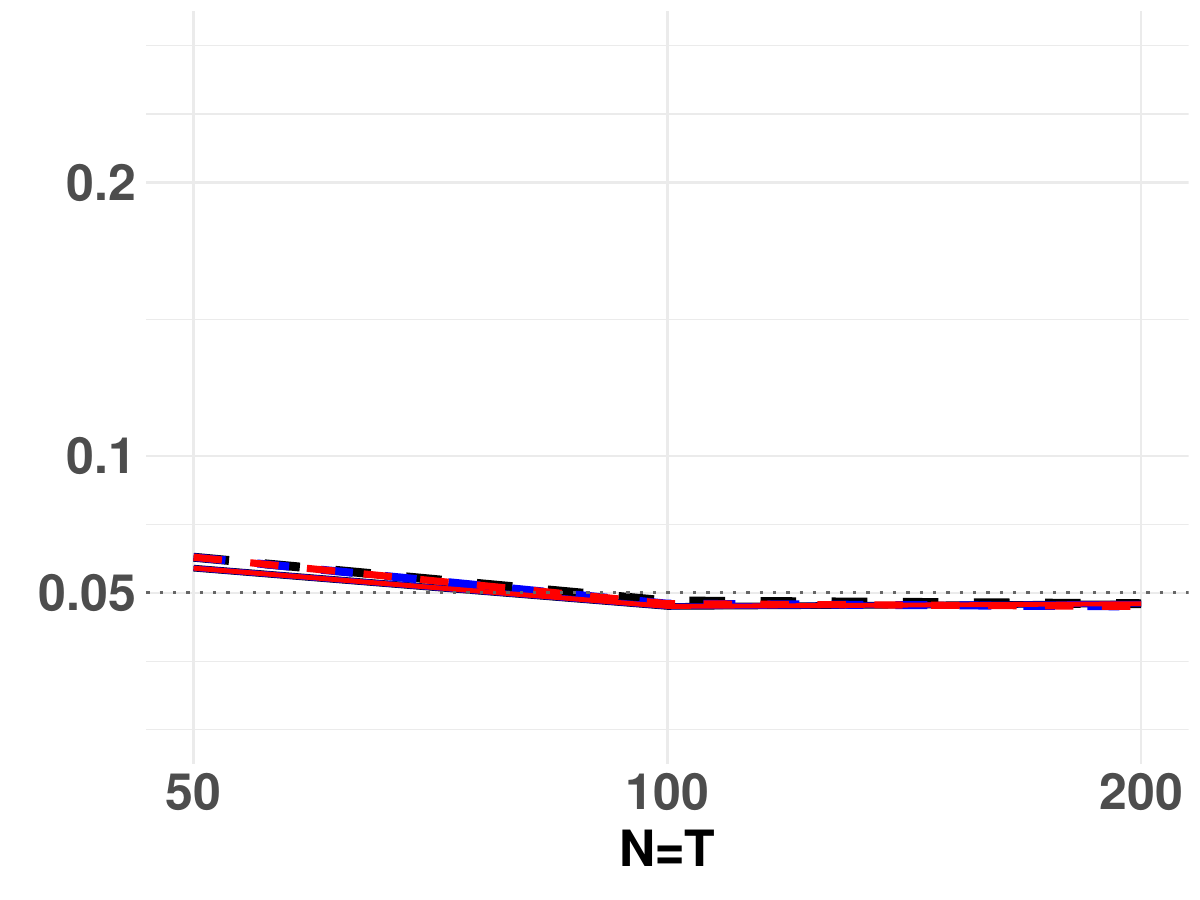}
\caption{$\rho_{wf}=0.0,\alpha_2 = 0.8$}
\label{fig:test_w_00_08}
\end{subfigure}
\hfill
\begin{subfigure}[b]{0.32\textwidth}
\centering
\includegraphics[width=\textwidth]{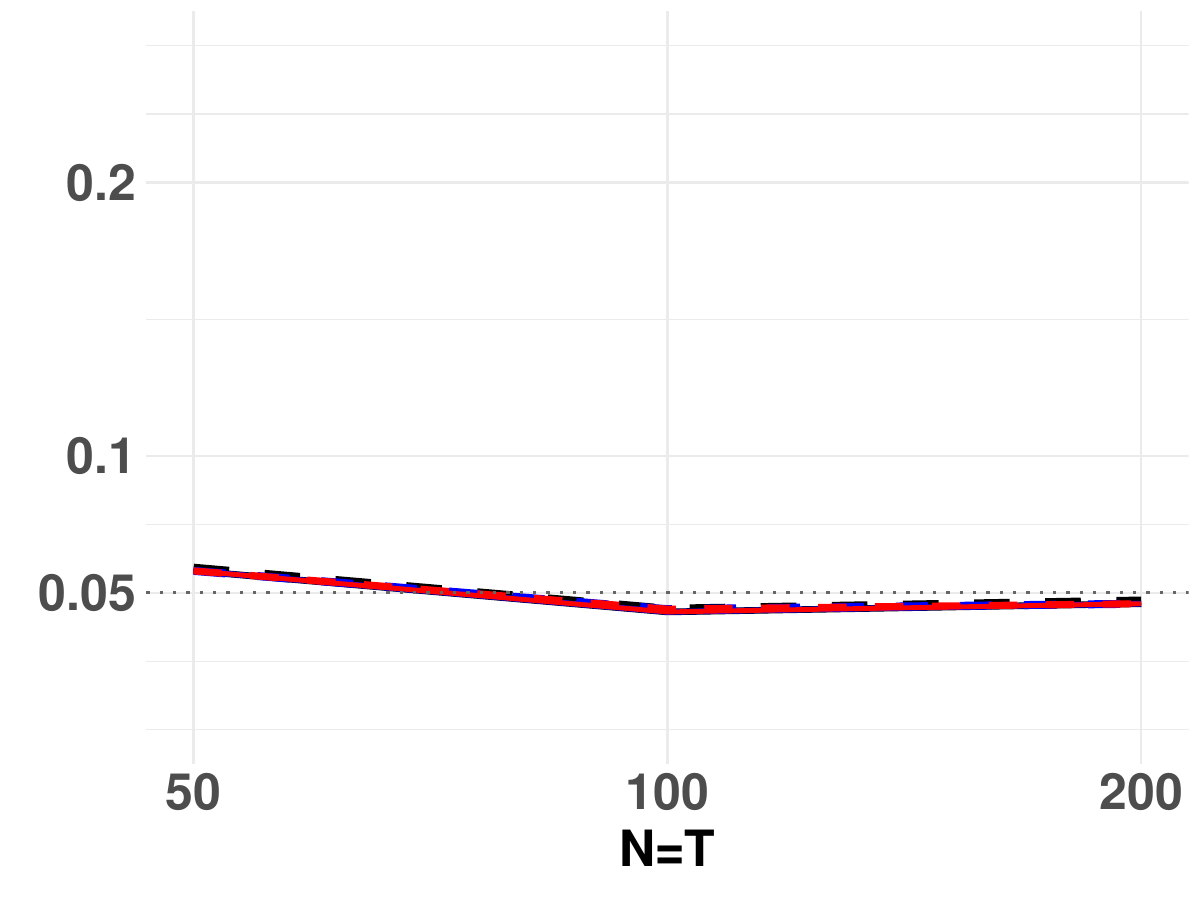}
\caption{$\rho_{wf}=0.0,\alpha_2 = 0.6$}
\label{fig:test_w_00_06}
\end{subfigure}

\centering
\begin{subfigure}[b]{0.32\textwidth}
\centering
\includegraphics[width=\textwidth]{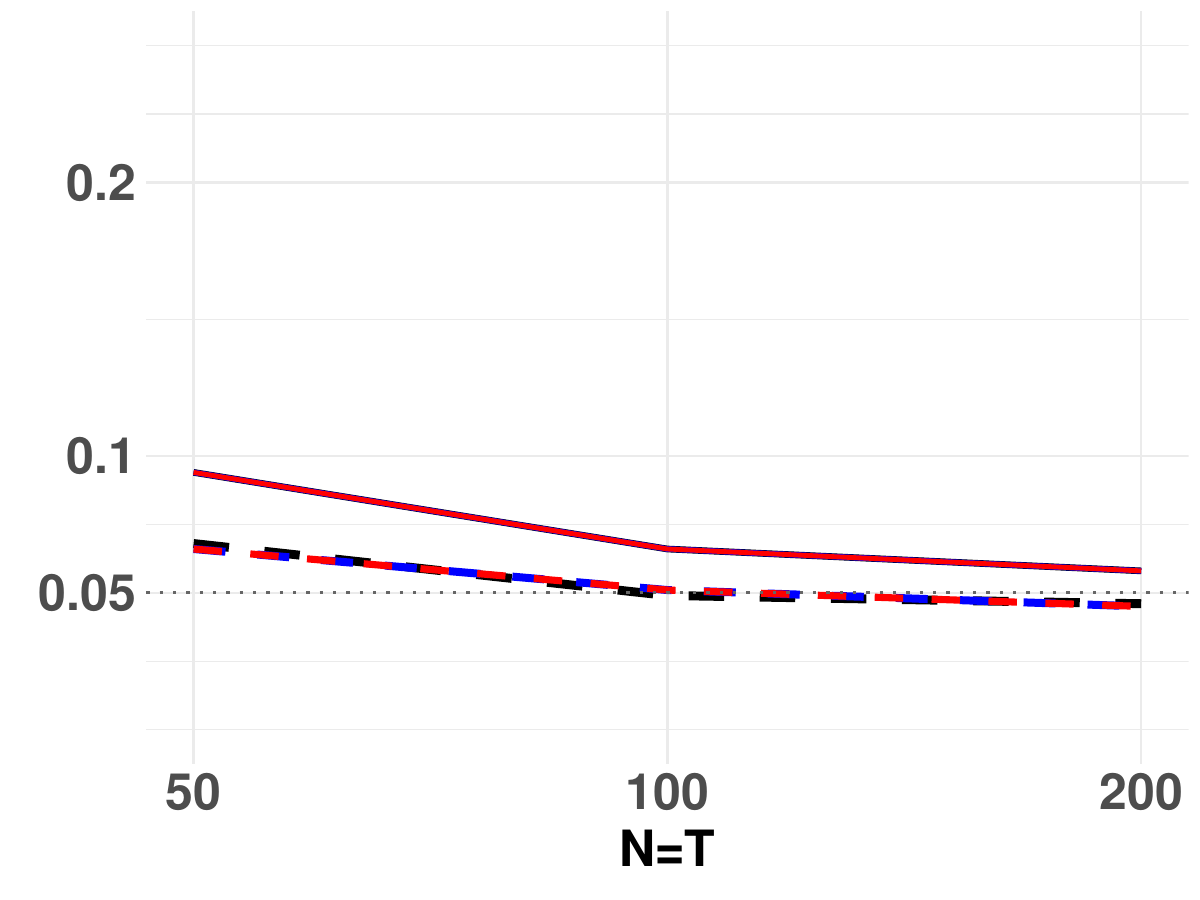}
\caption{$\rho_{wf}=0.6,\alpha_2 = 1.0$}
\label{fig:test_w_06_10}
\end{subfigure}
\hfill
\begin{subfigure}[b]{0.32\textwidth}
\centering
\includegraphics[width=\textwidth]{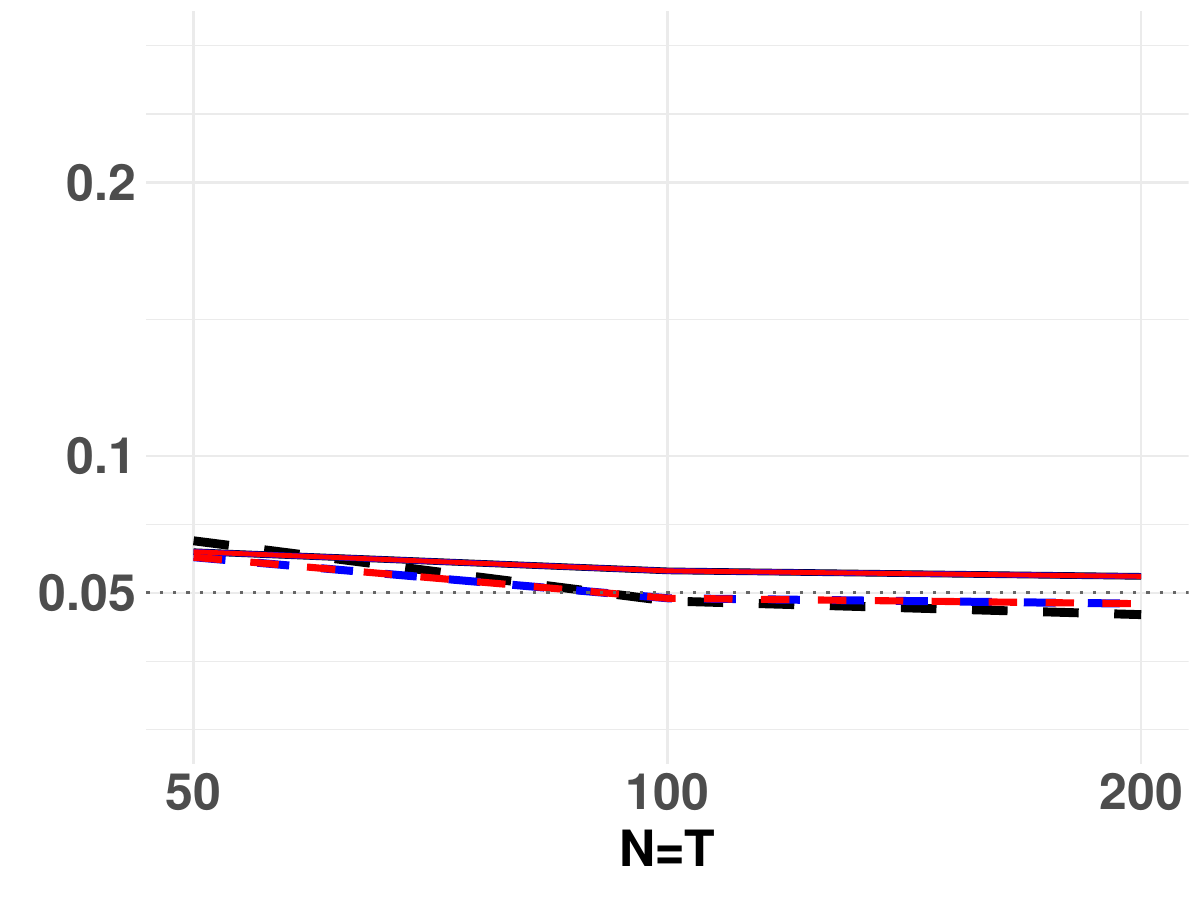}
\caption{$\rho_{wf}=0.6,\alpha_2 = 0.8$}
\label{fig:test_w_06_08}
\end{subfigure}
\hfill
\begin{subfigure}[b]{0.32\textwidth}
\centering
\includegraphics[width=\textwidth]{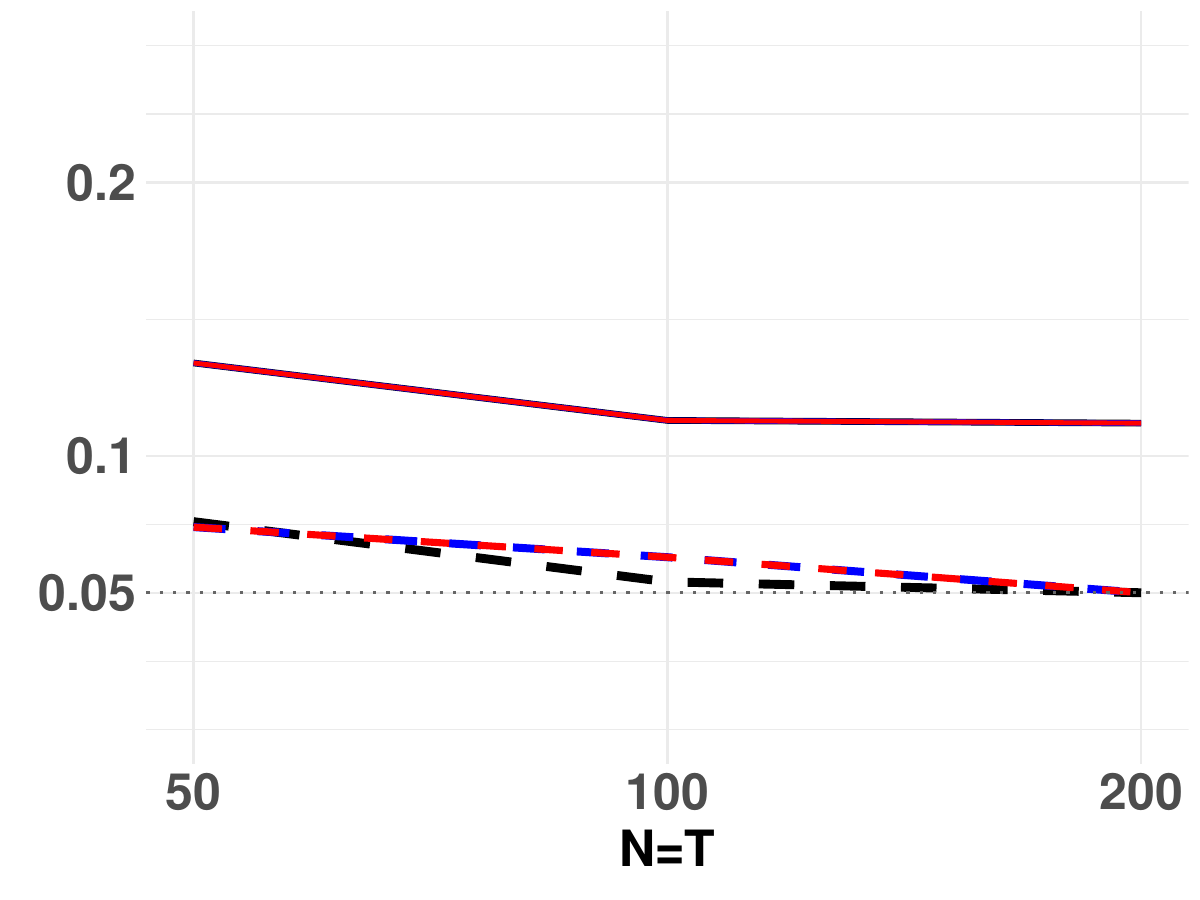}
\caption{$\rho_{wf}=0.6,\alpha_2 = 0.6$}
\label{fig:test_w_06_06}
\end{subfigure}

\centering
\includegraphics[width=0.50\textwidth]{newimages/legend_SigE3_rhoe2_bias_rhowf0_a1_w.pdf}
\caption{Size of the t-tests using $\hat{\beta}$ and its bias corrected versions for cross and serially correlated $e_{t,i}$}
\label{fig:test.w}

\end{figure}

The size of the tests is summarized in Figure \ref{fig:test.w}. The size of the tests based on all the estimators is correct when $\rho_{wf}=0.0$, while the size of the tests based on the uncorrected estimators tends to deviate from the nominal level, which is successfully corrected by all the bias-correction methods. 

Finally, the experimental results for the models augmented with factors extracted from the prediction variables orthogonalized to the observed factor, $\bM_w \bX$ are presented in Figures \ref{fig:bias.f2w}-\ref{fig:test.ww} in the online appendix. These results confirm that the $\hat{\bdelta}_w - \bdelta_{\hat{\bH}_{q,w}}$ has little bias for all the designs, including the case with $\rho_{wf}=0.6$, for the small sample size and for the weakest factor model, as predicted by our theory. However, if researchers are interested in estimating the parameter $\bdelta_w^0$ rather than the random vector $\bdelta_{\hat{\bH}_{q,w}}$, the jackknife estimator $\hat\bdelta_{bcjk,w}$ may be preferred.

\subsection{Size-adjusted power curve of the test of significance}
To see how the asymptotic bias can affect the power curve of the t-test of significance, $H_0 : \delta_{\bR,k}=0$ against $H_1 : \delta_{\bR,k}\neq 0$ for $\bR = \bH, \hat{\bH}_q, \hat{\bH}$, $k=1,\dots,r+p$, we conduct the following experiments. The data generating process (DGP) is identical to that described in Section \ref{sec:design}, except that we change the value of an element in $\bdelta^0$ under the test between $-0.4$ and $0.4$ by $0.025$, keeping other elements at unity. We have chosen the case for $(\alpha_1, \alpha_2)=(0.8,0.6)$, $N=T=100$ and $\rho_{fw}=0.6$.
The t-ratios for $\gamma_{\bR,2}$ and $\beta$ based on $\hat{\bdelta}$, $\hat{\bdelta}_{bcjk}$, $\hat{\bdelta}_{bc\hat{\bH}_q}$ and $\hat{\bdelta}_{bc\hat{\bH}}$ are calculated.
Table \ref{table:sizequantile} reports the size of the significance tests for the second factor and $w_t$, using the critical value from the standard normal distribution. 
There are moderate size distortions in the significance tests for the second factor based on $\hat{\gamma}_{bcjk,2}$, and higher size distortion based on $\hat{\gamma}_{bc\hat{\bH},2}$.
In view of the size distortion, the two-sided tests are then implemented to compare the size-adjusted power curves using the 95\% quantiles of the absolute values of the t-ratios under the null over the replications as the critical values, which are also reported in the table.


\begin{table}[!htb]
\centering
\caption{The size of the significance tests for the second factor and $w_t$ and the 95\% quantile of the absolute values of the associated t-ratios based on different estimators}
\label{table:sizequantile}
\begin{tabular}{lcccclccc}
\hline
&     & Size of the Test & 95\% quantile  &     &     &     & Size of the Test & 95\% quantile  \bigstrut\\
\cline{1-4}\cline{6-9}
$\hat{\gamma}_2$ &     & 6.1\% & 2.06 &     & $\hat{\beta}$ &     & 7.3\% & 2.14 \bigstrut[t]\\
$\hat{\gamma}_{bcjk,2}$ &     & 7.9\% & 2.15 &     & $\hat{\beta}_{bcjk}$ &     & 5.1\% & 1.97 \\
$\hat{\gamma}_{bc\hat{\bH}_q,2}$ &     & 6.4\% & 2.04 &     & $\hat{\beta}_{bc\hat{\bH}_q}$ &     & 5.2\% & 1.97 \\
$\hat{\gamma}_{bc\hat{\bH},2}$ &     & 10.0\% & 2.28 &     & $\hat{\beta}_{bc\hat{\bH}}$&     & 5.2\% & 1.97 \bigstrut[b]\\
\hline
\end{tabular}%
\end{table}

The estimated size-adjusted power curve for $\gamma_{\bR,2}$ and $\beta$ is shown in Figures \ref{fig:pwcrv}.
From Figure \ref{fig:pwcrv.g} we can see that the power curves of the significance tests for the second factor based on the three bias-corrected estimators are virtually identical, while the power curve of the test based on the least squares estimator is asymmetric most likely due to the bias in the parameter estimates. 
The similarity of the power curves for the bias-corrected estimators breaks down for testing the significance of the observed factor $w_t$, which is shown in Figure \ref{fig:pwcrv.b}. 
The power curve for the significance test based on the least squares estimator is biased, with the curve shifted substantially to the left. 
The bias correction towards the `parameters' with the data dependent rotations $\hat{\bH}_q$ and $\hat{\bH}$ mitigates the bias of the power curve, but the location shift still remains. In contrast, the jackknife bias correction successfully corrects the bias and restores the symmetry of the power curve at $\beta=0$.

To summarize the power curve analysis, it is generally recommended to use the bias-corrected estimators for the significance test, and the split-panel jackknife estimator seems to be the most reliable among the bias-corrected estimators considered.
The recommendation does not apply in the special case of joint significance tests for all latent factors, since there is no bias in the least squares estimator under the null hypothesis; to see this, substitute $\bgamma^* = \bzero$ into the results in Theorems \ref{thm:bias_Hhat}-\ref{thm:bias_Hw3} and all asymptotic biases disappear.

\begin{figure}[h!]	
\centering
\begin{subfigure}[b]{0.49\textwidth}
\centering
\includegraphics[width=\textwidth]{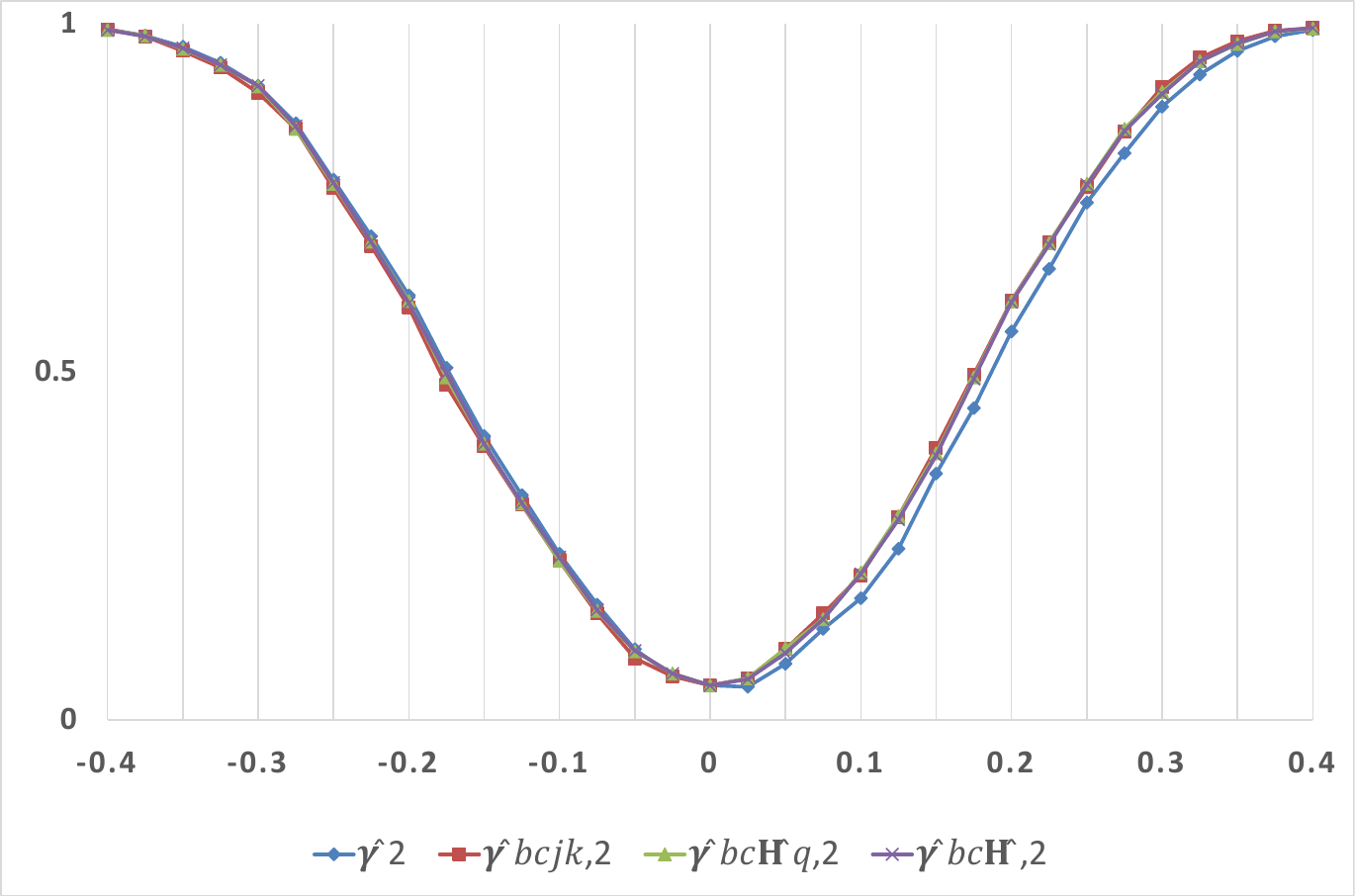}
\caption{Second factor}
\label{fig:pwcrv.g}
\end{subfigure}
\hfill
\begin{subfigure}[b]{0.49\textwidth}
\centering
\includegraphics[width=\textwidth]{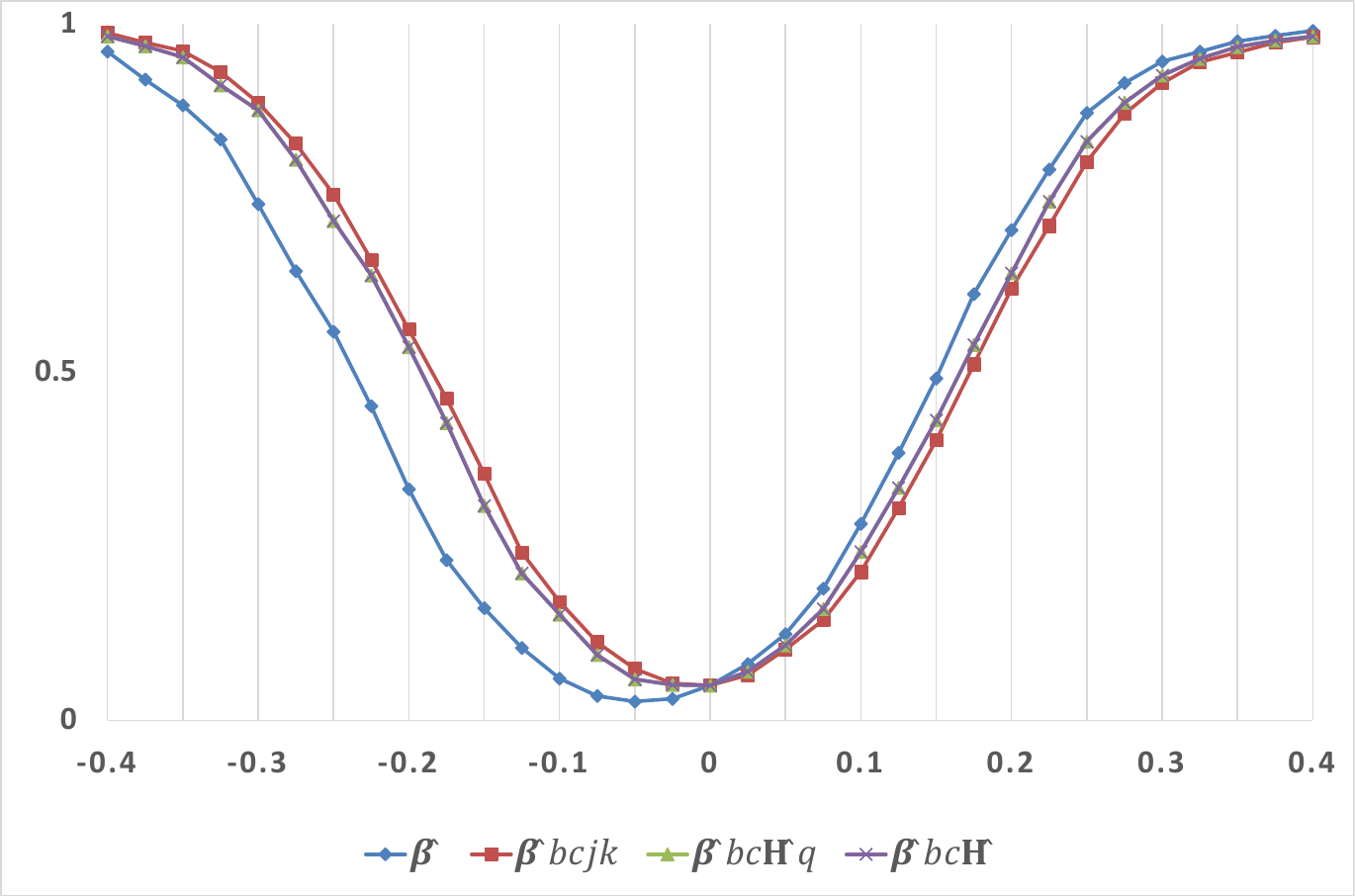}
\caption{Observed factor $w_t$}
\label{fig:pwcrv.b}
\end{subfigure}

\caption{Size-adjusted power curve for testing the significance of the second factor and $w_t$ based on $\hat{\bdelta}$, $\hat{\bdelta}_{bcjk}$, $\hat{\bdelta}_{bc\hat{\bH}_q}$ and $\hat{\bdelta}_{bc\hat{\bH}}$}
\label{fig:pwcrv}
\end{figure}


\section{Empirical application}\label{sec:emp}

It is important to test which factors are significant in factor-augmented regressions. This is because the extracted PC factors are ordered by their importance in the covariation of the predictors, which may not correspond to their predictive power for the particular series of interest; see further discussion in \citet{BaiNg2008,BaiNg2009} and \cite{ChengHansen2015}.

We consider the factor-augmented forecast regression of bond yields $y_{t+h}$ on
the extracted factors from a large number of predictor variables $x_{t,i}$ and the
observed predictor, $w_{t}$. We use the dataset used in \cite{LudvigsonNg2009}, which is provided by Sydney Ludvigson's website. The data consists of
the continuously compounded (log) annual excess returns on an $2$-year
discount bond at month $t$, $y_{t+12}$, and a balanced panel of
$i=1,\dots,131$ monthly macroeconomic series at month $t$, $x_{t,i}$, spanning
the period $t=$January 1982,..., December 2002, which are standardized. For
$x_{t,i}$, the Edge Distribution (ED) estimator of \cite{Onatski2010} with $r_{\max}=9$ gives $\hat{r}=3$. Before running the regression, three PC factors, $\mathbf{\hat{f}}_{t}=(\hat{f}_{t1},\hat{f}_{t2},\hat{f}_{t3})^{\prime}$, are extracted from $x_{t,i}$ and the \cite{CochranePiazzesi2005} (observed) factor, $CP_{t}$, which is a linear combination of five forward Treasury yield spreads, is obtained. The correlations between $CP_{t}$ and $\mathbf{\hat{f}}_{t}^{\prime}$
were \{-0.11, 0.14, -0.11\}, respectively, which are significant or
insignificant on the borderline at the 10\% level test. We run a regression of
$y_{t+12}$ on $\{\mathbf{\hat{f}}_{t}^{\prime},CP_{t},1\}$. We also run a
similar regression but on $\{\mathbf{\hat{f}}_{wt}^{\prime},CP_{t},1\}$, where
$\mathbf{\hat{f}}_{wt}$ is extracted from the data variable from which $CP$ is projected out.

\begin{table}[!htb]
\centering
\caption{Prediction regression results of the yield of the two-year maturity bond}
\label{table:emp1}
\begin{tabular}
[c]{lrccccc}\hline
$y_{t+12}$ &  & $\hat{f}_{t1}$ & $\hat{f}_{t2}$ & $\hat{f}_{t3}$ & $CP_{t}$ &
$R^{2}$\\\hline
$\boldsymbol{\hat{\delta}}$ &  & 0.47*** & 0.13 & -0.01 & 0.36*** & 0.14\\
\multicolumn{1}{r}{} &  & (3.24) & (0.92) & (-0.21) & (3.15) & \\
$\boldsymbol{\hat{\delta}}_{bc\hat{\bH}}$ &  & 0.50*** & 0.14 & -0.01 &
0.36*** & 0.14\\
\multicolumn{1}{r}{} &  & (3.43) & (1.00) & (-0.20) & (3.15) & \\
$\boldsymbol{\hat{\delta}}_{bc\hat{\bH}_{q}}$ &  & 0.47*** & 0.13 & -0.01 &
0.36*** & 0.14\\
\multicolumn{1}{r}{} &  & (3.24) & (0.92) & (-0.21) & (3.15) & \\
$\boldsymbol{\hat{\delta}}_{bcjk}$ &  & 0.49*** & 0.19 & 0.01 & 0.36*** &
0.14\\
\multicolumn{1}{r}{} &  & (3.33) & (1.29) & (0.15) & (3.18) & \\\hline
\multicolumn{1}{r}{} &  & $\hat{f}_{w,t1}$ & $\hat{f}_{w,t2}$ & $\hat
{f}_{w,t3}$ & $CP_{t}$ & $R^{2}$\\\hline
$\boldsymbol{\hat{\delta}}_{w}$ &  & 0.47*** & 0.13 & -0.02 & 0.35*** & 0.14\\
\multicolumn{1}{r}{} &  & (3.22) & (0.92) & (-0.28) & (3.09) & \\
$\boldsymbol{\hat{\delta}}_{bc\hat{\bH}_w}$ &  & 0.50*** & 0.14 & -0.02 &
0.35*** & 0.14\\
\multicolumn{1}{r}{} &  & (3.41) & (1.00) & (-0.28) & (3.09) & \\
$\boldsymbol{\hat{\delta}}_{bc\hat{\bH}_{q,w}}$ &  & 0.47*** & 0.13 & -0.02 &
0.35*** & 0.14\\
\multicolumn{1}{r}{} &  & (3.22) & (0.92) & (-0.28) & (3.09) & \\
$\boldsymbol{\hat{\delta}}_{bcjk,w}$ &  & 0.49*** & 0.18 & 0.00 & 0.34*** &
0.14\\
\multicolumn{1}{r}{} &  & (3.32) & (1.27) & (0.03) & (3.03) & \\\hline
\end{tabular}

\begin{minipage}{10cm}
\vspace{0.1cm}
\vspace{0.1cm}
\small  Notes: Values in parentheses are t-ratio using HAC s.e. *, **, *** indicate significant at 10, 5 and 1\% level respectively.
\end{minipage}

\end{table}

The estimation results are summarized in Table \ref{table:emp1}. 
The value in the parentheses is the t-ratio based on Newey-West heteroskedasticity and autocorrelation consistent (HAC) standard errors with the threshold being the integer part of $T^{1/4}$. 
As can be seen, the augmented regressions with $\hat{\bF}$ and $\hat{\bF}_w$ give very similar results. The first factor and CP are significant at the 1\% level, while the test fails to reject $H_0 : \gamma_2 = 0$, but the magnitude of the t-statistic is larger with the jackknife estimator due to the positive upward correction.

It may be of practical interest to find out whether the orthonormal latent factors have the same explanatory power of $y_{t+12}$. In this application, it seems
reasonable to consider the pair $(f_{t1},f_{t2})$ for such a comparison.
Accordingly, we have computed the value of  $\hat{\gamma}_{1} -\hat{\gamma}_{2}$ for different estimates and the corresponding t-test statistics for $H_{0}:\gamma_{1}=\gamma_{2}$ versus $H_{1}:\gamma_{1}\neq\gamma_{2}$. Very similar estimates are obtained for the transformed model, $\bM_w \bX$, which seems reasonable given that the correlations between $\hat{\bF}$ and $CP$ are not strong. 

The results are summarized in Table \ref{table:emp2}. 
As can be seen, the test of the equal explanatory power of the first and  the second factors is rejected with the estimators $\hat{\bdelta}$, $\hat{\bdelta}_{bc\hat{\bH}}$ and $\hat{\bdelta}_{bc\hat{\bH}_q}$, but not rejected with $\hat{\bdelta}_{bcjk}$. This is due to the jackknife bias-correction which makes the estimate of the difference $\gamma_1 - \gamma_2$ smaller.  
Very similar comments apply to the results for the transformed model, $\bM_w \bX$.

\begin{table}[!htb]
\centering
\caption{The estimated difference $\gamma_1 - \gamma_2$ and the t-test statistic for $H_0 : \gamma_1 = \gamma_2$}
\label{table:emp2}
\begin{tabular}[c]{cccccc}
\hline
& \multicolumn{1}{c}{} & $\boldsymbol{\hat{\delta}}$ & $\boldsymbol{\hat
{\delta}}_{bc\mathbf{\hat{\bH}}}$ & $\boldsymbol{\hat{\delta}}_{bc\mathbf{\hat
	{\bH}}_{q}}$ & $\boldsymbol{\hat{\delta}}_{bcjk}$\\\hline
	$\gamma_{1}-\gamma_{2}$ &  & 0.34* & 0.36* & 0.34* & 0.30\\
	&  & (1.75) & (1.84) & (1.75) & (1.55)\\\hline
	&  & $\boldsymbol{\hat{\delta}}_{w}$ & $\boldsymbol{\hat{\delta}%
	}_{bc\mathbf{\hat{\bH}}_w}$ & $\boldsymbol{\hat{\delta}}_{bc\mathbf{\hat{\bH}}%
_{q,w}}$ & $\boldsymbol{\hat{\delta}}_{bcjk,w}$\\\hline
$\gamma_{w1}-\gamma_{w2}$ &  & 0.35* & 0.37* & 0.35* & 0.31\\
&  & (1.78) & (1.88) & (1.78) & (1.61)\\\hline
\end{tabular}

\begin{minipage}{10cm}
\vspace{0.1cm}
\vspace{0.1cm}
\small  Notes: Values in parentheses are t-ratio using HAC s.e. *, **, *** indicate significant at 10, 5 and 1\% level respectively.
\end{minipage}
\end{table}

\section{Conclusion}\label{sec:con}
In this paper we have studied the asymptotic bias of the least squares (LS) estimator for the factor-augmented model, $y_{t+h}={\bgamma^*}'\bff_t^*+\bbeta' \bw_t+\epsilon_{t+h}, t=1,\dots,T$,
replacing the latent factor $\bff_t^*$ with the principal component (PC) estimator $\hat{\bff}_t$ extracted from a large set of predictors, $\{x_{t,i}\}_{i=1}^N$. Unlike the existing literature, we allow for the predictors $x_{t,i}$ follow more general weak factor (WF) models, in which the $r$ largest eigenvalues of the sample covariance matrix of $x_{t,i}$ may diverge at different rates, $N^{\alpha _{k}}$, $0<\alpha _{k}\leq 1$, $k=1,\dots,r$. The literature typically assumes the strong factor (SF) model, in which $\alpha_1=\dots=\alpha_r=1$.

As discussed in \cite{BaiNg2023} and \cite{jiang2023revisiting}, there are choices of rotation matrices $\bR$ for approximations $\hat{\bff}_t = {\bR}'\bff_t^* + o_p(1)$. Accordingly, the first term of the augmented model is approximated by ${\bgamma^*}'\bff_t^*=\bgamma_{{\bR}}'\hat{\bff}_t + o_p(1)$, where $\bgamma_{{\bR}}={\bR}^{-1}\bgamma^*$, hence resulting in different ``parameters'' $\bdelta_{\bR}=(\bgamma_{{\bR}}',\bbeta')'$ estimated by the least squares estimator $\hat{\bdelta}=(\hat{\bgamma}',\hat{\bbeta}')'$, obtained by regressing $y_{t+h}$ on $(\hat{\bff}_t',\bw_t')$.
Replacing the regressor $\bff_t^*$ with $\hat{\bff}_t$ generally results in non-zero correlation between the regressor $(\hat{\bff}_t',\bw_t')$ and the replacement error, which can lead to an asymptotic bias of $\hat{\bdelta}$.  
We have studied the asymptotic bias of $\sqrt{T}(\hat{\bdelta}-\bdelta_{\bR})$ 
for three different choices of $\bR$ for the approximation: $\hat{\bH}$ which is data ($x_{t,i}$) dependent and commonly used for the approximation in the literature including \cite{BaiNg2006}, \cite{GoncalvesPerron2014,gonccalves2020bootstrapping}; another data dependent matrix $\hat{\bH}_q$, whose estimation error appears to be orthogonal to $\hat{\bff}_t$; and $\bH$, which is the population matrix $\bH$ of $\hat{\bH}$ and $\hat{\bH}_q$.

Our study has shown that if $\sqrt{T}/N^{(3\alpha_r - \alpha_1)/2} \to c_1\in[0,\infty)$ as $N,T\to\infty$, $\sqrt{T}(\hat{\bdelta}-\bdelta_{\hat{\bH}})$ has an asymptotic bias for WF models, which generalizes the results in \cite{GoncalvesPerron2014,gonccalves2020bootstrapping} for SF models. It is shown that the asymptotic bias expression for WF models is more complicated than for SF models, because the former depends on how many exponents $(\alpha_1,\dots,\alpha_r)$ are the same as $\alpha_1$ and $\alpha_r$.
We have also shown that if $\sqrt{T}/N^{\alpha_r} \to c_2\in[0,\infty)$ as $N,T\to\infty$, $\sqrt{T}(\hat{\bdelta}-\bdelta_{\hat{\bH}_q})$ has an asymptotic bias that is generally smaller in magnitude than that of $\sqrt{T}(\hat{\bdelta}-\bdelta_{\hat{\bH}})$ and also the convergence rate is generally not slower than that of $\sqrt{T}(\hat{\bdelta}-\bdelta_{\hat{\bH}})$, and even the former is faster when $\alpha_1\neq\alpha_r$.
The structure of the asymptotic bias depends on how many exponents $(\alpha_1,\dots,\alpha_r)$ are the same as $\alpha_r$.
Importantly, the asymptotic biases of $\sqrt{T}(\hat{\bdelta}-\bdelta_{\hat{\bH}})$ and $\sqrt{T}(\hat{\bdelta}-\bdelta_{\hat{\bH}_q})$ are parametrically estimable; thus, analytical bias corrections are feasible in practice.
Moreover, it turns out that the asymptotic bias of $\sqrt{T}(\hat{\bdelta}-\bdelta_{\hat{\bH}_q})$ disappears completely when $\bw_t$ and $\bff_t^*$ are uncorrelated. To exploit this property, we propose to extract the factor from the predictors ($x_{t,i}$) after projecting out the observable factor, $\bw_t$.

We have also studied the asymptotic bias with the population rotation matrix, $\sqrt{T}(\hat{\bdelta}-\bdelta^0)$ with $\bdelta^0:=\bdelta_{{\bH}}$, where $\bdelta^0=(\bgamma^{0\prime},\bbeta')'$, $\bgamma^0:= {\bH}^{-1}\bgamma^{*}$. It is shown that if $\sqrt{T}/N^{(3\alpha_r - \alpha_1)/2} \to c_1\in[0,\infty)$ as $N,T\to\infty$, $\sqrt{T}(\hat{\bdelta}-\bdelta^0)$ has an asymptotic bias which, unlike the biases with $\hat{\bH}$ and $\hat{\bH}_q$, cannot be estimated parametrically.
In view of this, we have proposed to use a subsampling method, called a split-panel jackknife bias correction, which is generally less computationally expensive than bootstrapping, while allowing for more general cross and serial correlations in $e_{t,i}$.

The finite sample evidence has shown that $\sqrt{T}(\hat{\bdelta}-\bdelta_{\hat{\bH}_q})$ has the least bias, the least size distortion of t-tests, with the smallest standard errors, while the performance of $\sqrt{T}(\hat{\bdelta}-\bdelta_{\hat{\bH}})$ is far worse than others throughout the design. Our preferred jackknife bias-corrected estimator with respect to the parameter $\bdelta^0$ comes in second, closely following the performance of the bias-corrected estimator with respect to $\bdelta_{\hat{\bH}_q}$.
In addition, the empirical power curve of the significance test shows that, on balance, the split-panel jackknife estimator appears to be the most reliable among the estimators compared.

We apply the bias-corrected estimator to the factor-augmented forecast regression of bond yields $y_{t+h}$ on the factors extracted from 131 monthly macroeconomic series and one observed predictor, the \cite{CochranePiazzesi2005} factor, over the period January 1982 to December 2002. The results show that the jackknife appears to effectively correct the bias of the LS estimator, thus providing more reliable inference.

A couple of implications follow from the results of this paper. First, the different approximation for $\hat{\bff}_t$ with different rotation matrices that are asymptotically equivalent implies that the LS estimator $\hat{\bdelta}$ estimates different ``parameters'' and they may have different asymptotic biases. Therefore, the researchers should clarify which ``parameter'' they are estimating with $\hat{\bdelta}$. This is very important because in finite samples the ``parameters'' for different rotation matrices can take very different values, as shown in Figure \ref{fig:mean_gamma_hats}. Our recommendation is to primarily consider $\bdelta^0$ as the parameter estimated by $\hat{\bdelta}$. 
Second, the bias in $\hat{\bdelta}$ can be significant and should not be ignored in practice. The empirical results in Table \ref{table:emp1} illustrate the effectiveness of the jackknife bias correction of $\hat{\bdelta}$ relative to $\bdelta^0$. As shown in Table \ref{table:emp2}, the test results for parameter restrictions with jackknife bias correction are significantly different from those without bias correction.
Note that, as our theory tells that the bias correction is not necessary 
in the special case of joint significance testing of all factors, since there is no bias in the LS estimator under the null hypothesis of $\bgamma^*=\bzero$ (all asymptotic biases in Theorems \ref{thm:bias_Hhat}-\ref{thm:bias_Hw3} disappear).

Finally, while we have focused on the analysis of the asymptotic bias of $\hat{\bdelta}$ in this paper, it is also of great interest to extend our analysis to bootstrapping $\sqrt{T}(\hat{\bdelta}-\bdelta_{\bR})$ for different choices of $\bR$. In particular, the advantage of the bootstrapping is that it can provide a higher-order approximation not only for the asymptotic bias but also for the distribution of $\sqrt{T}(\hat{\bdelta}-\bdelta_{\bR})$. This line of research is pursued in a companion paper, \cite{jiang2024bootstrap}, to which interested readers may wish to refer.

\section*{Acknowledgment}
We are grateful to Jia Chen, Naoko Hara, Yohei Yamamoto and Yang Zu for helpful discussions and useful comments.

\section*{Funding}
This work was supported by JSPS KAKENHI (grant numbers 21H00700, 21H04397, 23K25501 and 24K16343).

\bibliographystyle{chicago}
\bibliography{references_wfr}

\newpage
\appendix
\setcounter{page}{1}
\setcounter{section}{0}

\begin{center}
{\Large Supplementary Material for \\[7mm]
{\LARGE Bias Correction in Factor-Augmented Regression Models with Weak Factors}} \\[10mm]
\textsc{\large Peiyun Jiang$^\dagger$,} 
\textsc{\large Yoshimasa Uematsu$^*$,} \textsc{\large Takashi Yamagata$^\ddagger$} \\[5mm]
$^\dagger$\textit{\large Faculty of Economics and Business Administration, Tokyo Metropolitan University} \\[1mm]
$^*$\textit{\large Department of Social Data Science, Hitotsubashi University} \\[1mm]
$^\ddagger$\textit{\large Department of Economics and Related Studies, University of York} \\[1mm]
$^\ddagger$\textit{\large Institute of Social Economic Research, Osaka University}
\end{center}

\section{Proofs of results in Section \ref{sec:bias}}
\setcounter{lem}{0}
\renewcommand{\thelem}{A.\arabic{lem}}
\renewcommand{\theequation}{A.\arabic{equation}}
\setcounter{equation}{0}

To simplify the proof, we introduce new rotation matrices defined as follows: 
\[
\tilde{\bH}= {\bB^0}'\bB^0\frac{{\bF^0}'\hat{\bF}}{T}{\hat{\bLambda}}^{-1}, \;	\tilde{\bH}_q= \left(\frac{1}{T}{\hat{\bF}}'\bF^0\right)^{-1},  \; \tilde{\bH}_b= \bB^{0\prime} {\hat{\bB}}\left( {\hat{\bB}}'\hat{\bB}\right)^{-1},\; \tilde{\bQ}=  \frac{1}{T}{\hat{\bF}}'\bF^0.
\]
As shown in \citep[Lemma B.4]{jiang2023revisiting}, since $\bF^0$ and $\bB^0$ satisfy the PC1 conditions in \cite{BaiNg2013}, these rotation matrices denoted with “tilde” are asymptotically equivalent to $\bI_r$.
Moreover, the equality $\hat{\bH}=\bH\tilde{\bH}$ implies the equivalences $ \bF^* \hat{\bH} = \bF^0 \tilde{\bH}$ and $ \bF^* \hat{\bH}_q = \bF^0 \tilde{\bH}_q.$ Given the good properties of the rotation matrices with “tilde”, we derive the following lemmas concerning these matrices rather than those with “hat” to prove the theorems.

\begin{lem}
\label{lem:Rotation} Define $\Delta_{NT}
=
\frac{N^{1-\alpha_r}}{T}+ N^{\frac{1}{2}\alpha_1- \alpha_r} \frac{N^{1- \alpha_r}}{T} + N^{\frac{1}{2}\alpha_{1}-\frac{3}{2}\alpha_{r}} +\frac{N^{\frac{1}{2}\alpha_{1}-\alpha_{r}}}{\sqrt{T}}.$
Suppose that Assumptions \ref{ass:eigen}--\ref{ass:Aug_errors} hold. If $\frac{N^{1-\alpha_r}}{T}\to 0$, then, we have
\begin{flalign*}
&(i) ~		\left\|  \frac{1}{\sqrt{T}} (\hat{\bF}- \bF^0\tilde{\bH}_b) \right\|_{\F} 
=  O_p \left( \frac{N^{1- \alpha_r}}{T}  \right)+ O_p \left(  N^{-\frac{1}{2}\alpha_{r}} \right), & \\
&(ii) ~		\left\|  \frac{1}{T} \bE'(\hat{\bF}- \bF^0\tilde{\bH}_b) \right\|_{\F} 
=  O_p \left( \frac{N^{1-\frac{1}{2}\alpha_r}}{T}  \right)+ O_p \left(  N^{-\frac{1}{2}\alpha_{r}} \right), & \\
& (iii) ~		\left\|  \frac{1}{T} {\bB^0}'\bE'(\hat{\bF}- \bF^0\tilde{\bH}_b) \right\|_{\F} 
= O_p \left( N^{\frac{1}{2}\alpha_{1}-\frac{1}{2}\alpha_{r}}\right)+ O_p \left(  N^{\frac{1}{2}\alpha_1-\frac{1}{2}\alpha_r} \frac{N^{1-\frac{1}{2}\alpha_r}}{T} \right), & \\
&(iv) ~	 	\left\|  \frac{1}{T} \bN^{-\frac{1}{2}}{\bB^0}'\bE'(\hat{\bF}- \bF^0\tilde{\bH}_b) \right\|_{\F} 
= O_p \left( N^{ -\frac{1}{2}\alpha_{r}}\right)+ O_p \left(    \frac{N^{1- \alpha_r}}{T} \right), & \\
&(v) ~		\left\| \frac{1}{T}{\bF^0}'(\hat{\bF}-\bF^0\tilde{\bH}) \right\|_{\F} 
= O_p \left(\Delta_{NT} \right), & \\
&(vi)~ \left\|	 \frac{{\hat{\bF}}'\bF^0}{T}-\bI_r \right\|_{\F} = O_p(\Delta_{NT}) .
\end{flalign*}	
\end{lem}

\begin{proof}[Proof of Lemma \ref{lem:Rotation}] The proofs have been given in \citep[Lemma 1 and Lemmas B.3-B.4]{jiang2023revisiting}.
\end{proof}

\begin{lem} Suppose that Assumptions \ref{ass:eigen}--\ref{ass:Aug_errors} hold.  If $\frac{N^{1-\alpha_r}}{T}\to 0$, then, we have, as $N, T \to \infty$,
\label{lem:F*epsilon}
\begin{flalign}
&(i) ~ \left\| \frac{1}{\sqrt{T}} (\hat{\bF}-\bF^0\tilde{\bH})'\bepsilon  \right\|_{\F}=O_p \left(\sqrt{T} N^{-\frac{3}{2}\alpha_r}\right) +O_p\left( \frac{N^{1-\alpha_r}}{\sqrt{T}} \right)+O_p\left( N^{\frac{1}{2}- \alpha_r} \right),\label{eq:F*ep_H}&\\
&(ii) ~ \left\| \frac{1}{\sqrt{T}} (\hat{\bF}-\bF^0\tilde{\bH}_q)'\bepsilon  \right\|_{\F}=O_p \left(\sqrt{T} N^{-\frac{3}{2}\alpha_r}\right) +O_p\left( \frac{N^{1-\alpha_r}}{\sqrt{T}} \right)+O_p\left( N^{\frac{1}{2}- \alpha_r} \right).\label{eq:F*ep_H3}&
\end{flalign}
\end{lem}

\begin{proof}[Proof of Lemma \ref{lem:F*epsilon}]
From the definition $\hat{\bB}= \bX^{\prime} \hat{\bF}/T$, we can express $\hat{\bB}$ as
\begin{align}
\label{Bhat}
\hat{\bB}&=\frac{1}{T} \bB^0 {\bF^0}^{\prime} \hat{\bF}+\frac{1}{T} \bE^{\prime} \hat{\bF}  \\
\nonumber
{\bB^0}' \hat{\bB}
&=\frac{1}{T} {\bB^0}' \bB^0 {\bF^0}' \hat{\bF}+\frac{1}{T} {\bB^0}' \bE^{\prime} \hat{\bF}  \\
{\bB^0}' \hat{\bB} \hat{\bLambda}^{-1}&=\frac{1}{T} {\bB^0}' \bB^0 {\bF^0}' \hat{\bF}\hat{\bLambda}^{-1}+\frac{1}{T} {\bB^0}' \bE^{\prime} \hat{\bF} \hat{\bLambda}^{-1}.
\label{eq:H4-H}
\end{align}
By pre-multiplying $\hat{\bB}'$ to \eqref{Bhat}, we obtain
\begin{eqnarray*}
{\hat{\bB}}' \hat{\bB}&=&\frac{1}{T} {\hat{\bB}}' \bB^0 {\bF^0}' \hat{\bF}+\frac{1}{T} {\hat{\bB}}' \bE^{\prime} \hat{\bF} \\
{\hat{\bB}}' \hat{\bB}&=&\frac{1}{T}\hat{\bF}'\bF^0 {\bB^0}' \hat{\bB} +\frac{1}{T} \hat{\bF}' \bE{\hat{\bB}}  \\
\tilde{\bH}_q &=& {\bB^0}' \hat{\bB} \hat{\bLambda}^{-1}+ \left( \frac{1}{T}\hat{\bF}' {\bF^0}\right)^{-1}\frac{1}{T} \hat{\bF}' \bE{\hat{\bB}}\left({\hat{\bB}}' \hat{\bB}\right)^{-1}.
\end{eqnarray*}
Since $\hat{\bLambda}$ and $\hat{\bF}$ are the eigenvalues and eigenvectors of $\bX\bX'/T$, respectively, we have
\begin{align*} 
T^{-1}\bX\bX'\hat{\bF}
= \hat{\bF}\hat{\bLambda}.
\end{align*}
Using the model $\mathbf{X}=\bF^0 \bB^{0\prime}+\mathbf{E}$, $\bX\bX'$ can be expanded as
\begin{align*}\label{XX}
\bX\bX'
= \bF^0{\bB^0}'\bB^0{\bF^0}' 
+ \bE\bB^0{\bF^0}'
+ \bF^0{\bB^0}'\bE'
+\bE\bE'.
\end{align*}
The estimated factor matrix $\hat{\bF}$ can be written as
\[
\hat{\bF}=\frac{1}{T}\left( \bF^0{\bB^0}'\bB^0{\bF^0}' 
+ \bE\bB^0{\bF^0}'
+ \bF^0{\bB^0}'\bE'
+\bE\bE' \right) \hat{\bF}\hat{\bLambda}^{-1}.
\]	
Thus, we obtain
\begin{align}
\label{fhat-fh_hat}
& \hat{\bF} - \bF^0 \tilde{\bH} =  \frac{1}{T} \bE\bE'\hat{\bF}\hat{\bLambda}^{-1}+\frac{1}{T} \bE\bB^0
{\bF^0 }'\hat{\bF} \hat{\bLambda}^{-1} + \frac{1}{T}  \bF^0
{\bB^0 }'\bE'\hat{\bF} \hat{\bLambda}^{-1}, \\
\nonumber
&\hat{\bF}  - \bF^0 \tilde{\bH}_q =  \frac{1}{T} \bE\bE'\hat{\bF}\hat{\bLambda}^{-1}+\frac{1}{T} \bE\bB^0
{\bF^0 }'\hat{\bF}   \hat{\bLambda}^{-1} - \bF^0  \left( \frac{1}{T}\hat{\bF}' {\bF^0}\right)^{-1}\frac{1}{T} \hat{\bF}' \bE{\hat{\bB}}\left({\hat{\bB}}' \hat{\bB}\right)^{-1}.\\
\label{fhat-fh_q}
\end{align}
(i) Applying the expansion in \eqref{fhat-fh_hat},
\begin{align*}
&    \frac{1}{\sqrt{T}} (\hat{\bF}-\bF^0\tilde{\bH})'\bepsilon   \\
&= \frac{1}{T^{3/2}} \hat{\bLambda}^{-1} {\hat{\bF}}' \bE\bE' \bepsilon  + \frac{1}{T^{3/2}}  \hat{\bLambda}^{-1}{\hat{\bF}}'\bF^0 {\bB^0}'{\bE}'\bepsilon    +\frac{1}{T^{3/2}}  \hat{\bLambda}^{-1}{\hat{\bF}}'\bE \bB^0 {\bF^0}'\bepsilon  \\
&= \bA_1+\bA_2+\bA_3 .
\end{align*}
We derive the order of magnitude for three terms $\bA_1$, $\bA_2$ and $\bA_3$.  By the argument given in \citep[Lemmas B.1--B.2]{jiang2023revisiting}, we know that $\bN^{-1}\hat{\bLambda}=O_p(1)$ and $\bN^{-\frac{1}{2}} \frac{1}{T} \hat{\bF}^{\prime} \bF^0 \bN^{\frac{1}{2}}=O_p(1)$. We decompose $\bA_1$ as 
\begin{align}
\nonumber
&  \left\|   \hat{\bF}^{\prime} \bE\bE' \bepsilon  \right\|_{\F} \\
\nonumber
& \le  \left\| (\hat{\bF}-\bF^0\tilde{\bH}_b)' \bE\bE' \bepsilon \right\|_{\F} +\left\| \tilde{\bH}_b' \bF^{0\prime} \bE\bE' \bepsilon \right\|_{\F} \\
\nonumber
& \le \left\| (\hat{\bF}-\bF^0\tilde{\bH}_b)' \bE \right\|_{\F} \left\| \bE' \bepsilon \right\|_{\F} 
+ \left\|  \bF^{0\prime} \bE \right\|_{\F} \left\|\bE' \bepsilon \right\|_{\F}\\
\nonumber
& =T O_p\left( \frac{N^{1-\frac{1}{2}\alpha_r}}{T}+N^{-\frac{1}{2}\alpha_r} \right) O_p((\sqrt{N}+\sqrt{T})\sqrt{T}) +O_p(\sqrt{NT}(\sqrt{N}+\sqrt{T})\sqrt{T}) \\
& = O_p(T^2N^{-\frac{1}{2}\alpha_r})+O_p(NT)+O_p(\sqrt{N}T^{3/2}),
\label{Feee}
\end{align}
where we used Lemma \ref{lem:Rotation}(ii) and $\frac{N^{1-\alpha_r}}{T}\to 0$.
Thus, we have
\begin{align*}
&   \left\|  \bA_1  \right\|_{\F}  
= O_p \left(\sqrt{T} N^{-\frac{3}{2}\alpha_r}\right) + O_p \left(\frac{N^{1-  \alpha_r}}{\sqrt{T}}\right)+ O_p \left(N^{\frac{1}{2}- \alpha_r} \right).
\end{align*}
For $\bA_2$,
\begin{align*}
\left\|  \bA_2  \right\|_{\F}  & =\left\|\frac{1}{\sqrt{T}} \hat{\bLambda}^{-1} \bN^{\frac{1}{2}}\left(\bN^{-\frac{1}{2}} \frac{1}{T} \hat{\bF}^{\prime} \bF^0 \bN^{\frac{1}{2}}\right)\left(\bN^{-\frac{1}{2}} \bB^{0\prime} \bE^{\prime} \bepsilon\right)\right\|_{\F} \\
& =O_p\left( \frac{1}{\sqrt{T}} N^{-\frac{1}{2} \alpha_r} \sqrt{T}\right)=O_p\left( \frac{1}{\sqrt{ N^{\alpha_r}}}\right),
\end{align*}
where we have used $ \|T^{-\frac{1}{2}}\bN^{-\frac{1}{2}} \bB^{0\prime} \bE^{\prime} \bepsilon\|_{\F}=O_p(1)$ in Assumption \ref{ass:2errors}. For $\bA_3$, Lemma \ref{lem:Rotation}(iii) implies that
\begin{align*}
\left\|  \bA_3  \right\|_{\F} &\le \left\|   \frac{1}{T^{3/2}} \hat{\bLambda}^{-1}  \left(\hat{\bF}-\bF^{0} \tilde{\bH}_b\right)^{\prime} \bE \bB^0  \bF^{0\prime} \bepsilon \right\|_{\F} 
+\left\|\frac{1}{T^{3/2}} \hat{\bLambda}^{-1} \tilde{\bH}_b^{\prime} \bF^0 \bE \bB^{0} \bF^{0\prime} \bepsilon\right\|_{\F} \\
&=
O_p\left(N^{-\alpha_r}\right) O_p \left( N^{\frac{1}{2}\alpha_{1}-\frac{1}{2}\alpha_{r}}+  N^{\frac{1}{2}\alpha_1-\frac{1}{2}\alpha_r} \frac{N^{1-\frac{1}{2}\alpha_r}}{T} \right)+O_p\left( \frac{1}{T^{3/2}} N^{-\alpha_r}\sqrt{TN^{\alpha_1}}\sqrt{T}  \right)\\
&=O_p\left( N^{\frac{1}{2}\alpha_{1}-\alpha_{r}}N^{-\frac{1}{2}\alpha_{r}}  \right)
+ O_p\left(N^{\frac{1}{2}\alpha_{1}-\alpha_{r}} \frac{ N^{1-\alpha_r}}{T} \right) 
+O_p\left(   \frac{N^{\frac{1}{2}\alpha_1-\alpha_r}}{\sqrt{T}}   \right),
\end{align*}
by $ \|T^{-\frac{1}{2}}\bN^{-\frac{1}{2}} \bB^{0\prime} \bE^{\prime} \bF^0\|_{\F}=O_p(1)$ in Assumption \ref{ass:factor and loadings}(iv).

Since $ \frac{N^{1-\alpha_r}}{T}\to 0$, the term $\bA_1$ is the dominant component. Thus, we obtain
\[
\left\| \frac{1}{\sqrt{T}} (\hat{\bF}-\bF^0\tilde{\bH})'\bepsilon  \right\|_{\F}
=O_p \left(\sqrt{T} N^{-\frac{3}{2}\alpha_r}\right) +O_p\left( \frac{N^{1-\alpha_r}}{\sqrt{T}} \right)
+O_p\left( N^{\frac{1}{2}- \alpha_r} \right).
\]
(ii) Using the expansion \eqref{fhat-fh_q}, we can write
\begin{align}
&    \frac{1}{\sqrt{T}} (\hat{\bF}-\bF^0\tilde{\bH}_q)'\bepsilon   \\
&= \frac{1}{T^{3/2}} \hat{\bLambda}^{-1} {\hat{\bF}}' \bE\bE' \bepsilon  + \frac{1}{T^{3/2}}  \hat{\bLambda}^{-1}{\hat{\bF}}'\bF^0 {\bB^0}'{\bE}'\bepsilon  -  \frac{1}{\sqrt{T}} \left({\hat{\bB}}' \hat{\bB}\right)^{-1}\frac{1}{T}{\hat{\bB}}' \bE'\hat{\bF}\left( \frac{1}{T}\hat{\bF}' {\bF^0}\right)^{-1}\bF^{0\prime}  \bepsilon  .
\label{fhat-fh_q*eplison}
\end{align}
Since the first two terms on the right-hand side of \eqref{fhat-fh_q*eplison} are the same as $\bA_1$ and $\bA_2$, we focus only on the third component. We first show the upper bound of the following term:
\begin{align*}
& \frac{1}{T} \hat{\bF}' \bE{\hat{\bB}}\left({\hat{\bB}}' \hat{\bB}\right)^{-1}
\\
&= 
\frac{1}{T} (\hat{\bF}-\bF^0\tilde{\bH}_b)' \bE(\hat{\bB}-\bB^0\tilde{\bQ}')\left({\hat{\bB}}' \hat{\bB}\right)^{-1}\\
& \quad\quad +
\frac{1}{T} (\hat{\bF}-\bF^0\tilde{\bH}_b)' \bE \bB^0\tilde{\bQ}' \left({\hat{\bB}}' \hat{\bB}\right)^{-1}
+
\frac{1}{T} \tilde{\bH}_b'\bF^{0\prime} \bE \hat{\bB}\left({\hat{\bB}}' \hat{\bB}\right)^{-1}
\\
&  = \bB_{1}+ \bB_{2}+ \bB_{3}.
\end{align*}
We analyze each term separately. The definition 
$
\hat{\bB} = \frac{1}{T} \bX'\hat{\bF} 
=\frac{1}{T}  \bB^0{\bF^0}'\hat{\bF} +\frac{1}{T} \bE'\hat{\bF}
$
implies 
\begin{align}
\hat{\bB}- \bB^0\tilde{\bQ}' = \frac{1}{T} \bE'(\hat{\bF}- \bF^0 \tilde{\bH}_b)+ \frac{1}{T} \bE'\bF^0 \tilde{\bH}_b. \label{B}
\end{align}	 
Since $\tilde{\bH}_b=O_p(1)$ in \citep[Lemma 2(iii)]{jiang2023revisiting} and $\frac{N^{1-\alpha_r}}{T}\to 0$, it follows that 	
\begin{align}
& \left\| \hat{\bB} - \bB^0 \tilde{\bQ}' \right\|_{\F}\\
& \le  \left\|\frac{1}{T} \bE'(\hat{\bF}- \bF^0 \tilde{\bH}_b)\right\|_{\F}+\left\| \frac{1}{T} \bE'\bF^0 \tilde{\bH}_b \right\|_{\F}\\
&=  O_p\left(\frac{N^{1-\frac{1}{2}\alpha_r}}{T}\right)+ O_p\left(N^{-\frac{1}{2}\alpha_{r}}\right) +   O_p\left(\sqrt{\frac{N}{T}} \right)\\
&= O_p\left(\sqrt{\frac{N}{T}}\right)+O_p\left(N^{-\frac{1}{2}\alpha_{r}}\right).
\label{consis of Bhat}
\end{align}
For the first term $\bB_1$, Lemma \ref{lem:Rotation}(ii) and \eqref{consis of Bhat} imply
\begin{align*}
\left\|  \bB_1 \right\|_{\F} 
& \le  O_p \left( \left(\frac{N^{1-\frac{1}{2}\alpha_r}}{T} +N^{- \frac{1}{2}\alpha_r}\right) \left(\sqrt{\frac{N}{T}}+N^{-\frac{1}{2}\alpha_{r}} \right)N^{-\alpha_r}\right)\\ 
&=
O_p \left(  \frac{N^{\frac{3}{2}-\frac{3}{2}\alpha_r}}{T^{\frac{3}{2}}}  \right) 
+
O_p \left(   \frac{N^{ \frac{1}{2} -\frac{3}{2}\alpha_r}}{\sqrt{T}}  \right)
+
O_p \left( \frac{N^{1-2\alpha_r}}{T}  \right)+
O_p \left(  N^{-2\alpha_r}   \right).
\end{align*}
The dominating term is $ \bB_{2}$:
\begin{align*}
\left\|  \bB_2 \right\|_{\F}
&=\left\|  \frac{1}{T} (\hat{\bF}-\bF^0\tilde{\bH}_b)' \bE \bB^0\tilde{\bQ}' \left({\hat{\bB}}' \hat{\bB}\right)^{-1} \right\|_{\F}\\
&=
O_p \left( \left( N^{ -\frac{1}{2}\alpha_{r}} +     \frac{N^{1- \alpha_r}}{T} \right)N^{-\frac{1}{2}\alpha_r}\right) \\
&=
O_p \left(  N^{ - \alpha_{r}}\right) +     O_p \left( \frac{N^{1- \alpha_r}}{T}  N^{-\frac{1}{2}\alpha_r}\right),
\end{align*}
where we have used Lemma \ref{lem:Rotation}(iv). For $\bB_3$, adding and subtracting terms, we have
\begin{align*}
\left\|  \bB_{3} \right\|_{\F} 
& \le \left\| \frac{1}{T} \tilde{\bH}_b'\bF^{0\prime} \bE (\hat{\bB}-\bB^0\tilde{\bQ}')\left(\hat{\bB}' \hat{\bB}\right)^{-1}\right\|_{\F}
+ \left\| \frac{1}{T} \tilde{\bH}_b'\bF^{0\prime} \bE \bB^0\tilde{\bQ}'\left(\hat{\bB}' \hat{\bB}\right)^{-1}\right\|_{\F}\\
& = 
O_p \left(  \frac{\sqrt{NT}}{T} \left( \sqrt{\frac{N}{T}}+ N^{- \frac{1}{2}\alpha_r}\right)N^{-\alpha_r} \right) 
+
O_p \left(  \frac{1}{\sqrt{T} }  N^{- \frac{1}{2}\alpha_r} \right) 
\\
&= O_p \left(   \frac{N^{1-\alpha_r}}{T}  \right) 
+O_p \left(  \frac{N^{\frac{1}{2}-\frac{3}{2}\alpha_r} }{\sqrt{T}} \right) 
+O_p \left(   \frac{1}{\sqrt{T N^{\alpha_r}}} \right) ,
\end{align*}
where $\|\bF^{0\prime}\bE\|_{\F}=O_p(\sqrt{NT})$ as shown in \citep[Lemma A.1]{BaiNg2023}. Thus, we obtain
\begin{align}
\label{FEB}
\left\| \frac{1}{T}  \hat{\bF}^{'}\bE \hat{\bB} ({\hat{\bB}}'\hat{\bB})^{-1}\right\|_{\F} 
= O_p \left(   \frac{N^{1-\alpha_r}}{T}  \right)  
+ O_p \left(  N^{ - \alpha_{r}}\right)
+O_p \left(   \frac{1}{\sqrt{T N^{\alpha_r}}} \right).
\end{align}
It follows that the third term on the right-hand side of \eqref{fhat-fh_q*eplison} is bounded by
\begin{align*}
&\left\| \frac{1}{\sqrt{T}} \left({\hat{\bB}}' \hat{\bB}\right)^{-1}\frac{1}{T}{\hat{\bB}}' \bE'\hat{\bF}\left( \frac{1}{T}\hat{\bF}' {\bF^0}\right)^{-1}\bF^{0\prime}  \bepsilon\right\|_{\F} \\
& \le \left\| \left({\hat{\bB}}' \hat{\bB}\right)^{-1}\frac{1}{T}{\hat{\bB}}' \bE'\hat{\bF}\right\|_{\F} \left\|\left( \frac{1}{T}\hat{\bF}' {\bF^0}\right)^{-1} \right\|_{\F} \left\|\frac{1}{\sqrt{T}} \bF^{0\prime}  \bepsilon\right\|_{\F}\\
& = O_p \left(   \frac{N^{1-\alpha_r}}{T}  \right)  
+ O_p \left(  N^{ - \alpha_{r}}\right)
+O_p \left(   \frac{1}{\sqrt{T N^{\alpha_r}}} \right),
\end{align*}
which is not larger than $\bA_1$. Thus, \eqref{fhat-fh_q*eplison} has the same bound as (i):
\[
\left\| \frac{1}{\sqrt{T}} (\hat{\bF}-\bF^0\tilde{\bH}_q)'\bepsilon  \right\|_{\F}
=O_p \left(\sqrt{T} N^{-\frac{3}{2}\alpha_r}\right) +O_p\left( \frac{N^{1-\alpha_r}}{\sqrt{T}} \right)
+O_p\left( N^{\frac{1}{2}- \alpha_r} \right).
\]
\end{proof}

\begin{lem} 
\label{lem:bias_Hhat}
Suppose Assumptions \ref{ass:eigen}--\ref{ass:Aug_errors} hold. If $\alpha_r>\frac{1}{2}$, $\frac{N^{1-\alpha_r}}{\sqrt{T}} \to  0$, and $\sqrt{T}N^{\frac{1}{2}\alpha_1-\frac{3}{2}\alpha_r} \to c_1 \in [0,\infty)$, as $N, T \to \infty$, we have
\begin{flalign*}
&(i) ~  \frac{1}{\sqrt{T}}  \hat{\bF}^{\prime}(\hat{\bF}-\bF^0\tilde{\bH})   =c_1   (\bG+\nu\bD^{-1}\bGamma\bD^{-1})+o_p\left( 1 \right),&\\
&(ii) ~  \frac{1}{\sqrt{T}}  \bW'(\hat{\bF}-\bF^0\tilde{\bH})   =c_1  \bSigma_{\bW \bF^0} \bG+o_p\left( 1 \right),&
\end{flalign*}
where $c_1 \bG = \lim_{N,T\to\infty} \sqrt{T}\bN^{\frac{1}{2}} \bGamma  \bD^{-2} \bN^{-\frac{3}{2}} $, $\nu = \lim_{N\to\infty} N^{-\frac{1}{2}(\alpha_1-\alpha_r)}$ 
and\\ $\bSigma_{\bW \bF^0} = \plim_{N,T \rightarrow \infty} T^{-1}\bW'\bF^0$.
\end{lem}

\begin{proof}[Proof of Lemma \ref{lem:bias_Hhat}] (i) Applying the expansion in \eqref{fhat-fh_hat}, we can rewrite
\begin{align*}
\frac{1}{\sqrt{T}} \hat{\bF} '(\hat{\bF}-\bF^0\tilde{\bH}) 
& =\frac{1}{\sqrt{T}} \hat{\bF}'\left( \frac{1}{T} \bE\bE'\hat{\bF}\hat{\bLambda}^{-1}+\frac{1}{T} \bE\bB^0
{\bF^0 }'\hat{\bF} \hat{\bLambda}^{-1} + \frac{1}{T}  \bF^0
{\bB^0 }'\bE'\hat{\bF} \hat{\bLambda}^{-1}\right)   \\
&=\bC_1+\bC_2+\bC_3.
\end{align*}
Consider the first term $\bC_1$, 
\begin{align*}
\left\| \bC_1 \right\|_{\F} 
& \le   \left\|   \frac{1}{T^{3/2}}(\hat{\bF}-\bF^0\tilde{\bH}_q)' \bE\bE'\hat{\bF} \hat{\bLambda}^{-1}\right\|_{\F} 
+ \left\|   \frac{1}{T^{3/2}}\tilde{\bH}_q' \bF^{0\prime} \bE\bE'\hat{\bF} \hat{\bLambda}^{-1}\right\|_{\F}\\
& = O_p \left( \left( \frac{N^{1-\frac{1}{2}\alpha_r}}{T}+N^{-\frac{1}{2}\alpha_r}  \right)(\sqrt{N}+\sqrt{T})N^{-\alpha_r}\right) 
+
O_p \left(   \frac{\sqrt{NT}}{T}  (\sqrt{N}+\sqrt{T})N^{-\alpha_r}\right) \\
&= O_p   \left( \frac{N^{1- \alpha_r}}{\sqrt{T}}+ \sqrt{T}N^{-\frac{3}{2}\alpha_r}+N^{\frac{1}{2}-\alpha_r}  \right) ,
\end{align*}
where we have used Lemma \ref{lem:Rotation}(ii) and Assumption \ref{ass:errors}(iv). The second term can be decomposed into
\begin{align*}
\bC_2 
&  =   \sqrt{T} \left[ \frac{1}{T}(\hat{\bF}-\bF^0\tilde{\bH}_b)' \bE\bB^0 \right] 
\frac{{\bF^0 }'\hat{\bF}}{T}  \hat{\bLambda}^{-1}+\frac{1}{\sqrt{T}}  \tilde{\bH}_b'   \bF^{0\prime}\bE\bB^0
\frac{{\bF^0 }'\hat{\bF}}{T}  \hat{\bLambda}^{-1}\\
& 
= \frac{1}{T^{3/2}} \hat{\bLambda}^{-1} {\hat{\bF}}' \bE\bE' \bE\bB^0 
\frac{{\bF^0 }'\hat{\bF}}{T} \hat{\bLambda}^{-1}  + \frac{1}{T^{3/2}}  \hat{\bLambda}^{-1}{\hat{\bF}}'\bF^0 {\bB^0}'{\bE}'\bE\bB^0 
\frac{{\bF^0 }'\hat{\bF}}{T} \hat{\bLambda}^{-1}   \\
&\quad \quad 
+\frac{1}{\sqrt{T}}  \tilde{\bH}'   \bF^{0\prime}\bE\bB^0
\frac{{\bF^0 }'\hat{\bF}}{T}  \hat{\bLambda}^{-1}
\\
& =\bC_{21}+\bC_{22}+\bC_{23}.
\end{align*}
Taking norms,
\begin{align*}
& \left\| \bC_{21} \right\|_{\F} 
=   O_p \left(  N^{-\alpha_r} \left( \frac{N}{T} + 1\right)\sqrt{T}N^{-\frac{1}{2}\alpha_r} \right) 
=
O_p \left(  \frac{N^{1-\alpha_r}}{\sqrt{T}}N^{-\frac{1}{2}\alpha_r} \right) 
+
O_p \left(  \sqrt{T}N^{ -\frac{3}{2} \alpha_r}  \right),\\
& \left\|\bC_{22}\right\|_{\F} 
= 
O_p \left(  \sqrt{T}N^{ - \alpha_r}  \right) ,\\
& \left\| \bC_{23} \right\|_{\F} 
=
O_p \left(   N^{-\frac{1}{2}\alpha_r} \right) ,
\end{align*}
since $\|\bE\bB^0\bN^{-\frac{1}{2}}\|_{\F}=O_p(\sqrt{T})$ under Assumption \ref{ass:factor and loadings}(ii).
Now considering
\begin{align*}
\bC_3  &   =
\frac{1}{\sqrt{T}} \frac{\hat{\bF}'\bF^0 }{T}
{\bB^0 }'\bE'(\hat{\bF}-\bF^0\tilde{\bH}_b ) \hat{\bLambda}^{-1} 
+
\frac{1}{\sqrt{T}}\frac{\hat{\bF}'\bF^0 }{T}
{\bB^0 }'\bE'\bF^0\tilde{\bH}_b \hat{\bLambda}^{-1} 
\\
&= 
\frac{1}{\sqrt{T}} \frac{\hat{\bF}'\bF^0 }{T}
{\bB^0 }'\bE'\left(\frac{1}{T} \bE\bE'\hat{\bF}\hat{\bLambda}^{-1} \right) \hat{\bLambda}^{-1}  
+
\frac{1}{\sqrt{T}} \frac{\hat{\bF}'\bF^0 }{T}
{\bB^0 }'\bE'\left( \frac{1}{T} \bE\bB^0
{\bF^0 }'\hat{\bF} \hat{\bLambda}^{-1} \right) \hat{\bLambda}^{-1} \\
& \quad\quad +
\frac{1}{\sqrt{T}}\frac{\hat{\bF}'\bF^0 }{T}
{\bB^0 }'\bE'\bF^0\tilde{\bH}_b \hat{\bLambda}^{-1} 
\\
&  = \bC_{31}+\bC_{32}+\bC_{33}.
\end{align*}
We analyze the order of the magnitude of each term in $\bC_3$ and obtain
\begin{align*}
\left\| \bC_{31} \right\|_{\F} 
&=  \frac{1}{\sqrt{T}}O_p \left(  \sqrt{TN^{\alpha_1}}\left( \frac{N}{T} + 1\right)\sqrt{T}N^{-2\alpha_r} \right) \\
& 
=
O_p \left(  \frac{N^{1-\alpha_r}}{\sqrt{T}}N^{\frac{1}{2}\alpha_1-\alpha_r} \right) 
+
O_p \left(  \sqrt{T}N^{ \frac{1}{2}\alpha_1-2\alpha_r}  \right) ,\\
\left\|\bC_{32}\right\|_{\F} 
&= 
O_p \left(  \sqrt{T}N^{\frac{1}{2}\alpha_1 -\frac{3}{2} \alpha_r}  \right) ,\\
\left\| \bC_{33} \right\|_{\F} 
&=  \frac{1}{\sqrt{T}}O_p \left(  \sqrt{TN^{\alpha_1}}  N^{- \alpha_r} \right) 
=
O_p \left(   N^{\frac{1}{2}\alpha_1-\alpha_r} \right) .
\end{align*}
\textbf{Case 1: $\alpha_1=\alpha_r=\alpha$.} Suppose that assumptions and conditions in Theorem \ref{thm:bias_Hhat} hold.  The terms $\bC_{22}$ and $\bC_{32}$ are the dominant terms of the same order $O_p( \sqrt{T}N^{  -\alpha_r})$ and the remaining terms are asymptotically negligible. Since $\sqrt{T}/N^{ \alpha} \to c_1 \in[0,\infty)$, $\nu=1$, and
\begin{align*}
& \bN \hat{\bLambda}^{-1} \left(\bN^{-\frac{1}{2}} \frac{\hat{\bF}' \bF^0}{T} \bN^{\frac{1}{2}}\right) \frac{\bN^{-\frac{1}{2}} \bB^{0\prime} \bE' \bE \bB^0 \bN^{-\frac{1}{2}}}{T}\left(\bN^{\frac{1}{2}} \frac{\bF^{0\prime} \hat{\bF}}{T} \bN^{-\frac{1}{2}}\right) \bN \hat{\bLambda}^{-1}
\CP \bD^{-1} \bGamma \bD^{-1}, \\
&\left(\bN^{-\frac{1}{2}} \frac{\hat{\bF}' \bF^0}{T} \bN^{\frac{1}{2}}\right) \frac{\bN^{-\frac{1}{2}} \bB^{0\prime} \bE' \bE \bB^0 \bN^{-\frac{1}{2}}}{T}\left(\bN^{\frac{1}{2}} \frac{\bF^{0\prime} \hat{\bF}}{T} \bN^{-\frac{1}{2}}\right) \bN^{2} \hat{\bLambda}^{-2}
\CP \bGamma \bD^{-2},
\end{align*}
where we have used $\bN \hat{\bLambda}^{-1} \CP \bN  {\bLambda}^{-1}=\bD^{-1}$ by \citep[Lemma B.3(iv)]{jiang2023revisiting} and Assumption \ref{ass:factor and loadings}(v), we can compute the probability limit of the dominants as follows:
\begin{align*}
\bC_{22}
& =\sqrt{T} \bN^{-\frac{1}{2}} \hat{\bLambda}^{-1} \bN\left(\bN^{-\frac{1}{2}} \frac{\hat{\bF}' \bF^0}{T} \bN^{\frac{1}{2}}\right) \frac{\bN^{-\frac{1}{2}} \bB^{0\prime} \bE' \bE \bB^0 \bN^{-\frac{1}{2}}}{T}\left(\bN^{\frac{1}{2}} \frac{\bF^{0\prime} \hat{\bF}}{T} \bN^{-\frac{1}{2}}\right) \bN \hat{\bLambda}^{-1}\bN^{-\frac{1}{2}} \\
&\CP  c_1   \bD^{-1}\bGamma  \bD^{-1},\\
\bC_{32}
& =\sqrt{T} \bN^{\frac{1}{2}}\left(\bN^{-\frac{1}{2}} \frac{\hat{\bF}' \bF^0}{T} \bN^{\frac{1}{2}}\right) \frac{\bN^{-\frac{1}{2}} \bB^{0\prime} \bE' \bE \bB^0 \bN^{-\frac{1}{2}}}{T}\left(\bN^{\frac{1}{2}} \frac{\bF^{0\prime} \hat{\bF}}{T} \bN^{-\frac{1}{2}}\right) \bN^{2} \hat{\bLambda}^{-2}\bN^{-\frac{3}{2}}\\
&\CP  c_1 \bG,
\end{align*}
where 
$c_1  \bD^{-1}\bGamma  \bD^{-1}= \lim\limits_{N, T \rightarrow \infty}\sqrt{T}\bN^{-\frac{1}{2}} \bD^{-1}\bGamma  \bD^{-1}\bN^{-\frac{1}{2}} $ and $c_1 \bG = \lim\limits_{N, T \rightarrow \infty} \sqrt{T}\bN^{\frac{1}{2}} \bGamma  \bD^{-2} \bN^{-\frac{3}{2}} =c_1 \bGamma  \bD^{-2}$. Thus, 
\[
\frac{1}{\sqrt{T}}  \hat{\bF}^{\prime}(\hat{\bF}-\bF^0\tilde{\bH})   =c_1   (\bG+ \bD^{-1}\bGamma  \bD^{-1})+o_p\left( 1 \right).
\]
\textbf{Case 2: $\alpha_1>\alpha_r$.} The term $\bC_{22}$ is no longer larger than $\bC_{32}$. 
If $\sqrt{T}N^{\frac{1}{2}\alpha_1-\frac{3}{2}\alpha_r} \to c_1 \in[0,\infty)$, we have
\begin{align*}
& \bC_{22} \CP  \mathbf{0},  \quad \bC_{32} \CP  c_1  \bG,
\end{align*}
where $\bG$ is defined by $\bG = (\be_{\alpha_1} \be_{\alpha_r}')
\odot\bGamma  \bD^{-2}$ in Corollary \ref{cor1}. 
Thus, 
\[
\frac{1}{\sqrt{T}}  \hat{\bF}^{\prime}(\hat{\bF}-\bF^0\tilde{\bH})   =c_1 \bG+o_p\left( 1 \right).
\]
(ii)  Applying the expansion in \eqref{fhat-fh_hat}, we can rewrite
\begin{align*}
\frac{1}{\sqrt{T}} \bW '(\hat{\bF}-\bF^0\tilde{\bH}) 
& =\frac{1}{\sqrt{T}}  \bW '\left( \frac{1}{T} \bE\bE'\hat{\bF}\hat{\bLambda}^{-1}+\frac{1}{T} \bE\bB^0
{\bF^0 }'\hat{\bF} \hat{\bLambda}^{-1} + \frac{1}{T}  \bF^0
{\bB^0 }'\bE'\hat{\bF} \hat{\bLambda}^{-1}\right)   \\
&=\bD_1+\bD_2+\bD_3.
\end{align*}
The first term is bounded by
\begin{align*}
\left\| \bD_1 \right\|_{\F} 
& \le   \left\|   \frac{1}{T^{3/2}}\bW' \bE\bE'(\hat{\bF}-\bF^0\tilde{\bH}_q) \hat{\bLambda}^{-1}\right\|_{\F} 
+ \left\|   \frac{1}{T^{3/2}} \bW'\bE\bE'\bF^{0}\tilde{\bH}_q \hat{\bLambda}^{-1}\right\|_{\F}\\
& = O_p \left( \sqrt{N}\left( \frac{N^{1-\frac{1}{2}\alpha_r}}{T}+N^{-\frac{1}{2}\alpha_r}  \right) N^{-\alpha_r}\right) 
+
O_p \left(   \frac{NT}{T^{3/2}} N^{-\alpha_r}\right) \\
&= O_p   \left( N^{\frac{1}{2}-\frac{3}{2}\alpha_r}   \right) 
+ O_p   \left(  \frac{N^{1- \alpha_r}}{\sqrt{T}} \right) ,
\end{align*}
by using Lemma \ref{lem:Rotation}(ii) and Assumption \ref{ass:factor and loadings}(iii). For the second term, 
\begin{align*}
& \left\| \bD_2 \right\|_{\F} 
=\left\| \frac{1}{\sqrt{T}}   \bW \bE\bB^0\bN^{-\frac{1}{2}}
\left(\bN^{\frac{1}{2}} \frac{{\bF^0 }'\hat{\bF}}{T} \bN^{-\frac{1}{2}}\right) \bN^{\frac{1}{2}}\hat{\bLambda}^{-1}\right\|_{\F}
= O_p \left(   N^{-\frac{1}{2} \alpha_r}  \right),
\end{align*}
by using Assumption \ref{ass:factor and loadings}(iv). The third term can be decomposed into
\begin{align*}
\bD_3  &   =
\frac{1}{\sqrt{T}} \frac{\bW'\bF^0 }{T}
{\bB^0 }'\bE'(\hat{\bF}-\bF^0\tilde{\bH}_b ) \hat{\bLambda}^{-1} 
+
\frac{1}{\sqrt{T}}\frac{\bW'\bF^0 }{T}
{\bB^0 }'\bE'\bF^0\tilde{\bH}_b \hat{\bLambda}^{-1} 
\\
&= 
\frac{1}{\sqrt{T}} \frac{\bW'\bF^0 }{T}
{\bB^0 }'\bE'\left(\frac{1}{T} \bE\bE'\hat{\bF}\hat{\bLambda}^{-1} \right) \hat{\bLambda}^{-1}  
+
\frac{1}{\sqrt{T}}\frac{\bW'\bF^0 }{T}
{\bB^0 }'\bE'\left( \frac{1}{T} \bE\bB^0
{\bF^0 }'\hat{\bF} \hat{\bLambda}^{-1} \right) \hat{\bLambda}^{-1} \\
& \quad\quad +
\frac{1}{\sqrt{T}}\frac{\bW'\bF^0 }{T}
{\bB^0 }'\bE'\bF^0\tilde{\bH}_b \hat{\bLambda}^{-1} 
\\
&  = \bD_{31}+\bD_{32}+\bD_{33}.
\end{align*}
Following the same arguments as for $\bC_{31}$ and $\bC_{33}$, we have
\begin{align*}
\left\| \bD_{31} \right\|_{\F} 
& =  \frac{1}{\sqrt{T}}O_p \left(  \sqrt{TN^{\alpha_1}}\left( \frac{N}{T} + 1\right)\sqrt{T}N^{-2\alpha_r} \right) 
\\
&=
O_p \left(  \frac{N^{1-\alpha_r}}{\sqrt{T}}N^{\frac{1}{2}\alpha_1-\alpha_r} \right) 
+
O_p \left(  \sqrt{T}N^{ \frac{1}{2}\alpha_1-2\alpha_r}  \right) ,\\
\left\| \bD_{33} \right\|_{\F} 
&=  \frac{1}{\sqrt{T}}O_p \left(  \sqrt{TN^{\alpha_1}}  N^{- \alpha_r} \right) 
=
O_p \left(   N^{\frac{1}{2}\alpha_1-\alpha_r} \right) .
\end{align*}
Since $\sqrt{T}N^{\frac{1}{2}\alpha_1-\frac{3}{2}\alpha_r} \to c_1 \in[0,\infty)$ and
\[
\frac{\bN^{-\frac{1}{2}} \bB^{0\prime} \bE' \bE \bB^0 \bN^{\frac{1}{2}}}{T}\left(\bN^{\frac{1}{2}} \frac{\bF^{0\prime} \hat{\bF}}{T} \bN^{-\frac{1}{2}}\right) \bN^{2} \hat{\bLambda}^{-2}
\CP \bGamma \bD^{-2},
\]
we calculate the probability limit of the dominant term:
\begin{align*}
\bD_{32}
& =\sqrt{T} \left(  \frac{\bW' \bF^0}{T}  \right)\bN^{\frac{1}{2}} \frac{\bN^{-\frac{1}{2}} \bB^{0\prime} \bE' \bE \bB^0 \bN^{-\frac{1}{2}}}{T}\left(\bN^{\frac{1}{2}} \frac{\bF^{0\prime} \hat{\bF}}{T} \bN^{-\frac{1}{2}}\right) \bN^{2} \hat{\bLambda}^{-2}\bN^{-\frac{3}{2}}\\
& \CP  c_1  \bSigma_{\bW \bF^0}\bG,
\end{align*}
where $\bG$ is defined by $\bG = (\be_{\alpha_1} \be_{\alpha_r}')
\odot\bGamma  \bD^{-2}$ in Corollary \ref{cor1}.
\end{proof}

\begin{proof}[Proof of Theorem \ref{thm:bias_Hhat}] Replacing $\hat{\bF}$ with ${\bF}^0$, we have
\begin{align*}
\frac{1}{T}  \hat\bZ'\hat\bZ   & =
\left( 
\begin{array}{cc}
\frac{1}{T} \hat\bF'\hat\bF  &\frac{1}{T} \hat\bF'\bW \\
\frac{1}{T} \bW'\hat\bF & \frac{1}{T} \bW'\bW
\end{array}
\right) \\
& =
\left( 
\begin{array}{cc}
\frac{1}{T}   \bF^{0\prime}\bF^0  &\frac{1}{T}   \bF^{0\prime}\bW \\
\frac{1}{T} \bW'\bF^0 & \frac{1}{T} \bW'\bW
\end{array}
\right)
+
\left( 
\begin{array}{cc}
\frac{1}{T} \hat{\bF}'(\hat\bF-\bF^0 )+\frac{1}{T} (\hat\bF-\bF^0 )'\bF^0  &\frac{1}{T} (\hat\bF-\bF^0 )'\bW \\
\frac{1}{T} \bW'(\hat\bF-\bF^0 ) &  \mathbf{0}
\end{array}
\right)\\
&=
\frac{1}{T}\bZ^{0\prime}\bZ^0 
+
O_p\left(\Delta_{NT}\right)+O_p\left(\frac{1}{\sqrt{T}}\right)
\end{align*}
in probability because Lemmas \ref{lem:Rotation}(v)(vi) and \ref{lem:bias_Hhat}(ii) imply
\begin{align*}
& \left\| \frac{1}{T} \hat{\bF}'\bF^0-\bI_r\right\|_{\F}=O_p\left( \Delta_{NT} \right), \quad 
\left\|\frac{1}{T} \bW'(\hat\bF-\bF^0)\right\|_{\F}=O_p\left( \Delta_{NT} \right)+O_p\left(\frac{1}{\sqrt{T}}\right).
\end{align*}
Then, we have
$ \frac{1}{T}  \hat\bZ'\hat\bZ 
\CP 
\bSigma_{\bZ^0 \bZ^0} $ by Assumption \ref{ass:Aug_errors} .

Next, we address that $\hat{\bH}$ has the same probability limit as $\bH$ because of the approximation
$  \hat{\bH}=\bH\tilde{\bH}=\bH +\bH (\tilde{\bH}-\bI_r)=\bH +O_p(\Delta_{NT})$ given by Lemma \ref{lem:Rotation}(v)(vi) and $  \Delta_{NT} $ is negligible under the conditions in Theorem \ref{thm:bias_Hhat}.  Combining with Lemmas \ref{lem:F*epsilon} and \ref{lem:bias_Hhat},
\begin{align*}
& \sqrt{T}(\hat\bdelta - \bdelta_{\hat{\bH}})  \\
&= (T^{-1}\hat\bZ'\hat\bZ)^{-1} T^{-\frac{1}{2}}\hat\bZ'\bepsilon
- (T^{-1}\hat\bZ'\hat\bZ)^{-1} T^{-\frac{1}{2}}\hat\bZ'(\hat\bF-{\bF}^0\tilde{\bH})\tilde{\bH}^{-1}\bgamma^0  \\
&  = (T^{-1}\hat\bZ'\hat\bZ)^{-1} T^{-\frac{1}{2}} \binom{ \tilde{\bH}'\bF^{0\prime }}{\bW} \bepsilon
+(T^{-1}\hat\bZ'\hat\bZ)^{-1} T^{-\frac{1}{2}} \binom{\hat{\bF}'-\tilde{\bH}'\bF^{0\prime }}{\mathbf{0}} \bepsilon \\
& 
\quad \quad - (T^{-1}\hat\bZ'\hat\bZ)^{-1} T^{-\frac{1}{2}}\hat\bZ'(\hat\bF-{\bF}^0\tilde{\bH})\tilde{\bH}^{-1}\bgamma^0 \\
&  = (T^{-1} \bZ^{0\prime} \bZ^0)^{-1} T^{-\frac{1}{2}} \bZ^{0\prime}\bepsilon
-(T^{-1} \bZ^{0\prime} \bZ^0)^{-1} T^{-\frac{1}{2}}\hat\bZ'(\hat\bF-{\bF}^0\tilde{\bH})\tilde{\bH}^{-1}\bgamma^0  \\
& 
\quad \quad +
O_p \left(\sqrt{T} N^{-\frac{3}{2}\alpha_r}\right) +O_p\left( \frac{N^{1-\alpha_r}}{\sqrt{T}} \right)+O_p\left( N^{\frac{1}{2}- \alpha_r} \right)\\
& \CD
N\left(-c_1 \bkappa_{\bdelta^*}, \bSigma_{\bdelta}\right),
\end{align*}
where $\bSigma_{\bdelta}= \bSigma_{\bZ^0 \bZ^0}^{-1} \bSigma_{\bZ^0 \bepsilon} \bSigma_{\bZ^0 \bZ^0}^{-1}$ and
\begin{align*}
\bkappa_{\bdelta^*} = \bSigma_{\bZ^0 \bZ^0}  ^{-1}\binom{\bG+\nu \bD^{-1}\bGamma\bD^{-1}}{\bSigma_{\bW \bF^0}  {\bG}}\bH_0^{-1} \bgamma^*. 
\end{align*} 
\end{proof}

\begin{proof}[Proof of Corollary \ref{cor1}] The proof follows directly from the arguments in Theorem \ref{thm:bias_Hhat} and Section \ref{subsec:Hhat}, and is therefore omitted here.
\end{proof}

\begin{lem} 
\label{lem:bias_Hhat3}
Suppose Assumptions \ref{ass:eigen}--\ref{ass:Aug_errors} hold. If $\alpha_r>\frac{1}{2}$, $\frac{N^{1-\alpha_r}}{\sqrt{T}} \to  0$, and $\sqrt{T}N^{-\alpha_r} \to c_2 \in [0,\infty)$, as $N, T \to \infty$, we have
\begin{align*}
&   \frac{1}{\sqrt{T}}  \bW'(\hat{\bF}-\bF^0\tilde{\bH}_q)   = c_2  \bSigma_{\bW \bF^0} \bar{\bG}+o_p\left( 1 \right),
\end{align*}
where $c_2 \bar{\bG} = \lim_{N, T \rightarrow \infty}\sqrt{T}\bN^{-\frac{1}{2}} \bD^{-1}\bGamma  \bD^{-1}\bN^{-\frac{1}{2}} $. 
If $\alpha_1=\alpha_r$, then $c_1=c_2$ and $c_2\bar{\bG}=c_2\bD^{-1}\bGamma  \bD^{-1}$. 
If $\alpha_1 > \alpha_r$, then   
$c_2\bar{\bG} = c_2(\be_{\alpha_r} \be_{\alpha_r}')
\odot\bD^{-1}\bGamma  \bD^{-1}$ .

\end{lem}

\begin{proof}[Proof of Lemma \ref{lem:bias_Hhat3}] Using the expansion \eqref{fhat-fh_q}, we can write
\begin{align}
&   \frac{1}{\sqrt{T}} \bW '(\hat{\bF}-\bF^0\tilde{\bH}_q) \\
& =\frac{1}{\sqrt{T}}  \bW '\left[  \frac{1}{T} \bE\bE'\hat{\bF}\hat{\bLambda}^{-1}+\frac{1}{T} \bE\bB^0
{\bF^0 }'\hat{\bF}   \hat{\bLambda}^{-1} - \bF^0  \left( \frac{1}{T}\hat{\bF}' {\bF^0}\right)^{-1}\frac{1}{T} \hat{\bF}' \bE{\hat{\bB}}\left({\hat{\bB}}' \hat{\bB}\right)^{-1} \right].  \label{eq:WF_Hhat3}
\end{align}
The upper bounds of the first and the second terms on the right-hand side of \eqref{eq:WF_Hhat3} are given by $\bD_1$ and $\bD_2$ in Lemma \ref{lem:bias_Hhat}. The third term is rewritten as
\begin{align}
&  {\sqrt{T}} \frac{\bW'\bF^0 }{T}
\left( \frac{1}{T}\hat{\bF}' {\bF^0}\right)^{-1}\frac{1}{T} \hat{\bF}' \bE{\hat{\bB}}\left({\hat{\bB}}' \hat{\bB}\right)^{-1}
\\
&= 
{\sqrt{T}} \frac{\bW'\bF^0 }{T}
\left( \frac{1}{T}\hat{\bF}' {\bF^0}\right)^{-1}\frac{1}{T} (\hat{\bF}-\bF^0\tilde{\bH}_b)' \bE(\hat{\bB}-\bB^0\tilde{\bQ}')\left({\hat{\bB}}' \hat{\bB}\right)^{-1}\\
& \quad\quad +
{\sqrt{T}} \frac{\bW'\bF^0 }{T}
\left( \frac{1}{T}\hat{\bF}' {\bF^0}\right)^{-1}\frac{1}{T} (\hat{\bF}-\bF^0\tilde{\bH}_b)' \bE \bB^0\tilde{\bQ}' \left({\hat{\bB}}' \hat{\bB}\right)^{-1}\\
& \quad\quad +
{\sqrt{T}} \frac{\bW'\bF^0 }{T}
\left( \frac{1}{T}\hat{\bF}' {\bF^0}\right)^{-1}\frac{1}{T} \tilde{\bH}_b'\bF^{0\prime} \bE \hat{\bB}\left({\hat{\bB}}' \hat{\bB}\right)^{-1}
\\
&= 
{\sqrt{T}} \frac{\bW'\bF^0 }{T}
\left( \frac{1}{T}\hat{\bF}' {\bF^0}\right)^{-1}\frac{1}{T} (\hat{\bF}-\bF^0\tilde{\bH}_b)' \bE(\hat{\bB}-\bB^0\tilde{\bQ}')\left({\hat{\bB}}' \hat{\bB}\right)^{-1}\\
& \quad \quad + 
{\sqrt{T}} \frac{\bW'\bF^0 }{T}
\left( \frac{1}{T}\hat{\bF}' {\bF^0}\right)^{-1} \hat{\bLambda}^{-1} \left(  \frac{1}{T} \hat{\bF}'{\bF^0 } \right) \frac{\bB^{0\prime}\bE' \bE \bB^0}{T} \tilde{\bQ}' \left({\hat{\bB}}' \hat{\bB}\right)^{-1 }\\
&\quad \quad
+
{\sqrt{T}} \frac{\bW'\bF^0 }{T}
\left( \frac{1}{T}\hat{\bF}' {\bF^0}\right)^{-1} \hat{\bLambda}^{-1}  \frac{1}{T} \left( \frac{1}{T} \hat{\bF}'\bE'\bE \right) \bE \bB^0\tilde{\bQ}' \left({\hat{\bB}}' \hat{\bB}\right)^{-1}\\
& \quad\quad +
{\sqrt{T}} \frac{\bW'\bF^0 }{T}
\left( \frac{1}{T}\hat{\bF}' {\bF^0}\right)^{-1}\frac{1}{T} \tilde{\bH}_b'\bF^{0\prime} \bE \hat{\bB}\left({\hat{\bB}}' \hat{\bB}\right)^{-1}
\\
&  = \bar{\bD}_{31}+ \bar{\bD}_{32}+ \bar{\bD}_{33}+ \bar{\bD}_{34}.\label{def of D}
\end{align}
Taking norms, 
\begin{align*}
\left\|  \bar{\bD}_{31} \right\|_{\F} 
& =   {\sqrt{T}}O_p \left( \frac{N^{1-\frac{1}{2}\alpha_r}}{T} +N^{- \frac{1}{2}\alpha_r}\right) \sqrt{N} O_p \left( \frac{1}{\sqrt{T}}+ N^{-\frac{1}{2}-\frac{1}{2}\alpha_r} \right) O_p \left(N^{-\alpha_r}\right)\\ 
&=
O_p \left(  \frac{N^{\frac{3}{2}-\frac{3}{2}\alpha_r}}{T^{\frac{3}{4}}}T^{-\frac{1}{4} } \right) 
+
O_p \left(   N^{ \frac{1}{2} -\frac{3}{2}\alpha_r}  \right)
+
O_p \left(  \sqrt{T} N^{  -2\alpha_r}  \right),
\end{align*}
\[
\left\|  \bar{\bD}_{33} \right\|_{\F}=
{\sqrt{T}} O_p\left( N^{-\alpha_r}\frac{1}{\sqrt{T}} \left(\frac{N}{T}+1\right)\sqrt{T}N^{-\frac{1}{2}\alpha_r} \right)
=
O_p\left( \frac{N^{1-\frac{3}{2}\alpha_r}}{\sqrt{T}} \right)
+
O_p\left( \sqrt{T} N^{-\frac{3}{2}\alpha_r}\right),
\]
\begin{align*}
\left\|  \bar{\bD}_{34} \right\|_{\F} 
& \le \left\| {\sqrt{T}} \frac{\bW'\bF^0 }{T}
\left( \frac{1}{T}\hat{\bF}' {\bF^0}\right)^{-1}\frac{1}{T} \tilde{\bH}_b'\bF^{0\prime} \bE (\hat{\bB}-\bB^0\tilde{\bQ}')\left(\hat{\bB}' \hat{\bB}\right)^{-1}\right\|_{\F}\\
&\quad \quad  + \left\| {\sqrt{T}} \frac{\bW'\bF^0 }{T}
\left( \frac{1}{T}\hat{\bF}' {\bF^0}\right)^{-1}\frac{1}{T} \tilde{\bH}_b'\bF^{0\prime} \bE \bB^0\tilde{\bQ}'\left(\hat{\bB}' \hat{\bB}\right)^{-1}\right\|_{\F}\\
& = 
{\sqrt{T}}O_p \left(  \frac{\sqrt{NT}}{T} \left( \sqrt{\frac{N}{T}}+ N^{- \frac{1}{2}\alpha_r}\right)N^{-\alpha_r} \right) 
+
{\sqrt{T}}O_p \left(  \frac{1}{\sqrt{T} }  N^{- \frac{1}{2}\alpha_r} \right) 
\\
&= O_p \left(   \frac{N^{1-\alpha_r}}{\sqrt{T}}  \right) 
+O_p \left(  N^{\frac{1}{2}-\frac{3}{2}\alpha_r}  \right) 
+O_p \left(    N^{ -\frac{1}{2}\alpha_r} \right) ,
\end{align*}
which are asymptotically negligible. Assuming that $\sqrt{T}N^{-\alpha_r} \to c_2$, the dominant term $\bar{\bD}_{32}$ has the probability limit:
\begin{align*}
& {\sqrt{T}} \frac{\bW'\bF^0 }{T}
\left( \frac{1}{T}\hat{\bF}' {\bF^0}\right)^{-1}  \hat{\bLambda}^{-1} \left( \frac{1}{T}  \hat{\bF}'{\bF^0 } \right) \bN^{\frac{1}{2}} \frac{\bN^{-\frac{1}{2}} \bB^{0\prime} \bE' \bE \bB^0 \bN^{\frac{1}{2}}}{T} \bN^{\frac{1}{2}} \tilde{\bQ}' \left({\hat{\bB}}' \hat{\bB}\right)^{-1}\\
&=\sqrt{T} \left(  \frac{\bW' \bF^0}{T}  \right)     \tilde{\bQ}^{-1}  \hat{\bLambda}^{-1} \bN^{\frac{1}{2}}  \left(\bN^{-\frac{1}{2}} \tilde{\bQ}\bN^{ \frac{1}{2}}\right)    \frac{\bN^{-\frac{1}{2}} \bB^{0\prime} \bE' \bE \bB^0\bN^{-\frac{1}{2}}}{T} \bN^{\frac{1}{2}} \tilde{\bQ}'\bN^{-\frac{1}{2}}  \bN  \hat{\bLambda}^{-1} \bN^{-\frac{1}{2}}\\
&  \CP  c_2  \bSigma_{\bW \bF^0}   \bar{\bG}.
\end{align*}
\end{proof}

\begin{proof}[Proof of Theorem \ref{thm:bias_H3}]
The definition of $\tilde{\bH}_{q}$ implies
\begin{align*}
&\hat\bF'(\hat\bF-{\bF}^0\tilde{\bH}_{q})\tilde{\bH}_{q}^{-1}\bgamma^0= \mathbf{0}.
\end{align*}
Combining with Lemmas \ref{lem:F*epsilon} and \ref{lem:bias_Hhat3},
\begin{align}
& \sqrt{T}(\hat\bdelta - \bdelta_{\hat{\bH}_q})  \\
&= (T^{-1}\hat\bZ'\hat\bZ)^{-1} T^{-\frac{1}{2}}\hat\bZ'\bepsilon
- (T^{-1}\hat\bZ'\hat\bZ)^{-1} T^{-\frac{1}{2}}\hat\bZ'(\hat\bF-{\bF}^0\tilde{\bH}_q)\tilde{\bH}_q^{-1}\bgamma^0  \\
&  = (T^{-1}\hat\bZ'\hat\bZ)^{-1} T^{-\frac{1}{2}} \binom{ \tilde{\bH}_q'\bF^{0\prime }}{\bW} \bepsilon
+(T^{-1}\hat\bZ'\hat\bZ)^{-1} T^{-\frac{1}{2}} \binom{\hat{\bF}'-\tilde{\bH}_q'\bF^{0\prime }}{\mathbf{0}} \bepsilon \\
& 
\quad \quad - (T^{-1}\hat\bZ'\hat\bZ)^{-1} T^{-\frac{1}{2}}\hat\bZ'(\hat\bF-{\bF}^0\tilde{\bH}_q)\tilde{\bH}_q^{-1}\bgamma^0 \\
&  = (T^{-1} \bZ^{0\prime} \bZ^0)^{-1} T^{-\frac{1}{2}} \bZ^{0\prime}\bepsilon
- (T^{-1} \bZ^{0\prime} \bZ^0)^{-1} T^{-\frac{1}{2}}\hat\bZ'(\hat\bF-{\bF}^0\tilde{\bH}_q)\tilde{\bH}_q^{-1}\bgamma^0  \\
& 
\quad \quad +
O_p \left(\sqrt{T} N^{-\frac{3}{2}\alpha_r}\right) +O_p\left( \frac{N^{1-\alpha_r}}{\sqrt{T}} \right)+O_p\left( N^{\frac{1}{2}- \alpha_r} \right)\\
& \CD
N\left(c_2 \bar{\bkappa}_{\bdelta^*}, \bSigma_{\bdelta}\right),\label{delta-de_q}
\end{align}
where 
\begin{align*}
\bar{\bkappa}_{\bdelta^*} = \bSigma_{\bZ^0 \bZ^0}^{-1}\binom{\mathbf{0}}{\bSigma_{\bW \bF^0}  \bar{\bG}}\bH_0^{-1} \bgamma^*. 
\end{align*}
\end{proof}

\begin{proof}[Proof of Theorem \ref{thm:bias_H}]
Theorem \ref{thm:bias_H3} implies that 
\begin{align}
\nonumber
\sqrt{T}( \hat{\bdelta}- \bdelta_{\hat{\bH}_q})
&= \left(\frac{1}{T}{\bZ}^{0'}{\bZ}^0\right)^{-1} \frac{1}{\sqrt{T}}{\bZ}^{0'}\bepsilon
+  \left(\frac{1}{T}{\bZ}^{0'}{\bZ}^0\right)^{-1}\left( 
\begin{array}{c}
\mathbf{0} \\
\bar{\bD}_{32}
\end{array}
\right) \tilde{\bH}_q^{-1}\bgamma^0  +o_p\left(\frac{\sqrt{T}}{N^{\alpha_r}}\right) \\
\nonumber
& \quad 
+O_p\left( \frac{N^{1-\alpha_r}}{\sqrt{T}} \right)+O_p\left( N^{\frac{1}{2}- \alpha_r} \right) , 
\end{align}
where the bias term $\bar{\bD}_{32}$ is defined in \eqref{def of D} as
\[
\bar{\bD}_{32} = \frac{1}{\sqrt{T}}\bW'\bF^0 \tilde{\bQ}^{-1}\tilde{\bLambda}^{-1}\tilde{\bQ}\frac{\bB^{0\prime}\bE'\bE\bB^0}{T}\tilde{\bQ}\tilde{\bLambda}^{-1} 
\CP 
c_2 \bSigma_{\bW\bF^0} \bar{\bG},
\quad \text{if} \quad \frac{\sqrt{T}}{N^{\alpha_r}} \to c_2.
\]
Next, we consider
\begin{align*}
& \sqrt{T}( \hat{\bdelta}- \bdelta_{\bH})\\
&= \left( 
\begin{array}{c}
\sqrt{T}(  \tilde{\bH}_q^{-1}- \bI_r)\bgamma^0 \\
\nonumber
\mathbf{0}
\end{array}
\right)
+\sqrt{T}( \hat{\bdelta}- \bdelta_{\hat{\bH}_q})\\
& =
\left( 
\begin{array}{c}
\sqrt{T}(  \tilde{\bH}_q^{-1}- \bI_r)\bgamma^0 \\
\nonumber
\mathbf{0}
\end{array}
\right)
+
\left(\frac{1}{T}{\bZ}^{0'}{\bZ}^0 \right)^{-1} \frac{1}{\sqrt{T}}{\bZ}^{0'}\bepsilon
+  \left(\frac{1}{T}{\bZ}^{0'}{\bZ}^0\right)^{-1}\left( 
\begin{array}{c}
\mathbf{0} \\
\bar{\bD}_{32}
\end{array}
\right) \tilde{\bH}_q^{-1}\bgamma^0   \\
\nonumber
& \quad 
+o_p\left(\frac{\sqrt{T}}{N^{\alpha_r}}\right) +O_p\left( \frac{N^{1-\alpha_r}}{\sqrt{T}} \right)+O_p\left( N^{\frac{1}{2}- \alpha_r} \right)  .
\end{align*} 
Although the explicit expansion of $ \sqrt{T}(  \tilde{\bH}_q^{-1}- \bI_r)\bgamma^0 $ is unknown, we know that $\|\sqrt{T}(  \tilde{\bH}_q^{-1}- \bI_r)\bgamma^0\|_{\F}=O_p(\sqrt{T}\Delta_{NT})$. Furthermore,  under the conditions in Theorem \ref{thm:bias_H}, this term is not larger than $O_p(\sqrt{T}N^{\frac{1}{2}\alpha_1-\frac{3}{2}\alpha_r})$. We assume that the first bias term $(\sqrt{T}\bgamma^{0\prime}(  \tilde{\bH}_q^{\prime-1}- \bI_r),  \mathbf{0}^{\prime} )' \CP c_1\bh_{\bgamma^*}$, where $ \bh_{\bgamma^*}$ is a constant vector, when $\sqrt{T}N^{\frac{1}{2}\alpha_1-\frac{3}{2}\alpha_r} \to c_1 \in [0,\infty)$ as $N, T \to \infty$. The second bias is smaller and stems from $ \sqrt{T}( \hat{\bdelta}- \bdelta_{\hat{\bH}_q})$, given by
\begin{align*}
\left(\frac{1}{T}{\bZ}^{0'}{\bZ}^0\right)^{-1}\left( 
\begin{array}{c}
\mathbf{0} \\
\bar{\bD}_{32}
\end{array}
\right) \tilde{\bH}_q^{-1}\bgamma^0 
\CP
c_2 \bSigma_{\bZ^0\bZ^0}^{-1}
\left( 
\begin{array}{c}
\mathbf{0} \\
\bSigma_{\bW\bF^0} \bar{\bG} \bH_0^{-1} \bgamma^*
\end{array}
\right) 
= c_2\bar{\bkappa}_{\bdelta^*},
\quad \text{if} \quad \frac{\sqrt{T}}{N^{\alpha_r}} \to c_2.
\end{align*}
If $\alpha_r < \alpha_1$ and $c_1 \in (0,\infty)$, then $c_2=0$, hence the asymptotic bias is $c_1\bh_{\bgamma^*}$. When $\alpha_1=\alpha_r$, we have $c_2=c_1$ and the asymptotic bias is $c_2( \bh_{\bgamma^*} + \bar{\bkappa}_{\bdelta^*})$. Thus, we complete the proof as
\begin{align*}
& \sqrt{T}(\hat\bdelta  - \bdelta^0)  \CD
N\left(c_1 \bh_{\bgamma^*}+ c_2 \bar{\bkappa}_{\bdelta^*} ,\bSigma_{\bdelta}\right).
\end{align*}
\end{proof}

\begin{proof}[Proof of Theorem \ref{thm:bias_JC}]
Define $
{\bPhi}_{{\bH}_j}=
\left( 
\begin{array}{cc}
{\bH}_j  &\mathbf{0} \\
\mathbf{0} &{\bI}_p
\end{array}
\right)$, where $\tilde{\bH}_{qj}$ and ${\bH_j} (j=1,2)$ are rotation matrices calculated from the half-panel $\mathcal{N}_j$, corresponding to $\tilde{\bH}_{q}$ and ${\bH}$ calculated from the whole-panel.
We assume that $\plim{\bH_j}=\bH_0$ $(j=1,2)$. Furthermore, if $(\sqrt{T}\bgamma^{0\prime}(  \tilde{\bH}_q^{\prime-1}- \bI_r),  \mathbf{0}^{\prime} )' \CP c_1\bh_{\bgamma^*}$ when $\sqrt{T}N^{\frac{1}{2}\alpha_1-\frac{3}{2}\alpha_r} \to c_1 \in [0,\infty)$, then, $(\sqrt{T}\bgamma^{*\prime}\bH_j^{\prime-1}(  \tilde{\bH}_{qj}^{\prime-1}- \bI_r),\mathbf{0}')' $ from the half-panel will converge to $2^{\frac{1}{2}(3\alpha_r-\alpha_1)}c_1\bh_{\bgamma^*}$ with $N/2$ replacing $N$.
Clearly,
\begin{align*}
\sqrt{T}( \hat{\bdelta}_{\mathcal{N}_j}-{\bPhi}_{\bH_j}^{-1}\bdelta^*)
&=
\left( 
\begin{array}{c}
\sqrt{T}(  \tilde{\bH}_{qj}^{-1}- \bI_r) \bH_j^{-1}\bgamma^* \\
\nonumber
\mathbf{0}
\end{array}
\right) 
+
\bPhi_{\bH_j}^{-1}\left(\frac{1}{T}{\bZ}^{*'}{\bZ}^*\right)^{-1} \frac{1}{\sqrt{T}}{\bZ}^{*'}\bepsilon\\
& \quad
+ \bPhi_{\bH_j}^{-1}\left(\frac{1}{T}{\bZ}^{*'}{\bZ}^*\right)^{-1}\left( 
\begin{array}{c}
\mathbf{0} \\
\bar{\bD}_{\mathcal{N}_j}
\end{array}
\right)\tilde{\bH}_j^{-1} \bH_j^{-1}\bgamma^*\\
& \quad 
+O_p\left( \frac{N^{1-\alpha_r}}{\sqrt{T}}\right) +o_p\left(N^{\frac{1}{2}-\alpha_r}\right) +o_p\left(\frac{\sqrt{T}}{N^{\alpha_r}}\right),
\end{align*}
where $\bar{\bD}_{\mathcal{N}_j}$ from  half-panel $\mathcal{N}_j$ preserves the same structure as $\bar{\bD}_{32}$ calculated form the whole panel. 

The first bias in $\sqrt{T}( \hat{\bdelta}_{bcjk}- \bdelta^0) $ is given by
\begin{align*}
& \left(2\sqrt{T}\bgamma^{0\prime}(  \tilde{\bH}_{q}^{\prime-1}- \bI_r)-\frac{1}{2}\left[\sqrt{T}\bgamma^{*\prime}\bH_1^{\prime-1}(  \tilde{\bH}_{q1}^{\prime-1}- \bI_r)+\sqrt{T}\bgamma^{*\prime}\bH_2^{\prime-1}(  \tilde{\bH}_{q2}^{\prime-1}- \bI_r)\right] ,\mathbf{0}'\right)' \\
&  \CP 
(2-2^{\frac{1}{2}(3\alpha_r-\alpha_1)})c_1 \bh_{\bgamma^*},
\end{align*}
which is exactly zero if $\alpha_r=1$ and smaller than $c_1 \bh_{\bgamma^*}$ if $\alpha_1 <3\alpha_r$.

Then, the second bias $2\bar{\bD}_{32}-\frac{1}{2}(\bar{\bD}_
{\mathcal{N}_1}+\bar{\bD}_
{\mathcal{N}_2})$ converges to $(2-2^{\alpha_r})c_2  \bSigma_{\bW\bF^0} \bar{\bG} $ in probability. The bias is removed if $\alpha_r=1$ and smaller than $ c_2 \bSigma_{\bW\bF^0} \bar{\bG}   $ for any positive $\alpha_r$.
Finally, we complete the proof as
\begin{align*}
\sqrt{T}(  \hat{\bdelta}_{bcjk}- \bdelta^0 )  \CD
N\left((2-2^{\frac{1}{2}(3\alpha_r-\alpha_1)})c_1 \bh_{\bgamma^*} + (2-2^{\alpha_r})c_2 \bar{\bkappa}_{\bdelta^*} ,\bSigma_{\bdelta}\right).
\end{align*}

\end{proof}

\section{Proofs of results in Section \ref{sec:Mw}}
\setcounter{lem}{0}
\renewcommand{\thelem}{B.\arabic{lem}}
\renewcommand{\theequation}{B.\arabic{equation}}
\setcounter{equation}{0}

\begin{lem} Suppose that Assumptions \ref{ass:eigen}--\ref{ass:Aug_errors} hold. If $\frac{N^{1-\alpha_r}}{T}\to 0$, then, we have, as $N, T \to \infty$,
\label{lem:F*epsilon_Mw}
\begin{flalign}
& ~ \left\| \frac{1}{\sqrt{T}} (\hat{\bF}_w-\bF^0_w\tilde{\bH}_{b,w})'\bepsilon  \right\|_{\F}=O_p \left(\sqrt{T} N^{-\frac{3}{2}\alpha_r}\right) +O_p\left( \frac{N^{1-\alpha_r}}{\sqrt{T}} \right)+O_p\left(N^{\frac{1}{2}-\alpha_r}\right) ,\label{eq:F*ep_H3w}&
\end{flalign}
where $\tilde{\bH}_{b,w}$ is the data dependent rotation matrix corresponding to $\tilde{\bH}_{b}$ but for the model of $\bX_w$.
\end{lem}

\begin{proof}[Proof of Lemma \ref{lem:F*epsilon_Mw}]
Following the same argument as for \eqref{eq:H4-H}, we have
\begin{align}
\label{eq:H4w-Hw}
\tilde{\bH}_{b,w}&=\tilde{\bH}_w+\frac{1}{T} \bB^{0\prime}_w \bE^{\prime}_w \hat{\bF}_w \hat{\bLambda}_w^{-1},
\end{align}
where $\tilde{\bH}_{w}$ is the data dependent rotation matrix corresponding to $\tilde{\bH}$ but for the model of $\bX_w$.
Since $\hat{\bLambda}_w$ and $\hat{\bF}_w$ are the eigenvalues and eigenvectors of $\bX_w\bX_w'/T$, respectively, we have
\begin{align*} 
T^{-1}\bX_w\bX_w'\hat{\bF}_w
= \hat{\bF}_w\hat{\bLambda}_w.
\end{align*}
The model is rewritten as
\begin{align*}
\mathbf{X}_w&=\bM_w\bF^*  \bB^{*\prime} +\mathbf{E}_w \\
&=\bM_w\bF^* \bH_w\bH_w^{-1} \bB^{*\prime} +\mathbf{E}_w \\
&= \bF^0_w  \bB^{0\prime}_w +\mathbf{E}_w ,
\end{align*}
where $\mathbf{X}_w=\bM_w\bX$, $\mathbf{F}_w^0=\bM_w \bF^* \bH_w$, $\mathbf{B}_w^0= \bB^* \bH_w^{'-1}$,  $\mathbf{E}_w=\bM_w \bE$, then, $\bX_w\bX_w'$ can be expanded as
\begin{align*} 
\bX_w\bX_w'
=\bF^0_w  \bB^{0\prime}_w  \bB^0_w  \bF^{0\prime}_w 
+  \bF^0_w  \bB^{0\prime}_w  \bE'_w  
+  \bE_w \bB^0_w  \bF^{0\prime}_w  +  \bE_w \bE'_w.
\end{align*}
We obtain
\begin{align}
\label{f-fH_4w}
& \hat{\bF}_w - \bF^0_w \tilde{\bH}_{b,w} =  \frac{1}{T} \bE_w\bE'_w\hat{\bF}_w\hat{\bLambda}_w^{-1}+\frac{1}{T} \bE_w\bB^0_w
{\bF^0 }_w'\hat{\bF}_w   \hat{\bLambda}_w^{-1} .
\end{align}
(i) By expanding \eqref{f-fH_4w}, 
\begin{align*}
&    \frac{1}{\sqrt{T}} (\hat{\bF}_w-\bF_w^0\tilde{\bH}_{b,w})'\bepsilon   \\
&= \frac{1}{T^{3/2}} \hat{\bLambda}_w^{-1} {\hat{\bF}_w}' \bE_w\bE_w' \bepsilon  + \frac{1}{T^{3/2}}  \hat{\bLambda}_w^{-1}{\hat{\bF}_w}'\bF_w^0  \bB_w^{0\prime}{\bE_w}'\bepsilon   \\
&= \tilde{\bA}_1+\tilde{\bA}_2  .
\end{align*}

For the first term,
\begin{align*}
&   \left\|  \tilde{\bA}_1  \right\|_{\F}  
\le \left\|    \frac{1}{T^{3/2}} \hat{\bLambda}_w^{-1} {\hat{\bF}_w}' \bE\bE' \bepsilon \right\|_{\F} 
+
\left\|    \frac{1}{T^{3/2}} \hat{\bLambda}_w^{-1} {\hat{\bF}_w}' \bP_w\bE\bE' \bepsilon \right\|_{\F} 
+
\left\|    \frac{1}{T^{3/2}} \hat{\bLambda}_w^{-1} {\hat{\bF}_w}' \bE\bE' \bP_w\bepsilon \right\|_{\F} \\
& \quad \quad\quad \quad 
+
\left\|    \frac{1}{T^{3/2}} \hat{\bLambda}_w^{-1} {\hat{\bF}_w}' \bP_w\bE\bE'\bP_w \bepsilon \right\|_{\F} \\
&\quad \quad \quad=O_p \left(\sqrt{T} N^{-\frac{3}{2}\alpha_r}\right) +O_p \left(\frac{N^{1-  \alpha_r}}{\sqrt{T}}\right)
+O_p\left(N^{\frac{1}{2}-\alpha_r}\right),
\end{align*}
since $ \hat{\bF}_w^{\prime} \bE\bE' \bepsilon =O_p \left(T^2 N^{-\frac{1}{2}\alpha_r}\right) + O_p(\sqrt{N}T^{3/2}) + O_p(NT) $
has the same order of magnitude of the term \eqref{Feee} and $ \left\|\bP_w \right\|_{\F} =O_p(1)$. For the second term,

\begin{align*}
\left\|  \tilde{\bA}_2  \right\|_{\F}  & =\left\|\frac{1}{\sqrt{T}} \hat{\bLambda}_w^{-1}  \left( \frac{1}{T} \hat{\bF}_w^{\prime} \bF_w^0  \right) \bH_w^{-1} \left( \bB^{0\prime} \bE^{\prime} \bepsilon-\bB^{0\prime} \bE^{\prime} \bP_w \bepsilon\right)\right\|_{\F} \\
& =O_p\left( \frac{1}{\sqrt{T}} N^{- \alpha_r} \sqrt{T N^{\alpha_1}}\right)=O_p\left( N^{\frac{1}{2}\alpha_1-\alpha_r} \right).
\end{align*}
Thus, we complete the proof as
\[
\left\|
\frac{1}{\sqrt{T}} (\hat{\bF}_w-\bF^0_w\tilde{\bH}_{b,w})'\bepsilon  \right\|_{\F}=
O_p \left(\sqrt{T} N^{-\frac{3}{2}\alpha_r}\right) +O_p\left( \frac{N^{1-\alpha_r}}{\sqrt{T}} \right)+O_p\left(N^{\frac{1}{2}-\alpha_r}\right).
\]
\end{proof}

\begin{proof}[Proof of Theorem \ref{thm:bias_Hw3}]
Because  
\begin{align*}
&  \bW'\bF_w^0=\bW'\bM_w\bF^0\bH_w=\mathbf{0},  \quad \bW'\hat{\bF}_w=T^{-1} \bW'\bX_w\bX_w'\hat{\bF}_w\hat{\bLambda}_w^{-1} =\mathbf{0},
\end{align*}
we have
\begin{align*}
&\hat\bF_w'(\hat\bF_w-{\bF}_w^0\tilde{\bH}_{q,w})\tilde{\bH}_{q,w}^{-1}\bgamma^0_w= \mathbf{0},\\
&  \bW'(\hat\bF_w-{\bF}_w^0\tilde{\bH}_{q,w})\tilde{\bH}_{q,w}^{-1}\bgamma^0_w= \mathbf{0},
\end{align*}
where $\tilde{\bH}_{q,w}$ is the data dependent rotation matrix corresponding to $\tilde{\bH}_q$ but for the model of $\bX_w$.
By the same arguments as given in the proof of Theorem \ref{thm:bias_H3}, we have
\begin{align*}
& \sqrt{T}(\hat\bdelta_w - \bdelta_{q,w})  \\
&= (T^{-1}\hat\bZ_w'\hat\bZ_w)^{-1} T^{-\frac{1}{2}}\hat\bZ_w'\bepsilon
- (T^{-1}\hat\bZ_w'\hat\bZ_w)^{-1} T^{-\frac{1}{2}}\hat\bZ_w'(\hat\bF-{\bF}_w^0\tilde{\bH}_{q,w})\tilde{\bH}_{q,w}^{-1}\bgamma^0_w  \\
&  = (T^{-1}\hat\bZ_w'\hat\bZ_w)^{-1} T^{-\frac{1}{2}} \binom{ \tilde{\bH}_{b,w}'\bF_w^{0\prime }}{\bW'} \bepsilon
+(T^{-1}\hat\bZ_w'\hat\bZ_w)^{-1} T^{-\frac{1}{2}} \binom{\hat{\bF}_w'-\tilde{\bH}_{b,w}'\bF_w^{0\prime }}{\mathbf{0}} \bepsilon \\
&  = (T^{-1}\bZ_w^{0\prime} \bZ_w^0)^{-1} T^{-\frac{1}{2}} \binom{ \bF_w^{0\prime }}{\bW'} \bepsilon+O_p \left(\sqrt{T} N^{-\frac{3}{2}\alpha_r}\right) +
+O_p\left( \frac{N^{1-\alpha_r}}{\sqrt{T}} \right)+O_p\left( N^{\frac{1}{2}-\alpha_r} \right).
\end{align*}
Thus,
$ \frac{1}{T}  \hat\bZ_w'\hat\bZ_w 
\CP 
\bSigma_{\bZ^0 \bZ^0}$.
By Assumption $T^{-1 / 2} \mathbf{Z}_w^{0 \prime} \boldsymbol{\epsilon} \xrightarrow{d} N\left(\mathbf{0}, \boldsymbol{\Sigma}_{\mathbf{Z}
_w^0 \boldsymbol{\epsilon}}\right)$, we prove that
\begin{align*}
& \sqrt{T}(\hat\bdelta_w - \bdelta_{q,w})  \CD
N\left(\mathbf{0},\bSigma_{\bdelta_w}\right),
\end{align*} 
where $\bSigma_{\bdelta_w}={\bSigma}_{\bZ^0 \bZ^0}^{-1} (  {\bSigma}_{\bZ_w^0 \bepsilon} ) {\bSigma}_{\bZ^0 \bZ^0}^{-1}$.
\end{proof}

\section{Additional Experimental Results}
\setcounter{lem}{0}
\renewcommand{\thelem}{C.\arabic{lem}}
\setcounter{equation}{0}
\renewcommand{\theequation}{C.\arabic{equation}}

\setcounter{table}{0}
\renewcommand{\thetable}{C.\arabic{table}}
\setcounter{figure}{0}
\renewcommand{\thefigure}{C.\arabic{figure}}

\subsection{Coefficient on the second factor, using orthoganised $x_{t,i}$ to $w_t$}


\begin{figure}[h!]	
\centering
\begin{subfigure}[b]{0.32\textwidth}
\centering
\includegraphics[width=\textwidth]{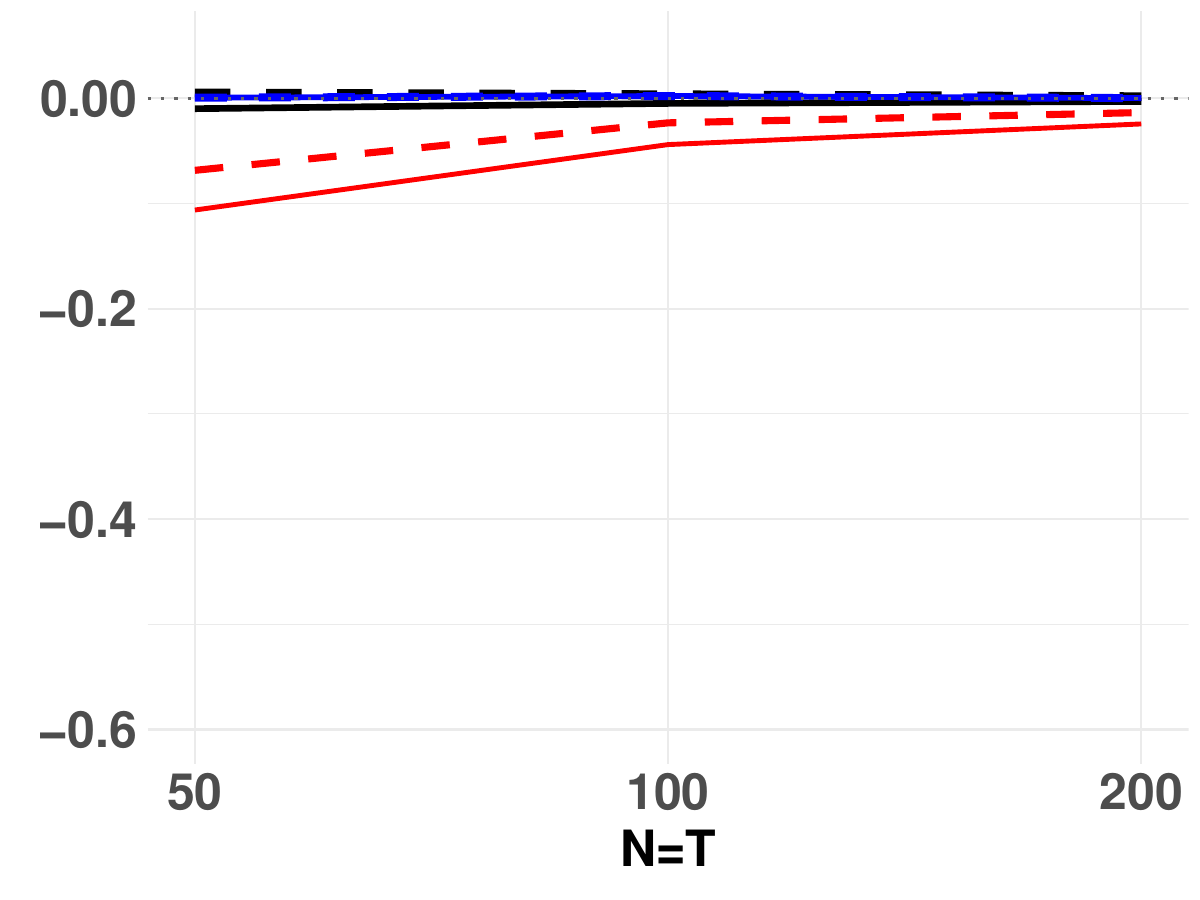}
\caption{$\rho_{wf}=0.0,\alpha_2 = 1.0$}
\label{fig:bias_f2w_00_10}
\end{subfigure}
\hfill
\begin{subfigure}[b]{0.32\textwidth}
\centering
\includegraphics[width=\textwidth]{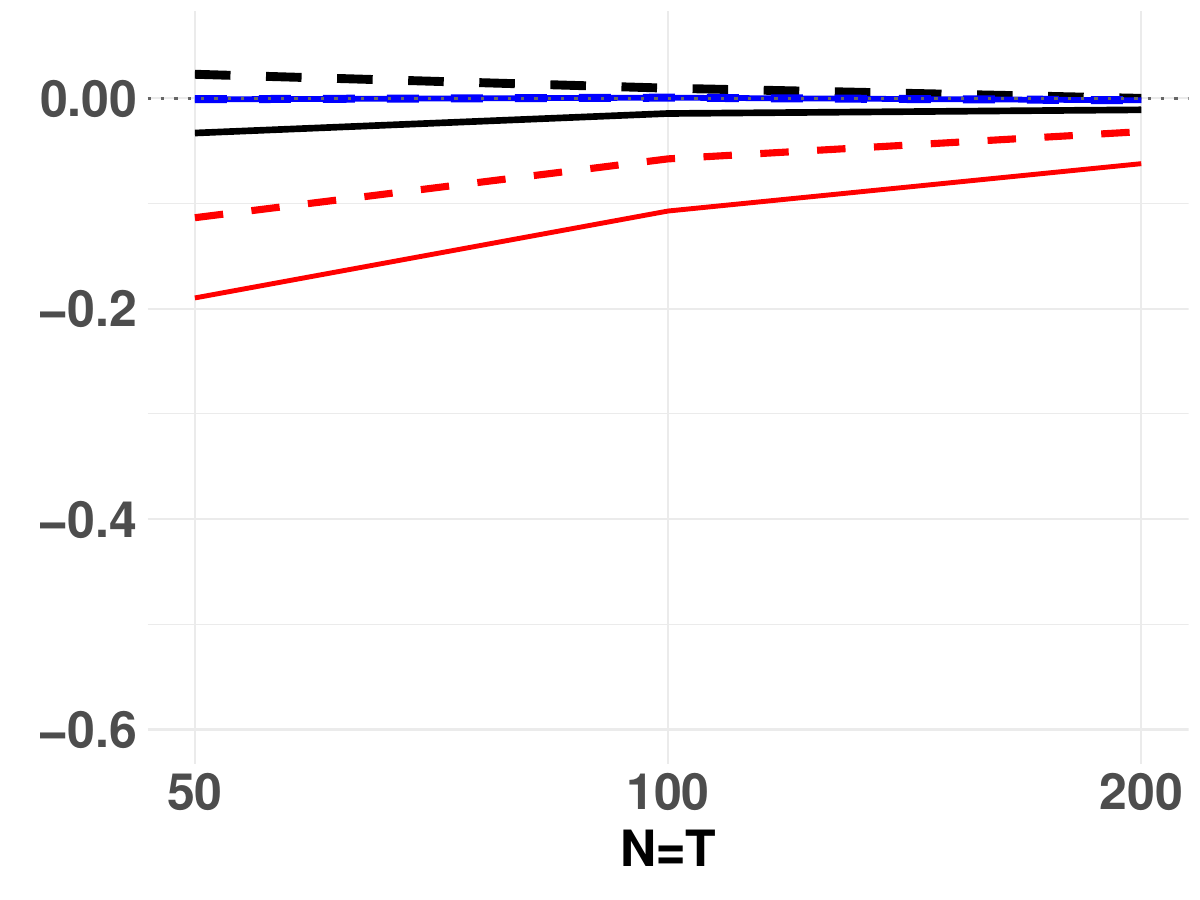}
\caption{$\rho_{wf}=0.0,\alpha_2 = 0.8$}
\label{fig:bias_f2w_00_08}
\end{subfigure}
\hfill
\begin{subfigure}[b]{0.32\textwidth}
\centering
\includegraphics[width=\textwidth]{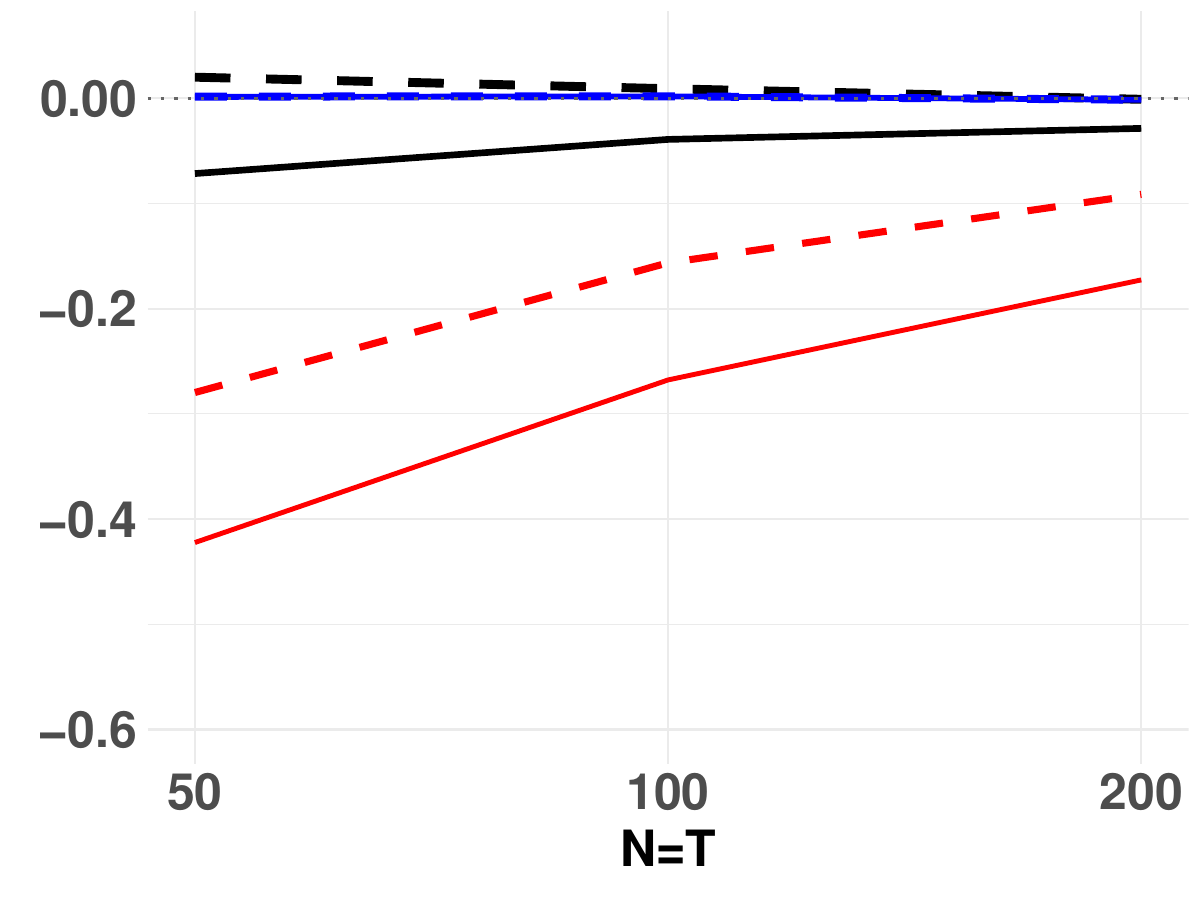}
\caption{$\rho_{wf}=0.0,\alpha_2 = 0.6$}
\label{fig:bias_f2w_00_06}
\end{subfigure}

\centering
\begin{subfigure}[b]{0.32\textwidth}
\centering
\includegraphics[width=\textwidth]{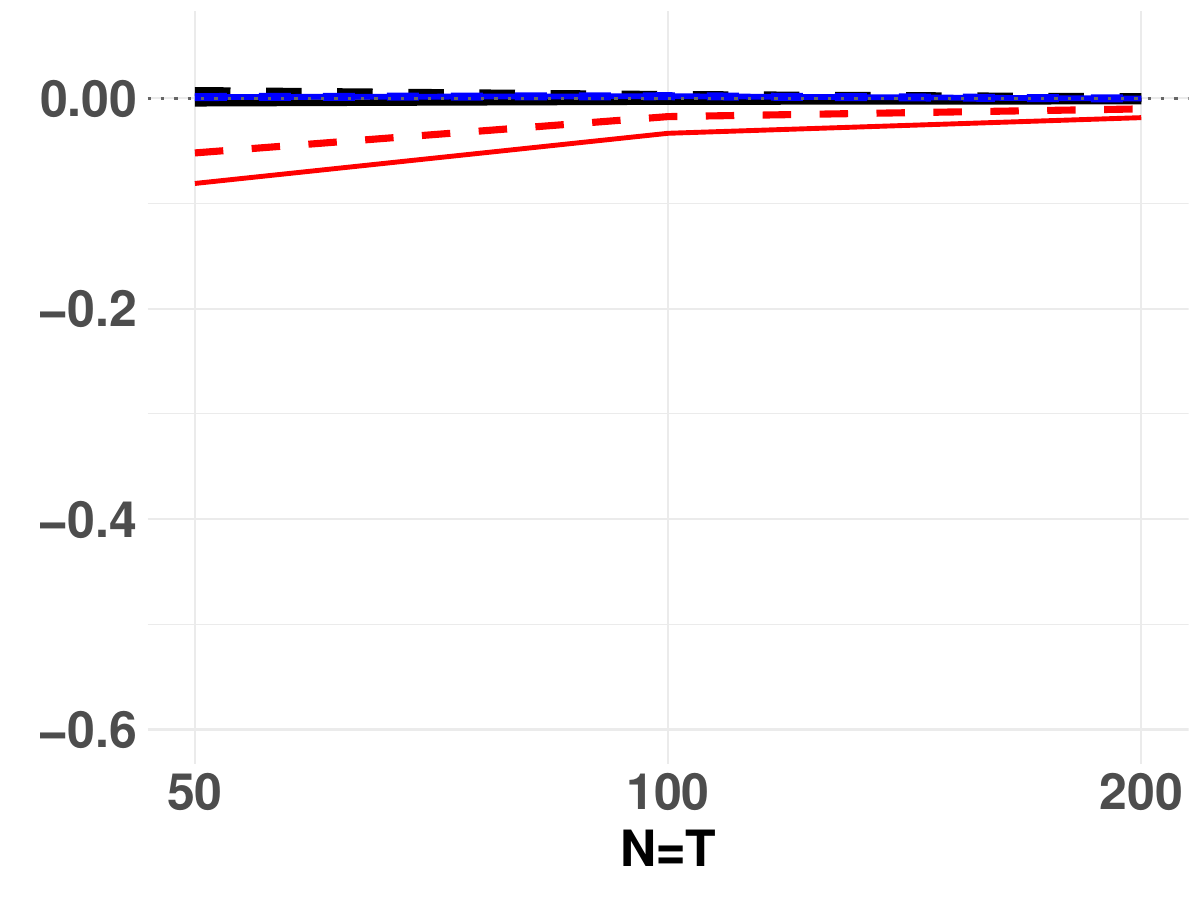}
\caption{$\rho_{wf}=0.6,\alpha_2 = 1.0$}
\label{fig:bias_f2w_06_10}
\end{subfigure}
\hfill
\begin{subfigure}[b]{0.32\textwidth}
\centering
\includegraphics[width=\textwidth]{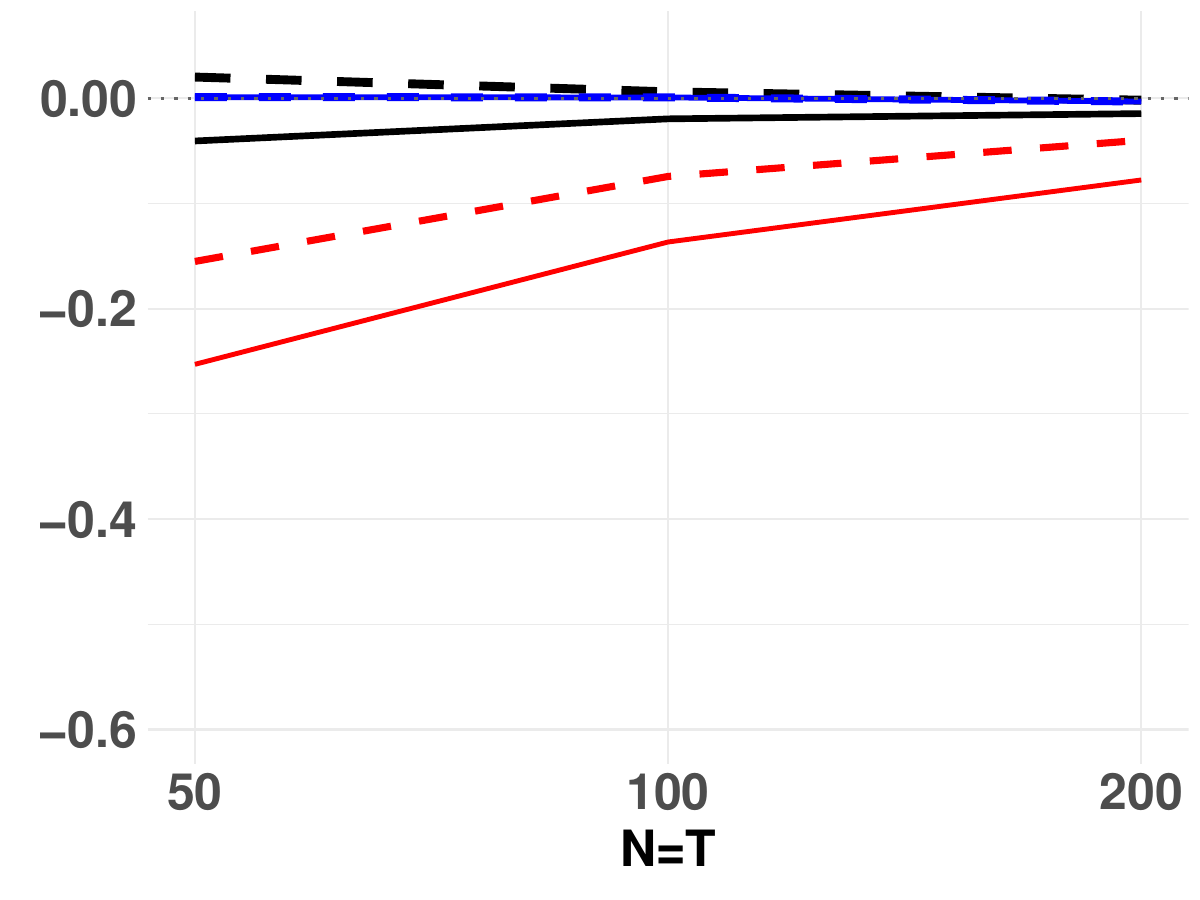}
\caption{$\rho_{wf}=0.6,\alpha_2 = 0.8$}
\label{fig:bias_f2w_06_08}
\end{subfigure}
\hfill
\begin{subfigure}[b]{0.32\textwidth}
\centering
\includegraphics[width=\textwidth]{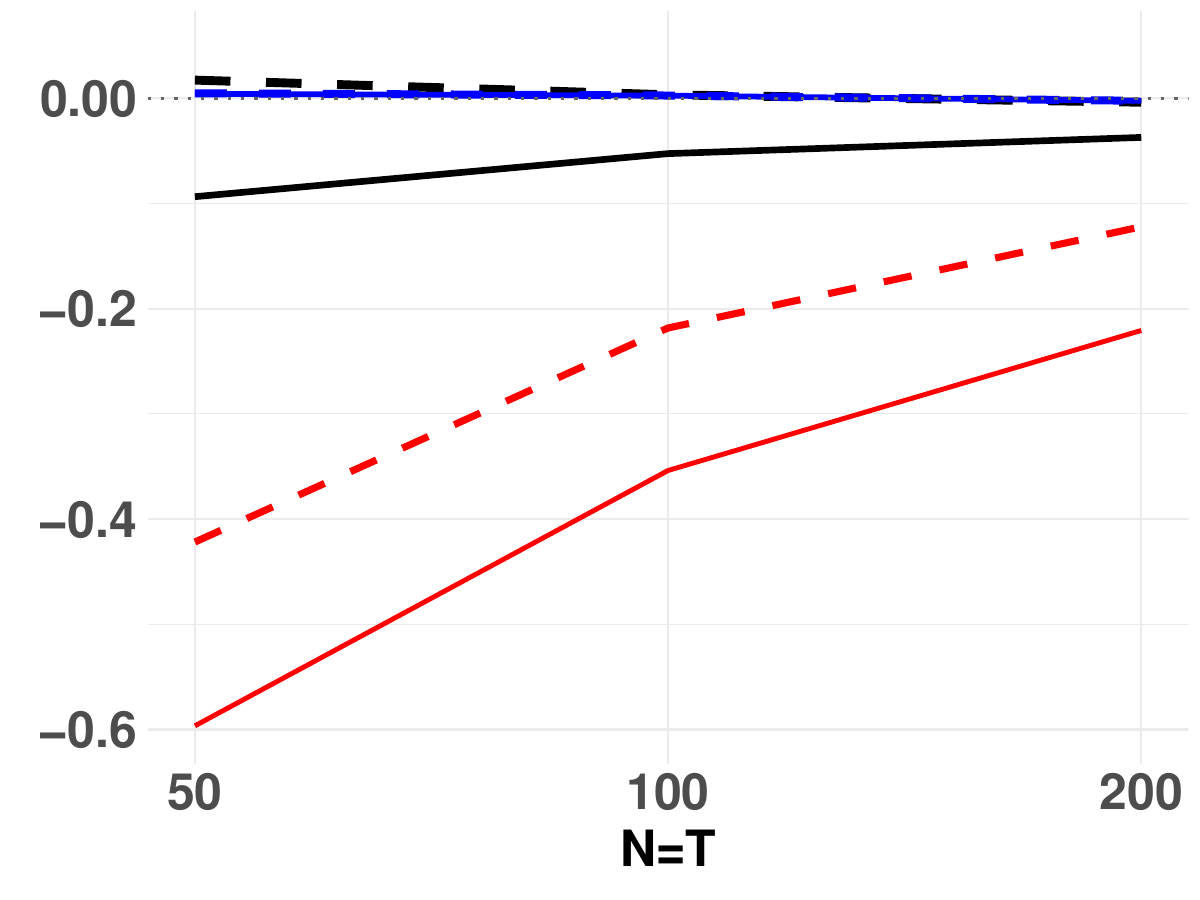}
\caption{$\rho_{wf}=0.6,\alpha_2 = 0.6$}
\label{fig:bias_f2w_06_06}
\end{subfigure}

\centering
\includegraphics[width=0.66\textwidth]{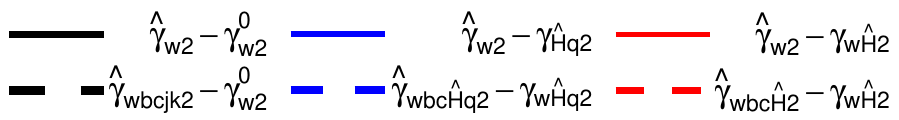}
\caption{Bias of $\hat{\gamma}_{w,2}$ and its bias corrected versions for cross and serially correlated $e_{t,i}$}
\label{fig:bias.f2w}

\end{figure}

\begin{figure}[h!]	
\centering
\begin{subfigure}[b]{0.32\textwidth}
\centering
\includegraphics[width=\textwidth]{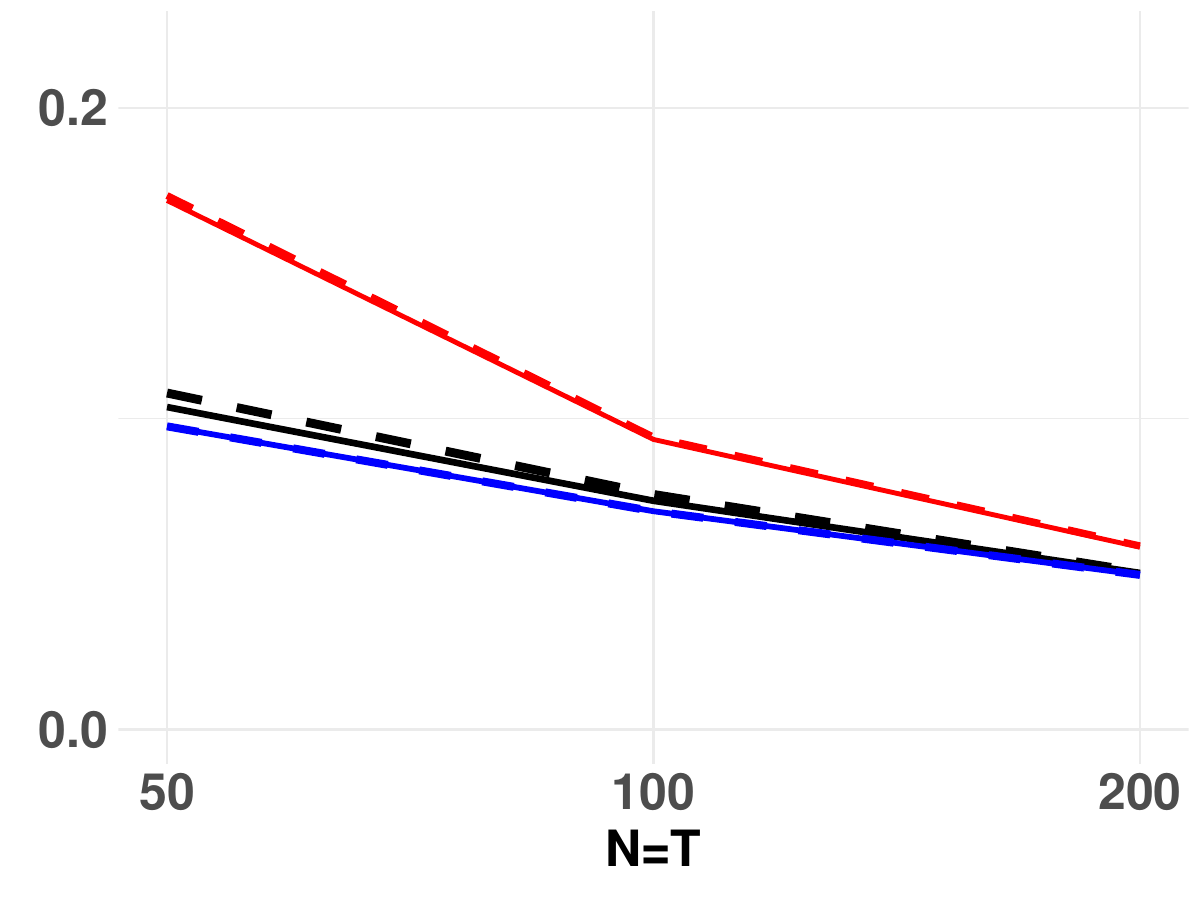}
\caption{$\rho_{wf}=0.0,\alpha_2 = 1.0$}
\label{fig:sd_f2w_00_10}
\end{subfigure}
\hfill
\begin{subfigure}[b]{0.32\textwidth}
\centering
\includegraphics[width=\textwidth]{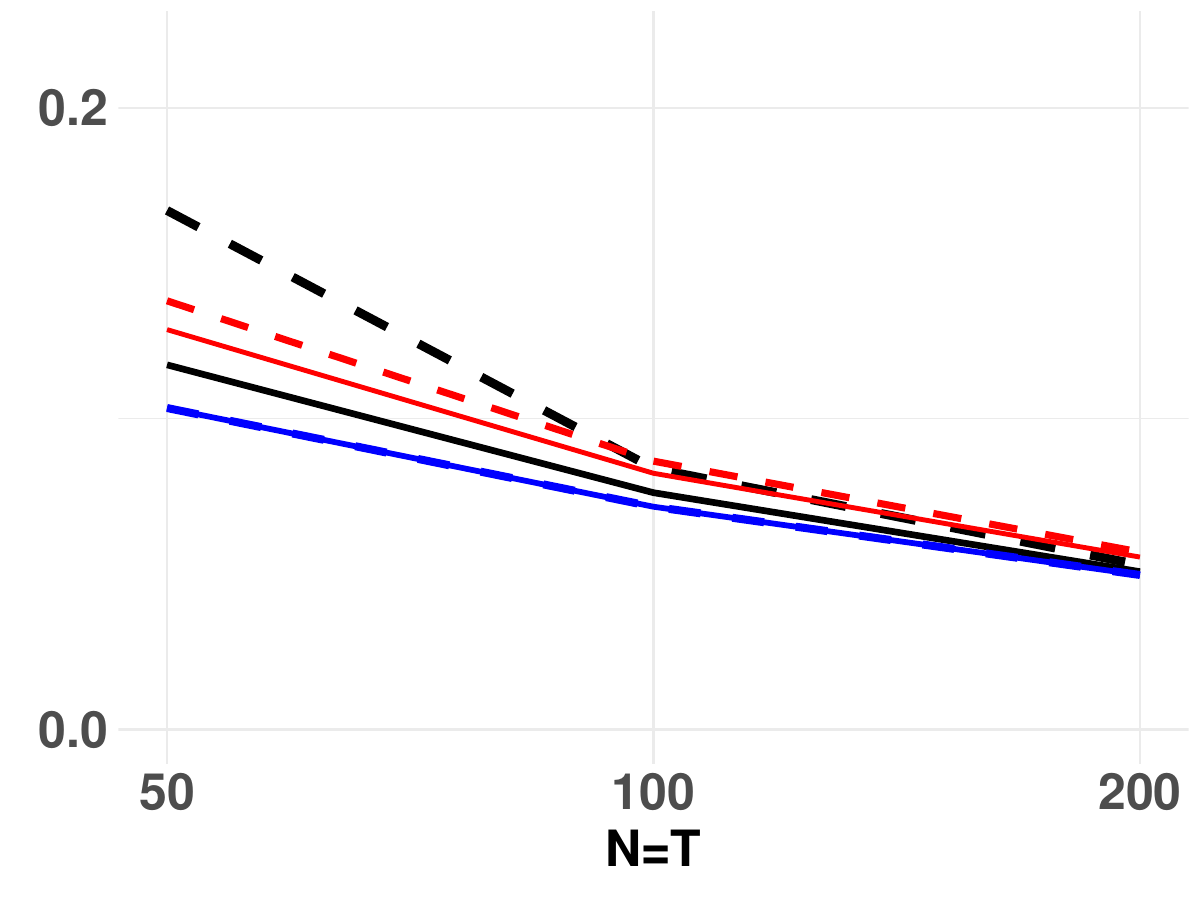}
\caption{$\rho_{wf}=0.0,\alpha_2 = 0.8$}
\label{fig:sd_f2w_00_08}
\end{subfigure}
\hfill
\begin{subfigure}[b]{0.32\textwidth}
\centering
\includegraphics[width=\textwidth]{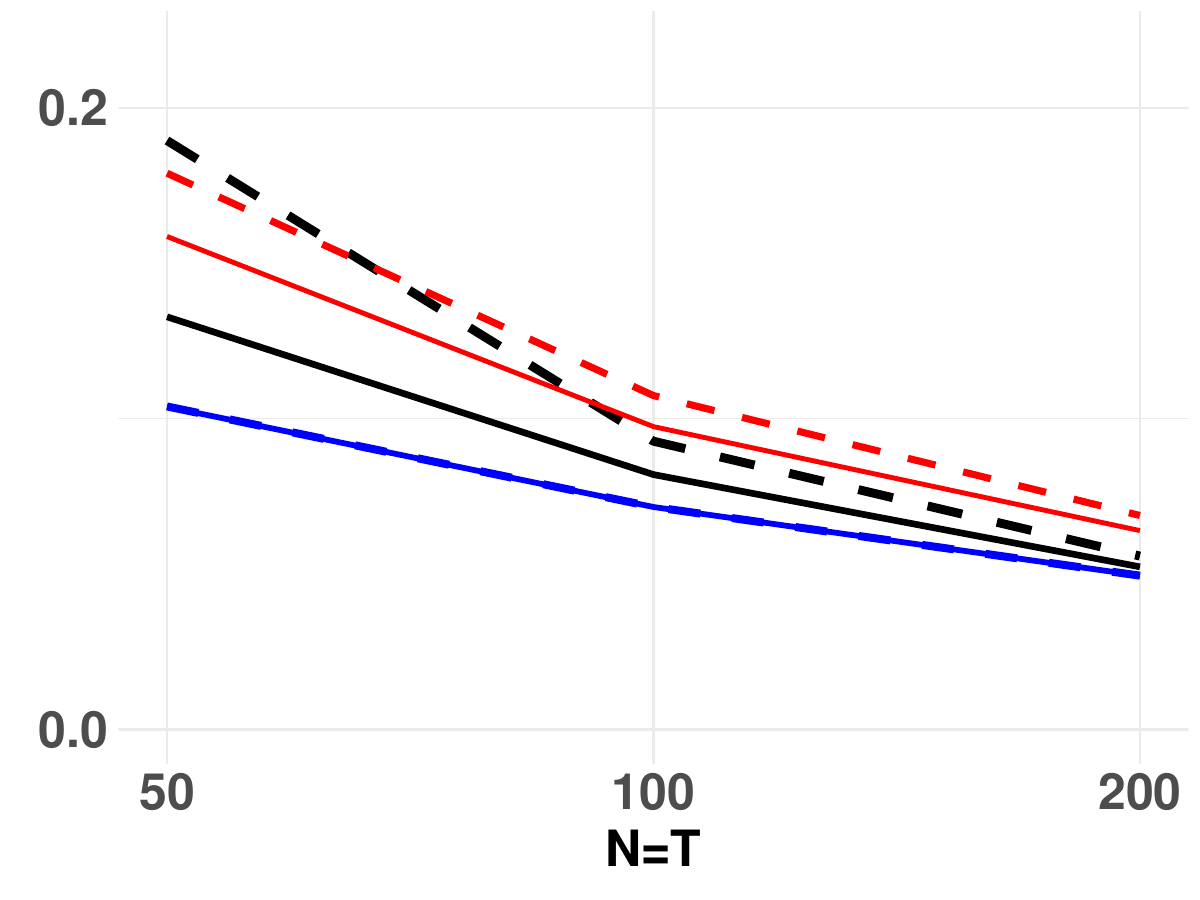}
\caption{$\rho_{wf}=0.0,\alpha_2 = 0.6$}
\label{fig:sd_f2w_00_06}
\end{subfigure}

\centering
\begin{subfigure}[b]{0.32\textwidth}
\centering
\includegraphics[width=\textwidth]{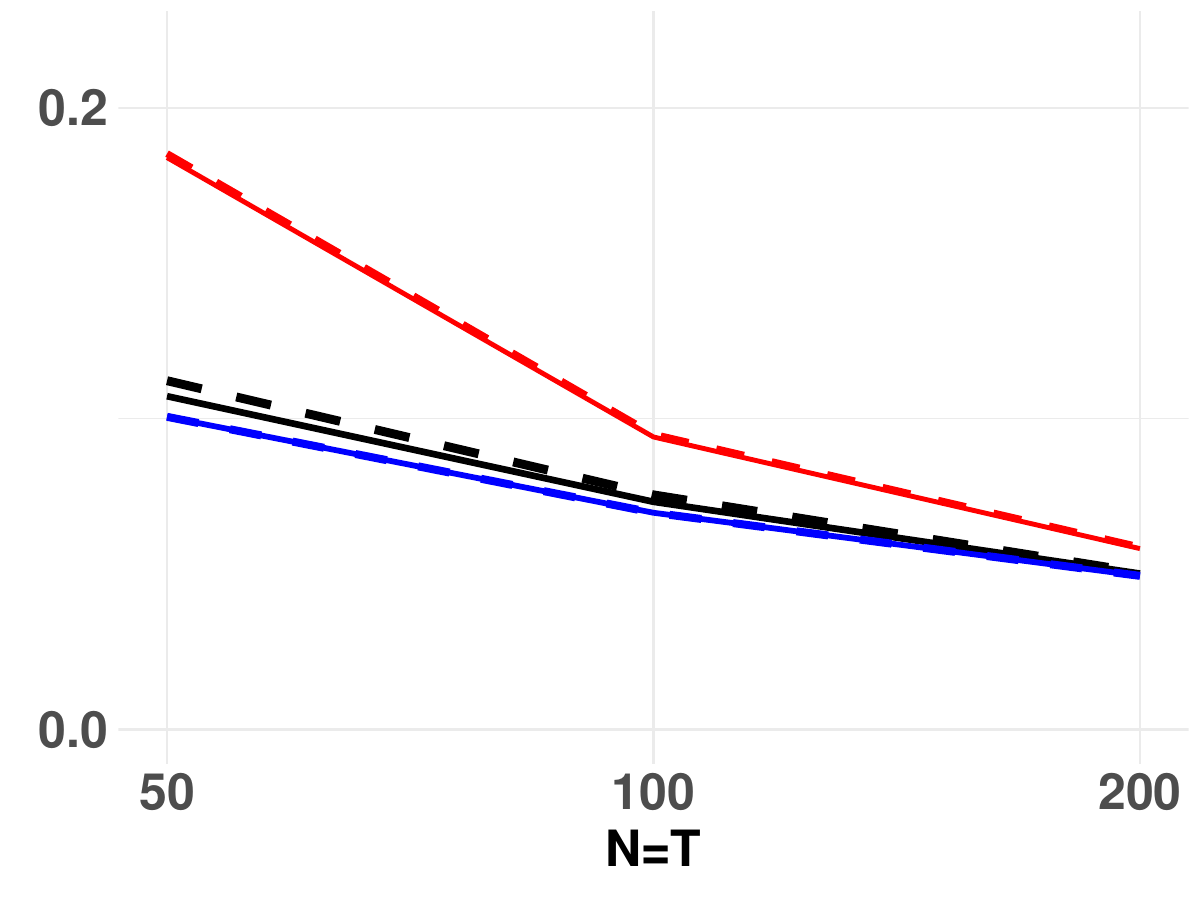}
\caption{$\rho_{wf}=0.6,\alpha_2 = 1.0$}
\label{fig:sd_f2w_06_10}
\end{subfigure}
\hfill
\begin{subfigure}[b]{0.32\textwidth}
\centering
\includegraphics[width=\textwidth]{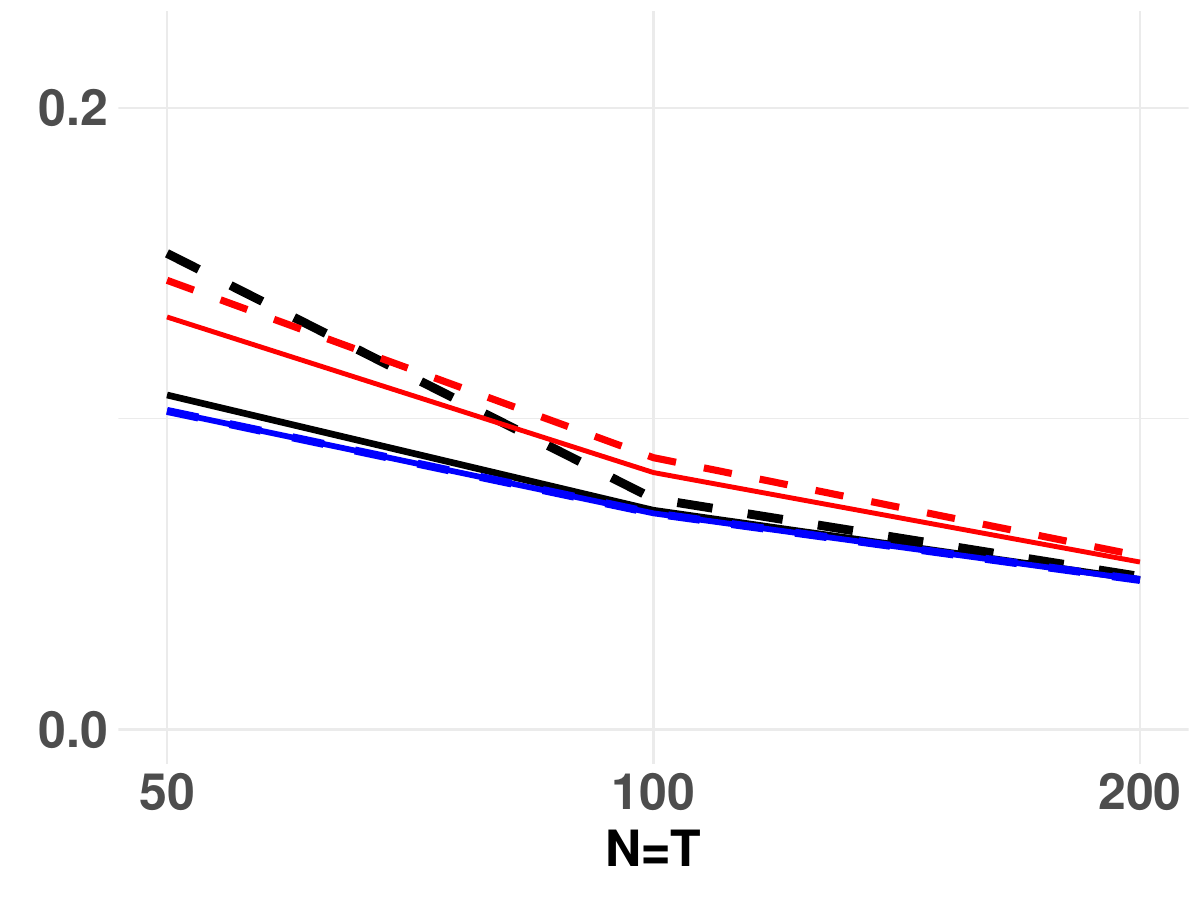}
\caption{$\rho_{wf}=0.6,\alpha_2 = 0.8$}
\label{fig:sd_f2w_06_08}
\end{subfigure}
\hfill
\begin{subfigure}[b]{0.32\textwidth}
\centering
\includegraphics[width=\textwidth]{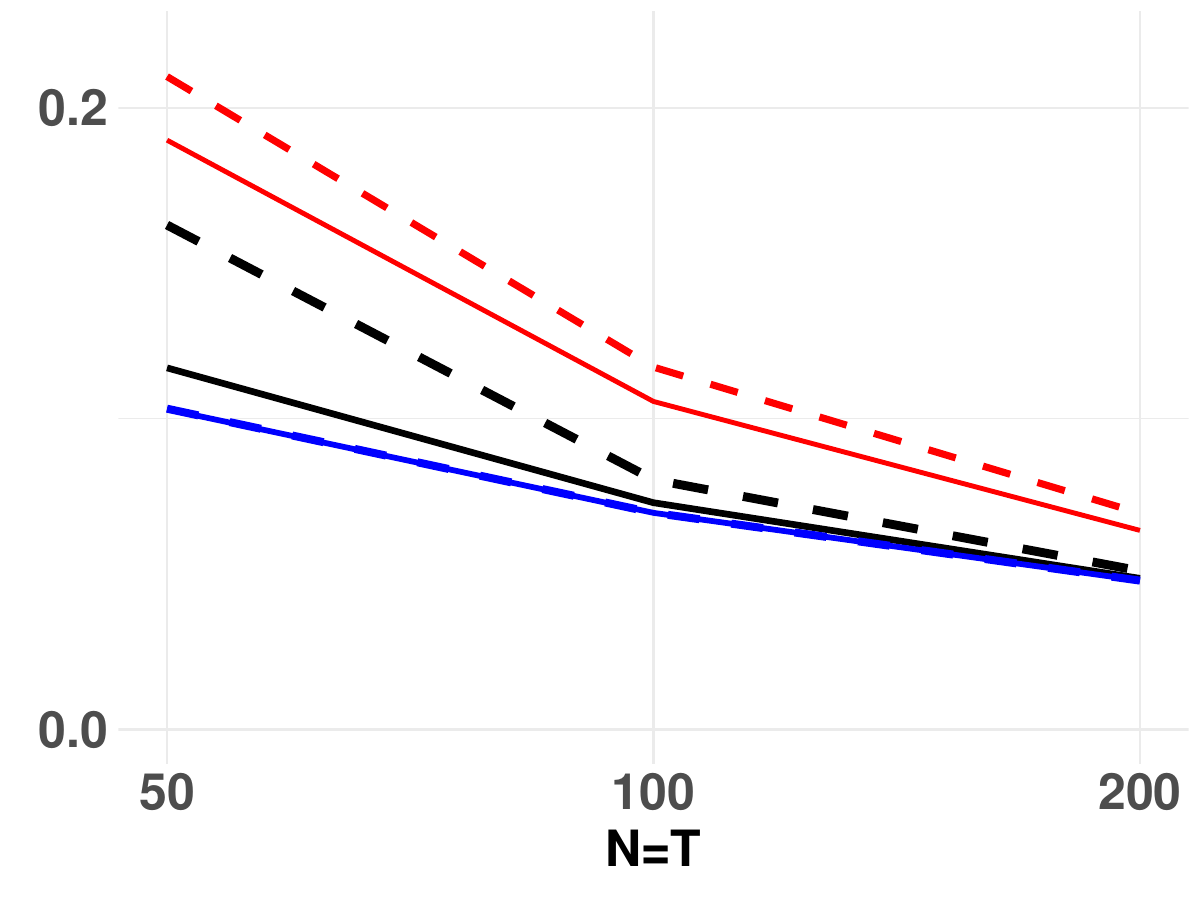}
\caption{$\rho_{wf}=0.6,\alpha_2 = 0.6$}
\label{fig:sd_f2w_06_06}
\end{subfigure}

\centering
\includegraphics[width=0.66\textwidth]{newimages/legend_SigE3_rhoe2_bias_rhowf0_a1_f2w.pdf}
\caption{Standard deviation of $\hat{\gamma}_{w,2}$ and its bias corrected versions for cross and serially correlated $e_{t,i}$}
\label{fig:sd.f2w}

\end{figure}

\begin{figure}[h!]	
\centering
\begin{subfigure}[b]{0.32\textwidth}
\centering
\includegraphics[width=\textwidth]{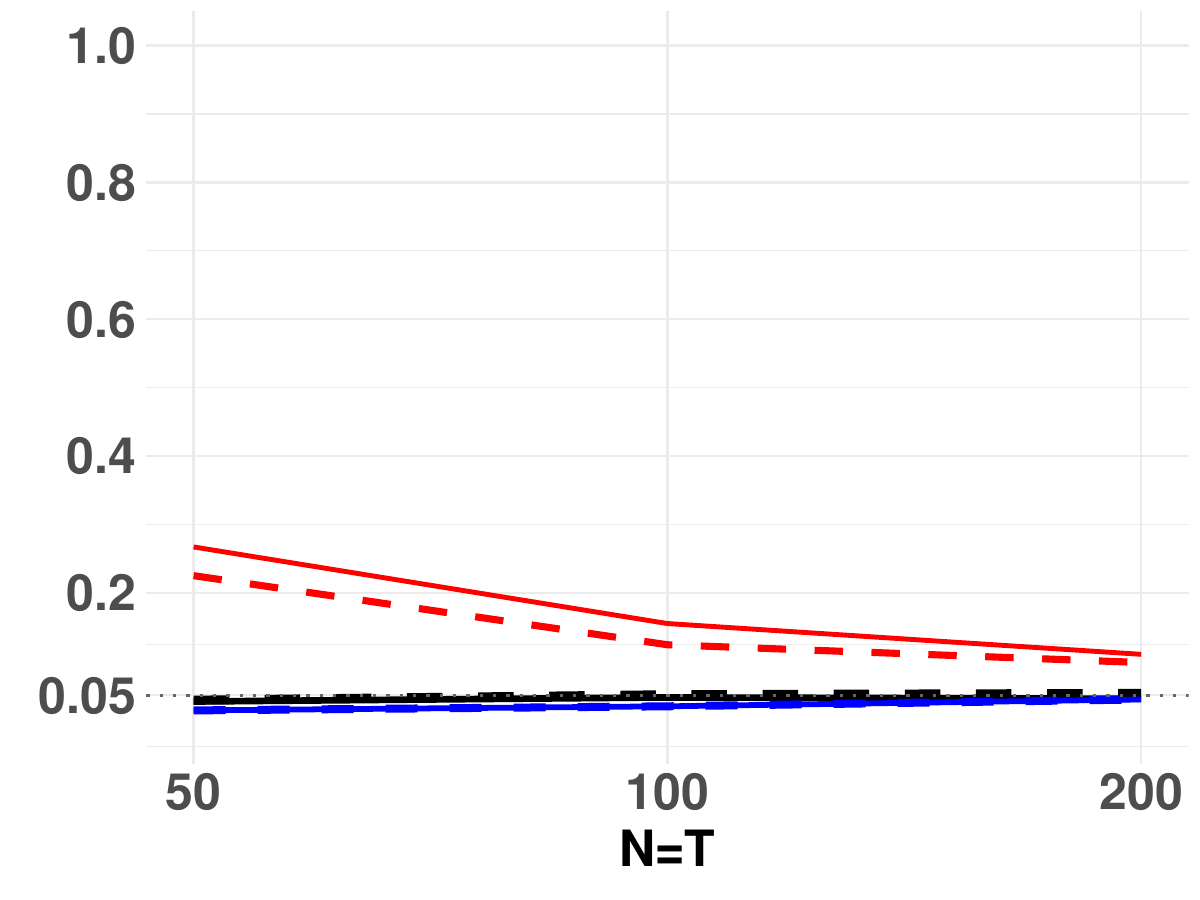}
\caption{$\rho_{wf}=0.0,\alpha_2 = 1.0$}
\label{fig:test_f2w_00_10}
\end{subfigure}
\hfill
\begin{subfigure}[b]{0.32\textwidth}
\centering
\includegraphics[width=\textwidth]{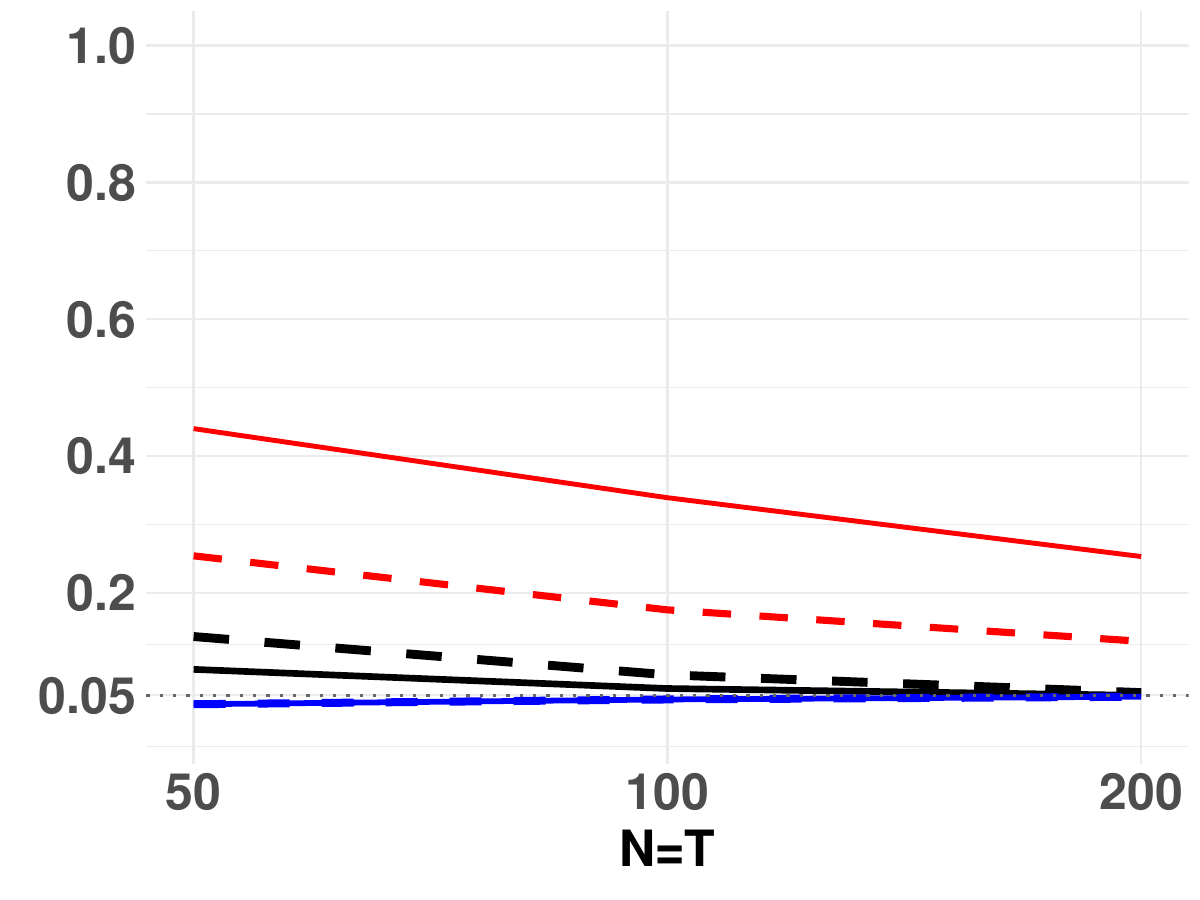}
\caption{$\rho_{wf}=0.0,\alpha_2 = 0.8$}
\label{fig:test_f2w_00_08}
\end{subfigure}
\hfill
\begin{subfigure}[b]{0.32\textwidth}
\centering
\includegraphics[width=\textwidth]{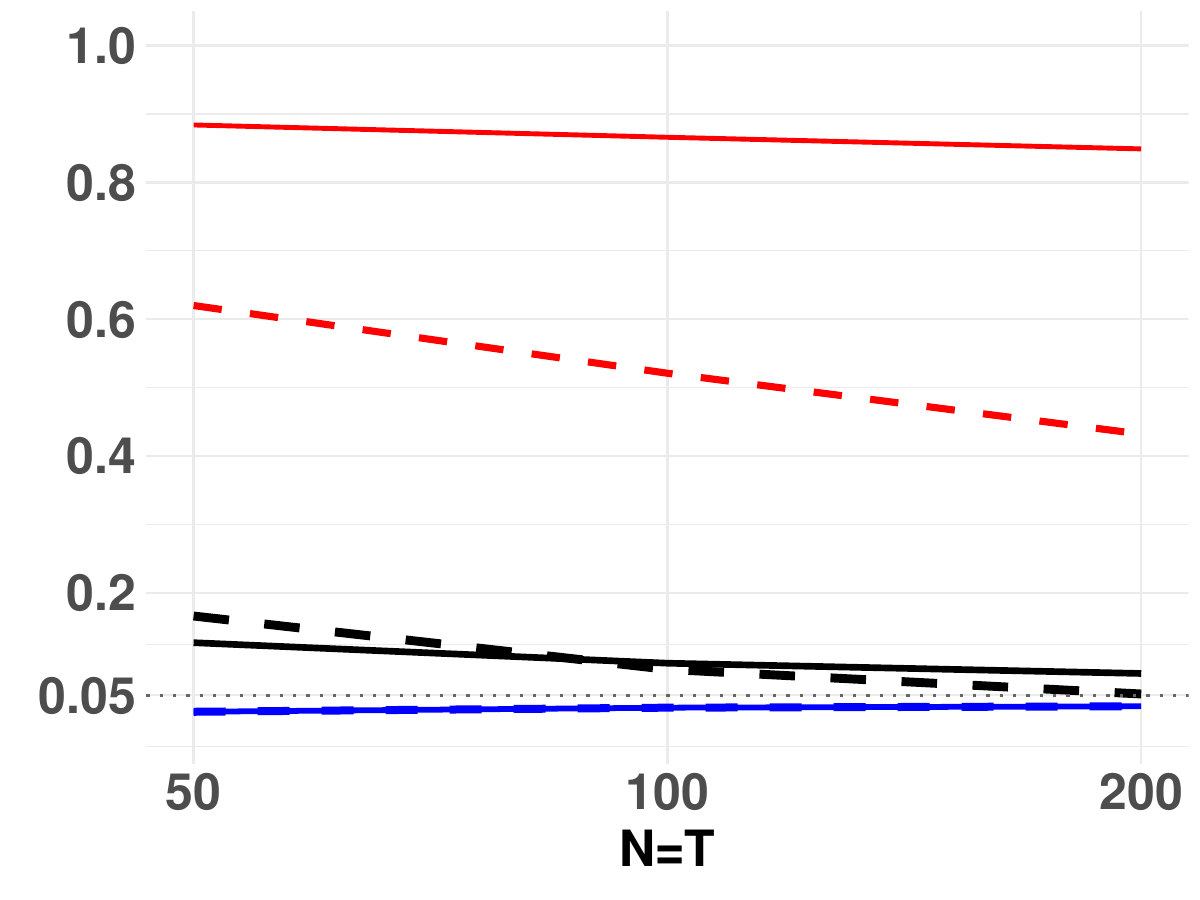}
\caption{$\rho_{wf}=0.0,\alpha_2 = 0.6$}
\label{fig:test_f2w_00_06}
\end{subfigure}

\centering
\begin{subfigure}[b]{0.32\textwidth}
\centering
\includegraphics[width=\textwidth]{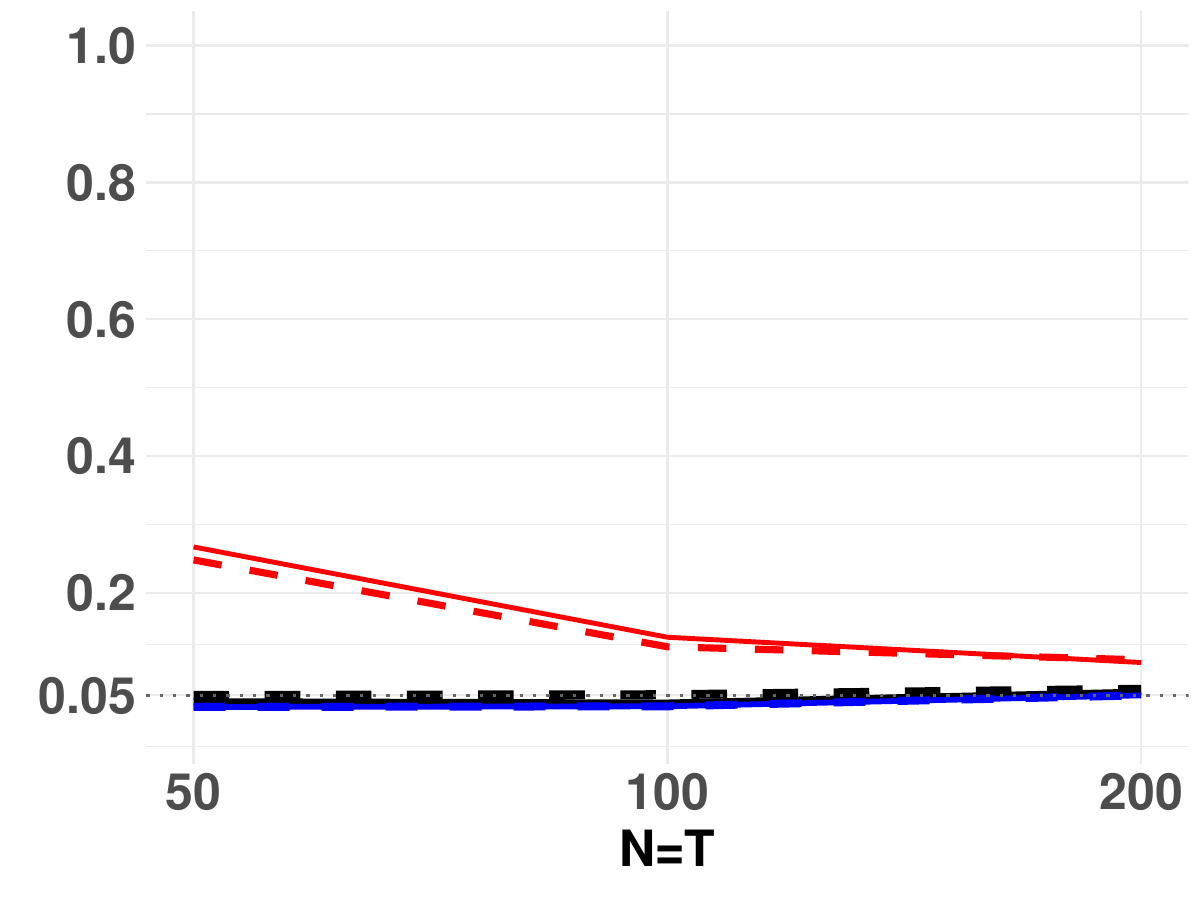}
\caption{$\rho_{wf}=0.6,\alpha_2 = 1.0$}
\label{fig:test_f2w_06_10}
\end{subfigure}
\hfill
\begin{subfigure}[b]{0.32\textwidth}
\centering
\includegraphics[width=\textwidth]{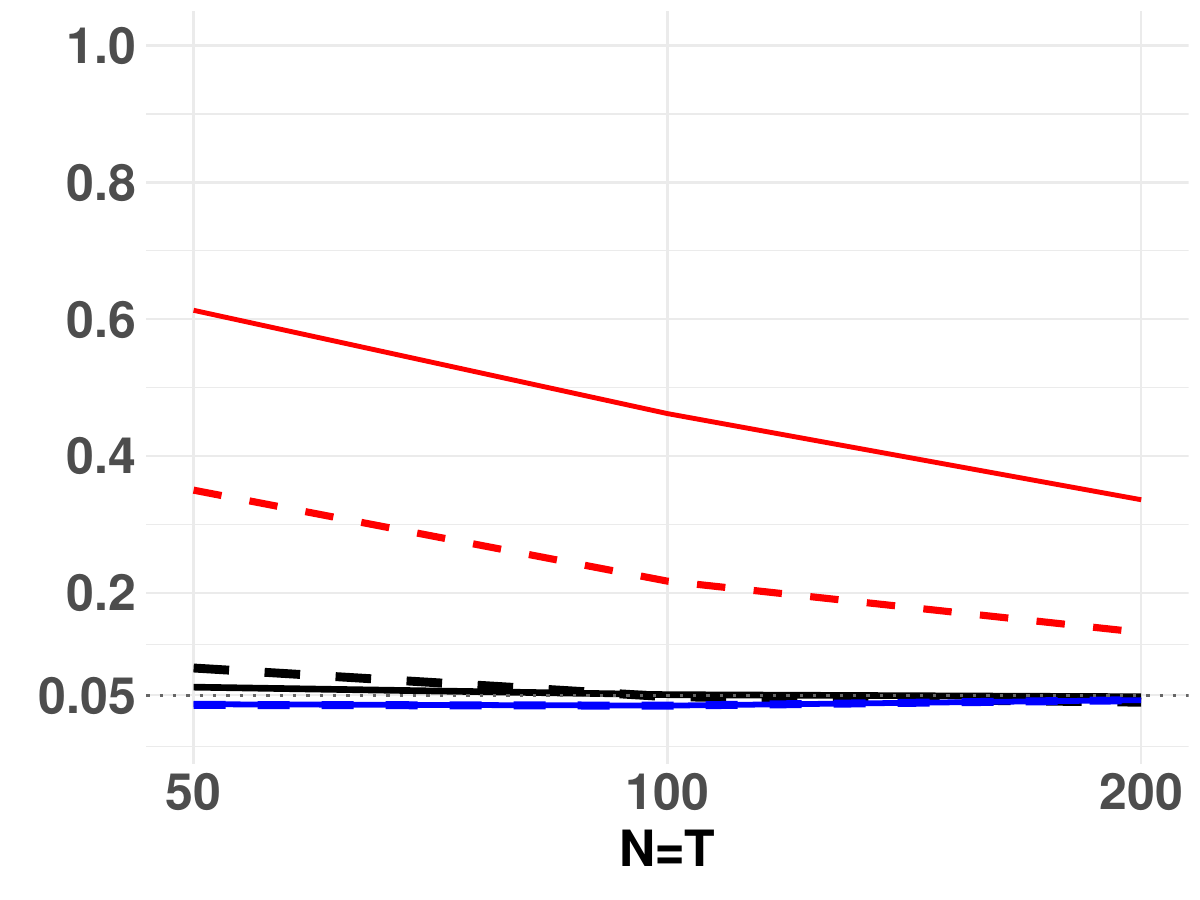}
\caption{$\rho_{wf}=0.6,\alpha_2 = 0.8$}
\label{fig:test_f2w_06_08}
\end{subfigure}
\hfill
\begin{subfigure}[b]{0.32\textwidth}
\centering
\includegraphics[width=\textwidth]{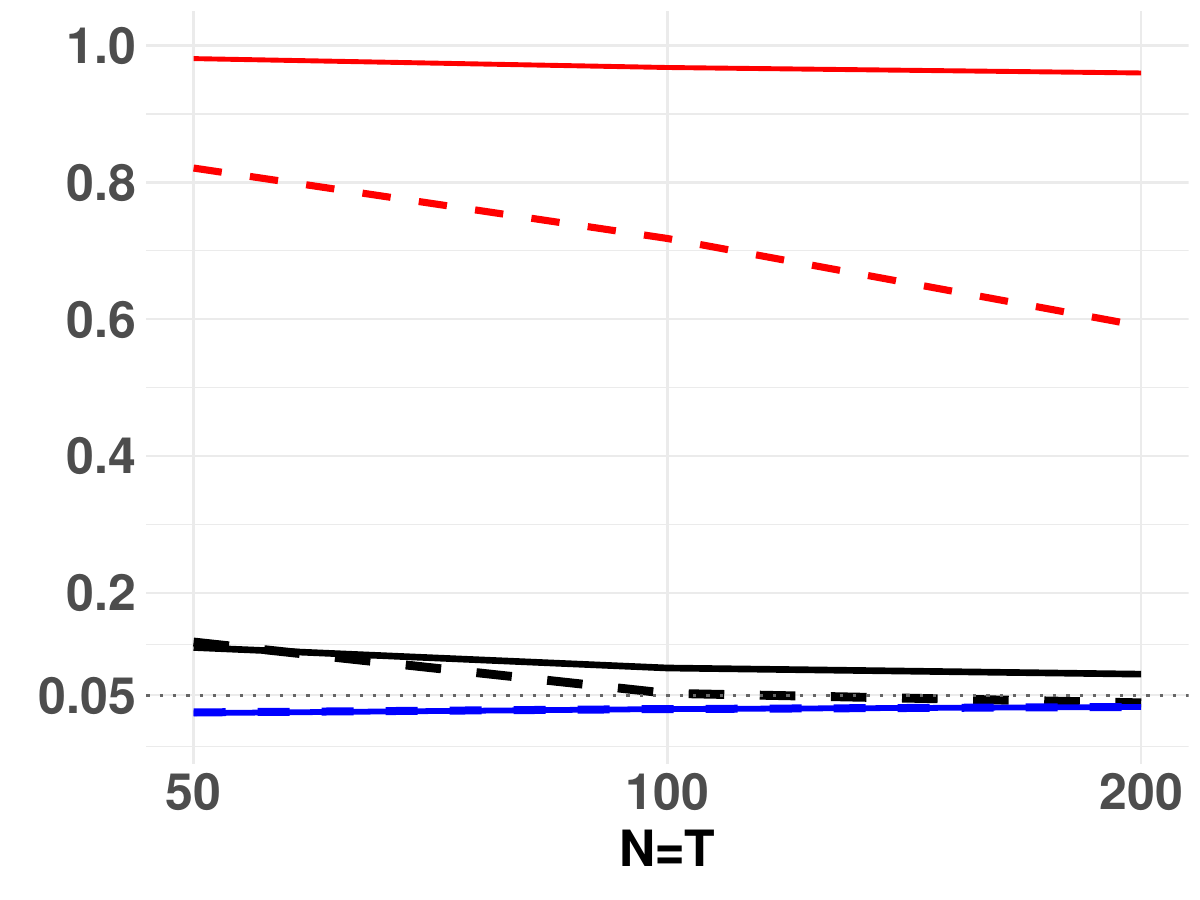}
\caption{$\rho_{wf}=0.6,\alpha_2 = 0.6$}
\label{fig:test_f2w_06_06}
\end{subfigure}

\centering
\includegraphics[width=0.66\textwidth]{newimages/legend_SigE3_rhoe2_bias_rhowf0_a1_f2w.pdf}
\caption{Size of the t-tests using $\hat{\gamma}_{w,2}$ and its bias corrected versions for cross and serially correlated $e_{t,i}$}
\label{fig:test.f2w}

\end{figure}

\subsection{Coefficient on the observed factor, using orthoganised $x_{t,i}$ to $w_t$}

\begin{figure}[h!]	
\centering
\begin{subfigure}[b]{0.32\textwidth}
\centering
\includegraphics[width=\textwidth]{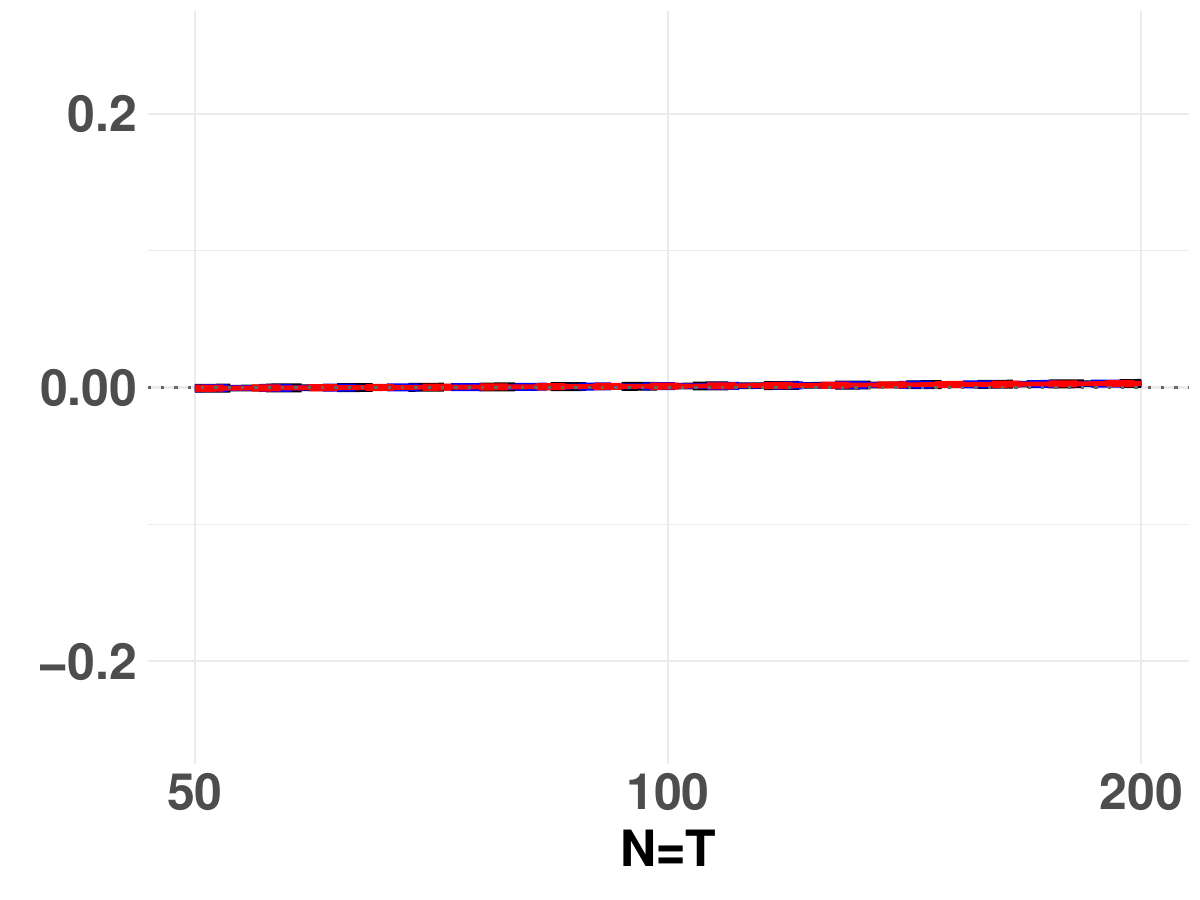}
\caption{$\rho_{wf}=0.0,\alpha_2 = 1.0$}
\label{fig:bias_ww_00_10}
\end{subfigure}
\hfill
\begin{subfigure}[b]{0.32\textwidth}
\centering
\includegraphics[width=\textwidth]{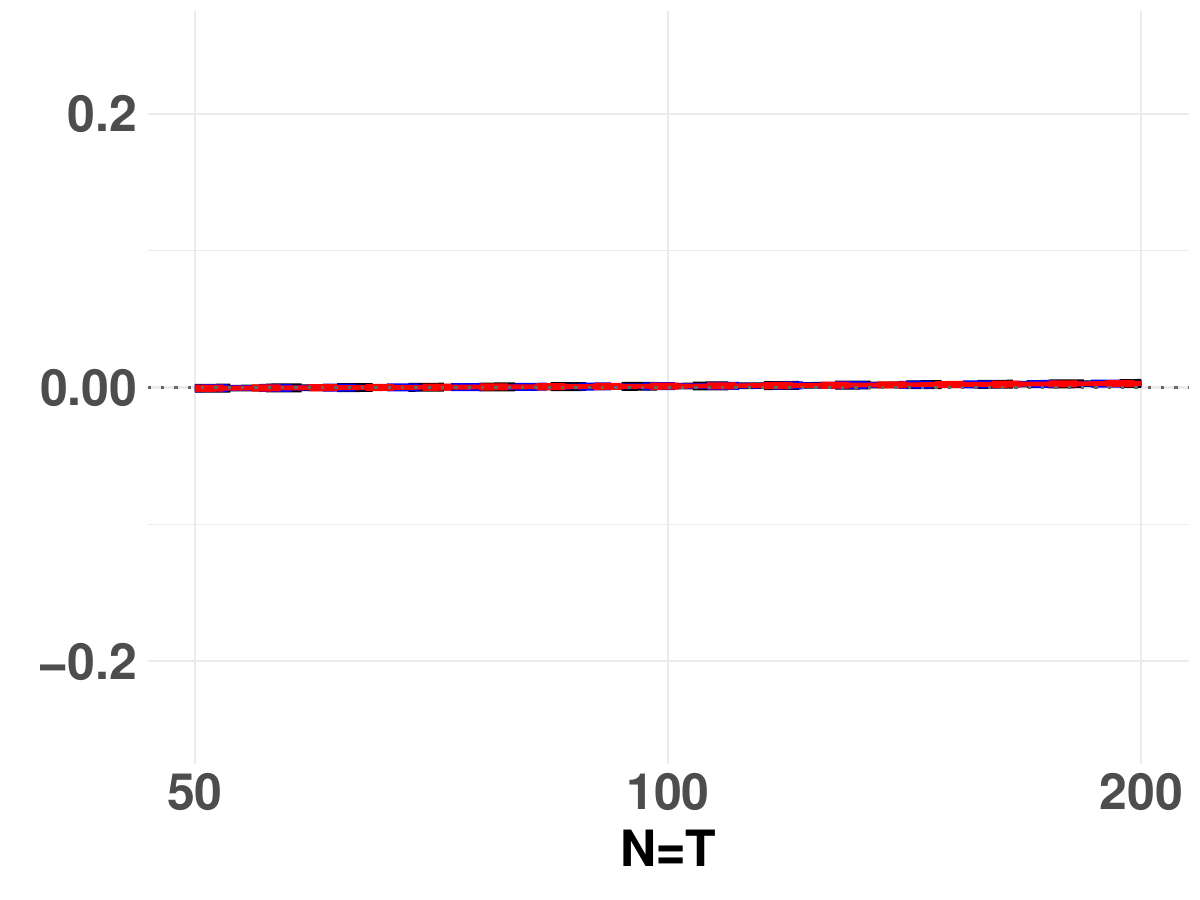}
\caption{$\rho_{wf}=0.0,\alpha_2 = 0.8$}
\label{fig:bias_ww_00_08}
\end{subfigure}
\hfill
\begin{subfigure}[b]{0.32\textwidth}
\centering
\includegraphics[width=\textwidth]{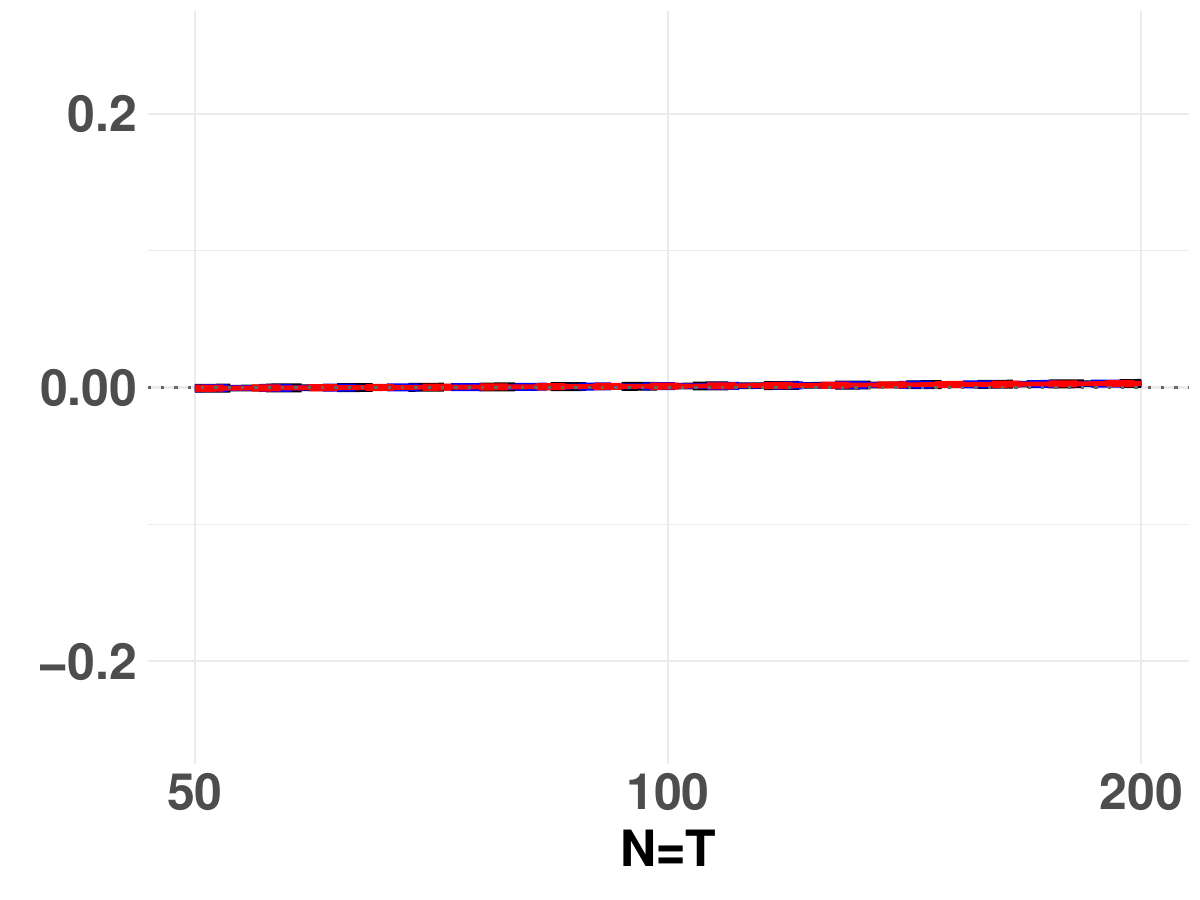}
\caption{$\rho_{wf}=0.0,\alpha_2 = 0.6$}
\label{fig:bias_ww_00_06}
\end{subfigure}

\centering
\begin{subfigure}[b]{0.32\textwidth}
\centering
\includegraphics[width=\textwidth]{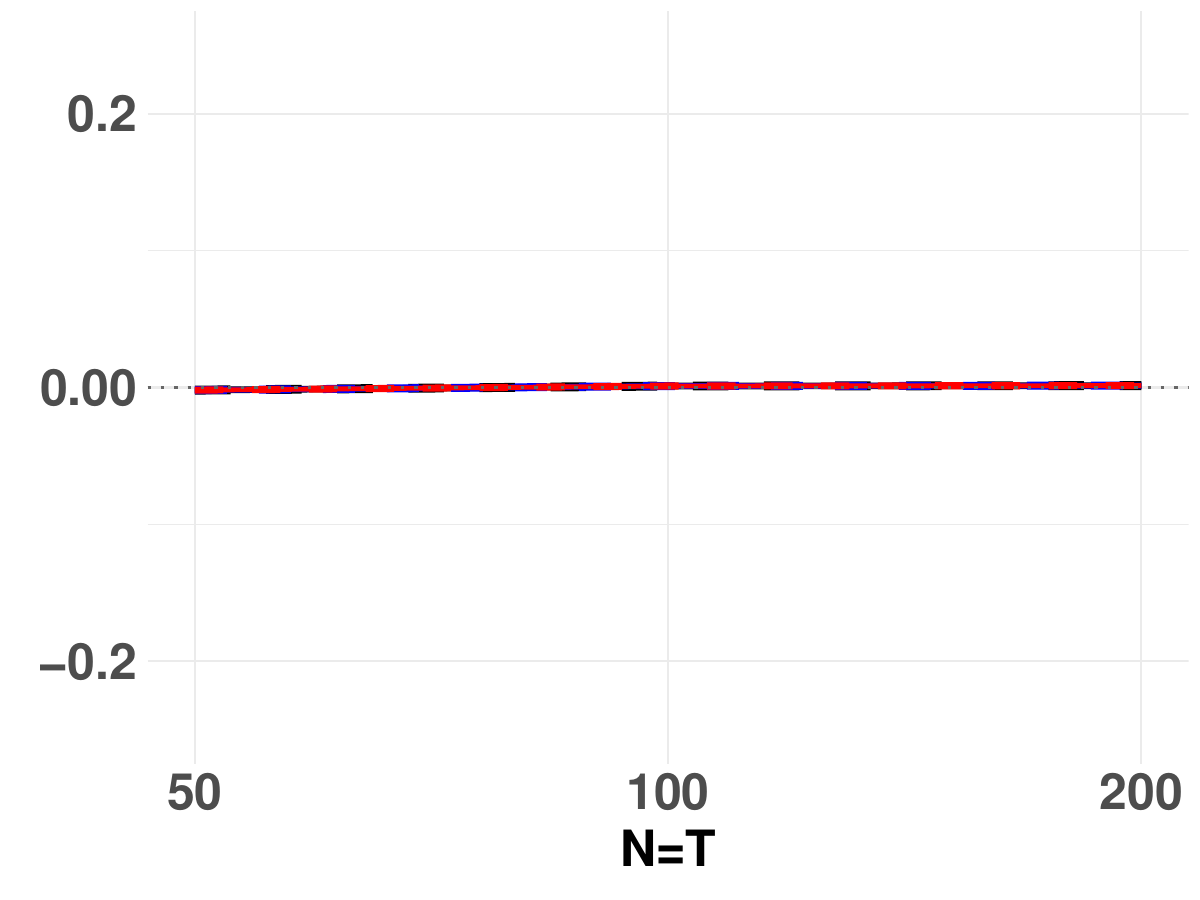}
\caption{$\rho_{wf}=0.6,\alpha_2 = 1.0$}
\label{fig:bias_ww_06_10}
\end{subfigure}
\hfill
\begin{subfigure}[b]{0.32\textwidth}
\centering
\includegraphics[width=\textwidth]{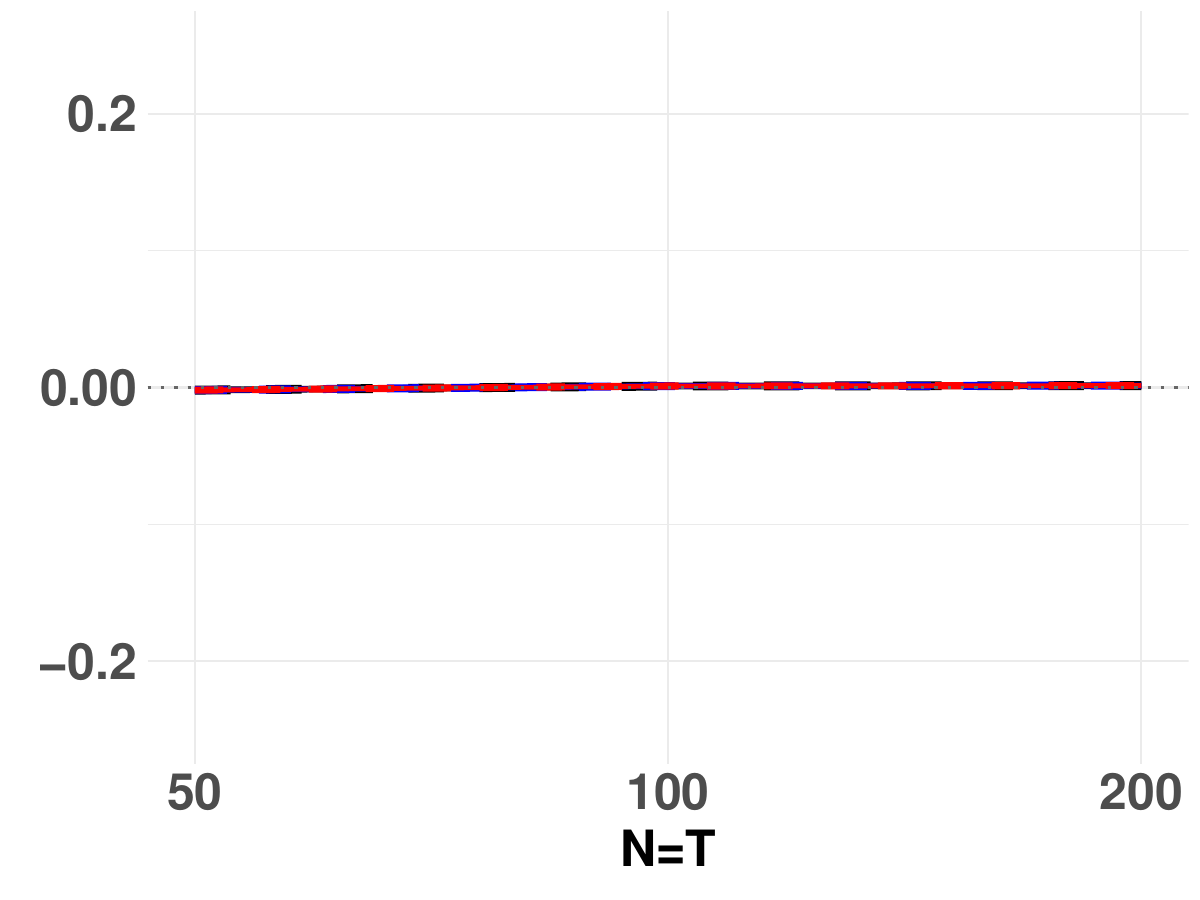}
\caption{$\rho_{wf}=0.6,\alpha_2 = 0.8$}
\label{fig:bias_ww_06_08}
\end{subfigure}
\hfill
\begin{subfigure}[b]{0.32\textwidth}
\centering
\includegraphics[width=\textwidth]{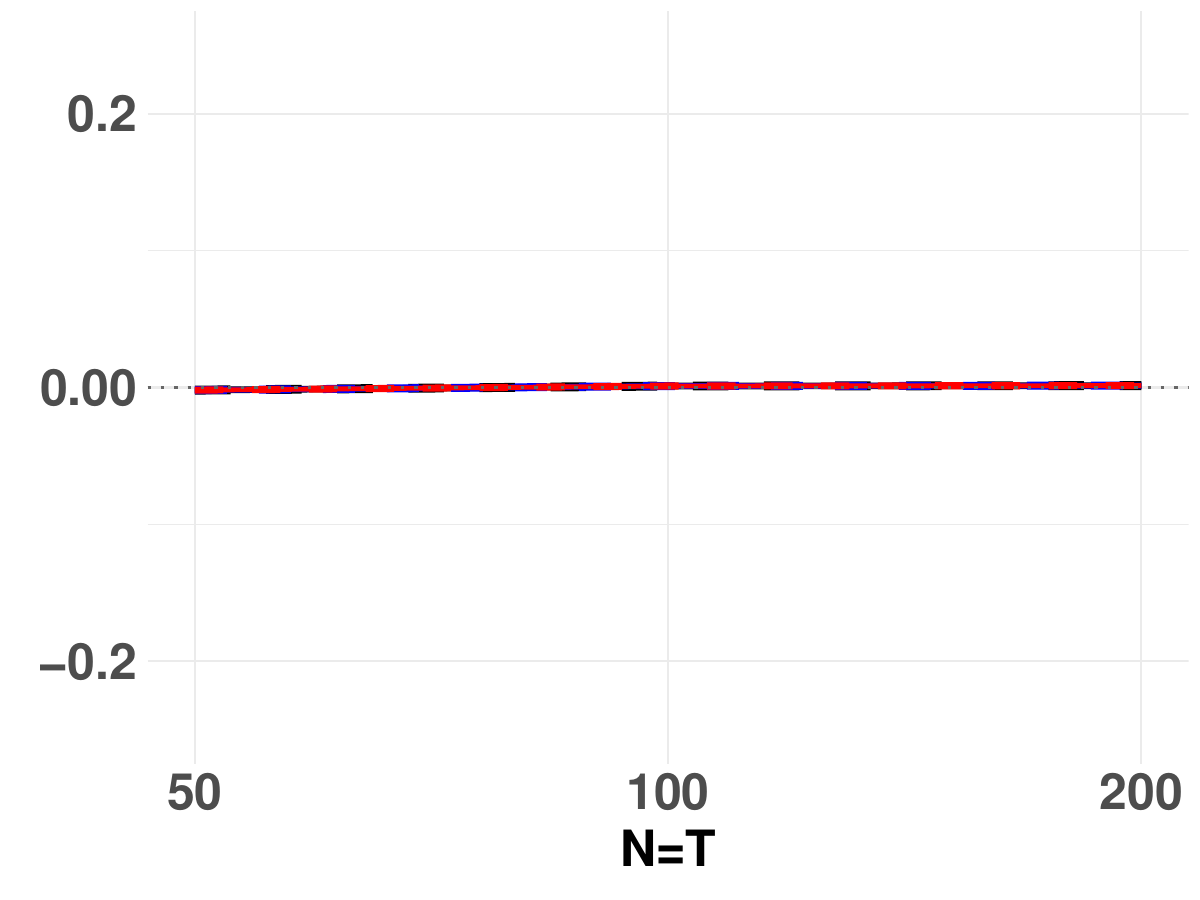}
\caption{$\rho_{wf}=0.6,\alpha_2 = 0.6$}
\label{fig:bias_ww_06_06}
\end{subfigure}

\centering
\includegraphics[width=0.66\textwidth]{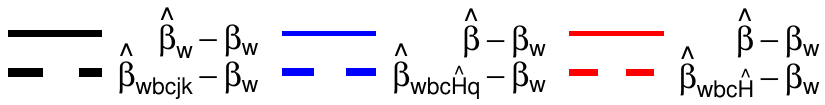}
\caption{Bias of $\hat{\beta}_w$ and its bias corrected versions for cross and serially correlated $e_{t,i}$}
\label{fig:bias.ww}

\end{figure}

\begin{figure}[h!]	
\centering
\begin{subfigure}[b]{0.32\textwidth}
\centering
\includegraphics[width=\textwidth]{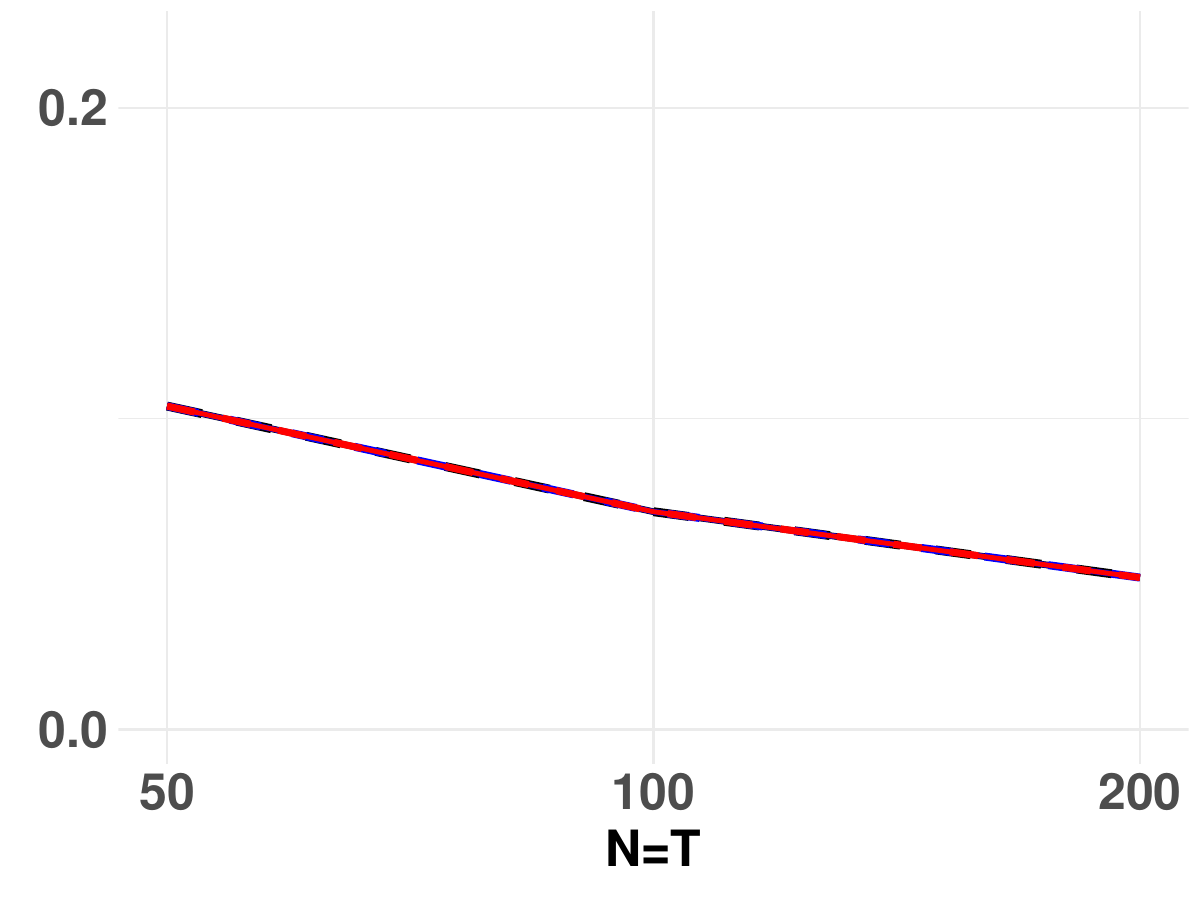}
\caption{$\rho_{wf}=0.0,\alpha_2 = 1.0$}
\label{fig:sd_ww_00_10}
\end{subfigure}
\hfill
\begin{subfigure}[b]{0.32\textwidth}
\centering
\includegraphics[width=\textwidth]{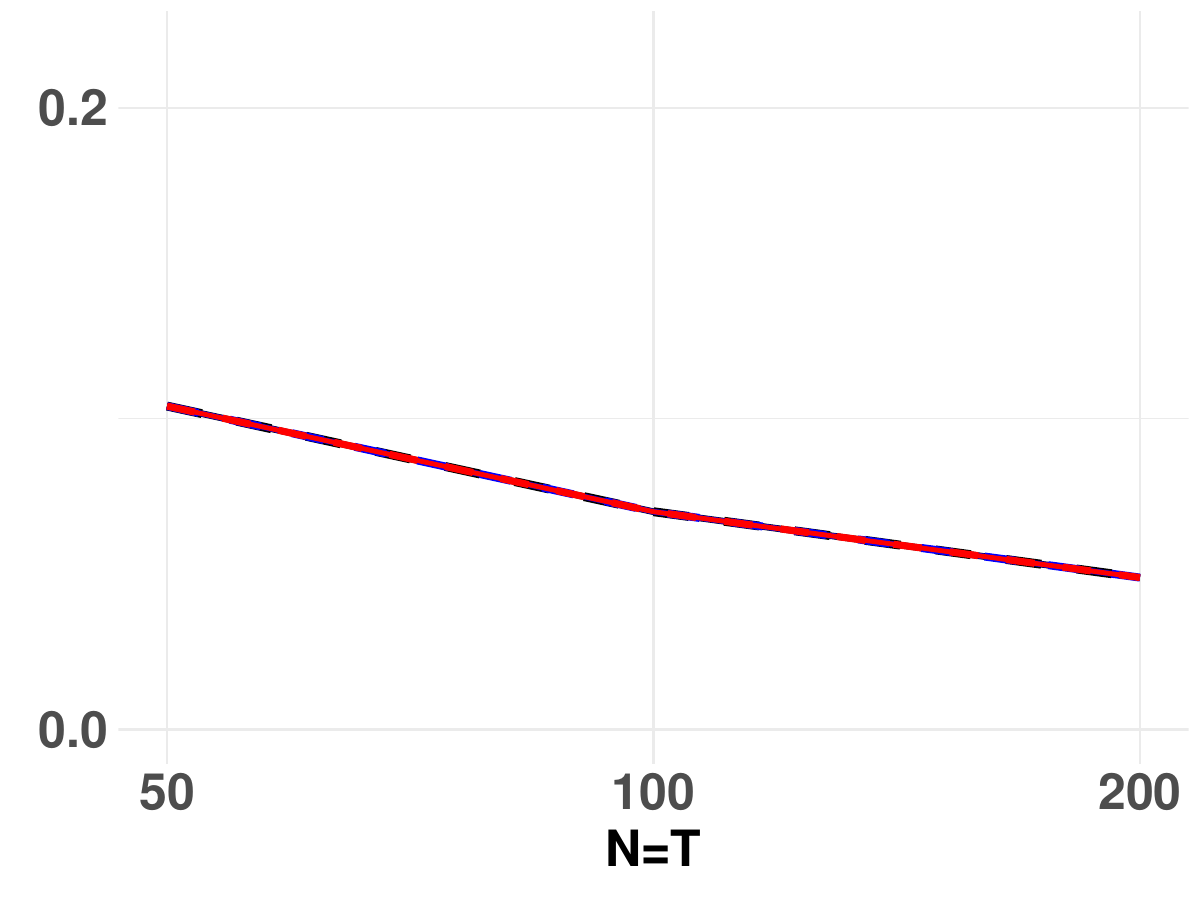}
\caption{$\rho_{wf}=0.0,\alpha_2 = 0.8$}
\label{fig:sd_ww_00_08}
\end{subfigure}
\hfill
\begin{subfigure}[b]{0.32\textwidth}
\centering
\includegraphics[width=\textwidth]{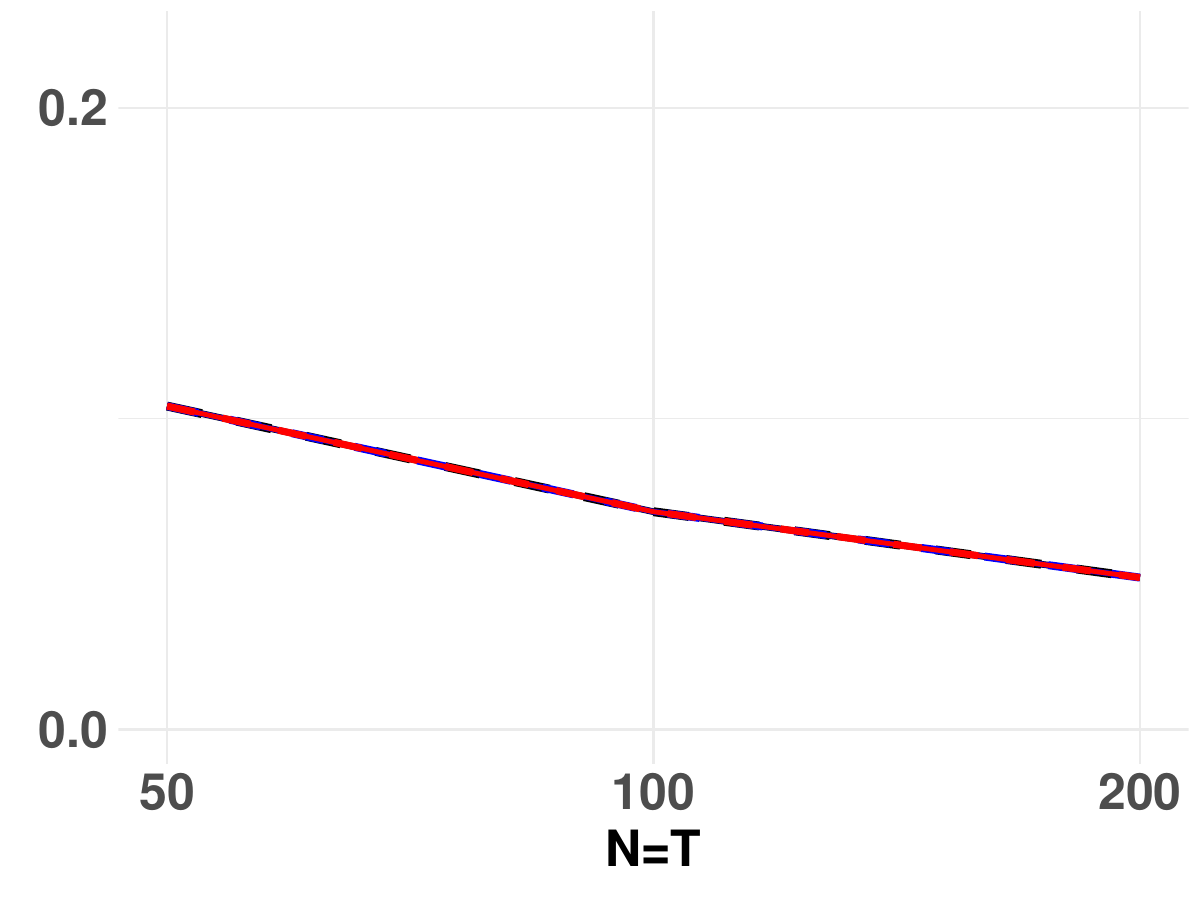}
\caption{$\rho_{wf}=0.0,\alpha_2 = 0.6$}
\label{fig:sd_ww_00_06}
\end{subfigure}

\centering
\begin{subfigure}[b]{0.32\textwidth}
\centering
\includegraphics[width=\textwidth]{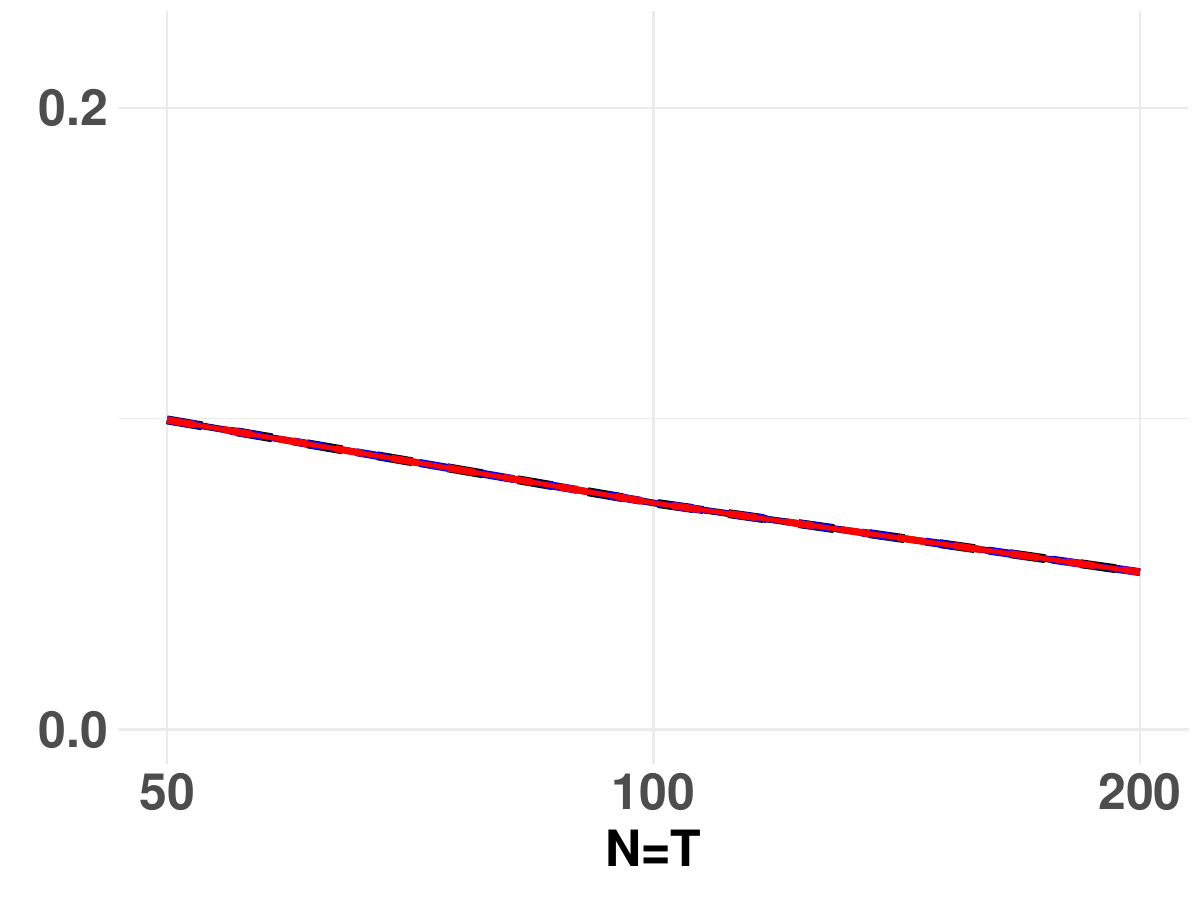}
\caption{$\rho_{wf}=0.6,\alpha_2 = 1.0$}
\label{fig:sd_ww_06_10}
\end{subfigure}
\hfill
\begin{subfigure}[b]{0.32\textwidth}
\centering
\includegraphics[width=\textwidth]{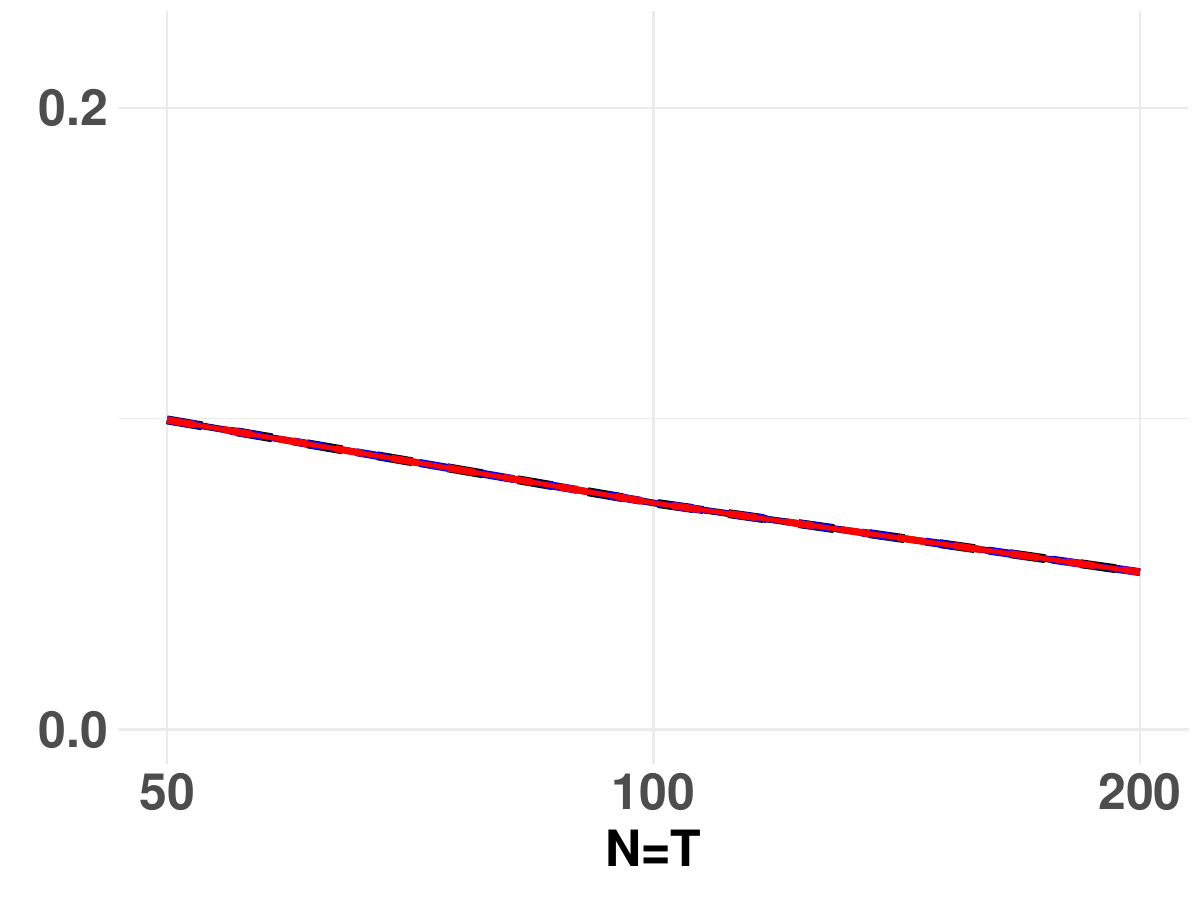}
\caption{$\rho_{wf}=0.6,\alpha_2 = 0.8$}
\label{fig:sd_ww_06_08}
\end{subfigure}
\hfill
\begin{subfigure}[b]{0.32\textwidth}
\centering
\includegraphics[width=\textwidth]{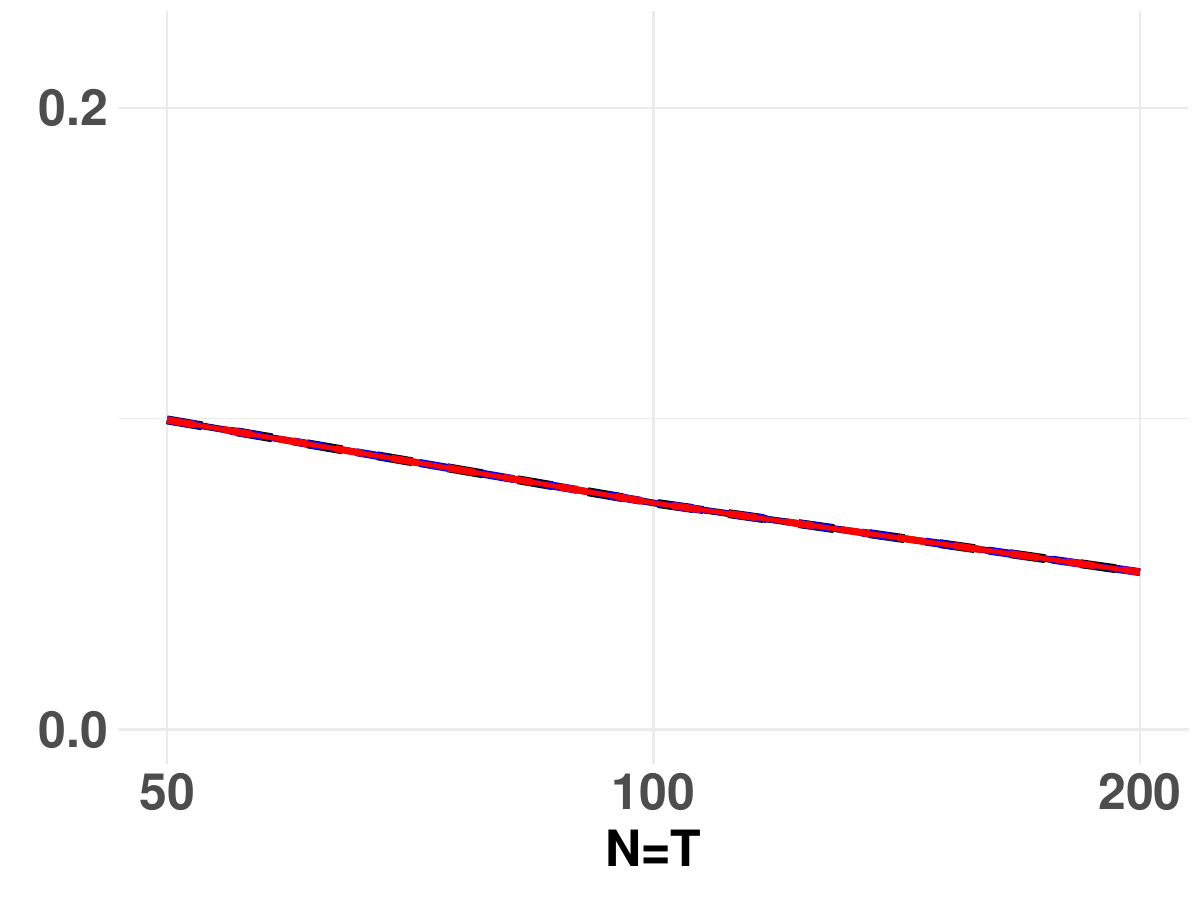}
\caption{$\rho_{wf}=0.6,\alpha_2 = 0.6$}
\label{fig:sd_ww_06_06}
\end{subfigure}

\centering
\includegraphics[width=0.66\textwidth]{newimages/legend_SigE3_rhoe2_bias_rhowf0_a1_ww.pdf}
\caption{Standard deviation of $\hat{\beta}_w$ and its bias corrected versions for cross and serially correlated $e_{t,i}$}
\label{fig:sd.ww}

\end{figure}

\begin{figure}[h!]	
\centering
\begin{subfigure}[b]{0.32\textwidth}
\centering
\includegraphics[width=\textwidth]{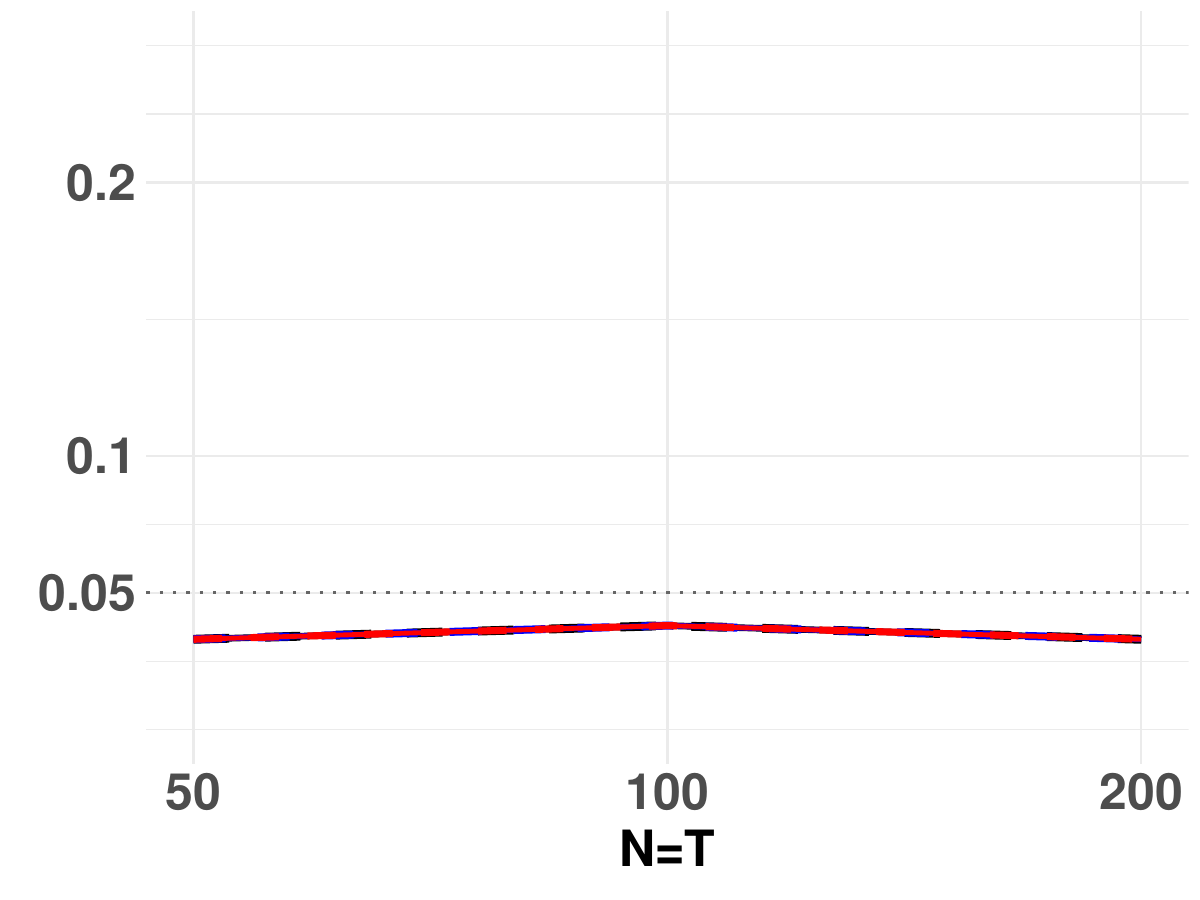}
\caption{$\rho_{wf}=0.0,\alpha_2 = 1.0$}
\label{fig:test_ww_00_10}
\end{subfigure}
\hfill
\begin{subfigure}[b]{0.32\textwidth}
\centering
\includegraphics[width=\textwidth]{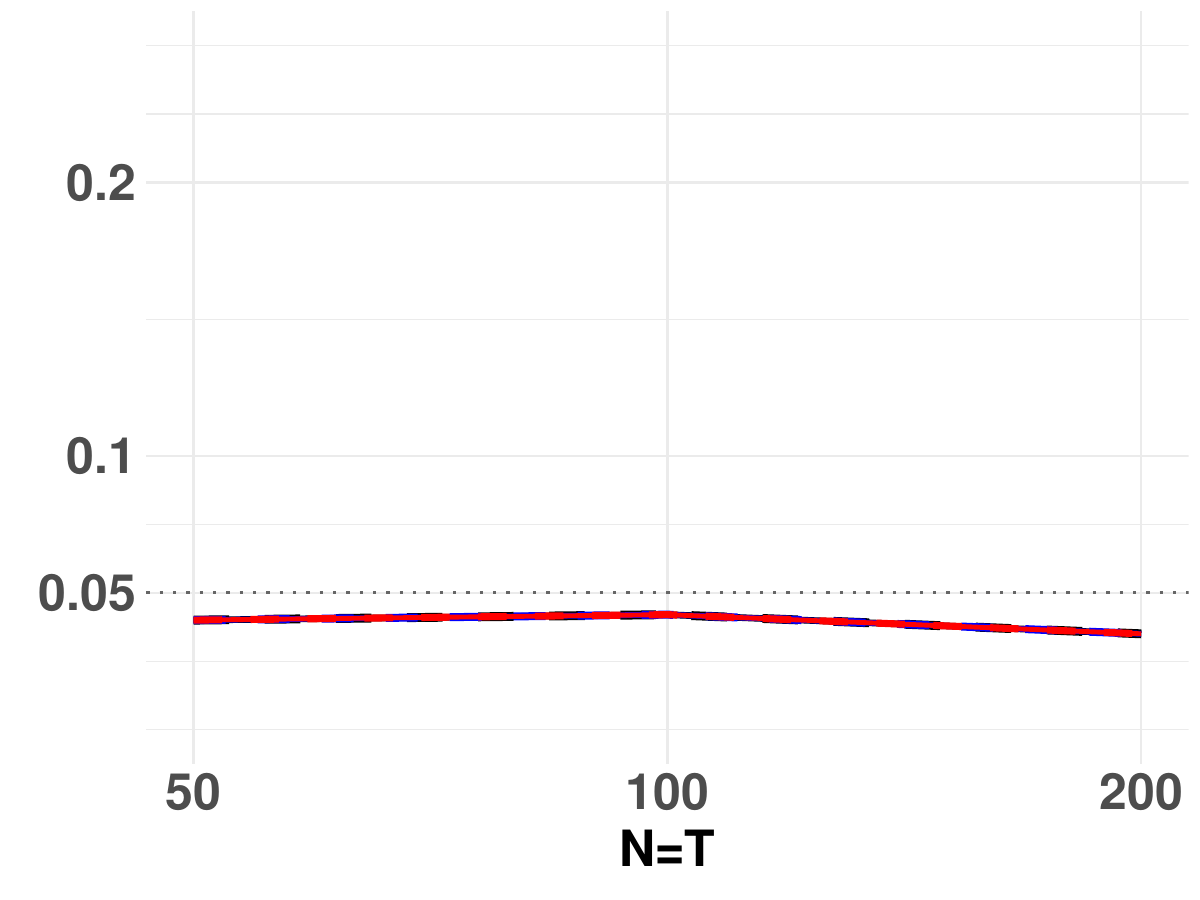}
\caption{$\rho_{wf}=0.0,\alpha_2 = 0.8$}
\label{fig:test_ww_00_08}
\end{subfigure}
\hfill
\begin{subfigure}[b]{0.32\textwidth}
\centering
\includegraphics[width=\textwidth]{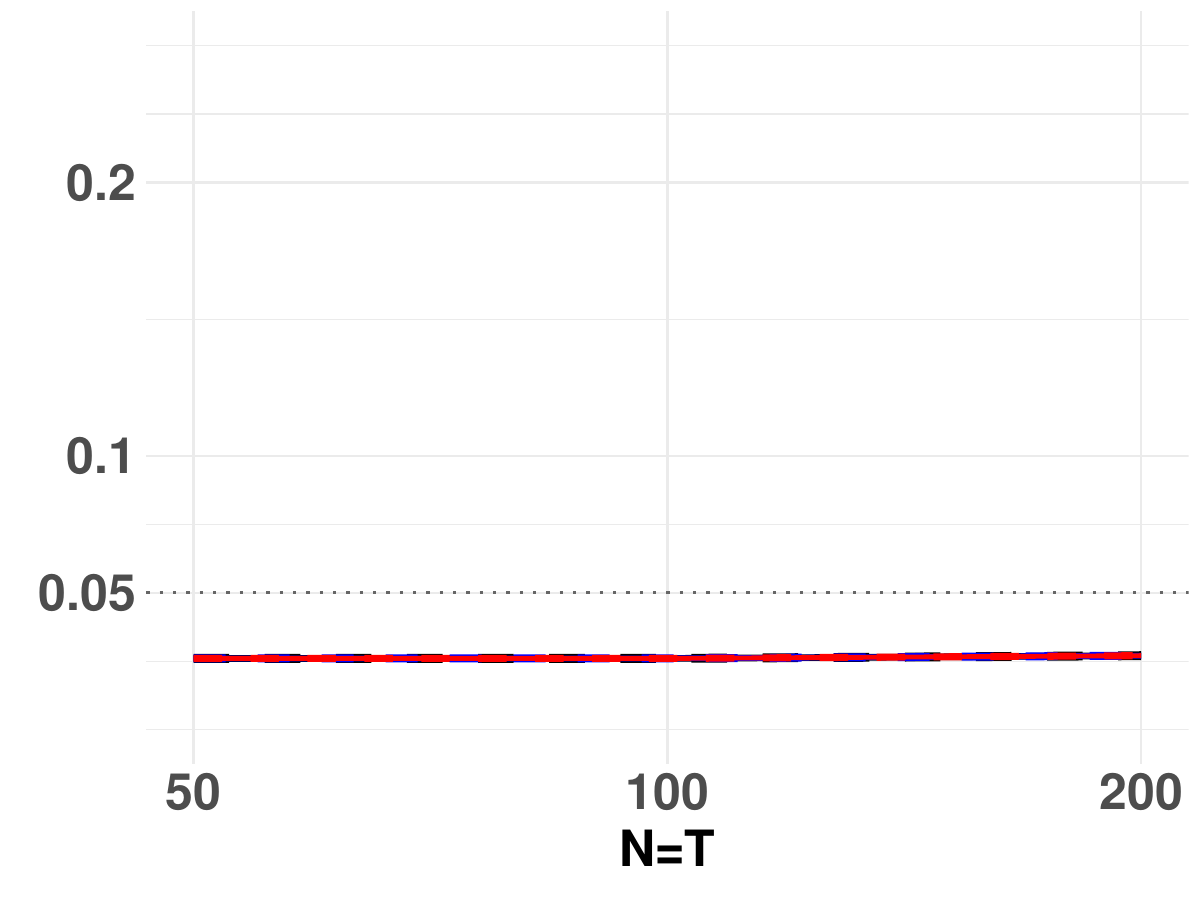}
\caption{$\rho_{wf}=0.0,\alpha_2 = 0.6$}
\label{fig:test_ww_00_06}
\end{subfigure}

\centering
\begin{subfigure}[b]{0.32\textwidth}
\centering
\includegraphics[width=\textwidth]{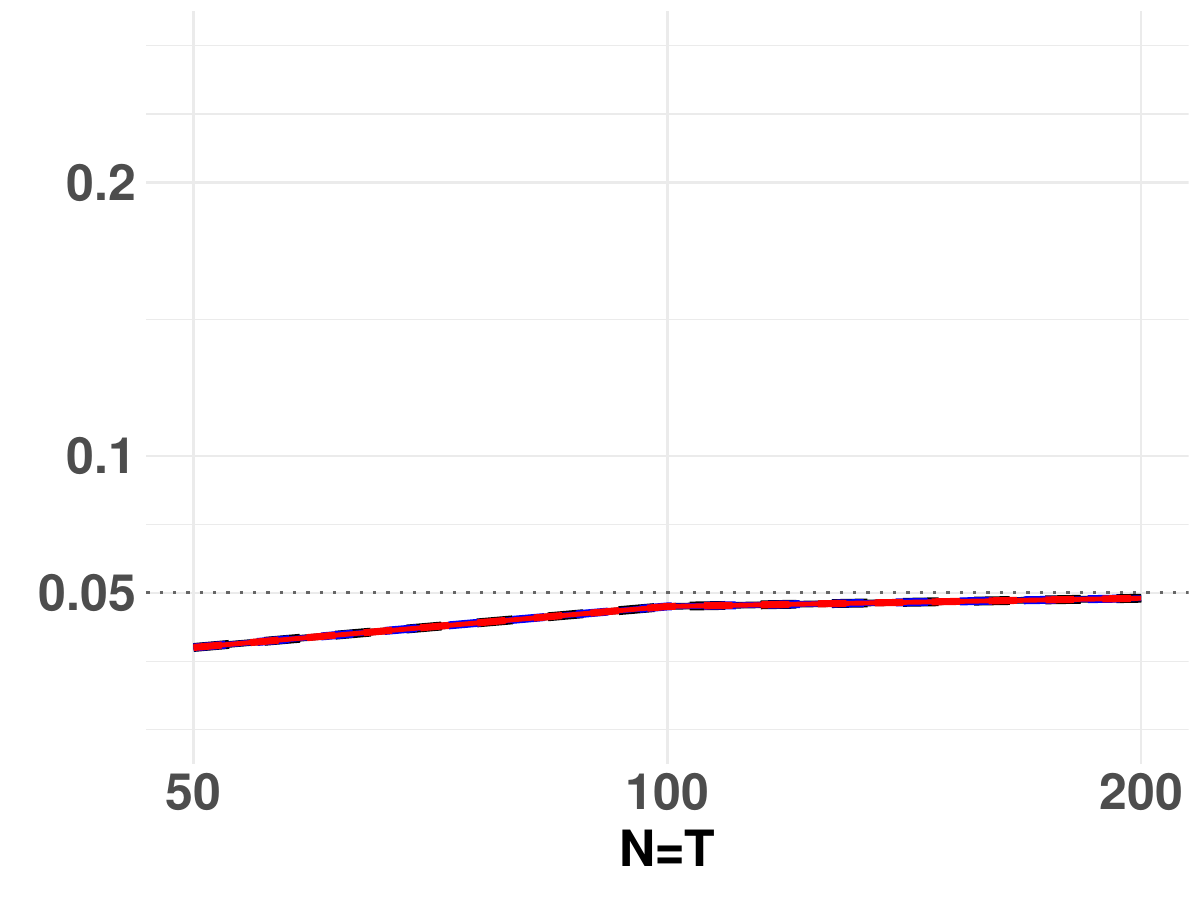}
\caption{$\rho_{wf}=0.6,\alpha_2 = 1.0$}
\label{fig:test_ww_06_10}
\end{subfigure}
\hfill
\begin{subfigure}[b]{0.32\textwidth}
\centering
\includegraphics[width=\textwidth]{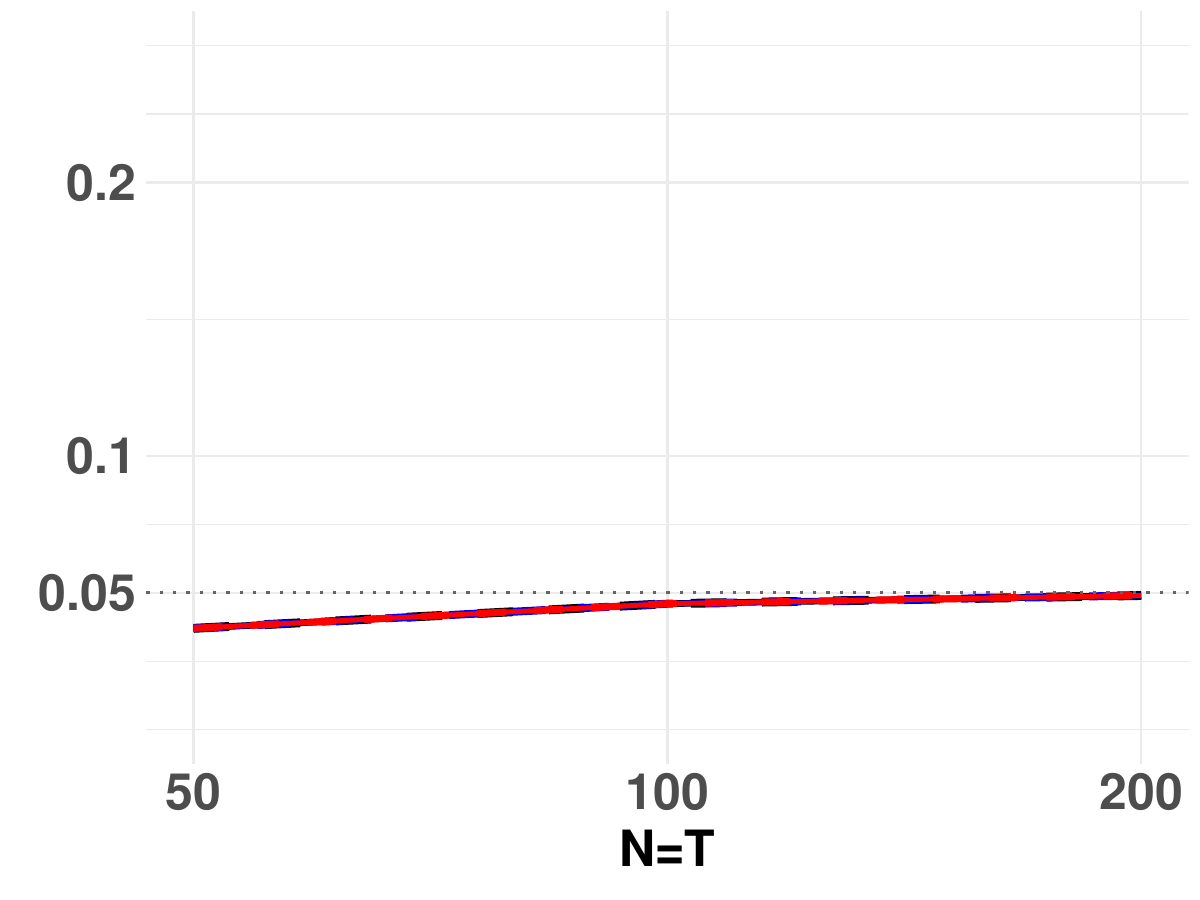}
\caption{$\rho_{wf}=0.6,\alpha_2 = 0.8$}
\label{fig:test_ww_06_08}
\end{subfigure}
\hfill
\begin{subfigure}[b]{0.32\textwidth}
\centering
\includegraphics[width=\textwidth]{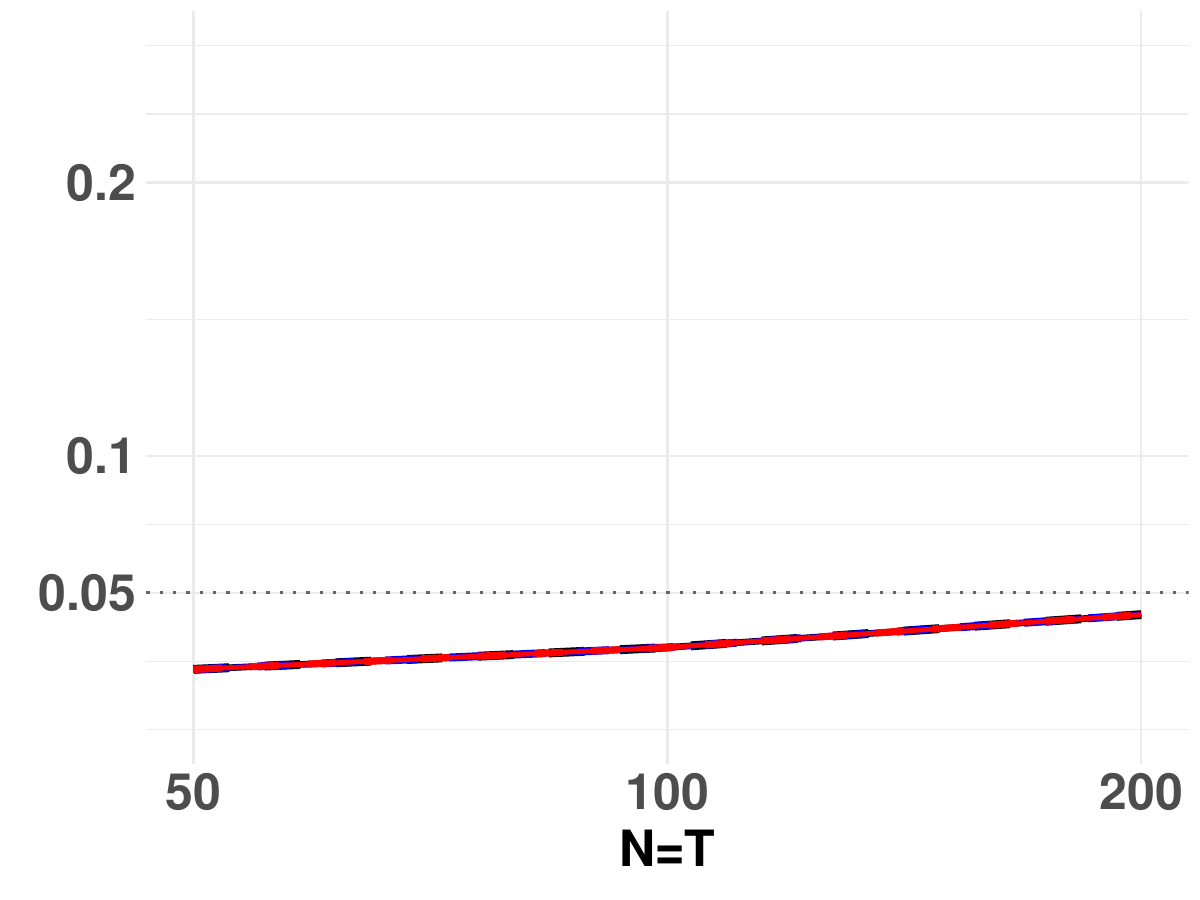}
\caption{$\rho_{wf}=0.6,\alpha_2 = 0.6$}
\label{fig:test_ww_06_06}
\end{subfigure}

\centering
\includegraphics[width=0.66\textwidth]{newimages/legend_SigE3_rhoe2_bias_rhowf0_a1_ww.pdf}
\caption{Size of the t-tests using $\hat{\beta}_w$ and its bias corrected versions for cross and serially correlated $e_{t,i}$}
\label{fig:test.ww}

\end{figure}

\end{document}